\documentclass[11pt]{amsart}
\usepackage{amsmath,amsthm,amssymb,mathtools}
\usepackage{mathabx} 
\usepackage{etex} 
\usepackage{stmaryrd}
\usepackage[normalem]{ulem}
\usepackage{longtable}
\usepackage[nosort]{cite}
\usepackage{graphicx}
\usepackage{caption}
\usepackage{subcaption}
\usepackage{pictexwd, dcpic}
\usepackage{float}
\usepackage{calligra} 
\usepackage{enumerate}
\usepackage[all]{xy} 
\input diagxy

\usepackage{url}
\usepackage[colorlinks=true]{hyperref}
\hypersetup{
	colorlinks=true,
    linktoc=all,
    linkcolor=blue,  
    citecolor=red,
    }

\newtheorem{thm}{Theorem}[section]
\newtheorem{cor}[thm]{Corollary}
\newtheorem{prop}[thm]{Proposition}
\newtheorem{lem}[thm]{Lemma}
\newtheorem{defn}[thm]{Definition}

\newtheorem{ex}[thm]{Example}
\newtheorem{rmk}[thm]{Remark}

\let\a=\alpha \let\b=\beta \let\g=\gamma \let\de=\delta 
\let\e=\epsilon
\let\z=\zeta \let\h=\eta \let\q=\theta  
\let\l=\lambda 
\let\s=\sigma \let\t=\tau    
\let\w=\omega      \let\G=\Gamma \let\D=\Delta  
\let\L=\Lambda
\let\X=\Xi  \let\S=\Sigma  
 
\let\C=\Chi \let\W=\Omega

\newcommand{\be}{\begin{equation}}
\newcommand{\ee}{\end{equation}}

\let\p=\partial \let\ov=\overline
\let\un=\underline
\newcommand{\<}{\langle}
\renewcommand{\>}{\rangle}

\newcommand{\id}{\mathrm{id}}

\newcommand{\aand}{\qquad \& \qquad}

\def\R{{{\mathbb R}}}
\def\C{{{\mathbb C}}}

\def\Z{{{\mathbb Z}}}

\def\triv{{{\mathrm{triv}}}}
\def\tra{{{\mathrm{tra}}}}
\def\Gr{{{\mathrm{Gr}}}}
\def\GTor{{{G\text{-}\mathrm{Tor}}}}

\DeclareMathAlphabet{\mathcalligra}{T1}{calligra}{m}{n}
\DeclareFontShape{T1}{calligra}{m}{n}{<->s*[2.2]callig15}{}
\newcommand{\scripty}[1]{\ensuremath{\mathcalligra{#1}}}

\title{Gauge invariant surface holonomy and monopoles}
\author[A.~Parzygnat]{Arthur J. Parzygnat}
\address{Physics Department, City College 
of the CUNY \\ 
New York, NY 10031, USA 
and
Physics Department, 
The Graduate Center of the CUNY \\
New York, NY 10016, USA\\
{{\em Email:} {\tt\href{mailto:aparzygnat@gradcenter.cuny.edu}
{aparzygnat@gradcenter.cuny.edu}}}}

\begin{document}
\maketitle

\begin{abstract}
There are few known computable examples of non-abelian surface 
holonomy. In this paper, we give 
several examples whose structure 2-groups are covering 2-groups 
and show that the surface 
holonomies can be computed via a simple formula in terms of paths 
of 1-dimensional holonomies 
inspired by earlier work of Chan Hong-Mo and Tsou Sheung Tsun on 
magnetic monopoles. 
As a consequence of our work and that of Schreiber and 
Waldorf, this formula 
gives a rigorous meaning to non-abelian magnetic flux for magnetic 
monopoles. In the process, we 
discuss gauge covariance of surface holonomies for spheres for any 
2-group, therefore generalizing 
the notion of the reduced group introduced by Schreiber and 
Waldorf. Using these 
ideas, we also prove that magnetic monopoles form an abelian 
group.
\end{abstract}

\tableofcontents

\section{Introduction}

\subsection{Background, motivation, and overview}

Ordinary holonomy along paths for principal group bundles has been 
studied for over 40 years in 
the context of gauge theories in physics and in the context of 
fiber bundles in mathematics. 
Recently, with ideas from higher category theory, it has been 
possible to extend these ideas to 
holonomy along surfaces. Although higher holonomy, and more 
generally higher gauge theory, has 
been studied in the context of \emph{abelian} gauge theory for
higher-dimensional manifolds, it was 
thought for some time that non-abelian generalizations were not
possible \cite{Te}. Today, we 
understand this as being due to the fact that a group object in
the category of groups is an abelian 
group. By ``categorifying'' well-known concepts, and considering
group objects in the category of 
categories, one can avoid this restriction. The language of higher
categories allows us to give a 
resolution to this problem. 

The data needed for defining surface holonomy for abelian structure
groups has been known for 
quite some time under the name
\emph{abelian gerbes with connection} with a formal presentation 
offered by Gawedzki \cite{Ga} in 1988 in the context of the WZW
model, with further work in 2002 
with Reis \cite{GaRe}. Further development under the name of
\emph{non-abelian gerbes}, 
\emph{higher bundles}, and so on were carried out in the following
years starting with the 
foundational work of Breen and Messing \cite{BM} in 2001, where the
data for connections on non-abelian 
gerbes first appeared. In \cite{BS}, Baez and Schreiber
gave a definition of non-abelian 
gerbes with connection in terms of parallel transport using the
notion of a 2-group. The most up-to-date 
theoretical framework in terms of category theory, which
provides a language easily adaptable 
for non-abelian generalizations, was established by Schreiber and
Waldorf in \cite{SW4}. In this 
categorical setting, higher principal bundles with connections are
described by transport functors. 

The motivation for transport functors comes from observations
originally made by Barrett in \cite{Ba} 
and expanded on by Caetano and Picken in \cite{CP} by describing
a bundle with connection in 
terms of its holonomies. In \cite{SW1}, Schreiber and Waldorf use
a categorical perspective to prove 
that a principal group bundle with connection over a smooth
manifold determines, and is determined 
by, a transport functor defined on the thin path groupoid of that
manifold with values in a fattened 
version of the structure group viewed as a one-object category.
The upshot of this equivalence is 
that it is conceptually simple to go from categories and functors
to 2-categories and 2-functors. In 
\cite{SW2}, \cite{SW3}, and \cite{SW4}, Schreiber and Waldorf take
advantage of this equivalence 
and abstract the definition so that it can be used to \emph{define}
principal 2-group 2-bundles with 
connection allowing a conceptually simple formulation of surface
holonomy. 

In the present article, we review the theory of transport functors
formalized by Schreiber and 
Waldorf in \cite{SW1}, \cite{SW2}, \cite{SW3}, and \cite{SW4} with
an emphasis on examples and 
explicit computations. Besides this, we accomplish several new
results. First, we provide a definition 
of holonomy along spheres modulo thin homotopy \emph{without}
representing a sphere as a bigon 
(Definition \ref{defn:gaugeinv2hol}). The target of this holonomy
is an analogue of conjugacy 
classes, which is used for ordinary holonomy along loops, called
$\a$-conjugacy classes. To prove 
this, we introduce a procedure that turns an arbitrary transport
functor into a group-valued transport 
functor. In \cite{SW4}, the authors forced their surface holonomy
to land in a rather restrictive 
quotient of the structure 2-group to prove gauge invariance of
holonomy. Our perspective is to take 
the smallest quotient possible, and we show our quotient surjects
onto the one of \cite{SW4}. 

We then focus on transport functors with a particular class of
2-groups, termed \emph{covering 2-groups}, 
given by a Lie group $G$ and a covering space of $G.$   
We provide a \emph{simple} formula, motivated by constructions in
\cite{HS}, for holonomy along 
surfaces in a local trivialization and show that this formula
agrees with the surface-ordered integral 
in \cite{SW2}. This gives an interesting relationship between (i)
well-known formulas in the physics 
literature for computing the magnetic flux in terms of a loop of
holonomies and (ii) non-abelian 
surface-ordered integrals in terms of 1- and 2-forms of \cite{SW2}.
Physically, we argue that the 
latter is the correct analogue to computing the magnetic flux as a
surface integral and our formula 
tells us that this agrees with the usual definition given in the
physics literature. This is all done 
without the introduction of a Higgs field, completing the ideas
in \cite{GNO}. 

Then we consider an entire collection of examples of transport
2-functors constructed from an 
ordinary principal $G$-bundle with connection along with a choice
of a subgroup $N$ of 
$\pi_{1} (G),$ the fundamental group of $G$ (such a choice of
subgroup determines a covering 2-group). 
We show that when the subgroup $N$ is chosen to be $\pi_{1} (G)$ 
itself, our example reduces to 
the curvature 2-functor defined by Schreiber and Waldorf in
\cite{SW4}. We instead focus on the 
other extreme, namely when the subgroup $N$ is chosen to be the
trivial group $\{1\},$ to calculate 
four examples of surface holonomies associated to both abelian and
non-abelian magnetic 
monopoles. But just as ordinary holonomy is not exactly
group-valued on the space of all loops (due 
to conjugation issues), surface holonomy isn't in general either.
Using our results on gauge 
invariance of sphere holonomy for arbitrary 2-groups, we prove that
the surface holonomies for 
magnetic monopoles are not only gauge invariant but also form an
abelian group.

\subsection{Outline of paper along with main results}
In Section \ref{sec:review1}, we review the main definitions of
transport functors along with an 
equivalence between local descent data and global transport
functors. We follow the recent work of 
Schreiber and Waldorf \cite{SW1} who describe it precisely and
categorically in a framework that is 
suitable for generalizations to surfaces. We briefly discuss the
relationship to principal $G$-bundles 
with connection, where $G$ is a Lie group, in their usual
formulation by introducing the category of 
$G$-torsors (manifolds with free and transitive right $G$-actions).
The equivalence between the two 
descriptions was proved in \cite{SW1}. We also review the
relationship between local descent data 
and differential cocycle data for principal group bundles,
recalling the well-known formula for parallel 
transport in terms of a path-ordered integral. To obtain
group-valued holonomies, we introduce a 
procedure (\ref{eq:groupvaluedholonomyfunctor}) described as a
functor that takes an arbitrary 
transport functor and produces a group-valued transport functor
in Section \ref{sec:1-holonomy}. 
The presentation differs a bit from that of \cite{SW1} so we
describe it in some detail. 

In Section \ref{sec:review2}, we review how to `categorify' the
definitions and statements of Section 
\ref{sec:review1} in order to define transport 2-functors. The
main references for this section include 
\cite{SW2}, \cite{SW3}, and \cite{SW4}. We only briefly review
the technical points but spend more 
time on a computational understanding of surface holonomy and
also supply an iterated integral 
expression for surface holonomy including a picture (Figure
\ref{fig:doublepathorderedintegral}) that 
we think will be useful for lattice gauge theory. As in the case
of holonomy along loops, we introduce 
a procedure (\ref{eq:groupvalued2holonomyfunctor}) to obtain
group-valued surface holonomy. This 
lets us discuss gauge covariance and gauge invariance simply
and in full detail without referring to 
the \emph{reduced group} of \cite{SW4}. 
However, we restrict ourselves
to holonomy along spheres as 
opposed to surfaces of arbitrary genus. We show, in Theorem
\ref{thm:invariancesurfaceholonomy}, 
that our holonomy along spheres lands in a set that surjects onto
the reduced group and give a 
simple example, in Lemma \ref{lem:alphaconjtoreducedgroup}, where
this surjection has nontrivial 
kernel. 

In Section \ref{sec:path-curvature}, we consider transport
2-functors with structure 2-group given by 
a covering 2-group. We give a new and simple formula valid for
\emph{all} such transport 2-functors 
in Corollary \ref{cor:main} for surface holonomy in a local
trivialization in terms of homotopy classes 
of paths of holonomies along loops. This construction was inspired
by work of physicists for 
computing magnetic charge as a topological number \cite{HS}. In
Definition 
\ref{defn:pathcurvature2functor}, we give our main construction
of a transport 2-functor, called the 
\emph{path-curvature 2-functor}, associated to every principal
$G$-bundle with connection and to 
any subgroup of $\pi_{1} (G).$ We prove that this assignment is
functorial. Furthermore, the path-curvature 
2-functor is shown to reduce to the example of Schreiber
and Waldorf known as the 
\emph{curvature 2-functor} in \cite{SW4} when the subgroup of
$\pi_{1} (G)$ is chosen to be 
$\pi_{1} (G)$ itself. We describe this construction on four levels:
(i) global transport functors (ii) 
functors with smooth trivialization data chosen (iii) descent data
(iv) differential cocycle data. This 
allows one to work with either construction at whatever level he
or she pleases. We then summarize 
our result as a list of commutative diagrams of functors in
(\ref{eq:commutingpathcurvature1}), 
(\ref{eq:thediagram}), and (\ref{eq:thediagramP}).  

In Section \ref{sec:examples}, we consider special cases of
covering 2-groups and give several 
examples all of which are known as magnetic monopoles \cite{HS}.
The first example is obtained 
from any principal $U(1)$-bundle with connection over the
two-sphere $S^2.$ It is shown that the 
surface holonomy along this sphere coming from the path-curvature
2-functor defined in Section 
\ref{sec:path-curvature} is precisely the integral of the curvature
form of the principal $U(1)$-bundle 
along this sphere, which in this case is the integral of the first
Chern class over the sphere. This 
example is precisely the Dirac monopole \cite{Di}  and the surface
holonomy gives the magnetic 
charge as the integral of a magnetic flux. We then discuss
non-abelian examples starting with a 
principal $SO(3)$-bundle with connection over the sphere and
compute the surface holonomy 
explicitly using both our simple formula and the formula in terms
of path-ordered integrals using 
differential forms. In the case of a non-trivial bundle, the
surface holonomy along the sphere is given 
by the element
$\left(\begin{smallmatrix} -1&0\\0&-1\end{smallmatrix}\right)$ in
$SU(2),$ the 
universal cover of $SO(3),$ which is the nontrivial element in the
kernel of the covering map 
$\t : SU(2) \to SO(3).$ We do this same computation in other
examples including 
$SU(n) \to SU(n) / Z(n),$ where $Z(n)$ is the center of $SU(n),$
and also for the case 
$SU(n) \times \R \to U(n).$ This gives a rigorous meaning to the
notion of 
\emph{non-abelian magnetic flux} as a surface holonomy along a
sphere (see Definition 
\ref{defn:magflux}). Furthermore, it is shown that magnetic flux
is a gauge-invariant quantity in 
Corollary \ref{cor:gaugeinvariantflux}.

Finally, the Appendix includes an overview of diffeological
spaces which are used to describe 
several of the constructions involving infinite-dimensional
manifolds and smooth maps between 
them. 

In short, this article contains the following results.
\begin{itemize}
\item
Theorem \ref{thm:gaugeinvarianthol} allows one to define
gauge-invariant holonomy along loops in 
the language of transport functors via Definition
\ref{defn:gaugeinvhol}. The image lands in 
conjugacy classes instead of the abelianization. 

\item
Theorem \ref{thm:invariancesurfaceholonomy} accomplishes the
analogous result for surface 
holonomy \emph{along spheres} in Definition
\ref{defn:gaugeinv2hol}. The image lands in 
$\a$-conjugacy classes (Definition \ref{defn:alphaconj}) 
instead of the
reduced group of \cite{SW4}. The 
set of $\a$-conjugacy classes surjects to the reduced group but is
not in general injective as shown 
in Lemma \ref{lem:alphaconjtoreducedgroup}. We also prove that the
fixed points of this $\a$ action 
form a central subgroup of the group of surface holonomies in
Lemma \ref{lem:fixedpoints}. 

\item
The rest of the paper focuses on transport
2-functors whose structure 2-groups 
are covering 2-groups (Definition \ref{defn:covering2groups}).
They are called \emph{path-curvature 
2-functors} (Definition
\ref{defn:pathcurvature2functor}). These transport 2-functors are 
defined \emph{without} using surface integrals, and we show, in
Theorem \ref{thm:main} and 
Corollary \ref{cor:main}, that locally, \emph{any} transport
2-functor (defined as in \cite{SW2} using 
surface integrals) with structure 2-group a covering 2-group,
coincides with ours, thus enabling a 
simple formula for calculating surface holonomy. 

\item
Section \ref{sec:examples} includes several examples and explicit
computations of surface 
holonomy. Due to the previously mentioned theorem, these examples
can rightfully be called 
magnetic fluxes of magnetic monopoles from physics. We include
several examples of non-abelian 
surface holonomy. We conclude with Corollary
\ref{cor:gaugeinvariantflux} that shows that the 
magnetic flux is a fixed point under the $\a$ action and therefore
lands in the central subgroup 
mentioned earlier. In particular, this implies that the magnetic
charge is an abelian group-valued 
quantity known as a \emph{topological number}.

\end{itemize}

\subsection{Acknowledgments}
Firstly, we thank Scott O. Wilson who helped greatly during the
entire process of this work, 
providing ideas and proofreading drafts. Secondly, we thank V.
Parameswaran Nair who made 
suggestions related to this work and informed us of references
including \cite{GNO}. We have  
benefited from conversations with Gregory Ginot, Jouko Mickelsson,
Urs Schreiber, Stefan Stolz, 
Rafal Suszek, Steven Vayl, Konrad Waldorf, and Christoph Wockel. 
We are also grateful to the referee of TAC for making several
useful suggestions and corrections to 
our first draft. 
We thank Aaron Lauda for his tutorial on xypic, which we relied on
to make many diagrams in this 
paper. All other figures were done in Gimp. 
This material is based on work supported by the National Science
Foundation Graduate Research 
Fellowship under Grant No. 40017-06 05 and 40017-06 06.

\subsection{Notations and conventions}

We assume the reader is familiar with some basic concepts of
2-categories (the Appendix of 
\cite{SW3} explains most details needed for this paper) but our
notation differs from the norm so we 
set it now. 

Compositions of 1-morphisms is usually written from right to left
as in 
\be
\xymatrix{
z && y 
	\ar[ll]_{\a} 	&& x 
	\ar[ll]_{\b} 
	}
\qquad
\mapsto
\qquad
\xymatrix{
z  && x 
		\ar[ll]_{\a \circ \b} 
	}
	.
\ee
Vertical composition is written from top to bottom as 
\be
\xymatrix{
y  && x
\ar@/_2.5pc/[ll]_{\b}="1"
\ar[ll]|-{\de}="2"
\ar@/^2.5pc/[ll]^{\z}="3"
\ar@{}"1";"2"|(.2){\,}="4"
\ar@{}"1";"2"|(.8){\,}="5"
\ar@{=>}"4";"5"^{\S}
\ar@{}"2";"3"|(.2){\,}="6"
\ar@{}"2";"3"|(.8){\,}="7"
\ar@{=>}"6";"7"^{\W}
}
\qquad
\mapsto
\qquad
\xymatrix{
y  && x
\ar@/_2pc/[ll]_{\b}="1"
\ar@/^2pc/[ll]^{\z}="2"
\ar@{}"1";"2"|(.2){\,}="3"
\ar@{}"1";"2"|(.8){\,}="4"
\ar@{=>}"3";"4"^{\begin{smallmatrix}
\S \\
\overset{\circ}{\W} 
\end{smallmatrix}}
}.
\ee
Horizontal composition is written as 
\be
\xymatrix{
z &&
\ar@/_1.5pc/[ll]_{\a}="8"
\ar@/^1.5pc/[ll]^{\g}="9"
\ar@{}"8";"9"|(.2){\,}="10"
\ar@{}"8";"9"|(.8){\,}="11"
\ar@{=>}"10";"11"^{\S}
y  && x
\ar@/_1.5pc/[ll]_{\b}="1"
\ar@/^1.5pc/[ll]^{\de}="2"
\ar@{}"1";"2"|(.2){\,}="3"
\ar@{}"1";"2"|(.8){\,}="4"
\ar@{=>}"3";"4"^{\W}
}
\qquad
\mapsto
\qquad 
\xymatrix{
z  && x
\ar@/_1.5pc/[ll]_{\a \circ \b}="1"
\ar@/^1.5pc/[ll]^{\g \circ \de}="2"
\ar@{}"1";"2"|(.2){\,}="3"
\ar@{}"1";"2"|(.8){\,}="4"
\ar@{=>}"3";"4"^{\S \circ \W}
}.
\ee

Sources, targets, and identity-assigning functions are denoted by
$s,t,$ and $i,$ respectively. We 
will always write the identity $i(x)$ at an object $x$ as
$\id_{x},$ $\id_{\a}$ for the vertical identity at 
a 1-morphisms $\a,$ and $\id_{\id_x}$ for the horizontal identity
at an object $x.$ Given a 2-category 
$\mathcal{C},$ the set of objects is typically denoted
by $\mathcal{C}_{0},$ 1-morphisms 
by $\mathcal{C}_{1}$ and 2-morphisms by $\mathcal{C}_{2}.$ In
general, an overline such as 
$\ov f$ will denote weak inverses, vertical inverses, and reversing
paths/bigons. It will be clear from 
context which is which. The first form of 2-categories appeared
under the name \emph{bi-categories} and 
were introduced by B\'enabou \cite{Be}.

\section{Principal bundles with connection are transport functors}
\label{sec:review1}

In this section, we review the notion of transport functors mainly 
following \cite{SW1}. We split up 
the discussion into several parts. We first discuss a \v Cech 
description of principal $G$-bundles 
(without connection), where $G$ is a Lie group, in terms of 
smooth 
functors. Then we attempt a 
guess for describing principal $G$-bundles with connections in 
terms 
of smooth functors. This 
attempt fails as it only gives trivialized bundles, motivating the 
need 
to use transport functors. We 
then proceed to describing local trivialization data, descent 
data, 
and 
finally transport functors. The 
key feature of descent data is that it enables us to encode 
smoothness while still allowing the 
`bundle' to have nontrivial topology. We then discuss a 
reconstruction 
functor that takes us from the 
category of descent data to the category of transport functors 
with 
chosen trivializations. It is here 
that we discuss a version of the \v Cech groupoid incorporating 
paths 
and `jumps' that are 
necessary for transition functions. Then we move in the other 
direction and go from smooth descent 
data to locally defined differential forms, or more generally 
differential 
cocycle data. We also 
describe how to go from differential cocycle data back to smooth 
descent data. We then summarize 
the four different levels describing transport functors and their 
relationship to one another. Finally, 
we use these results to formulate a procedure that sends an 
arbitrary 
transport functor to a transport 
functor with group-valued parallel transport and discuss its gauge 
covariance and invariance 
stressing the use of conjugacy classes. 

\subsection{A \v Cech description of principal $G$-bundles}
\label{sec:1Cech}

Let $G$ be a Lie group. Principal $G$-bundles over a smooth 
manifold $M$ can be described 
simply in terms of functors. Furthermore, an isomorphism of 
such 
bundles corresponds to a natural 
transformation of the corresponding functors. This is done as follows 
(this is an expansion of 
Remark II.13. in \cite{Wo}). 

\begin{defn}
\label{defn:Cechgroupoid}
Given an open cover $\{ U_{i} \}_{i \in I}$ of $M,$ the 
\emph{\uline{\v Cech groupoid}} $\mathfrak{U}$ is 
the category whose set of objects is given by 
\be
\mathfrak{U}_{0} := \coprod_{i \in I} U_{i} 
\ee
and whose morphisms, called `jumps,' are given by 
\be
\mathfrak{U}_{1} := \coprod_{i, j \in I } U_{ij},
\ee
where $U_{ij} := U_{i} \cap U_{j}$ and the order of the index is kept 
track of in the disjoint union. 
Explicitly, elements of $\mathfrak{U}_{0}$ are written as $(x,i)$ and 
elements of $\mathfrak{U}_{1}$ 
are written as $(x,i,j).$ 
The source and target maps are given by $s ( (x,i,j ) ) := (x,i)$ and 
$t ( (x,i,j) ) := (x,j)$ for 
$(x,i,j) \in \mathfrak{U}_{1}.$ The identity-assigning map is given by%
\footnote{Our apologies for this double usage of the letter $i$ to 
mean both the identity-inclusion 
map and the index letter. We hope that it is not too confusing. Later, 
we will also use the letter $i$ 
for several other purposes.}
$i( (x,i) ) := (x,i,i).$ Let $(x,i,j)$ and $(x',i',j')$ be two
morphisms with 
$t ( (x,i,j) ) = s ( (x',i',j') ),$ i.e. 
$(x,j) = (x',i').$ Renaming the index $j'$ to $k,$ the composition is 
defined to be 
\be
\label{eq:ijk}
(x,j,k) \circ (x,i,j) := (x,i,k). 
\ee
\end{defn}

\begin{defn}
\label{defn:BG}
For every Lie group $G,$ there is a one-object groupoid 
$\mathcal{B} G$ defined as follows. Denote 
the one object by $\bullet.$ Let the set of morphisms from $\bullet$ to
itself be given by the set $G.$ 
Composition is given by group multiplication. 
\end{defn}

The previous two groupoids have a smooth structure, formalized in 
the following definition. 

\begin{defn}
\label{defn:2-space}
A \emph{\uline{Lie groupoid}} is a (small) category, typically denoted by 
$\Gr,$ whose objects, 
morphisms, and sets of composable morphisms all form smooth 
manifolds. Furthermore, the 
source, target, identity-assigning, and composition maps are all 
smooth. In addition, every 
morphism has an inverse and the map that sends a morphism to its 
inverse is smooth. 
\end{defn}

\begin{ex}
The  \v Cech groupoid of Definition \ref{defn:Cechgroupoid} and 
$\mathcal{B} G$ of Definition 
\ref{defn:BG} 
are Lie groupoids with the appropriate (obvious) smooth structures.
\end{ex}

\begin{defn}
A \emph{\uline{smooth functor}} from one Lie groupoid to another is an 
ordinary functor that is smooth 
on objects and morphisms. Likewise, a 
\emph{\uline{smooth natural transformation}} is a natural 
transformation whose function from objects to morphisms is smooth. 
\end{defn}

Any smooth functor $\mathfrak{U} \to \mathcal{B} G$ gives the
\v Cech cocycle data of a principal 
$G$-bundle over $M$ subordinate to the cover $\{ U_{i} \}_{i \in I}.$ 
To see this, simply recall what a 
functor does. To each object $(x,i)$ in $\mathfrak{U},$ it assigns the 
single object $\bullet$ in 
$\mathcal{B} G.$ To each jump $(x,i,j),$ it assigns an element 
denoted by $g_{ij} (x) \in G$ in such 
a way that we get a smooth 1-cochain $g_{ij} : U_{ij} \to G$ 
\be
\xy 0;/r.25pc/:
(-10,3)*{j};
(10,3)*{i};
(-10,0)*\xycircle(15,15){.};
(10,0)*\xycircle(15,15){.};
(-10,0)*{\bullet}="j";
(10,0)*{\bullet}="i";
"i";"j"**\dir{-} ?(0.55)*\dir{>}+(1,3)*{ij};
\endxy
\qquad
\mapsto
\qquad  
\xy 0;/r.25pc/:
(-10,0)*{\bullet}="j";
(10,0)*{\bullet}="i";
"i";"j"**\dir{-} ?(0.55)*\dir{>}+(1,3)*{g_{ij}};
\endxy
\quad . 
\ee
This picture should be interpreted as follows. To each $x \in U_{ij},$ 
we draw the jump $(x,i,j)$ as 
the figure on the left. Its image under 
$\mathfrak{U} \to \mathcal{B} G$ is $g_{ij} (x)$ drawn on the 
right (without explicitly writing $x$).  To each triple intersection 
$U_{ijk},$ which corresponds to the 
composition of $(x,i,j)$ in $U_{ij}$ with $(x,j,k)$ in $U_{jk}$ as in 
(\ref{eq:ijk}), functoriality gives a 
cocycle condition
\be
\xy 0;/r.25pc/:
(-13,-10)*{k};
(0,12)*{j};
(13,-10)*{i};
(-10,-9)*\xycircle(15,15){.};
(10,-9)*\xycircle(15,15){.};
(0,8.32)*\xycircle(15,15){.};
(-10,-9)*{\bullet}="k";
(0,8.32)*{\bullet}="j";
(10,-9)*{\bullet}="i";
"i";"j"**\dir{-} ?(0.55)*\dir{>}+(3,0)*{ij};
"j";"k"**\dir{-} ?(0.55)*\dir{>}+(-3,2)*{jk};
"i";"k"**\dir{-} ?(0.55)*\dir{>}+(1,-3)*{ik};
\endxy
\qquad
\mapsto
\qquad  
\xy 0;/r.30pc/:
(-10,-9)*{\bullet}="k";
(0,8.32)*{\bullet}="j";
(10,-9)*{\bullet}="i";
"i";"j"**\dir{-} ?(0.55)*\dir{>}+(4,0)*{g_{ij}};
"j";"k"**\dir{-} ?(0.55)*\dir{>}+(-3,2)*{g_{jk}};
"i";"k"**\dir{-} ?(0.55)*\dir{>}+(1,-3)*{g_{ik}};
\endxy
\quad , 
\ee
which says 
\be
g_{jk} g_{ij} = g_{ik}.
\ee
This convention was chosen to match that of \cite{SW1} and 
\cite{SW4} so that the reader who is 
interested in further details can consult without too much trouble. 

We now discuss refinements and morphisms between two such 
functors. Let $\{ U_{i'} \}_{i' \in I'}$ 
be another cover of $M$ with associated \v Cech groupoid 
$\mathfrak{U}'.$  Let 
$P: \mathfrak{U} \to \mathcal{B} G$ and 
$P' : \mathfrak{U}' \to \mathcal{B} G$ be two smooth 
functors. A morphism from $P$ to $P'$ consists of a common 
refinement $\{ V_{\a} \}_{\a \in A},$ 
with associated \v Cech groupoid $\mathfrak{V},$ of both 
$\{ U_{i} \}_{i \in I}$ and 
$\{ U_{i'} \}_{i' \in I'}$ along with a smooth natural transformation 
\be
\xy 0;/r.20pc/:
(0,-15)*+{\mathfrak{U}'}="3";
(15,0)*+{\mathcal{B} G}="4";
(,15)*+{ \mathfrak{U}}="2";
(-15,0)*+{ \mathfrak{V}}="1";
{\ar^{\a} "1";"2" };
{\ar^{P} "2";"4" };
{\ar_{\a'} "1";"3"};
{\ar_{P} "3";"4"};
{\ar@{=>}^{h} (0,10);(0,-10)};
\endxy
.
\ee
The refinement condition means that there are associated functions 
$\a : A \to I$ and $\a' : A \to I'$ 
so that $V_{a} \subset U_{\a ( a) }$ and $V_{a} \subset U'_{\a (a) }$ 
for all $a \in A.$ These 
functions determine the functors $\a : \mathfrak{V} \to \mathfrak{U}$ 
and 
$\a' : \mathfrak{V} \to \mathfrak{U}'$ drawn above.  We denote the 
restrictions of $g_{\a (a) \a(b)}$ 
and $g'_{\a' (a) \a' (b)}$  to $V_{ab}$ by $g_{ab}$ and  $g'_{ab},$ 
respectively. Any such smooth 
natural transformation gives an equivalence of \v Cech cocycle data 
of principle $G$-bundles. To 
see this, simply recall what a natural transformation does.  To each 
object $(x,a)$ in $\mathfrak{V}$ 
it assigns a group element $h_{a} (x) \in G$ in a smooth way. In other 
words, it gives a smooth 
function $h_{a} : V_{\a} \to G.$ To each jump $(x,a,b)$ in 
$\mathfrak{V},$ the naturality condition 
\be
\xy 0;/r.15pc/:
(15,15)*+{\bullet}="1";
(-15,15)*+{\bullet}="2";
(15,-15)*+{\bullet}="3";
(-15,-15 )*+{\bullet}="4";
{\ar_{g_{a b}} "1";"2" };
{\ar_{h_{b}} "2";"4" };
{\ar^{h_{a}} "1";"3"};
{\ar^{g'_{a b}} "3";"4"};
(-10,10)*+{}="a";
(10,-10)*+{}="b";
{\ar@{=}"a";"b"};
\endxy
\ee
says that 
\be
h_{b} g_{ab} = g'_{ab} h_{a}
\ee
on $V_{ab}.$ 
This is precisely the condition that says the principal $G$-bundles 
$P$ and $P'$ are isomorphic 
\cite{St}.

\subsection{A naive guess for transport functors}

Our goal in this section is to guess what a connection on a principal 
$G$-bundle over $M$ should 
be in terms of functors. We will fail at this attempt, but will learn
an 
important lesson which will 
motivate the modern definition in terms of transport functors. First, 
recall that in a principal $G$-bundle 
$P \to M,$ every fiber is a right $G$-torsor.

\begin{defn}
\label{defn:GTor}
Let $G$ be a Lie group. Let $\GTor$ be the category whose objects 
are right 
\emph{\uline{$G$-torsors}}, i.e. smooth manifolds equipped with a free 
and transitive right $G$-action, 
and whose morphisms are right $G$-equivariant maps. 
\end{defn}

Furthermore, a connection on a principal $G$-bundle over $M$ gives 
an assignment from paths in 
$M$ to isomorphisms of fibers between the endpoints. This 
assignment is independent of the 
parametrization of the path, but it is even independent of the thin 
homotopy class of a path as 
discussed in \cite{CP}. To define this, we use the theory of smooth 
spaces, reviewed in Appendix 
\ref{sec:smooth}, which give natural definitions for smooth structures 
on subsets, mapping spaces, 
and quotient spaces.

\begin{defn}
\label{defn:pwsi}
Let $X$ be a smooth manifold. A \emph{\uline{path with sitting instants}} 
is a smooth map $\g : [0,1] \to X$ such that there exists 
an $\e$ with $\frac{1}{2} > \e > 0$ and $\g ( t )$ is constant for
all $t \in [0, \e] \cup [ 1 - \e, 1 ].$ For such paths $\g$ with
$\g (0) = x$ 
and $\g(1) = y,$ we write 
\be
\xy 0;/r.25pc/:
(-10,0)*+{y}="y";
(10,0)*+{x}="x";
{\ar"x";"y"_{\g}};
\endxy
.
\ee
The set of paths with sitting instants in $X$ will be denoted by $PX.$ 
\end{defn}

Paths with sitting instants were first described in \cite{CP}. We 
reserve the notation $X^{[0,1]}$ for the set of (ordinary) smooth 
paths in $X.$ Thus, $PX \subset X^{[0,1]}.$ 

\begin{defn}
\label{defn:thinhomotopy}
Two paths in $X$ with sitting instants $\g$ and $\g'$ with the same 
endpoints, i.e. $\g (0) = \g' (0) = x$ and $\g (1) = \g'(1) = y,$ are
said 
to be \emph{\uline{thinly homotopic}} if there exists a smooth map 
$\G : [ 0 , 1 ] \times [ 0, 1 ] \to X$ with the following two
properties. 
\begin{enumerate}[(a)]
\item
First, there exists an  $\e$ with $\frac{1}{2} > \e > 0$ such that 
\be
\label{eq:bigon}
\G ( t, s ) = 
	\begin{cases}
	x & \mbox{ for all } (t,s) \in [ 0 , \e ] \times [ 0 , 1 ] \\
	y & \mbox{ for all } (t,s) \in [ 1- \e, 1 ] \times [ 0 , 1 ] \\
	\g ( t) & \mbox{ for all } (t,s) \in [ 0 , 1 ] \times [ 0 , \e ]
\\
	\g' ( t) & \mbox{ for all } (t,s) \in [ 0 , 1 ] \times
[ 1 - \e, 1 ] \\
	\end{cases}  
\ee
A map $\G: [0,1] \times [0,1] \to X$ satisfying just (\ref{eq:bigon}) is
called a \emph{\uline{bigon}} in $X$ and is typically denoted by 
\be
\xymatrix{
y & & \ar@/_1.5pc/[ll]_{\g}="1" 
\ar@/^1.5pc/[ll]^{\g'}="2" \ar@{}"1";"2"|(.15){\,}="3" 
\ar@{}"1";"2"|(.85){\,}="4" \ar@{=>}"3";"4"^{\G}  x
}
\ee
The set of bigons in $X$ is denoted by $BX.$ 
\item
Second, the rank of $\G$ is strictly less than 2, i.e. the differential 
$D_{(t,s)} \G : T_{(t,s)} ([0,1]\times [0,1]) \to T_{\G(t,s)} X,$ where 
$T_{y} Y$ denotes the tangent space to $Y$ at the point $y \in Y,$ 
has kernel of dimension at least one for all
$(t,s) \in [0,1]\times [0,1].$  
\end{enumerate}
Thin homotopy is an equivalence relation and the equivalence 
classes are called \emph{\uline{thin paths}}.  
Denote the set of thin paths in $X$ by $P^{1} X.$
\end{defn}

$P^{1} X$ is naturally a smooth space since it is a quotient of $PX,$ 
which is itself a subset of $X^{[0,1]},$ which has a natural smooth 
space structure as a mapping space. With these preliminaries, the 
definition of the thin path-groupoid of a smooth manifold $X$ can be 
given (we refer the reader to \cite{CP} and \cite{SW1} for more 
details).  

\begin{defn}
\label{defn:thinpathgroupoid}
Let $X$ be a smooth manifold. Let $\mathcal{P}_{1} (X)$ be the 
category whose objects are the points of the smooth manifold $X$ 
and whose morphisms are the thin paths of $X.$ The source and 
target of 
a thin path are defined by choosing a representative and taking the 
source and target, respectively. The identity at each point $x \in X$ is
the thin path associated to the constant path at $x.$ 
The composition of two thin paths is defined by choosing 
representatives and concatenating with double-speed 
parametrization. Namely, given two thin paths
\be
\xy 0;/r.25pc/:
(-20,0)*+{z}="z";
(0,0)*+{y}="y";
(20,0)*+{x}="x";
{\ar"x";"y"_{\g}};
{\ar"y";"z"_{\g'}};
\endxy
,
\ee
the composition is given by the thin homotopy class associated to 
\be
(\g' \circ \g) (t) := 
	\begin{cases}
	\g(2t) & \mbox{ for } 0 \le t \le \frac{1}{2} \\
	\g' (2t-1) & \mbox{ for } \frac{1}{2} \le t \le 1
	\end{cases} 
	.
\ee
\end{defn}

Under the sitting instants assumption and the thin homotopy 
equivalence relation, the composition is well-defined, smooth, 
associative, has left and right units given by constant paths, and right
and left inverses by reversing paths. By replacing the word ``smooth 
manifold'' with ``smooth space'' in Definition \ref{defn:2-space}, 
$\mathcal{P}_{1} (X)$ is therefore a Lie groupoid. 

With this definition of the thin path-groupoid of $M,$ one might guess 
that a transport functor should be a smooth functor 
$\mathcal{P}_{1} (M) \to \GTor.$ However, since $\GTor$ is not a Lie 
groupoid, there is no obvious way of demanding such a functor to be 
smooth. One might therefore try to use $\mathcal{B} G$ instead of 
$\GTor.$ Indeed, notice that there is a natural functor 
$i : \mathcal{B} G \to \GTor$ defined by 
\be
\label{eq:Lg}
\begin{split}
\bullet &\mapsto G \\
g &\mapsto L_{g} , 
\end{split}
\ee
where $G$ is viewed as a right $G$-torsor and $L_{g}$ is left
multiplication on $G$ by $g.$ One can think of $\GTor$ as a 
`thickening' of $\mathcal{B} G$ because $i$ is an equivalence of 
categories. We can then try to use $\mathcal{B} G$ for our target 
instead of $\GTor$ so that we can ask for smoothness. Then one 
might guess that a transport functor should be a smooth functor 
$\mathcal{P}_{1} (M) \to \mathcal{B} G.$ Unfortunately, now that we 
have smoothness, we've lost non-triviality because such smooth 
functors describe parallel transport on \emph{trivialized} principal 
$G$-bundles (this fact follows from Section \ref{sec:diffcocycledata} 
particularly around equation (\ref{eq:integratedifferentiate})). 

In order to encode local instead of global triviality, we have to 
combine these ideas with those of the previous section in terms 
of the \v Cech groupoid (we will also return to a more suitable 
combination of the path groupoid and the \v Cech groupoid in 
Section \ref{sec:reconstruction}). To avoid a huge collection of 
indices again, we denote our open cover $\{ U_{i} \}_{i \in I}$ of $M$ 
simply by $Y := \coprod_{i \in I } U_{i}$ and we let $\pi : Y \to M$ be 
the inclusion of these open sets into $M.$ Note that $\pi$ is a 
surjective submersion. Then, the next guess might be that we need 
to have a smooth functor $\mathcal{P}_{1} (Y) \to \mathcal{B} G,$ 
but we still need an assignment of fibers 
$\mathcal{P}_{1} (M) \to \GTor.$ These assignments should be 
compatible in terms of the functor $i : \mathcal{B} G \to \GTor$ and 
the submersion $\pi.$ This is exactly what is done in \cite{SW1} and 
we therefore now proceed to discussing local triviality of functors.

\subsection{Local triviality of functors}

Our first goal is to discuss local triviality of functors without making
any assumptions on smoothness, which is left to the next section. 
The fibers of principal $G$-bundles were right $G$-torsors, which led 
us to consider the category $\GTor$ of $G$-torsors. One of the great 
features of Schreiber's and Waldorf's work \cite{SW1} is their 
generality on the different flavors of bundles. If one wants to work 
with vector bundles one simply replaces $\GTor$ with 
$\mathrm{Vect},$ the category of vector spaces (over some 
appropriate field such as $\R$ or $\C$), and if this vector bundle is 
an associated bundle for some representation of $G,$ then this 
representation is precisely encoded by a functor 
$i : \mathcal{B} G \to \mathrm{Vect}.$ Fiber bundles can be defined 
similarly. Therefore, we've made two important observations. The 
first is that fibers of a bundle are objects of some category $T.$ The 
second is that the structure group of the bundle is encoded by a 
functor $i : \mathcal{B} G \to T.$ Schreiber and Waldorf generalize 
this even further by considering any Lie groupoid $\Gr$ instead of the 
special one $\mathcal{B} G.$ They define a $\pi$-local trivialization 
as follows (Definition 2.5. of \cite{SW1}).

\begin{defn}
\label{defn:trivialization1}
Let $\Gr$ be a Lie groupoid, $T$  a category,  $i : \mathrm{Gr} \to T$ 
a functor, and $M$ a smooth manifold. Fix a surjective submersion 
$\pi : Y \to M.$ A \emph{\uline{$\pi$-local $i$-trivialization}} of a
functor 
$F : \mathcal{P}_{1} (M) \to T$ is a pair $(\triv, t)$ of a functor 
$\triv : \mathcal{P}_{1} (Y) \to \Gr$ and a natural isomorphism 
$t : \pi^* F \Rightarrow \triv_{i}$ as in the diagram 
\be
\xy 0;/r.15pc/:
(15,-15)*+{\Gr}="3";
(-15,-15 )*+{T}="4";
(-15,15)*+{\mathcal{P}_{1} (M)}="2";
(15,15)*+{\mathcal{P}_{1} (Y)}="1";
{\ar_(0.45){\pi_{*}} "1";"2" };
{\ar_{F} "2";"4" };
{\ar^{\triv} "1";"3"};
{\ar^{i} "3";"4"};
{\ar@{=>}^{t} "2";"3"};
\endxy
.
\ee
The groupoid $\Gr$ is called the 
\emph{\uline{structure groupoid}} for $F.$ 
\end{defn}

In this definition $\pi_{*}$ is the pushforward defined sending 
points $y \in Y$ to $\pi(y)$ and sending thin paths $\g \in P^1 Y$ to 
the thin homotopy class of $\pi \circ \g$ (by choosing a 
representative). $\pi^*F := F \circ \pi_{*}$ is the pullback 
of $F$ along $\pi$ and 
$\triv_{i} := i \circ \triv.$ 
Functors $F : \mathcal{P}_{1} (M) \to T$ equipped 
with $\pi$-local $i$-trivializations $(\triv, t)$ form the objects,
written 
as triples $(F, \triv, t),$ of a category denoted by 
$\mathrm{Triv}^{1}_{\pi} (i).$ 

\begin{defn}
\label{defn:trivialization1mor}
A \emph{\uline{morphism}} $\a : ( F, \triv, t ) \to ( F' , \triv', t')$ in 
$\mathrm{Triv}^{1}_{\pi} (i)$ \emph{\uline{of $\pi$-local 
$i$-trivializations}} is a natural transformation
$\a : F \Rightarrow F'.$ 
Composition is given by vertical composition of natural 
transformations. 
\end{defn}

\begin{rmk}
One might expect a morphism $( F, \triv, t ) \to ( F' , \triv', t')$ to 
consist of $\a : F \Rightarrow F'$ as well as a natural transformation 
$h : \triv \Rightarrow \triv'$ satisfying some compatibility condition 
with $\a,$ $t,$ and $t'.$ This natural compatibility condition 
completely determines $h$ which is why it is excluded in the 
definition. 
\end{rmk}

In this description, it's not immediately obvious what transition 
functions are. This is part of the motivation for introducing descent 
objects (Definition 2.8. of \cite{SW1}). We 
use the notation $Y^{[n]}$ associated to a surjective submersion 
$\pi : Y \to M$ to mean the $n$-fold fiber product defined by 
\be
\label{eq:Yn}
Y^{[n]} := \left \{ (y_1, \dots, y_n) \in Y \times \cdots \times Y \ 
		| \ \pi(y_1)=\cdots=\pi(y_n) \right \}.
\ee
There are several projection maps 
$\pi_{i_{1} \cdots i_{k}} : Y^{[n]} \to Y^{[n-k]}$ for all $n \ge 2$ and
$0 < k < n$ with $1< i_{1} < \cdots < i_{k} < n$ that are defined by 
\be
Y^{[n]} \ni (y_1, \dots, y_n) \mapsto (y_{i_{1}}, \dots, y_{i_{k}}).
\ee
$Y^{[n]}$ is a smooth manifold for all $n$ and all 
$\pi_{i_{1} \cdots i_{k}}$ are smooth since $\pi$ is a surjective 
submersion. 

\begin{defn}
\label{defn:descent1}
Let $\Gr$ be a Lie groupoid, $T$ a category, and $i : \Gr \to T$ a 
functor.  Fix a surjective submersion $\pi : Y \to M.$ A 
\emph{\uline{descent object}} is a pair $(\triv, g)$ consisting of a
functor 
$\triv : \mathcal{P}_{1} (Y) \to \Gr,$ a natural isomorphism 
\be
\xy 0;/r.15pc/:
(17.5,-15)*+{\mathcal{P}_{1} (Y)}="3";
(-17.5,-15 )*+{T}="4";
(-17.5,15)*+{\mathcal{P}_{1} (Y)}="2";
(17.5,15)*+{\mathcal{P}_{1} (Y^{[2]})}="1";
{\ar_{\pi_{1*}} "1";"2" };
{\ar_{\triv_{i}} "2";"4" };
{\ar^{\pi_{2*}} "1";"3"};
{\ar^{\triv_{i}} "3";"4"};
{\ar@{=>}^{g} "2";"3"};
\endxy
.
\ee
The pair $(\triv, g)$ must satisfy
\be
\begin{smallmatrix}  
	\pi^{*}_{12} g \\ 
	\overset{\circ}{\pi^{*}_{23}} g 
	\end{smallmatrix} 
	=  \pi^{*}_{13} g , 
\ee
where the left-hand-side is vertical composition of natural 
transformations (read from top to bottom), 
and
\be
\id_{\triv_{i}} = \D^* g, 
\ee
where $\D$ is the diagonal $\D : Y \to Y^{[2]}$ sending $y$ to $(y,y).$
\end{defn}

Descent objects form the objects of a category denoted by 
$\mathfrak{Des}^{1}_{\pi} (i).$ 

\begin{defn}
\label{defn:descent1mor}
A \emph{\uline{descent morphism}} $h : (\triv, g) \to (\triv', g')$ is a 
natural transformation $h : \triv_{i} \Rightarrow \triv'_{i}$ satisfying
\be
\begin{matrix}
\pi_{1}^{*} h \\
\overset{\circ}{g'}
\end{matrix}
=
\begin{matrix}
g \\
\overset{\circ}{\pi_{2}^{*} h}
\end{matrix}
.
\ee
\end{defn}

There is a functor 
$\mathrm{Ex}^{1}_{\pi} : \mathrm{Triv}^{1}_{\pi} (i) 
\to \mathfrak{Des}^{1}_{\pi} (i)$ 
that extracts descent data from trivialization data. At the level of 
objects, this functor is defined as follows. Let $(F, \triv, t)$ be an 
object in $\mathrm{Triv}^{1}_{\pi} (i).$ For the pair $(\triv, g),$ take
$\triv$ to be exactly the same. For $g$ take the composition 
$g := \begin{smallmatrix} \pi_{1}^{*} \ov t  \\ 
\overset{\circ}{\pi_{2}^{*} t} \end{smallmatrix}$ 
coming from the composition in the diagram 
\be
\xy 0;/r.15pc/:
(17.5,0)*+{\mathcal{P}_{1} (Y)}="3";
(-17.5,0 )*+{\mathcal{P}_{1} (M)}="4";
(-17.5,30)*+{\mathcal{P}_{1} (Y)}="2";
(17.5,30)*+{\mathcal{P}_{1} (Y^{[2]})}="1";
(-12.5,25)*+{}="a";
(12.5,5)*+{}="b";
{\ar_{\pi_{1*}} "1";"2" };
{\ar_{\pi_{*}} "2";"4" };
{\ar^{\pi_{2*}} "1";"3"};
{\ar^{\pi_{*}} "3";"4"};
{\ar@{=}^{\id} "a";"b"};
(-42.5,10)*+{\Gr}="5";
(-42.5,-20)*+{T}="7";
(-7.5,-20)*+{\Gr}="6";
{\ar^{\triv} "3";"6"};
{\ar^{i} "6";"7"};
{\ar^{F} "4";"7"};
{\ar_{i} "5";"7"};
{\ar_{\triv} "2";"5"};
{\ar@{=>}^(0.4){\ov t} "5";"4"};
{\ar@{=>}^{t} "4";"6"};
\endxy
,
\ee
where $\ov t$ is the (vertical) inverse of $t.$  This defines a descent 
object (Section 2.2 of \cite{SW1}). On a morphism 
$\a : (F, \triv, t) \to (F', \triv', t'),$ the functor
$\mathrm{Ex}_{\pi}^{1}$ 
is defined by setting 
\be
h: = 
\begin{smallmatrix}
\ov t \\
\overset{\circ}{\pi^{*} \a} \\
\overset{\circ}{t'}
\end{smallmatrix}
\ee
coming from the composition in the diagram 
\be
\xy 0;/r.15pc/:
(35,0)*+{\mathcal{P}_{1} (Y) }="1";
(0,0)*+{\mathcal{P}_{1}(M)}="2";
(-40,0)*+{T}="3";
{\ar@/_2.5pc/_{\triv_{i}} "1";"3"};
{\ar_{\pi_{*}} "1";"2"};
{\ar@/_1pc/|-{F} "2";"3"};
{\ar@/^1pc/|-{F'} "2";"3"};
{\ar@/^2.5pc/^{\triv'_{i}} "1";"3"};
{\ar@{=>}^{\ov t} (-5,15);(-5,5)};
{\ar@{=>}^{t'} (-5,-5);(-5,-15)};
{\ar@{=>}^{\a} (-22,4);(-22,-4)};\
\endxy
.
\ee
The functor $\mathrm{Ex}_{\pi}^{1}$ is part of an equivalence of 
categories between $\mathrm{Triv}_{\pi}^{1} (i)$ and 
$\mathfrak{Des}_{\pi}^{1} (i).$ 
This is done by constructing a weak inverse functor 
$\mathrm{Rec}_{\pi}^{1} : \mathfrak{Des}_{\pi}^{1} (i) 
\to \mathrm{Triv}_{\pi}^{1} (i),$ which we will describe in Section 
\ref{sec:reconstruction}. 

\begin{defn}
Let $(F, \triv, t)$ be a $\pi$-local $i$-trivialization of a functor 
$F : \mathcal{P}_{1} (M) \to T,$ i.e. an object of 
$\mathrm{Triv}^{1}_{\pi} (i).$ The \emph{\uline{descent object 
associated to the $\pi$-local $i$-trivialization of $F$}} is 
$\mathrm{Ex}^{1}_{\pi} (F, \triv, t).$ Let 
$\a : (F, \triv, t) \to ( F', \triv', t')$ be a morphism in 
$\mathrm{Triv}^{1}_{\pi} (i).$ The \emph{\uline{descent morphism 
associated to the $\pi$-local $i$-trivialization of $\a$}} is 
$\mathrm{Ex}^{1}_{\pi} (\a).$  
\end{defn}

\subsection{Transport functors}

We now discuss smoothness of descent data and finally give a 
definition of transport functors. 

\begin{defn}
\label{defn:smoothdescent1}
A descent object $(\triv, g)$ as above is said to be 
\emph{\uline{smooth}} if $\triv : \mathcal{P}_{1} (Y) \to \Gr$ is a smooth 
functor and if there exists a smooth natural isomorphism 
$\tilde{g} : \pi_{1}^{*} \triv \Rightarrow \pi_{2}^{*} \triv$ with 
$g = \id_{i} \circ \tilde{g},$ the horizontal composition of natural 
transformations $\id_{i}$ and $\tilde{g}.$ A descent morphism 
$h : (\triv, g) \to (\triv', g')$ as above is said to be
\emph{\uline{smooth}} 
if  there exists a smooth natural isomorphism 
$\tilde{h} :  \triv \Rightarrow \triv'$ with $h = \id_{i} \circ
\tilde{h}.$ 
\end{defn}

Smooth descent objects and morphisms form the objects and 
morphisms of a category denoted by 
$\mathfrak{Des}^{1}_{\pi} (i)^{\infty}$ and form a sub-category of 
$\mathfrak{Des}^{1}_{\pi} (i).$ 

\begin{defn}
\label{defn:smoothtrivialization1}
A $\pi$-local $i$-trivialization $(F, \triv, t)$ is said to be 
\emph{\uline{smooth}} if the associated descent object 
$\mathrm{Ex}^{1}_{\pi} ( F, \triv, t)$ is smooth. A morphism 
$\a : (F, \triv, t) \to (F', \triv', t')$ is said to be
\emph{\uline{smooth}} if the  
associated descent morphism $\mathrm{Ex}^{1}_{\pi} ( \a)$ is 
smooth. 
\end{defn}

Smooth local trivializations and their morphisms form the objects and 
morphisms of a category denoted by 
$\mathrm{Triv}^{1}_{\pi} (i)^{\infty}$ and form a sub-category of 
$\mathrm{Triv}^{1}_{\pi} (i).$ $\mathrm{Ex}^{1}_{\pi}$ restricts to an 
equivalence of categories 
$\mathrm{Triv}^{1}_{\pi} (i)^{\infty} \xrightarrow{\simeq}
\mathfrak{Des}^{1}_{\pi} (i)^{\infty}$ of smooth data. 
Again, we will discuss an inverse functor in Section 
\ref{sec:reconstruction} since it will be necessary in discussing gauge 
invariant holonomy in Section \ref{sec:1-holonomy}. We now come to 
the definition of a transport functor (Definition 3.6 of \cite{SW1}). 

\begin{defn}
\label{defn:transportfunctor1}
Let $\mathrm{Gr}$ be a Lie groupoid, $T$ a category, 
$i : \mathrm{Gr} \to T$ a functor, and $M$ a smooth manifold. A 
\emph{\uline{transport functor on $M$ with values in a category $T$ and 
with $\mathrm{Gr}$-structure}} is a functor 
$\mathrm{tra} : \mathcal{P}_{1} (M) \to T$ such that 
there exists a surjective submersion $\pi : Y \to M$ and a smooth 
$\pi$-local $i$-trivialization $(\mathrm{triv}, t)$ of $\tra.$ 
\end{defn}

Transport functors with values in $T$ with $\mathrm{Gr}$-structure 
form the objects of a category 
$\mathrm{Trans}^{1}_{\mathrm{Gr}} (M,T).$  We also define the 
morphisms of transport functors. 

\begin{defn}
\label{defn:mortran} 
A \emph{\uline{morphism $\h$ of transport functors on $M$ from $\tra$ 
to $\tra'$}} is a natural transformation $\h : \tra \Rightarrow \tra'$
such 
that there exists a surjective submersion $\pi : Y \to M$ and smooth 
$\pi$-local $i$-trivializations $(\triv, t),$ $(\triv', t'),$ and 
$h : (\triv, t) \to (\triv', t')$ of $\tra,$ $\tra',$ and $\eta$
respectively. 
\end{defn}

By using pullbacks, one can define the composition of such 
morphisms. We will not explicitly describe this now since we will 
come back to it later when discussing limit categories over surjective 
submersions in Section \ref{sec:1limits}.

\subsection{The reconstruction functor: local to global}
\label{sec:reconstruction}

In many situations, one works locally and pieces together data to 
construct globally defined quantities. In the case of parallel
transport, 
one obtains group elements. An explicit construction of a (weak) 
inverse 
$\mathrm{Rec}_{\pi}^{1} : \mathfrak{Des}^{1}_{\pi} (i) \to
\mathrm{Triv}^{1}_{\pi} (i)$ to $\mathrm{Ex}_{\pi}^{1}$ will assist in 
this direction. Following Section 2.3 of \cite{SW1}, we introduce a 
category that combines the \v Cech groupoid with the path groupoid
utilizing the surjective submersion $\pi : Y \to M.$ 

\begin{defn}
\label{defn:Cechpath1}
Let $\mathcal{P}_{1}^{\pi} (M)$ be the category, called the 
\emph{\uline{\v Cech path groupoid}}, whose set of objects are the 
elements of $Y.$ The set of morphisms are freely generated by two 
types of morphisms (the generators) which are given as follows
\begin{enumerate}[i)]
\item
thin paths (see Definition \ref{defn:thinhomotopy}) $\g$ in $Y$ with 
sitting instants and
\item
points $\a$ in $Y^{[2]}$ (thought of as morphisms 
$\pi_{1} (\a) \xrightarrow{\a} \pi_{2} (\a)$ and called 
\emph{\uline{jumps}}). 
\end{enumerate}
There are several relations imposed on the set of morphisms. 
\begin{enumerate}[(a)]
\item
For any thin path $\Theta : \a \to \b$ in $Y^{[2]}$ the diagram 
\be
\xy 0;/r.15pc/:
(-15,15)*+{\pi_{1} (\b)}="3";
(-15,-15 )*+{\pi_{2} (\b) }="4";
(15,-15)*+{\pi_{2} (\a) }="2";
(15,15)*+{\pi_{1} (\a) }="1";
{\ar^{\a} "1";"2" };
{\ar^{\pi_{2} (\Theta)} "2";"4" };
{\ar_{\pi_{1} (\Theta)} "1";"3"};
{\ar_{\b} "3";"4"};
\endxy
\ee
commutes (see Figure  \ref{fig:pathjumpcondition} for a visualization 
of this). 

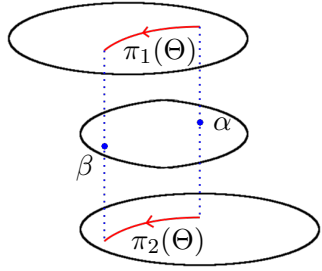
\begin{figure}[h]
\centering
	\begin{picture}(0,0)
	\put(38,43){$\b$}
	\put(90,60){$\a$}
	\put(56,87){$\pi_{1} ( \Theta )$}
	\put(59,15){$\pi_{2} (\Theta)$}
	\end{picture}
    \includegraphics[width=0.30\textwidth]{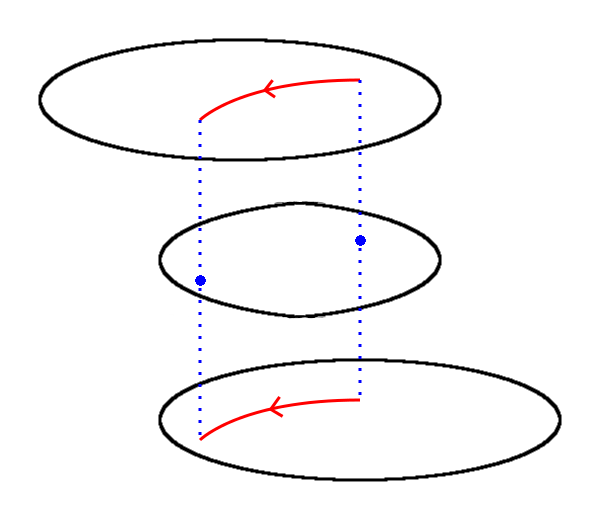}
    \vspace{-3mm}
    \caption{Thinking in terms of an open cover as a submersion, 
    condition i) above says that if a path 
    $\Theta : \a \rightarrow \b$ is in a 
    double intersection, it doesn't matter whether or not the jump is 
    performed first and then the thin path is traversed or vice versa.}
    \label{fig:pathjumpcondition}
\end{figure} 

\item
For any point $\X \in Y^{[3]}$ the diagram 
\be
\xy 0;/r.15pc/:
(-25,-12.5)*+{\pi_{3} (\X)}="3";
(0,12.5)*+{\pi_{2} (\X)}="2";
(25,-12.5)*+{\pi_{1} (\X)}="1";
{\ar_{\pi_{12}(\X)} "1";"2" };
{\ar_{\pi_{23}(\X)} "2";"3" };
{\ar^{\pi_{13}(\X)} "1";"3"};
\endxy
\ee
commutes.

\item
The free composition of two thin free paths is the usual composition 
of thin paths and for every point $y \in Y,$ the thin homotopy class 
representing the constant path at $y$ is equal to $\D (y) \in Y^{[2]}$ 
which is the formal identity for the composition. 
\end{enumerate}

The notation for the free composition will be $*.$
\end{defn}

Item (b) together with item (c) demands that the jumps 
$\a \in Y^{[2]}$ are isomorphisms. A typical morphism in 
$\mathcal{P}_{1}^{\pi} (M)$ is depicted in Figure \ref{fig:pathjump}. 

\begin{figure}[h]
\centering
    \includegraphics[width=0.55\textwidth]{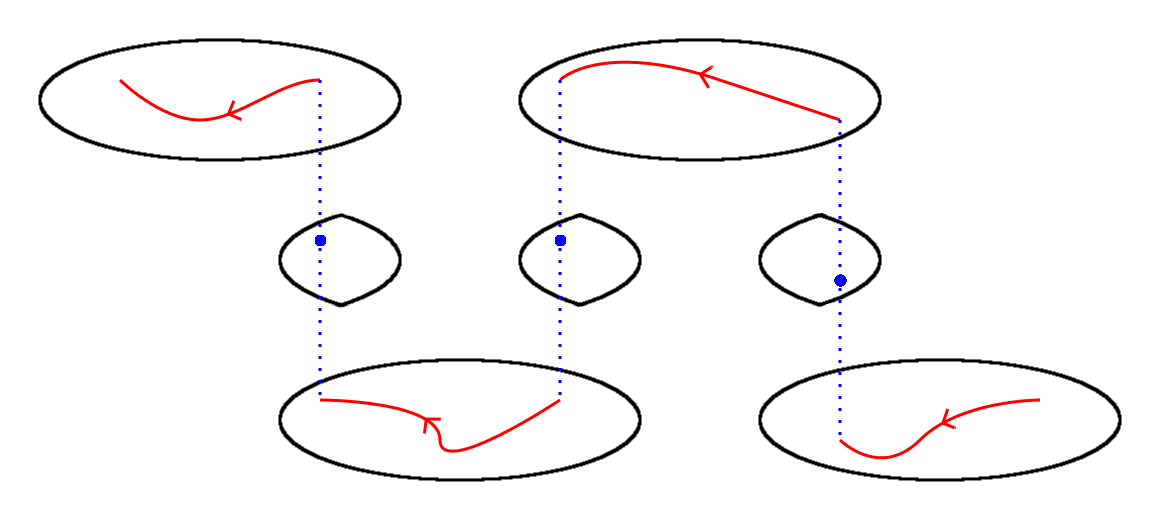}
    \vspace{-3mm}
    \caption{A generic representative of a morphism in 
    $\mathcal{P}_{1}^{\pi} (M)$ is shown above for 
    $Y = \coprod_{i \in I} U_{i},$ the disjoint union over an open
cover. 
    The larger ellipses indicate open sets and the smaller ones in the 
    middle indicate intersections. The curves in the open sets indicate 
    the paths and the dotted vertical lines indicate the jumps.}
    \label{fig:pathjump}
\end{figure} 

Associated to every descent object $(\triv, g)$ in 
$\mathfrak{Des}_{\pi}^{1} (i)$ is a functor 
$R_{(\triv, g)} : \mathcal{P}_{1}^{\pi} (M) \to T$ defined (on objects 
and generators) by  
\be
\begin{split}
Y \ni y &\mapsto \triv_{i} (y) , \\ 
P^{1}Y \ni \g &\mapsto \triv_{i} (\g) , \quad \mbox{ and } \\
Y^{[2]} \ni \a &\mapsto \Big ( g ( \a ) : \triv_{i} (  \pi_{1} (\a) )
\to
\triv_{i} ( \pi_{2} (\a) ) \Big ) .
\end{split}
\ee
This assignment extends to a functor 
$R : \mathfrak{Des}_{\pi}^{1} ( i ) \to 
\mathrm{Funct} ( \mathcal{P}_{1}^{\pi} ( M ) , T )$ (Lemma 2.14. of 
\cite{SW1}). To a descent morphism $h : (\triv, g) \to (\triv ' , g')$
it 
gives a natural transformation 
$R_{h} : R_{(\triv, g)} \Rightarrow R_{(\triv', g')}$ defined by sending
$y \in Y$ to $h(y)$ for all $y \in Y.$ 

The functor 
$\mathrm{Rec}_{\pi}^{1} :  \mathfrak{Des}_{\pi}^{1} ( i ) \to
\mathrm{Triv}_{\pi}^{1} (i)$ will be defined so that it factors through 
$R.$ What will then remain is to define a functor 
$ \mathrm{Funct} ( \mathcal{P}_{1}^{\pi} ( M ) , T ) \to 
\mathrm{Funct} ( \mathcal{P}_{1} ( M ) , T ).$ In order to do this, we 
need to ``lift'' paths. First, notice that there is a canonical 
projection functor $p^{\pi} : \mathcal{P}_{1}^{\pi} ( M )  \to 
\mathcal{P}_{1} ( M )$ which sends objects $y \in Y$ to $\pi (y),$ thin 
paths $\g$ to $\pi (\g),$ and points $\a \in Y^{[2]}$ to the identity.
We 
will construct a weak inverse $s^{\pi} : \mathcal{P}_{1} (M) \to
\mathcal{P}_{1}^{\pi} ( M ).$ 

Since $\pi : Y \to M$ is surjective, for every $x \in M,$ there exists a
$y \in Y$ such that $\pi (y) = x.$ Therefore, define 
$s^{\pi} : \mathcal{P}_{1} (M) \to \mathcal{P}_{1}^{\pi} ( M )$ on 
objects to be this assignment. Because $\pi : Y \to M$ is a surjective 
submersion, there exists an open cover $\{ U_{i} \}_{i \in I}$ of $M$ 
with local sections $s_{i} : U_{i} \to Y$ of $\pi.$ 
Using these local sections, we can define 
$s^{\pi} : \mathcal{P}_{1} (M) \to \mathcal{P}_{1}^{\pi} ( M )$ on 
morphisms as follows. For every thin path $\g : x \to x'$ in $M$ there 
exists a collection of thin paths $\g_{1}, \cdots, \g_{n}$ with 
(representatives of) $\g_{j}$ inside $U_{i_{j}}$ for all
$j = 1, \dots, n$ 
and 
\be
x' \xleftarrow{\g} x 
\quad = \quad 
x' \xleftarrow{\g_n} x_{n-1} \xleftarrow{\g_{n-1}} \cdots 
\xleftarrow{\g_{2}} x_{1} \xleftarrow{\g_{1}} x
.
\ee
For such a choice define (we write $s_{j}$ instead of $s_{i_{j}}$ to 
avoid too many indices)
\be
\label{eq:pathlifting}
s^{\pi} ( \g ) := \a_{x'} * s_{n} ( \g_{n} ) * \a_{n-1} * s_{n-1}
( \g_{n-1} )
*  \cdots * s_{2} ( \g_{2} ) * \a_{1} * s_{1} ( \g_{1} ) * \a_{x} , 
\ee
where $\a_{x}$ is the unique isomorphism from $s^{\pi}(x)$ to 
$s_{1} (x),$ $\a_{j}$ is the unique isomorphism from $s_{j-1} (x_{j})$ 
to $s_{j} (x_{j}),$ and $\a_{x'}$ is the unique isomorphism from 
$s_{n} (x)$ to $s^{\pi}(x').$  
This definition comes from Figure \ref{fig:pathlift}. 
\begin{figure}
\centering
    \includegraphics[width=0.5\textwidth]{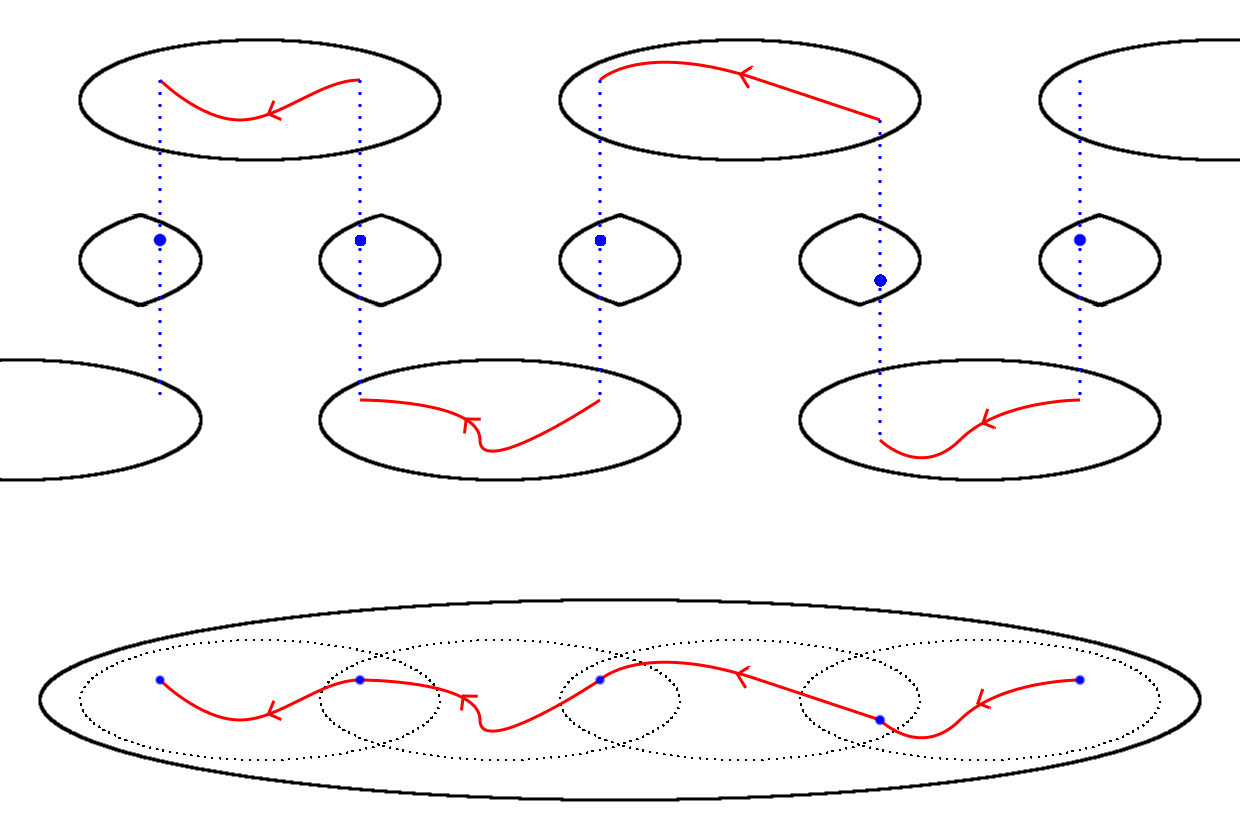}
    \vspace{-3mm}
    \caption{By choosing a decomposition of every path to land in 
    open sets one can lift using the locally defined sections. At the 
    beginning and end of the path, one must apply a jump since the 
    map $s$ defined on objects  might not coincide with the lift of the 
    endpoint of the path.}
    \label{fig:pathlift}
\end{figure} 

The functor $s^{\pi}$ is a weak inverse to $p^{\pi}$ (Lemma 2.15. of 
\cite{SW1}). For reference, by definition this means there exists a 
natural isomorphism 
\be
\z : s^{\pi} \circ p^{\pi} \Rightarrow \id_{\mathcal{P}_{1}^{\pi} (M)}.
\ee
that is part of an adjoint equivalence given by the quadruple 
$(s^{\pi}, p^{\pi}, \z, \id_{p^{\pi} \circ s^{\pi} } )$ since  
$p^{\pi} \circ s^{\pi} = \id_{\mathcal{P}_{1} (M) }.$
This natural isomorphism $\z$ is the one that sends 
$y \in Y$ to the unique jump, an isomorphism, from $y$ to 
$s^{\pi} (\pi(y)).$ It is natural by relation i) in Definition
\ref{defn:Cechpath1}.

\begin{rmk}
Note that we have not put a smooth structure on
$\mathcal{P}^{\pi}_{1} (M)$ nor will we (although it is done in 
\cite{SW1}). Indeed, the choice of lifts for the points could be 
sporadic. All the smoothness for transport functors is contained in the 
descent data. 
\end{rmk}

The functor 
$s^{\pi} : \mathcal{P}_{1} (M) \to \mathcal{P}^{\pi}_{1} (M)$ induces a 
pullback functor $s^{\pi *} : {\mathrm{Funct} 
( \mathcal{P}^{\pi}_{1} (M), T )} \rightarrow {\mathrm{Funct} 
( \mathcal{P}_{1} (M), T )}$ 
defined by $s^{\pi *} (F) :=
F \circ s^{\pi}$ 
on functors $F : \mathcal{P}^{\pi}_{1} (M) \to T$ and by 
$s^{\pi *} ( \rho) := \rho \circ \id_{s^{\pi}}$ on natural
transformations 
$\rho : F \Rightarrow G.$ Finally, $\mathrm{Rec}^{1}_{\pi}$ is defined 
as the composition in the diagram 
\be
\label{eq:Rec1pi}
\xy 0;/r.15pc/:
(-25,10)*+{\mathrm{Funct} ( \mathcal{P}_{1} (M) , T ) }="2";
(25,10)*+{\mathfrak{Des}^{1}_{\pi} (i) }="1";
(0,-10)*+{\mathrm{Funct} ( \mathcal{P}^{\pi}_{1} (M) , T ) }="3";
{\ar_{\qquad \mathrm{Rec}^{1}_{\pi}} "1";"2"};
{\ar^{R} "1";"3"};
{\ar^{s^{\pi *}} "3";"2"};
\endxy
.
\ee
The image of  $\mathfrak{Des}^{1}_{\pi} (i)$ under 
$\mathrm{Rec}_{\pi}^{1}$ is actually in $\mathrm{Triv}^{1}_{\pi} (i).$ 
This means at the level of objects that  associated to 
$R_{(\triv, g )} \circ s^{\pi}$   there exists a $\pi$-local $i$-
trivialization. We take $\triv$ itself for the first part of this datum.
To 
define $t : \pi^{*} \big ( s^{\pi *} ( R_{(\triv, g)} ) \big )
\Rightarrow
\triv_{i}$ we take the composition defined by the diagram 
\be
\xy 0;/r.15pc/:
(40,35)*+{\mathcal{P}_{1} (Y) }="1";
(-40,35)*+{\mathcal{P}_{1} (M) }="2";
(40,-35)*+{\mathrm{Gr}}="3";
(0,15)*+{\mathcal{P}_{1}^{\pi} (M)}="4";
(-40,-5)*+{\mathcal{P}_{1}^{\pi} (M)}="5";
(-40,-35)*+{T}="6";
(0,35)*+{}="a1";
(0,20)*+{}="b1";
(-5,5)*+{}="a2";
(30,-25)*+{}="b2";
{\ar_{\pi_{*}} "1";"2"};
{\ar^{\triv} "1";"3"};
{\ar_{p^{\pi}} "4";"2"};
{\ar@{^{(}->} "1";"4"}; 
{\ar^{\id} "4";"5"};
{\ar_{R_{(\triv, g)}} "5";"6"};
{\ar_{s^{\pi}} "2";"5"};
{\ar_{i} "3";"6"};
{\ar@{=>}_{\z} "2";(-21,7)};
{\ar@{=}"a1";"b1"^{\id}};
{\ar@{=}"a2";"b2"_{\id}};
\endxy
,
\ee
where the functor $\mathcal{P}_{1} (Y) \hookrightarrow 
\mathcal{P}_{1}^{\pi} (M)$ is the inclusion. The rest of the proof, 
namely the fact that the image of a morphism lands in 
$\mathrm{Triv}_{\pi}^{1} (i)$ under $\mathrm{Rec}_{\pi}^{1},$ is 
explained in Appendix B.1. of \cite{SW1}.

\subsection{Differential cocycle data}
\label{sec:diffcocycledata}

We now switch gears a bit and go in the other (infinitesimal) 
direction. We describe this in several parts. We focus on a local 
description first in terms of `trivialized' transport functors. We
extract 
the differential cocycle data from functors and then we construct 
functors from differential cocycle data. This is a brief and simplified 
account of the material covered in Section 4 of \cite{SW1} since we 
do not prove any results.

\subsubsection{From functors to 1-forms} 

Throughout this article, let $\un G$ denote the Lie algebra of $G.$ 
Given a smooth functor $F : \mathcal{P}_{1} (X) \to \mathcal{B} G,$ 
we will define a $\un G$-valued 1-form $A$ pointwise for 
every $x \in X$ and $v \in T_{x} X$ as follows. Let $\g : \R \to X$ be a
curve in $X$ with $\g(0) = x$ and $\frac{d \g}{dt} (0) = v.$ 
$\g : \R \to X$ induces a smooth pushforward functor 
$\g_{*} : \mathcal{P}_{1} (\R) \to \mathcal{P}_{1} (X).$ At the level of
morphisms, the space $P^{1} \R$ of thin homotopy 
classes of paths in $\R$ is actually smoothly equivalent to 
$\R \times \R.$ The diffeomorphism $\g_{\R} : \R \times \R \to 
P^{1} \R$  is defined by sending $(s,t)$ to the thin homotopy class of 
a path in $\R$ determined by its source point $s$ and target $t$ as 
shown schematically in Figure \ref{fig:gammar}. 

\begin{figure}[h]
\centering
\begin{picture}(0,0)
\put(-162,-38){$t$}
\put(-93,-38){$s$}
\put(111,-38){$t$}
\put(180,-38){$s$}
\put(-30,-20){$\xymatrix{\ar@/^0.5pc/@{|->}[rr]^{\g_{\R}} & & }$}
\end{picture}
    \includegraphics[width=1.00\textwidth]{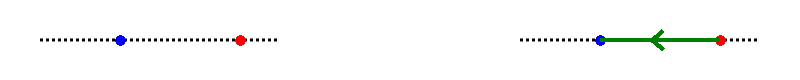}
    \vspace{-3mm}
    \caption{A point $(s,t)$ in $\R^{2}$ is drawn as two points on $\R$
and gets mapped to the thin path in $\R$ from the point $s$ to the point
$t$ with a representative shown on the right under the map $\g_{\R}.$}
    \label{fig:gammar}
\end{figure}
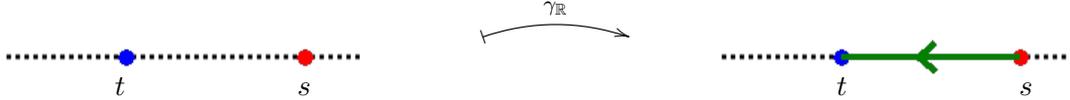

Therefore, we obtain a function $F_{1} \circ \g_{*} \circ \g_{\R}$ from
the composition 
\be
G \xleftarrow{F_{1}} P^{1} X \xleftarrow{\g_{*}} P^{1} \R
\xleftarrow{\g_{\R}} \R \times \R . 
\ee 
Here $F_{1}$ is $F$ restricted to the set of morphisms $P^{1} X.$ 
Using this, we define
\be
\label{eq:DXF}
A_{x} (v) := - \frac{d}{dt} \bigg |_{t=0} F_{1} \Big ( \g_{*} 
\big ( \g_{\R} (0,t)  \big) \Big). 
\ee
$A_{x}(v)$ is independent of $\g$ and only depends on $x$ and $v.$
Furthermore, it defines a 1-form $A \in \W^{1} (X; \un G).$

\subsubsection{From 1-forms to functors} 
\label{sec:1formstofunctors}

Starting with a $\un G$-valued 1-form $A \in \W^{1} (X; \un G)$ on
$X$ we want to define a smooth functor
$\mathcal{P}_{1} (X) \to \mathcal{B} G.$ To do this, we first define a
function, referred to as the \emph{\uline{path transport}},
$k_{A} : P X \to G$ on paths in $X$ with sitting instants
(we do \emph{not} mod out by thin homotopy).  Given $\g \in PX,$ we
can pull back the 1-form $A$ to $\R,$ namely
$\g^{*} (A) \in \W^{1} ([0,1]; \un G).$ We then define
$k_{A} (\g)$ using the path-ordered-exponential
\be
\label{eq:kA}
k_{A} (\g) :=  \mathcal{P} \exp \left \{ \int_{0}^{1} A_{t}
\left ( \frac{\p}{\p t} \right ) dt \right \} . 
\ee
Recall that this path-ordered exponential is defined by%
\footnote{In this expression, we are assuming that $G$ is a matrix
Lie group.}
\be
\label{eq:pathexp}
\mathcal{P} \exp \left \{ \int_{0}^{1} A_{t} \left ( \frac{\p}{\p t}
\right ) dt \right \} := \sum_{n=0}^{\infty} \frac{1}{n!} \int_{0}^{1}
dt_{n} \cdots \int_{0}^{1} dt_{1} \;  \mathcal{T} \left [  A_{t_{n}}
\left ( \frac{\p}{\p t} \right ) \cdots A_{t_{1}} \left (
\frac{\p}{\p t} \right ) \right ] , 
\ee
where the \emph{time-ordering operator} 
$\mathcal{T}$ is defined by
\be
\mathcal{T} \left [ A_{t} A_{s} \right ]  := 
\begin{cases}
A_{t} A_{s} &\mbox{ if } t \ge s \\
A_{s} A_{t} & \mbox{ if } s \ge t 
\end{cases}
.
\ee
The $n=0$ term on the right-hand side of equation (\ref{eq:pathexp})
is the identity. 
We can picture the path-ordered exponential schematically as a power
series of graphs with marked points as in Figure
\ref{fig:pathorderedintegral}.

\begin{figure}[h]
\centering
    \includegraphics[width=1.00\textwidth]{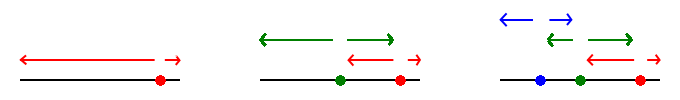}
    \vspace{-3mm}
    \caption{The path-ordered integral is depicted as a power series
over integrals. The first term (not drawn) is the identity. The second
term is the integral of $A_{t}$ (depicted as a bullet on the interval)
over all $t$ from the right to the left (the orientation goes from right
to left). The third term is the integral of $A_{t} A_{s}$ over the
interval but keeping earlier operators to the right. This is drawn by
showing the bullet on the right being able to move along the interval
provided it stays behind the bullet to its left. The fourth term
involves three operators. Higher terms are not drawn. All terms are
summed with appropriate factors.}
    \label{fig:pathorderedintegral}
\end{figure}

$k_{A}$ only depends on the thin homotopy class of $\g$
and therefore factors through a smooth map $F_{A} : P^{1} X \to G$ on
thin paths (see Definition \ref{defn:thinhomotopy}). This map defines
a smooth functor $F_{A} : \mathcal{P}_{1} (X) \to \mathcal{B} G$
(see Proposition 4.3. and Lemma 4.5. of \cite{SW1}).

\subsubsection{Local differential cocycles for transport functors}
\label{sec:diffcocycles}

The above constructions can be extended to smooth natural
transformations between smooth functors. Given a smooth natural
transformation $h : F \Rightarrow F'$ of smooth functors
$F , F' : \mathcal{P}_{1} (X) \to \mathcal{B} G$ we obtain a function,
written somewhat abusively also as $h : X \to G$ satisfying 
\be
h(y) F(\g) = F'(\g) h(x) 
\ee
for all thin paths $\g : x \to y$ in $X.$ If we differentiate this
condition, we obtain 
\be
\label{eq:A'toA}
A' = \mathrm{Ad}_{h} (A) - h^{*} \ov \theta,
\ee
where $\ov \theta$ is right Maurer-Cartan form, sometimes written as
$dg g^{-1}$ for matrix groups, $A$ is the 1-form corresponding to
$F,$  $A'$ is the 1-form corresponding to $F',$ and $\mathrm{Ad}$ is the
adjoint action on the Lie algebra $\un G$  defined by 
\be
\label{eq:Adjointaction}
\mathrm{Ad}_{h} ( T ) := \frac{d}{dt} \bigg |_{t=0} \Big ( h
\exp \{ t T \} h^{-1} \Big )
\ee 
for all $T \in \un G.$ This motivates the following definition. 
\begin{defn}
\label{defn:globalconnection}
Let $Z^{1}_{X} (G)^{\infty}$ be the category whose objects are 1-forms
$A \in \W^{1} (X ; \un G )$ and a morphism from $A$ to $A'$ is a
function $h :  X \to G$ satisfying 
\be
A' = \mathrm{Ad}_{h} (A) - h^{*} \ov \theta.
\ee
The composition is defined by 
\be
\Big ( A'' \xleftarrow{h'} A' \xleftarrow{h} A \Big ) \quad \mapsto
\quad  \Big ( A'' \xleftarrow{h' h} A \Big ),
\ee
where $h' h$ is (pointwise) multiplication of $G$-valued functions. 
\end{defn}

This (and the previous section) defines two functors 
\be
\label{eq:integratedifferentiate}
\xymatrix{ 
Z_{X}^{1} (G)^{\infty} \ar@<+1ex>[r]^(0.4){\mathcal{P}_{X}} &
\ar@<+1ex>[l]^(0.6){\mathcal{D}_{X}} \mathrm{Funct}^{\infty}
(\mathcal{P}_{1}(X), \mathcal{B} G )
}
,
\ee
where $\mathrm{Funct}^{\infty} (\mathcal{P}_{1}(X), \mathcal{B} G)$ is
the category of smooth functors and smooth natural transformations from
the thin path groupoid of $X$ to $\mathcal{B} G.$ These
functors are defined on objects by $\mathcal{D}_{X}(F) := A$ from
(\ref{eq:DXF}) and $\mathcal{P}_{X}(A) := F_{A}$ from (\ref{eq:kA}). 
These two functors are \emph{inverses} of each other, and not just an
\emph{equivalence} pair (Proposition 4.7. of \cite{SW1}). 

All of this was for globally defined smooth functors. Transport functors
on $M$ are not necessarily smooth globally. However, there must exist a
surjective submersion $\pi : Y \to M$ with a smooth $\pi$-local
$i$-trivialization. The smooth functor $\triv : \mathcal{P}_{1} (Y)
\to \mathcal{B} G$ corresponds to a 1-form $A \in \W^{1} (Y; \un G),$
which is an object in $Z_{Y}^{1} (G)^{\infty}.$ The natural
transformation $g : \pi_{1}^{*} \triv_{i} \Rightarrow \pi_{2}^{*}
\triv_{i}$ factors through a smooth natural transformation $\tilde{g} :
\pi_{1}^{*} \triv \Rightarrow \pi_{2}^{*} \triv,$ which is a morphism
in the category $Z_{Y^{[2]}}^{1} (G)^{\infty}$ from
$\pi_{1}^{*} A$ to $\pi_{2}^{*} A.$ This means 
\be
\pi_{2}^{*} A = \mathrm{Ad}_{\tilde{g}} ( \pi_{1}^{*} A ) -
\tilde{g}^{*} \ov \theta.
\ee
The condition 
\be
\begin{smallmatrix}
\pi_{12}^{*} g \\
\overset{\circ}{\pi_{23}^{*} g}
\end{smallmatrix}
=
\pi_{13}^{*} g
\ee
translates to 
\be
\pi_{23}^{*} \tilde{g} \ \pi_{12}^{*} \tilde{g} = \pi_{13}^{*}\tilde{g},
\ee
where the concatenation indicates multiplication of $G$-valued
functions.  
A morphism of transport functors subordinate to the same surjective
submersion is a natural transformation
$h : \triv_{i} \Rightarrow \triv'_{i}$ that factors through a smooth
natural transformation $\tilde{h} : \triv \Rightarrow \triv'$ and
therefore defines a morphism from $A$ to $A'$ in
$Z_{Y}^{1} (G)^{\infty}.$ This motivates the following definition of
local differential cocycles. 

\begin{defn}
\label{defn:diffcocyclespi}
Let $\pi : Y \to M$ be a surjective submersion. Define the category
$Z_{\pi}^{1} (G)^{\infty}$ of \emph{\uline{differential cocycles
subordinate to $\pi$}} as follows. An object of
$Z_{\pi}^{1} (G)^{\infty}$  is a pair $(A, g),$ where $A$ is an object
in $Z_{Y}^{1} (G)^{\infty},$ $g$ is a morphism from $\pi_{1}^{*} A$ to
$\pi_{2}^{*} A$ in $Z_{Y^{[2]}}^{1} (G)^{\infty}.$ A morphism from
$(A,g)$ to $(A',g')$ is a morphism $h$ from $A$ to $A'$ in $Z_{Y}^{1}
(G)^{\infty}.$ The composition of morphisms in
$Z_{\pi}^{1} (G)^{\infty}$ is defined by 
\be
\Big ( (A'',g'') \xleftarrow{h'} (A',g') \xleftarrow{h} (A,g) \Big )
\quad \mapsto \quad \Big ( (A'',g'') \xleftarrow{h' h} (A,g) \Big ).
\ee
\end{defn}

The above generalizations produce functors
\be
\label{eq:localintegratedifferentiate}
\xymatrix{ 
Z_{\pi}^{1} (G)^{\infty} \ar@<+1ex>[r]^(0.65){\mathcal{P}_{\pi}} & 
\ar@<+1ex>[l]^(0.35){\mathcal{D}_{\pi}} 
} \mathfrak{Des}_{\pi}^{1} (i)^{\infty}
\ee
exhibiting an equivalence of categories whenever
$i : \mathcal{B} G \to T$ is an equivalence (Corollary 4.9. in
\cite{SW1}).

\subsection{Limit over surjective submersions}
\label{sec:1limits}

Here we give a brief summary of the four levels of construction
introduced and the notation of the functors relating these categories.
To do this, we get rid of the dependence on the surjective submersion
in the categories introduced in the prequel. Several of our categories
depended on the choice of a surjective submersion. These categories
were $ \mathrm{Triv}_{\pi}^{1} (i)^{\infty}, \mathfrak{Des}^{1}_{\pi}
(i)^{\infty}, $ and $Z^{1}_{\pi} ( G )^{\infty}.$ On the contrast, the
category of transport functors
$\mathrm{Trans}_{\mathcal{B} G}^{1} ( M , T)$ does not depend on $\pi.$
To relate these categories better, we will take limits over $\pi.$
Changing the surjective submersion gives a collection of categories
depending on such surjective submersions. One can take a limit over
the collection of surjective submersions in this case. 

The general construction is done as follows. Let $S_{\pi}$ be a family
of categories parametrized by surjective submersions $\pi : Y \to M$
and let $F (\z) : S_{\pi} \to S_{\pi \circ \z}$ be a family of functors
for every refinement $\z : Y' \to Y$ of $\pi$ satisfying the condition
that for any iterated refinement $\z ' : Y'' \to Y'$ and $\z : Y' \to Y$
of $\pi : Y \to M$ then $F ( \z' \circ \z ) = F (\z') \circ F (\z).$ In
this case, an object of $\varinjlim_{\pi} S_{\pi}$ is given by a pair
$(\pi, X)$ of a surjective submersion $\pi : Y \to M$ and an object $X$
of $S_{\pi}.$  A morphism from $(\pi_{1}, X_{1} )$ to $(\pi_{2},X_{2})$
consists of an equivalence class of a common refinement 
\be
\xy 0;/r.15pc/:
(0,20)*+{Z}="Z";
(-20,0)*+{Y_{1}}="Y1";
(20,0)*+{Y_{2}}="Y2";
(0,-20)*+{M}="M";
{\ar^{\z} "Z";"M"};
{\ar_{y_{1}} "Z";"Y1"};
{\ar^{y_{2}} "Z";"Y2"};
{\ar_{\pi_{1}} "Y1";"M"};
{\ar^{\pi_{2}} "Y2";"M"};
\endxy
\ee
together with a morphism
$f : (F (y_{1})) (X_{1}) \to (F(y_{2})) (X_{2})$ in $S_{\z}.$
It is written as a pair $(\z, f).$ Two such $(\z, f)$ and $(\z', f')$
are equivalent if they agree (on the nose) on their common pullback.
The composition 
\be
(\pi_{3}, X_{3} ) \xleftarrow{(\z_{23}, g)} (\pi_{2}, X_{2} )
\xleftarrow{(\z_{12}, f)} (\pi_{1}, X_{1} ) 
\ee
is defined by choosing representatives and taking the pullback
refinement 
\be
\xy 0;/r.15pc/:
(0,30)*+{Z_{13}}="Z13";
(-20,10)*+{Z_{12}}="Z12";
(20,10)*+{Z_{23}}="Z23";
(-30,-10)*+{Y_{1}}="Y1";
(0,-10)*+{Y_{2}}="Y2";
(30,-10)*+{Y_{3}}="Y3";
(0,-30)*+{M}="M";
{\ar_{\z_{12}} "Z12";"M"};
{\ar^{\z_{23}} "Z23";"M"};
{\ar "Z12";"Y1"};
{\ar "Z12";"Y2"};
{\ar "Z23";"Y2"};
{\ar "Z23";"Y3"};
{\ar_{\pi_{1}} "Y1";"M"};
{\ar|-{\pi_{2}} "Y2";"M"};
{\ar^{\pi_{3}} "Y3";"M"};
{\ar_{\mathrm{pr}_{12}} "Z13";"Z12"};
{\ar^{\mathrm{pr}_{23}} "Z13";"Z23"};
\endxy
\ee
along with the composition
$(F(\mathrm{pr}_{23})) (g) \circ (F ( \mathrm{pr}_{12} ) ) (f ).$ One
can check this definition does not depend on the equivalence class
chosen. 

After getting rid of the specific choices of the surjective submersions,
we can take the limits of all the categories we have introduced. We set
the following notation, slightly differing from that of \cite{SW4}:
\begin{align}
\mathrm{Triv}^{1}_{M} (i)^{\infty} &:= \varinjlim_{\pi}
\mathrm{Triv}_{\pi}^{1} (i)^{\infty} \label{eq:triv1M}  \\
\mathfrak{Des}^{1}_{M} (i)^{\infty} &:= \varinjlim_{\pi}
\mathfrak{Des}^{1}_{\pi} (i)^{\infty} \label{eq:desc1M} \\
Z^{1} ( M; G )^{\infty} &:= \varinjlim_{\pi} Z^{1}_{\pi} ( G )^{\infty}
\label{eq:Z1MG}. 
\end{align}
Because a limit of such equivalences is still an equivalence, the
following facts, summarizing the several previous sections, hold. 
The categories $Z^{1} ( M; G )^{\infty}$ and
$\mathfrak{Des}^{1}_{M} (i)^{\infty}$ are equivalent under the condition
that $i : \mathcal{B} G \to T$ is an equivalence of categories.
$\mathfrak{Des}^{1}_{M} (i)^{\infty}$ and
$\mathrm{Triv}^{1}_{M} (i)^{\infty}$ are equivalent for \emph{any}
$i.$
Let $v : \mathrm{Triv}^{1}_{M} (i)^{\infty} \to
\mathrm{Trans}^{1}_{\mathcal{B} G} (M, T)$ be the forgetful functor
which forgets the specific local trivialization. If $i$ is full and
faithful, then $v : \mathrm{Triv}^{1}_{M} (i)^{\infty} \to
\mathrm{Trans}^{1}_{\mathcal{B} G} (M, T)$ is part of an equivalence of
categories. All these statements are proved in \cite{SW1} (except the
last one, but it follows from Lemma 3.3 in \cite{SW1}).

For the reader's convenience, we collect the categories and equivalences
(assuming $i$ is an equivalence)
introduced in the past few sections
\be
\label{eq:allequivalences1}
\xymatrix{
Z^{1} ( M; G )^{\infty} \ar@<+1ex>^{\mathcal{P}}[r]
& \ar@<+1ex>^{\mathcal{D}}[l]  \mathfrak{Des}^{1}_{M} (i)^{\infty}
\ar@<+1ex>^{\mathrm{Rec}^{1}}[r]   & \ar@<+1ex>^{\mathrm{Ex}^{1}}[l]
\mathrm{Triv}^{1}_{M} (i)^{\infty}  \ar@<+1ex>^(0.45){v}[r]
& \ar@<+1ex>^(0.55){c}[l] \mathrm{Trans}^{1}_{\mathcal{B} G} (M, T)
}
,
\ee
where we've introduced the notation $\mathcal{P} := \varinjlim_{\pi}
\mathcal{P}_{\pi}$ and similarly for the other functors. $c$ is a weak
inverse to $v$ and chooses a $\pi$-local $i$-trivialization for
transport functors.

\subsection{Parallel transport, holonomy, and gauge invariance}
\label{sec:1-holonomy} 

Holonomy for principal $G$-bundles with connection is defined in several
different ways. In all cases, it is a special case of parallel transport
where one restricts attention to paths whose target match their source,
i.e. \emph{marked loops}.%
\footnote{The terminology ``marked'' is chosen over ``based'' to avoid
confusion with the based loop space, which is the space of loops with a
single base point. We allow our basepoints to vary.}
 Holonomy along a marked loop is an isomorphism of the fiber over the
endpoint. However, for computational purposes, it is convenient to
express such isomorphisms as group elements. One common way of doing
this is to choose an open cover over which the bundle trivializes,
choose a trivialization, and for each path, choose a decomposition of
that path subordinate to the cover and parallel transport along each
piece while patching the terms together using the transition functions.
This is the procedure we discussed in Section \ref{sec:reconstruction}.
The problem with this procedure is that it depends on several choices.
One purpose of this section is to analyze the dependence on
these choices. The second purpose is to discuss (and make precise) the
dependence of such group elements on the marking chosen for loops. The
punchline is that to obtain a well-defined holonomy independent of such
choices, one needs to pass to conjugacy classes in $G.$ 

The first goal is accomplished by starting with a transport functor
$F:\mathcal{P}_{1} (M) \to T,$ choosing a local trivialization,
extracting the descent data, and using the descent data to reconstruct
a transport functor. This procedure can be described as a functor, which
we call $\scripty{t} \;,$ from
$\mathrm{Trans}_{\mathcal{B} G}^{1} (M, T)$ to itself (see Definition
\ref{defn:groupvaluedholonomyfunctor}). Although all the ingredients for
the functor $\scripty{t} \;$ were described in \cite{SW1}, this
procedure was not discussed. Here, we formulate this procedure and use
it to analyze holonomy along loops. Thus, starting with a transport
functor $F$ we obtain a new transport functor $\scripty{t} \; (F)$ that
produces group-valued holonomies along loops under suitable assumptions.
The first choice we made in this procedure is the transport functor $F$
itself. One could have chosen a different, but naturally isomorphic,
transport functor $F'$ to obtain $\scripty{t} \; (F')$. The other
choices made were those defining $\scripty{t} \;.$
Abstract nonsense tells us there is a morphism
$F \to \scripty{t} \; (F)$ of transport functors.
Different choices of local trivializations and reconstructions are thus
described in terms of natural isomorphisms. Formulated this way, it
becomes a tautology that holonomy along loops is independent of these
choices once one passes to conjugacy classes in $G.$ 

\begin{rmk}
One might argue that such a complicated formalism to obtain the
well-known fact that holonomy is defined only with respect to conjugacy
classes of $G$ is overkill. While this is true for holonomy along loops,
this formalism extends naturally to holonomy along surfaces, which is
our main objective, and the proofs are similar since they are expressed
in terms of category theory. In the case of surfaces, we will use these
ideas to generalize the results of Section 5.2 of \cite{SW4}. It is
therefore important to study the simpler case of holonomy along loops
first. 
\end{rmk}

The second goal, namely the dependence on markings, is accomplished by
showing that for any two loops that are thinly homotopic, but not
necessarily thinly homotopic preserving their marking, the group-valued
holonomy using $\scripty{t} \; (F)$ is well-defined up to conjugation.
Using all these observations, we define, for every isomorphism class of
transport functors, a holonomy map $L^{1}M \to G / \mathrm{Inn} (G)$
from the space of thin homotopy classes of free loops (see Definition
\ref{defn:thinbasedloopspace}) to the conjugacy classes of $G.$

We now define precisely what we mean by (functorially) extracting
group-valued parallel transport from arbitrary transport functors. In
order to accomplish this, we restrict our discussion to transport
functors with $\mathcal{B} G$-structure and with values in $T$ and
assume once and for all that $i : \mathcal{B} G \to T$ is full and
faithful. 

\begin{defn}
\label{defn:groupvaluedholonomyfunctor}
A \emph{\uline{group-valued transport extraction}} is a composition of
functors (starting at the left and moving clockwise)
\be
\label{eq:groupvaluedholonomyfunctor}
\xy 0;/r.15pc/:
(-30,0)*+{\mathrm{Trans}_{\mathcal{B} G}^{1} (M, T)}="1";
(0,15)*+{\mathrm{Triv}^{1} (i)^{\infty}}="2";
(30,0)*+{\mathfrak{Des}^{1} (i)^{\infty}}="3";
(0,-15)*+{\mathrm{Triv}^{1} (i)^{\infty}}="4";
{\ar@/^1pc/^(0.4){c} "1";"2"};
{\ar@/^1pc/^(0.6){\mathrm{Ex}^{1}} "2";"3"};
{\ar@/^1pc/^(0.35){\mathrm{Rec}^{1}} "3";"4"};
{\ar@/^1pc/^(0.6){v} "4";"1"};
\endxy
\ee
and consists of a choice of a weak inverse $c$ of the forgetful functor
$v$ and a reconstruction functor $\mathrm{Rec}^{1}$ (which itself
depends on the choice of a lifting of paths as in
(\ref{eq:pathlifting})). Such a functor is written as
$\scripty{t} \;  := v \circ \mathrm{Rec}^{1} \circ \mathrm{Ex}^{1}
\circ c.$ The notation $\scripty{t} \;$ stands for (local)
trivialization.
\end{defn}

\begin{rmk}
\label{rmk:trivialization}
Although the functor $\scripty{t} \;$ depends on both $c$ and $s^{\pi}$
(which defines $\mathrm{Rec}^{1}$) we suppress the notation. The reason
is because if we change $c$ and/or $s^{\pi},$ the functor
$\scripty{t} \;$ will change to a naturally isomorphic one and only this
fact will matter in any computation.
\end{rmk}

The purpose of $\scripty{t} \;$ is that it assigns group elements to
thin paths for \emph{every} transport functor $F$ and \emph{also}
assigns group-valued gauge transformations for \emph{every} morphism
$\eta : F \to F'$ of transport functors (this will be reviewed in the
following paragraphs). Furthermore, we know that a natural isomorphism
$\scripty{r} : \id \Rightarrow \scripty{t}\;$ exists because all the
functors in (\ref{eq:groupvaluedholonomyfunctor}) are (part of)
equivalences of categories. 
\emph{Choosing} such a natural isomorphism specifies isomorphisms from
the original fibers to the fiber $G$ viewed as a $G$-torsor and relates
our original parallel transports to the group elements defined from
$\scripty{t}\;.$ 

To see this, first recall what $\scripty{t}\;$ does. For a transport
functor $F,$ $c$ chooses a local trivialization
$c(F) := (\pi, F, \triv, t).$ Then we extract the smooth local descent
object $\mathrm{Ex}^{1} ( \pi, F , \triv, t ) := (\pi, \triv, g ).$
Then, we reconstruct a $\pi$-local $i$-trivialization $\mathrm{Rec}^{1}
( \pi, \triv, g)$ and then forget the trivialization data keeping just
the transport functor $v ( \mathrm{Rec}^{1} (\pi, \triv, g) ).$ The
resulting transport functor, written as $\scripty{t}_{F}$ (as opposed to
$\scripty{t}\;(F)$), is defined by (see the paragraph after Definition
\ref{defn:Cechpath1})
\be
\begin{split}
\mathcal{P}_{1}(M) &\xrightarrow{\scripty{t}_{F}} T \\
M \ni x &\mapsto i(\bullet) =: \triv_{i} ( s^{\pi} ( x) )  \\
P^{1} M \ni \g &\mapsto	R_{\mathrm{Ex}^{1} ( c (F) )} ( s^{\pi} (\g) ).
\end{split}
\ee
Here $\triv: \mathcal{P}_{1}(Y) \to \mathcal{B} G$ is the ``local''
transport, $s^{\pi} : \mathcal{P}_{1} (M) \to \mathcal{P}^{\pi}_{1} (M)$
is a choice of lifting points and paths, and
$R_{\mathrm{Ex}^{1} (c(F))} (s^{\pi}(\g)) : i(\bullet) \to i(\bullet)$
is an element of $G$ because $i$ is full and faithful.  
This element of $G$ is defined by choosing a lift of the path $\g$ (see
Figure \ref{fig:pathlift}) and applying trivialized transport on the
pieces and transition functions on the jumps (see Section
\ref{sec:reconstruction}). Note that in the special case that
$T = G\text{-}\mathrm{Tor},$ $i(\bullet)$ can be taken to be $G$ itself
and then $\scripty{t}_{F}(\g)$ for a thin path $\g$ is left
multiplication by some uniquely specified group element. 

To a morphism $\eta : F \to F'$ of transport functors, the resulting
morphism of transport functors, written as $\scripty{t}_{\eta},$ is
defined as follows. First, $c$ chooses surjective submersions
$\pi : Y \to M$ and $\pi' : Y' \to M$ for $F$ and $F',$ respectively,
along with local trivializations $(\triv, t)$ and $(\triv', t').$  This
means that under $c$ the morphism $c(\eta)$ is defined on a common
refinement $\z : Z \to M$ of both $\pi$ and $\pi'.$ The same thing
applies to the extracted descent morphism $\mathrm{Ex}^{1} (c (\eta) )
= ( \z, h ).$ Since our domain is changed under the refinement, $h$ is
defined by the composition 
\be
\label{eq:hforR}
\xy 0;/r.15pc/:
(30,15)*+{\mathcal{P}_{1} (Y) }="1";
(30,-15)*+{\mathcal{P}_{1} (Y') }="1p";
(0,0)*+{\mathcal{P}_{1}(M)}="2";
(60,0)*+{\mathcal{P}_{1}(Z)}="4";
(-40,0)*+{T}="3";
{\ar@/_2.5pc/_{\triv_{i}} "1";"3"};
{\ar_{y_{*}} "4";"1"};
{\ar^{y'_{*}} "4";"1p"};
{\ar_{\pi_{*}} "1";"2"};
{\ar^{\pi'_{*}} "1p";"2"};
{\ar@/_1pc/|-{F} "2";"3"};
{\ar@/^1pc/|-{F'} "2";"3"};
{\ar@/^2.5pc/^{\triv'_{i}} "1p";"3"};
{\ar@{=>}^{\ov t} (-5,22.5);(-5,5)};
{\ar@{=>}^{t'} (-5,-5);(-5,-22.5)};
{\ar@{=>}^{\eta} (-22,4);(-22,-4)};
{\ar@{=>}^{\id} "1"+(0,-8);"1p"+(0,8)};
\endxy
.
\ee
This composition satisfies the condition
\be
\begin{matrix}
y^{[2] *} g \\
\overset{\circ}{\z_{2}^{*} h}
\end{matrix}
\quad
=
\quad
\begin{matrix}
\z_{1}^{*} h \\
\overset{\circ}{y'^{[2] *} g'}
\end{matrix}
.
\ee
The notation means the following. A map $y : Z \to Y$ (and similarly for
$y' : Z \to Y'$) determines a unique map $y^{[2]}: Z^{[2]} \to Y^{[2]}$
defined by $y^{[2]}(z,z') := (y(z),y(z')).$ The maps
$\z_{1}, \z_{2} : Z^{[2]} \to Z$ are the two projections. 

The reconstruction functor $\mathrm{Rec}^{1} : \mathfrak{Des}_{M}^{1}(i)
\to \mathrm{Triv}_{M}^{1} (i) $ sends the morphism $h$ in
(\ref{eq:hforR}) to $\mathrm{Rec}^{1} ( \z, h ) := s^{\z *}
R_{( \z, h )}$ which is a morphism of transport functors
from $\mathrm{Rec}^{1} ( y^* (\pi, \triv, g ) )$ to $\mathrm{Rec}^{1}
( y'^* (\pi', \triv', g' ) )$ with respect to this common refinement and
where $s^{\z} : \mathcal{P}_{1} (M) \to \mathcal{P}_{1}^{\z} (M).$
$\mathrm{Rec}^{1} ( \z, h )$ is defined by sending $x \in M$ to
$h ( s^{\z} (x) )$ which is a morphism from
$\triv_{i} ( y ( s^{\z} (x) ) )$ to $\triv'_{i} ( y' ( s^{\z} (x ) ) ).$

Now, the natural isomorphism $\scripty{r} : \id \Rightarrow
\scripty{t}\;$ assigns to every transport functor $F$ a morphism of
transport functors $\scripty{r}_{F} : F \to  \scripty{t}_{F}.$  This
means (see Definition \ref{defn:mortran}) that associated to every
$x \in M$ is an isomorphism $\scripty{r}_{F} (x) : F(x) \to i(\bullet)$
satisfying naturality, which means that to every thin path
$\g \in P^{1} M$ from $x$ to $y,$ the diagram 
\be
\label{eq:rcsF}
\xy 0;/r.15pc/:
(0,-15)*+{F(y)}="3";
(-30,-15 )*+{i(\bullet)}="4";
(-30,15)*+{i(\bullet)}="2";
(0,15)*+{F(x)}="1";
{\ar_{\scripty{r}_{F} (x)} "1";"2" };
{\ar_{\scripty{t}_{F} (\g) } "2";"4" };
{\ar^{F(\g)} "1";"3"};
{\ar^{\scripty{r}_{F} (y)} "3";"4"};
\endxy
\ee
commutes. 

\begin{rmk}
In Section 3.2, \cite{SW1} \emph{define} the Wilson line, what we're
calling $\scripty{t}_{F} (\g),$ in terms of (\ref{eq:rcsF}) as the
composition 
$\scripty{r}_{F} (y) \circ F(\g) \circ \scripty{r}_{F} (x)^{-1} :
i(\bullet) \to i(\bullet)$ using that $i$ is full and faithful so that
this composition defines a unique group element. Our viewpoint is to
define the Wilson line functorially and globally by using the
group-valued transport extraction procedure $\scripty{t}\;.$ 
\end{rmk}

Since $\scripty{r}\;$ itself is a natural transformation, to every
morphism $\eta : F \to F'$ of transport functors, the diagram 
\be
\xy 0;/r.15pc/:
(0,-15)*+{F'}="3";
(-30,-15 )*+{\scripty{t}_{F'} }="4";
(-30,15)*+{\scripty{t}_{F}}="2";
(0,15)*+{F}="1";
{\ar_(0.4){\scripty{r}_{F}} "1";"2" };
{\ar_{\scripty{t}_{\eta}} "2";"4" };
{\ar^{\eta} "1";"3"};
{\ar^(0.4){\scripty{r}_{F'}} "3";"4"};
\endxy
\ee
commutes. 

To analyze holonomy, we need to restrict parallel transport to thin
paths whose source and target are the same, i.e. thin marked loops,
and eventually thin free loops. 

\begin{defn}
\label{defn:thinbasedloopspace}
The \emph{\uline{marked loop space of $M$}} is the set
\be
\label{eq:markedloopspace}
\mathfrak{L} M := \{ \g \in PM \ | \ s(\g) = t(\g) \}
\ee
equipped with the subspace smooth structure (see Example
\ref{ex:subspace}). Elements of $\mathfrak{L} M$ are called
\emph{\uline{marked loops}}. 
The \emph{\uline{thin marked loop space of $M$}} is the set
\be
\label{eq:thinmarkedloopspace}
\mathfrak{L}^{1} M := \{ \g \in P^1M \ | \ s(\g) = t(\g) \}
\ee
equipped with the subspace smooth structure. 
Elements of $\mathfrak{L}^{1} M$ are called
\emph{\uline{thin marked loops}}. 
\end{defn}

\begin{defn}
\label{defn:tcsholonomyF}
The \emph{\uline{$\scripty{t}\;$-holonomy of $F$}}, written as
$\mathrm{hol}_{\scripty{t}}^{F},$ is defined as the restriction of
parallel transport of a transport functor $F$ to the thin marked loop
space $\mathfrak{L}^{1} M$ of $M$:
\be
\mathrm{hol}_{\scripty{t}}^{F}
:=\scripty{t}_{F} \Big|_{\mathfrak{L}^{1}M} : \mathfrak{L}^{1} M \to G.
\ee
\end{defn}

We now pose three questions that will motivate the rest of our
discussion on holonomy along thin marked loops. 

\begin{enumerate}[i)]
\item
How does $\mathrm{hol}_{\scripty{t}}^{F}$ depend on the choice of
basepoint? Namely, suppose that two thin marked loops $\g:x \to x$
and $\g':x' \to x'$ are thinly homotopic \emph{without preserving the
marking}%
\footnote{The notion of thin homotopy introduced in Definition
\ref{defn:thinhomotopy} does not make sense when $x \ne x'.$}
(see Definition \ref{defn:thinloop}). Then, how is
$\mathrm{hol}_{\scripty{t}}^{F} (\g)$ related to
$\mathrm{hol}_{\scripty{t}}^{F} (\g')$? 
\item
How does $\mathrm{hol}_{\scripty{t}}^{F}$ depend on $F$? Namely, suppose
that $\eta : F \to F'$ is a morphism of transport functors. How is
$\mathrm{hol}_{\scripty{t}}^{F}$ related to
$\mathrm{hol}_{\scripty{t}}^{F'}$ in terms of $\eta$?
\item
How does $\mathrm{hol}_{\scripty{t}}^{F}$ depend on $\scripty{t}\;,$
i.e. the choices of $c$ and $s^{\pi}$? Namely, suppose that
$\scripty{t}\;\;{}'$ is another trivialization. Then how is 
$\mathrm{hol}_{\scripty{t}}^{F}$ related to
$\mathrm{hol}_{\scripty{t}\;{}'}^{F}$? 
\end{enumerate}

We first define what we mean by the thin free loop space and then we
proceed to answer the above questions. 
Denote the smooth space of loops in $M$ by
$LM := \{ \g : S^1 \to M \; | \; \g \text{ smooth}\}.$

\begin{defn}
\label{defn:thinloop}
Two smooth loops $\g, \g' \in LM$ are \emph{\uline{thinly homotopic}} if
there exists a smooth map $h: S^1 \times [0,1] \to M$ such that 
\begin{enumerate}[i)]
\item
there exists an $\e > 0$ with $h ( t , s ) = \g ( t)$ for $s \le \e$ and
$h (t, s ) = \g' (t)$ for $s \ge 1- \e$ and for all $t \in S^1$ and  
\item
the smooth map $h$ has rank $\le 1.$ 
\end{enumerate}
Such a smooth map $h$ is called an \emph{\uline{unmarked thin homotopy}}. 
The smooth space of such thin homotopy classes of loops is denoted by
$L^{1} M$ and is called the \emph{\uline{thin free loop space of $M$}}.
Elements of $L^{1}M$ are called \emph{\uline{thin free loops}} or just
\emph{\uline{thin loops}}. 
\end{defn}

The first condition guarantees that unmarked thin homotopy defines an
equivalence relation and $L^1 M$ is well-defined. The second condition
is where the thin structure is buried. We need to discuss a few
definitions and facts before relating thin loops to thin marked loops.
For the purposes of being absolutely clear, from Lemma
\ref{lem:forgetmarking} through Lemma \ref{lem:markingforloops} we will
distinguish between representatives of loops and thin homotopy
equivalence classes by using brackets $[ \ ].$ However, afterwards, we
will abuse notation and will rarely make the distinction. 

\begin{defn}
\label{defn:forgetmarking}
The function $f : \mathfrak{L} M \to LM$ defined by sending a marked
loop $\g : [0,1] \to M$ to the associated map $f(\g) : S^1 \to M$
obtained from identifying the endpoints of $[0,1]$ is called the
\emph{\uline{forgetful map}}. 
\end{defn}

\begin{lem}
\label{lem:forgetmarking}
There exists a unique map $f^{1} : \mathfrak{L}^{1}M \to L^1 M$ such
that the diagram 
\be
\xy 0;/r.25pc/:
(-10,7.5)*+{\mathfrak{L}M}="1";
(10,7.5)*+{\mathfrak{L}^{1} M}="2";
(-10,-7.5)*+{LM}="3";
(10,-7.5)*+{L^1 M}="4";
{\ar"1";"2"};
{\ar"1";"3"_{f}};
{\ar"3";"4"};
{\ar"2";"4"^{f^{1}}};
\endxy
\ee
commutes (the horizontal arrows are the projections onto thin homotopy
classes). 
\end{lem}

\begin{proof}
The map is constructed by choosing a representative, applying $f,$ and
then projecting to $L^1 M.$ Let $[ \g ] : x \to x$ be an element of
$\mathfrak{L}^{1} M$ and let $\g : x \to x$ and $\g' : x \to x$ be two
representatives in $\mathfrak{L} M.$ Then there exists a thin homotopy
$h : [0,1] \times [0,1] \to M$ from $\g$ to $\g'.$ Because $h(t,s) = x$
for all $s \in [0,1]$ and all $t \in [0,\e] \cup [1-\e, 0]$ for some
$\e > 0,$ the two ends of the first $[0,1]$ factor can be identified
resulting in a smooth map $\tilde{h} : S^1 \times [0,1].$ This gives
the desired homotopy from $f(\g)$ to $f(\g').$ 
\end{proof}

Note that there is also a function
$\mathrm{ev}_{0} : \mathfrak{L}^{1} M \to M$ given by evaluating a thin
loop at its endpoint. This function forgets the loop and remembers only
the basepoint. 

\begin{defn}
\label{defn:loopmarking}
A \emph{\uline{marking of thin loops}} is a section (\emph{not} necessarily
smooth) $\mathfrak{m} : L^1 M \to \mathfrak{L}^{1} M $ of $f^{1} :
\mathfrak{L}^{1} M \to L^1 M,$ i.e. $f^{1} \circ \mathfrak{m} =  \id.$ 
\end{defn}

\begin{rmk}
A marking of ordinary loops cannot be defined in this way as a section
of $f : \mathfrak{L} M \to LM$ because an arbitrary smooth map
$S^1 \to M$ need not have a sitting instant at any point. 
\end{rmk}

\begin{prop}
A marking of thin loops exists. 
\end{prop}

Actually, much more is true. Because the fact is somewhat surprising and
interesting (and only holds due to the thin homotopy equivalence
relation), we include it here. Let $\pi_{0} M$ denote the set of
components of $M$ and $p : M \to \pi_{0} M$ the canonical function
sending a point to its component. Let $c_0 : L^1 M \to \pi_0 M$ denote
the canonical function sending a thin loop to the component in which it
(every representative) lies. A marking of thin loops
$\mathfrak{m}$ determines a function $\b : L^1 M \to M$ given by
$\b := \mathrm{ev}_{0} \circ \mathfrak{m}$ that satisfies the
condition that 
\be
\label{eq:choosingpointsforthinloops}
\xy 0;/r.25pc/:
(-10,-7.5)*+{L^1 M}="1";
(10,7.5)*+{M}="2";
(10,-7.5)*+{\pi_{0}M}="3";
{\ar"1";"2"^{\b}};
{\ar"1";"3"_{c_0}};
{\ar"2";"3"^{p}};
\endxy
\ee 
commutes. 

\begin{lem}
\label{lem:markingforloops}
Let $\b : L^1 M \to M$ be \emph{any} function such that the diagram in
(\ref{eq:choosingpointsforthinloops}) commutes. Then there exists a
marking of thin loops $\mathfrak{m} :  L^1 M \to \mathfrak{L}^{1} M$
such that the diagram 
\be
\label{eq:thinloopmarkingev}
\xy 0;/r.25pc/:
(-10,7.5)*+{\mathfrak{L}^{1} M}="2";
(-10,-7.5)*+{L^1 M}="1";
(10,7.5)*+{M}="3";
{\ar"1";"2"^{\mathfrak{m}}};
{\ar"2";"3"^{\mathrm{ev}_{0}}};
{\ar"1";"3"_{\b}};
\endxy
\ee
commutes. 
\end{lem}

\begin{proof}
A function $\mathfrak{m}$ can be defined as follows. For any thin loop
$[ \g ] \in L^1 M,$ let $\g : S^1 \to M$ be a representative. Then there
exists an unmarked thin homotopy $h$ from $\g$ to a loop $\g_{\b}$ with
sitting instants at $\b ([\g])$ because
(\ref{eq:choosingpointsforthinloops}) commutes. To see this, one can
simply pick a point on the loop and extend the loop out to the basepoint
and come back without sweeping out any area (see Figure
\ref{fig:markingloops}). 
\begin{figure}[h!]
\centering
	\begin{picture}(0,0)
	\put(115,95){$h$}
	\put(175,75){$\g$}
	\put(230,55){$\g_{\b}$}
	\put(275,5){$x$}
	\end{picture}
    \includegraphics[width=0.60\textwidth]{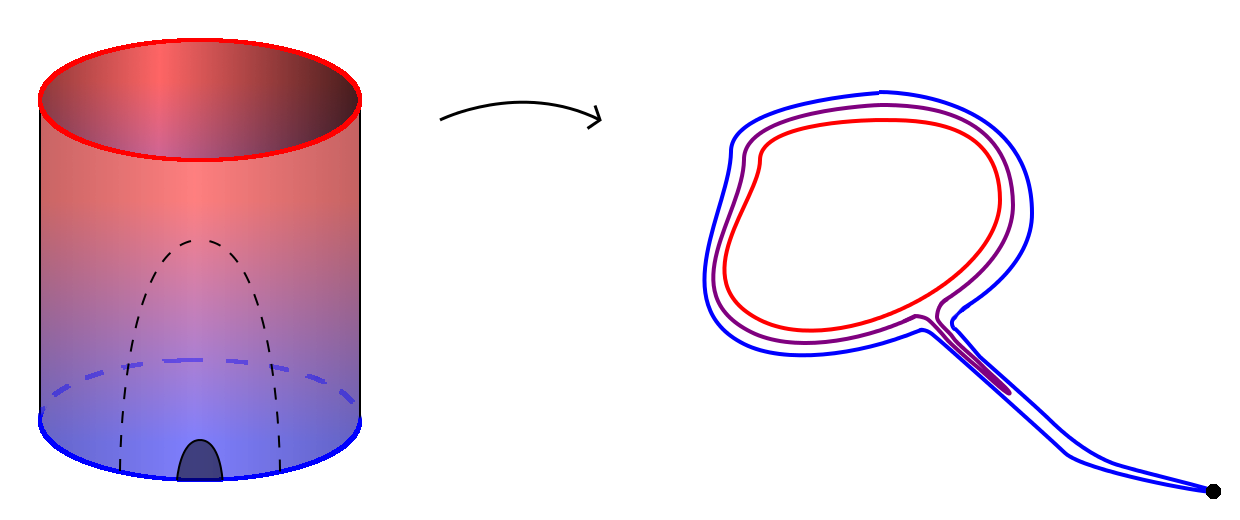}
    \vspace{-3mm}
    \caption{Let $[\g]$ be a thin free loop, $x:=\b([\g])$ a point in
the same connected component as $[\g],$ and $\g$ a representative loop
(in red). Then there exists a path $\g' : x \rightarrow x$ with sitting instants
(in blue) and an unmarked thin homotopy $h : \g \Rightarrow \g'.$ The
cylinder depicts such a homotopy with the middle loop (in purple)
indicating an intermediate loop. The dashed line on the cylinder
indicates that the loops begin to extend outwardly towards the marking
without sweeping any area. The ``mouse-hole'' on the cylinder indicates
that the loops from the homotopy eventually sit at $x.$}
    \label{fig:markingloops}
\end{figure}
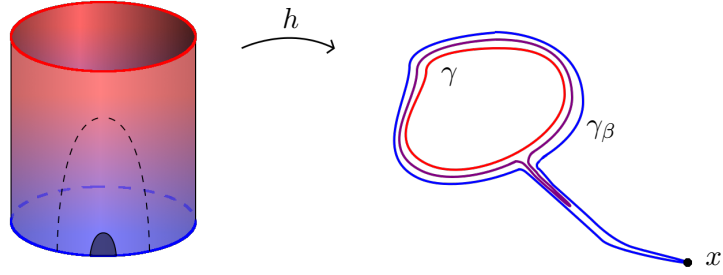 
Then project $\g_{\b}$ to $\mathfrak{L}^{1} M.$ Thus, set
$\mathfrak{m}([\g]) := [\g_{\b}].$ To see that this is well-defined, let
$\g'$ be another representative of $[\g]$ and let $\tilde{h}$ be an
unmarked thin homotopy from $\g'$ to $\g.$ Then composing the two
unmarked thin homotopies $h \circ \tilde{h}$ gives an unmarked thin
homotopy from $\g'$ to $\g_{\b}.$ Of course, there are many possible
choices for $\g_{\b}$ for a given $\b$ that will give different markings
$\mathfrak{m}.$ 
\end{proof}

\begin{rmk}
If $\b$ is chosen so that the diagram in
(\ref{eq:choosingpointsforthinloops}) does not commute, a marking
$\mathfrak{m}$ satisfying (\ref{eq:thinloopmarkingev}) does not exist.
\end{rmk}

We now proceed to answering the above questions in order. 

\begin{enumerate}[i)]
\item
Let $\mathfrak{m} ,\mathfrak{m}' : L^{1} M \to \mathfrak{L}^{1} M$ be
two markings of thin loops in $M.$ Let $[\g]\in L^{1}M$ and denote
$x:=\mathrm{ev}_{0}(\mathfrak{m}([\g]))$ and
$x' :=\mathrm{ev}_{0}(\mathfrak{m}'([\g])).$ A choice of representatives
$\g:x\to x$ and $\g':x'\to x'$ as paths with sitting instants of
$\mathfrak{m} ( [ \g ] )$ and $\mathfrak{m}' ( [ \g ] ),$ respectively,
need not have the same image. In particular, $x$ and $x'$ might not lie
on each others images. Figure \ref{fig:twoloops} gives an example. This
makes it impossible to compare their holonomies using thin bigons in the
usual way (because no such bigon exists). 

\begin{figure}[h!]
\centering
	\begin{picture}(0,0)
	\put(18,13){$x$}
	\put(185,95){$x'$}
	\put(120,116){$\g'$}
	\put(48,13){$\g$}
	\end{picture}
    \includegraphics[width=0.4\textwidth]{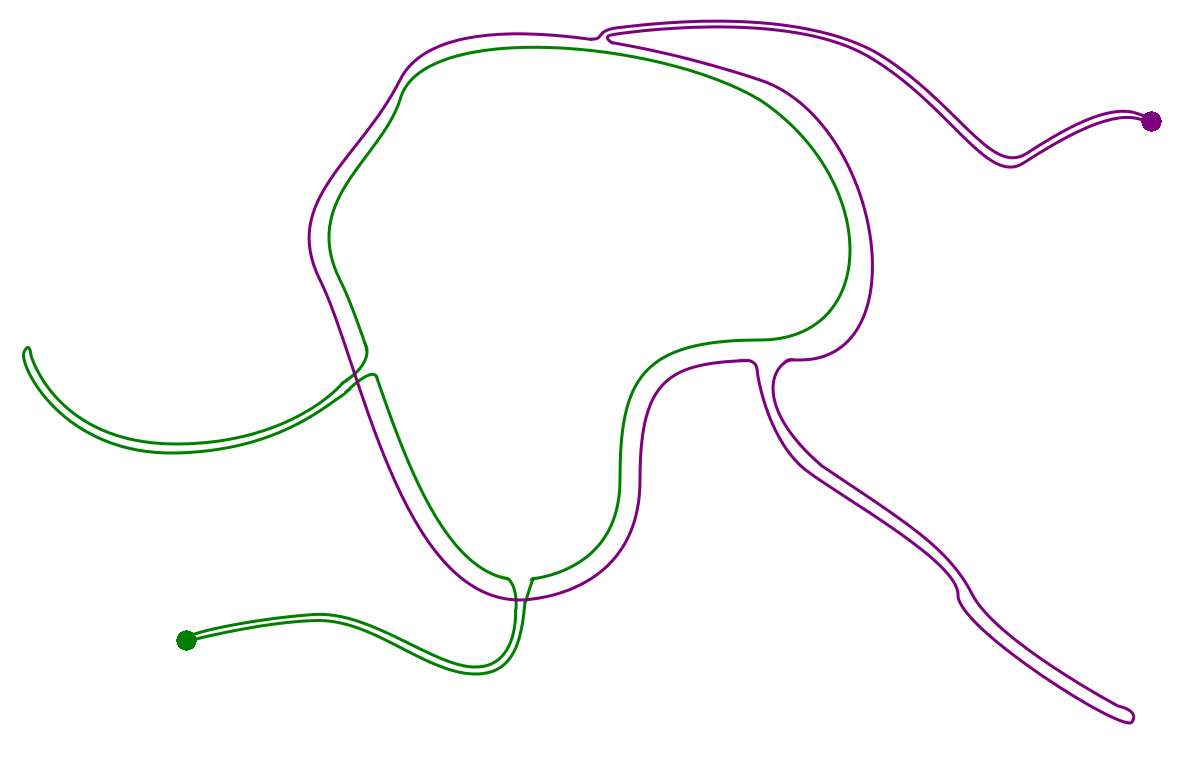}
    \vspace{-3mm}
    \caption{Two representatives $\g$ and $\g'$ of two markings of a
single thin loop are shown. Their respective basepoints $x$ and $x'$ do
not lie on each others images.}
    \label{fig:twoloops}
\end{figure}
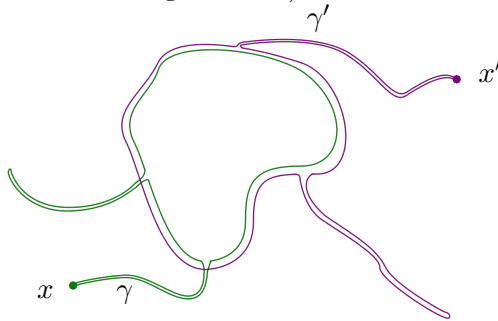 

However, there is an \emph{unmarked} thin homotopy
$h : S^{1} \times [0,1] \to M$ with $h (t,s) = \g (t)$ for $s \le \e$
and $h(t,s) = \g' (t)$ for $s \ge 1-\e$ for some $\e > 0.$ Therefore,
one can choose a loop $\tilde{\g}$ and two \emph{paths} with sitting
instants $\g_{x' x} : x \to x'$ and $\g_{x x'} : x' \to x$ with the
following three properties. First, as a loop, $\tilde{\g}$ can be
written as the composition $\g_{x' x} $ and $\g_{x x'}$ in some order,
i.e. using the map $f$ of Definition \ref{defn:forgetmarking},
$\tilde{\g} = f(\g_{x' x} \circ \g_{x x'})$ or
$f(\g_{x x'} \circ \g_{x' x})$.
Second, the composition $\g_{x x' } \circ \g_{x' x}$ is thinly homotopic
to $\g$ preserving the basepoint $x.$  Third, the composition
$\g_{x' x } \circ \g_{x x'}$ is thinly homotopic to $\g'$ preserving the
basepoint $x'.$ This is depicted in Figure \ref{fig:thinhomotopyloop}. 

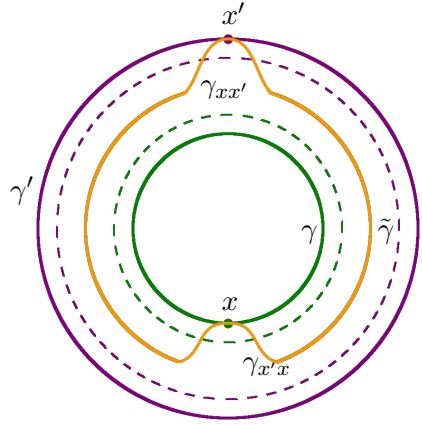
\begin{figure}[h]
\centering
	\begin{picture}(0,0)
	\put(80,48){$x$}
	\put(80,156){$x'$}
	\put(0,90){$\g'$}
	\put(110,75){$\g$}
	\put(139,75){$\tilde{\g}$}
	\put(88,25){$\g_{x' x}$}
	\put(72,130){$\g_{x x'}$}
	\end{picture}
    \includegraphics[width=0.35\textwidth]{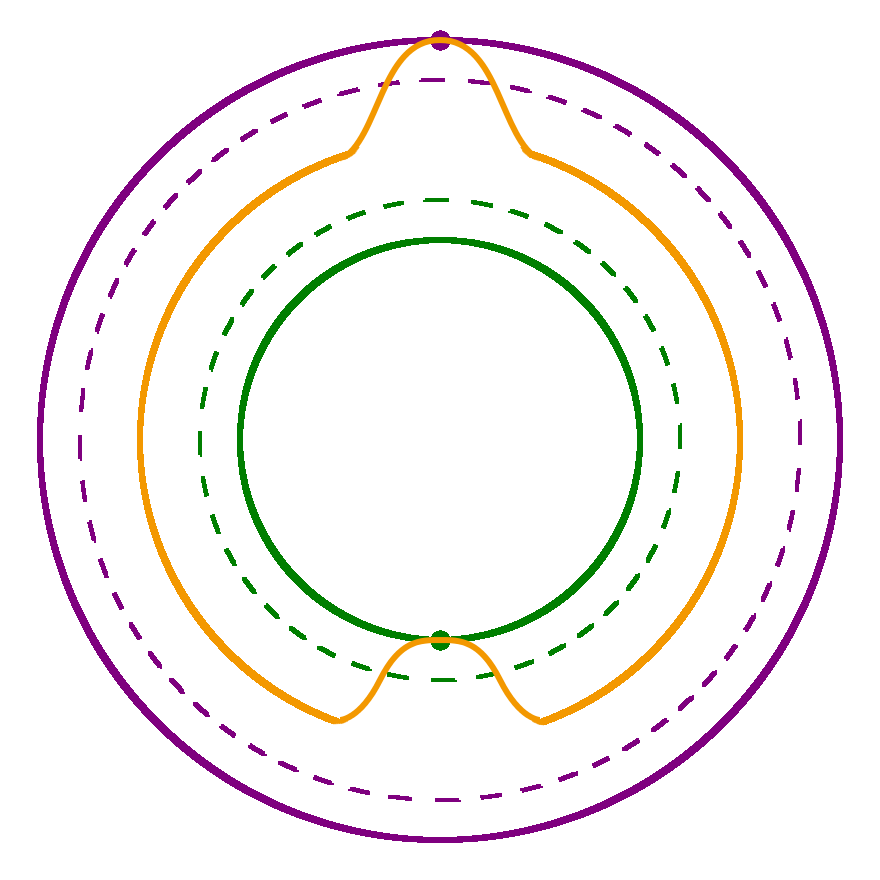}
    \vspace{-3mm}
    \caption{The domain of the unmarked thin homotopy
$h : S^{1} \times [0,1] \rightarrow M$ 
is drawn as an annulus depicting the
domain of $\g$ as the inner circle and that of $\g'$ as the outer
circle. The homotopy allows us to choose a loop $\tilde{\g},$ drawn
somewhat in the middle (in orange), that contains both $x$ and $x'$
and is thinly homotopic to both $\g$ and $\g'.$ This loop $\tilde{\g}$
is decomposed into two paths
$\g_{x'x} : x \rightarrow x'$ and $\g_{x x'} : x' \rightarrow x.$ The dashed lines
indicate the regions of sitting instants. All paths are oriented
counter-clockwise. Note that, by a reparametrization if necessary, the
homotopy $h$ may be chosen to separate the two basepoints into the
northern and southern hemispheres as drawn.}
    \label{fig:thinhomotopyloop}
\end{figure} 

This says that given two marked loops, with possibly different markings,
that are thinly homotopic \emph{without} preserving the marking, one can
always choose a representative of such a thin loop in $M$ with
\emph{two} marked points so that the associated two \emph{marked} loops
(coming from starting at either marking) are thinly homotopic to the
original two with a thin homotopy that preserves the marking. More
precisely, we proved the following fact. 

\begin{lem}
\label{lem:choosethinmarkedloop}
Let $\mathfrak{m},\mathfrak{m}': L^{1} M \to \mathfrak{L}^{1} M$ be two
markings. Let $[ \g ] \in L^{1} M$ be a thin loop in $M$ and write
$x:= \mathrm{ev}_{0} ( \mathfrak{m} ( [ \g ] ) )$ and
$x' := \mathrm{ev}_{0} ( \mathfrak{m}' ( [ \g ] ) ).$ 
Then, there exist two paths $\g_{x' x} : x \to x'$ and
$\g_{x x'} : x' \to x$ with sitting instants such that the following
three properties hold (see Figure \ref{fig:gammaxy}). 
\begin{enumerate}[i)]
\item
The composition of $\g_{x x' }$ and $\g_{x' x}$ (in either order) and
forgetting the marking is a representative of $[\g].$
\item
$\g_{x x' } \circ \g_{x' x}$ is a representative of
$\mathfrak{m} ([\g])$ as a path with sitting instants.  
\item
$\g_{x' x } \circ \g_{x x'}$ is a representative of
$\mathfrak{m}' ([\g])$ as a path with sitting instants.
\end{enumerate}
\end{lem}

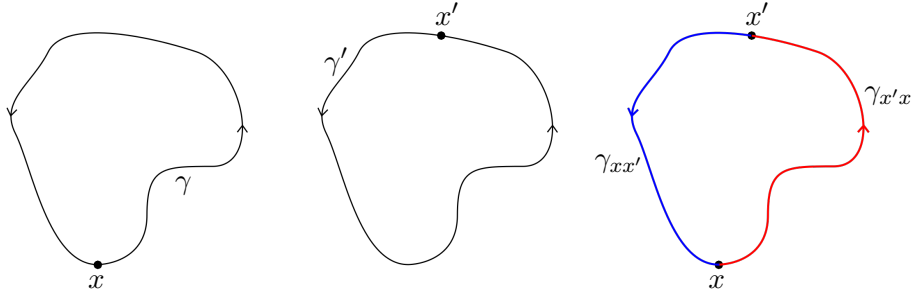
\begin{figure}[h]
	\centering
	\begin{picture}(0,0)
	\put(70,35){$\g$}
	\put(37,-2.5){$x$}
	\end{picture}
	\begin{subfigure}[b]{0.22\textwidth}
		\includegraphics[width=\textwidth]{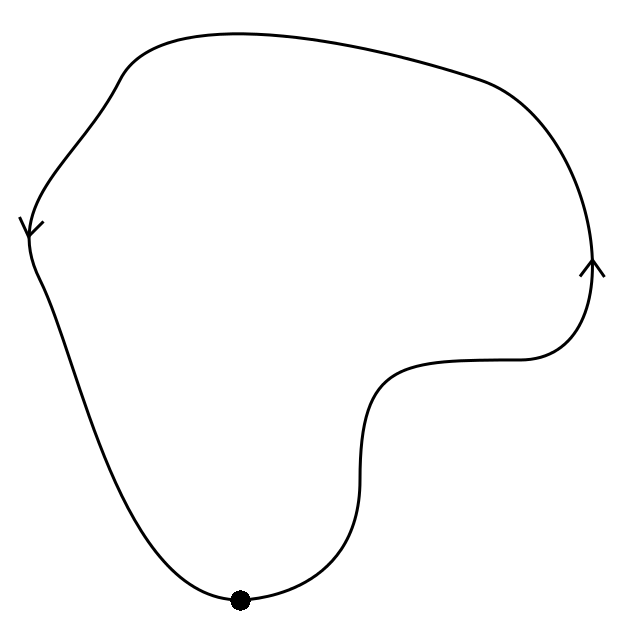}
	\end{subfigure}
	\quad
	\begin{picture}(0,0)
	\put(9,80){$\g'$}
	\put(51,97){$x'$}
	\end{picture}
	\begin{subfigure}[b]{0.22\textwidth}
		\includegraphics[width=\textwidth]{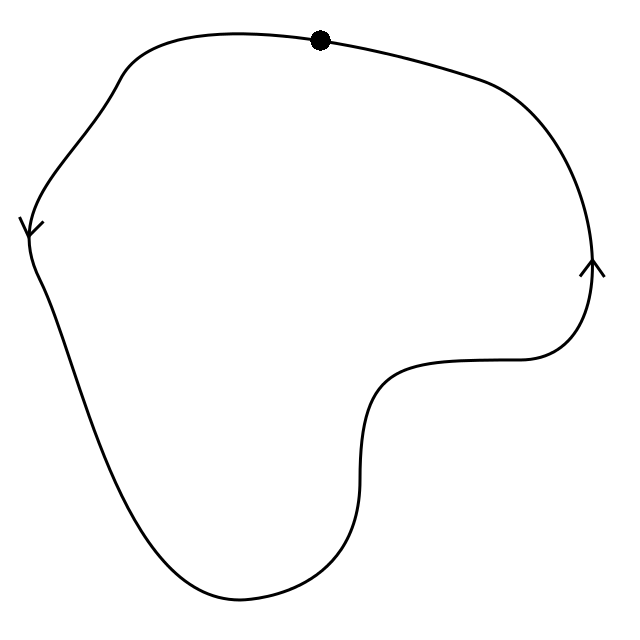}
	\end{subfigure}
	\quad
	\begin{picture}(0,0)
	\put(96,69){$\g_{x'x}$}
	\put(-6,43){$\g_{xx'}$}
	\put(51,97){$x'$}
	\put(37,-2.5){$x$}
	\end{picture}
	\begin{subfigure}[b]{0.22\textwidth}
		\includegraphics[width=\textwidth]{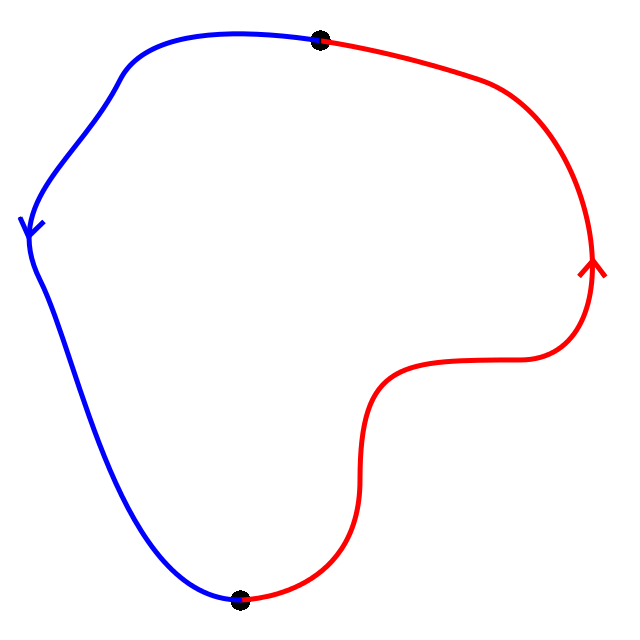}
	\end{subfigure}
	\caption{For two markings with associated basepoints $x$ and
$x'$ of a thin loop $[\g],$ there exist representatives paths with
sitting instants (shown on the right) $\g_{x'x} : x \rightarrow x'$ (in red)
and $\g_{xx'} : x' \rightarrow x$ (in blue) such that
$\g := \g_{xx'} \circ \g_{x'x}$ (shown on the left) represents one
marking and $\g' := \g_{x'x} \circ \g_{xx'}$ (shown in the middle)
represents the other. Note that $\g'$ and $\overline{\g_{xx'}} \circ \g
\circ \g_{xx'}$ are thinly homotopic.}
	\label{fig:gammaxy}
\end{figure}

Therefore, without loss of generality, we can choose a \emph{single}
representative $\tilde{\g}$ of a thin free loop $[\g]$ with a
decomposition as in the Lemma. We denote $\g':=\g_{x' x } \circ
\g_{x x'}$ and $\g:=\g_{x x'} \circ \g_{x' x}.$ Thus $\tilde{\g}$ is one
of $f(\g)$ or $f(\g').$  Note that $\g'$ and $\overline{\g_{xx'}} \circ
\g \circ \g_{xx'}$ are thinly homotopic. For convenience, from now on we
abuse notation often and do not distinguish between the actual paths
versus the thin homotopy classes as elements of $P^{1}M.$ 

By functoriality of the transport functor $\scripty{t}_{F},$ we have
\be
\begin{split}
\mathrm{hol}_{\scripty{t}}^{F} (\g') &= \scripty{t}_{F} (\g') \\
	&=\scripty{t}_{F}(\overline{\g_{xx'}}\circ\g\circ\g_{xx'}) \\
	&= \scripty{t}_{F} ( \overline{\g_{xx'}} ) \scripty{t}_{F} (\g)
\scripty{t}_{F} ( \g_{xx'} ) \\
	&= ( \scripty{t}_{F} (\g_{xx'} ) )^{-1}
\mathrm{hol}_{\scripty{t}}^{F} (\g) \scripty{t}_{F} (\g_{xx'} )
\end{split}
\ee
so that $\mathrm{hol}_{\scripty{t}}^{F} $ changes by conjugation in $G$
when the marking is changed. 

\item
Suppose that $\h : F \to F'$ is a morphism of transport functors. Then,
for every thin path $\g : x \to y$  we have a commutative diagram 
\be
\xy 0;/r.15pc/:
(20,-15)*+{\scripty{t}_{F}(y)}="3";
(-20,-15 )*+{\scripty{t}_{F'}(y)}="4";
(-20,15)*+{\scripty{t}_{F'}(x)}="2";
(20,15)*+{\scripty{t}_{F}(x)}="1";
{\ar_{\scripty{t}_{\eta} (x) } "1";"2" };
{\ar_{\scripty{t}_{F'} (\g) } "2";"4" };
{\ar^{\scripty{t}_{F} (\g) } "1";"3"};
{\ar^{\scripty{t}_{\eta} (y)} "3";"4"};
\endxy
,
\ee
which says
\be
\scripty{t}_{\eta} (y) \scripty{t}_{F} (\g) = \scripty{t}_{F'} (\g)
\scripty{t}_{\eta} (x). 
\ee
If we restrict this to a thin marked loop $\g$ with $y=x,$ then 
\be
\mathrm{hol}_{\scripty{t}}^{F'} (\g) = ( \scripty{t}_{\eta} (x) )^{-1}
\mathrm{hol}_{\scripty{t}}^{F} (\g) \scripty{t}_{\eta} (x)
\ee
so that again, $\mathrm{hol}^{F}_{\scripty{t}}$ changes under
conjugation when the functor $F$ is changed to an isomorphic one. 

\item
Suppose that another trivialization $\scripty{t}\;\;{}'$ was chosen.
Then following the comments after Remark \ref{rmk:trivialization}, we
can choose natural isomorphisms $\scripty{r} : \id \Rightarrow
\scripty{t}\;$ and $\scripty{r}\;{}' : \id \Rightarrow
\scripty{t}\;\;{}'$ resulting in a natural isomorphism
$\scripty{s} := \begin{smallmatrix} \overline{ \scripty{r}\;{}'}  \\
\overset{\circ}{\scripty{r}\;}\end{smallmatrix} : \scripty{t}\;\;{}'
\Rightarrow \scripty{t}\;.$ This means every transport functor $F$
gets assigned a morphism of transport functors $\scripty{s}_{F}:
\scripty{t}_{F} {}' \to \scripty{t}_{F}$ satisfying naturality. This
means to every $x \in M$ we have a morphism $\scripty{s}_{F} (x) :
\scripty{t}_{F}{}'(x) \to \scripty{t}_{F} (x)$ satisfying naturality,
i.e. to every path $\g : x \to y$ the diagram 
\be
\xy 0;/r.15pc/:
(20,-15)*+{\scripty{t}_{F}{}' (y)}="3";
(-20,-15 )*+{\scripty{t}_{F} (y)}="4";
(-20,15)*+{\scripty{t}_{F} (x)}="2";
(20,15)*+{\scripty{t}_{F}{}' (x)}="1";
{\ar_{\scripty{s}_{F} (x) } "1";"2" };
{\ar_{\scripty{t}_{F} (\g) } "2";"4" };
{\ar^{\scripty{t}_{F}{}'  (\g) } "1";"3"};
{\ar^{\scripty{s}_{F} (y)} "3";"4"};
\endxy
\ee
commutes. In case $\g$ is a thin loop at $x,$ this gives
\be
\mathrm{hol}_{\scripty{t}\;\;{}'}^{F}(\g)=(\scripty{s}_{F}(x))^{-1}
\mathrm{hol}_{\scripty{t}}^{F} (\g) \scripty{s}_{F} (x) .
\ee

\end{enumerate}

In conclusion, the answer to every one of the three questions is
conjugation. This is what is called \emph{gauge covariance}. To get
something gauge invariant, we first denote the quotient map from $G$ to
its conjugacy classes by $q : G \to G / \mathrm{Inn}(G),$ where
$\mathrm{Inn}(G)$ stands for the inner automorphisms of $G$ and the
quotient $G / \mathrm{Inn} (G)$ is given by the conjugation action of
$G$ on itself.  All of the above considerations show that the following
theorem holds. 

\begin{thm}
\label{thm:gaugeinvarianthol}
Let $M$ be a smooth manifold, $G$ be a Lie group, $T$ a category, and
suppose that $i : \mathcal{B} G \to T$ is full and faithful. 
Let $F \in \mathrm{Trans}^{1}_{\mathcal{B}G}(M,T)$ be a transport
functor and $\scripty{t} \;$ a group-valued transport extraction.
Let $L^{1} M, \mathfrak{L}^{1} M, \mathfrak{m},
\mathrm{hol}_{\scripty{t}}^{F}$ and $q$ be defined as above. Then the
composition 
\be
G / \mathrm{Inn} (G) \xleftarrow{q} G
\xleftarrow{\mathrm{hol}^{F}_{\scripty{t}}} \mathfrak{L}^{1} M
\xleftarrow{\mathfrak{m}} L^{1} M
\ee
is 
\begin{enumerate}[i)]
\item
independent of $\mathfrak{m},$ 
\item
independent of the isomorphism class of $F,$
\item
and independent of the isomorphism class of $\scripty{t} \; .$ 
\end{enumerate}
\end{thm}

Notice that this theorem lets us make the following definition.  

\begin{defn}
\label{defn:gaugeinvhol}
Let $[F]$ be an isomorphism class of transport functors. 
The \emph{\uline{gauge invariant holonomy}} of $[F]$ is defined to be the
map in the previous theorem, namely 
\be
\label{eq:gaugeinvarianthol}
\mathrm{hol}^{[F]}
:= q \circ \mathrm{hol}^{F}_{\scripty{t}} \circ \mathfrak{m} : L^{1} M
\to G / \mathrm{Inn} (G)
\ee
where $F$ is a representative of $[F],$ $\scripty{t} \;$ is a 
group-valued transport extraction, 
and $\mathfrak{m} : L^{1} M \to \mathfrak{L}^{1} M$ is a marking of
thin loops in $M.$ Let $\g \in L^{1} M.$ If $\mathrm{hol}^{[F]} (\g)$
is such that $q^{-1} ( \mathrm{hol}^{[F]} (\g))$ is a single element,
we will say that $\mathrm{hol}^{[F]} (\g)$ is
\emph{\uline{gauge invariant}} and abusively write
$\mathrm{hol}^{[F]} (\g)$ instead of $q^{-1} (\mathrm{hol}^{[F]} (\g)).$
\end{defn}

\section{Transport 2-functors and gauge invariant surface holonomy}
\label{sec:review2}

In the present section, we review the basics of transport 2-functors and
also provide some new and interesting results. As a preliminary, we
briefly set our notation and review some facts about (strict) 2-groups
and crossed modules. Then we split up the discussion into several parts
and follow a similar pattern to the transport functor case. However,
since we are now aware of what local triviality should mean, we skip the
guess-work and head straight to the correct theory. We start with a
\v Cech description of ordinary principal (strict) 2-group 2-bundles
(without connection) in terms of smooth 2-functors. We then discuss how
to add connection data by introducing transport functors, local
triviality, and descent data. The discussion of the reconstruction
functor is more involved, and because it is important for the
calculation, we spend some time on it. Nevertheless, we skip some
technical details (such as compositors and unifiers). Then we consider
the differential cocycle data and discuss a formula for higher holonomy
in terms of an iterated surface integral. We summarize the results as
before. Sections \ref{sec:2-groups} through \ref{sec:2limit} are a
summary of \cite{SW2}, \cite{SW3}, and \cite{SW4}. 

Finally, in Section \ref{sec:2-holonomy}, we discuss some results on
surface holonomy and its gauge covariance. We introduce a notion of
\emph{$\a$-conjugacy classes} for a 2-group in Definition
\ref{defn:alphaconj} and prove in Theorem
\ref{thm:invariancesurfaceholonomy} that surface holonomy along
\emph{spheres} is well-defined in $\a$-conjugacy classes generalizing
the reduced group of \cite{SW4} (it is not yet known whether this
generalization will work for more general surfaces). In the process, the
procedure of group-valued transport extraction is categorified for the
purposes of (i) proving this theorem and (ii) providing a functorial
description for computing transport locally, which we utilize in Section
\ref{sec:examples}. 

We assume the reader is familiar with the basics of 2-categories. A
review sufficient for most of our purposes can be found in Appendix A of
\cite{SW3}. 

\subsection{2-group conventions} 
\label{sec:2-groups}

The theory of 2-groups is discussed in great detail in the article
\cite{BL}. However, to simplify the discussion, we will define a
(strict) \emph{\uline{2-group}} as a strict one-object 2-groupoid, i.e. a
strict 2-category with inverses for all 1- and 2-morphisms. Normally,
one defines a 2-group as a groupal groupoid as in \cite{BL}, but we find
this unnecessary. However, to be consistent with notation in the
literature, we will write our 2-groups as $\mathcal{B}\mathfrak{G}$ and
use the notation $\mathfrak{G}$ where appropriate. 

There is a 2-category of strict 2-groups denoted by $2$-$\mathrm{Grp}$
whose 1-morphisms and 2-morphisms are functors and natural
transformations, respectively. It is useful to relate this
higher-categorical definition to one involving ordinary groups.
Although this is standard, we set the notation, which may differ from
some authors. 

\begin{defn}
\label{defn:crossedmodule}
A \emph{\uline{crossed module}} is a quadruple $(H, G , \t, \a)$ of two
groups, $G$ and $H,$ and group homomorphisms $\t : H \to G$ and
$\a : G \to \mathrm{Aut} (H)$ satisfying the two conditions 
\be
\label{eq:Pfeiffer}
\a_{\t(h)} (h') = h h' h^{-1} 
\ee
and
\be
\label{eq:crossedmodule2}
\t ( \a_{g} ( h) ) = g \t ( h ) g^{-1} . 
\ee
\end{defn}
In this definition, $\mathrm{Aut} (H)$ is the automorphism group of $H.$
The collection of crossed modules form the objects of a 2-category
$\mathrm{CrsMod}.$ 

\begin{thm}
\label{thm:CrsMod2group}
The 2-categories  $\mathrm{CrsMod}$ and $2$-$\mathrm{Grp}$ are
equivalent. 
\end{thm}

This theorem has been known for quite some time in several different
forms. A simple place to start for this is in the article \cite{BaHu}
with more information in \cite{BL}. 

\begin{proof}
We only prove the equivalence at the level of objects and in only one
direction. This will set up our conventions throughout the paper. Given
a crossed module $(H, G, \t, \a)$ the associated 2-group $\mathcal{B}
\mathfrak{G}$ is defined to have a single object $\bullet,$ $G$ as its
set of 1-morphisms, and $H \rtimes G$ as its set of 2-morphisms.
Composition of 1-morphisms is given by multiplication in $G.$ The source
and target maps of 2-morphisms are defined pictorially by 
\be
\label{eq:(h,g)}
\xymatrix{
\bullet & & \ar@/_1.5pc/[ll]_{g}="1" \ar@/^1.5pc/[ll]^{\t(h) g}="2"
\ar@{}"1";"2"|(.15){\,}="3" \ar@{}"1";"2"|(.85){\,}="4"
\ar@{=>}"3";"4"^{(h,g)}  \bullet
}
.
\ee
Vertical and horizontal compositions are defined by 
\be
\label{eq:vertical}
\xymatrix{
\bullet & & & \ar@/_3pc/[lll]_{g}="1" \ar[lll]|-{\t(h) g}="2"
\ar@{}"1";"2"|(.25){\,}="3" \ar@{}"1";"2"|(.85){\,}="4"
\ar@{=>}"3";"4"^{(h,g)}  \ar@/^3pc/[lll]^{\t(\tilde{h}) \t(h) g}="5"
\ar@{}"2";"5"|(.15){\,}="6" \ar@{}"2";"5"|(.75){\,}="7"
\ar@{=>}"6";"7"^(0.3){(\tilde{h}, \t(h) g)}   \bullet
}
\quad \mapsto \quad 
\xymatrix{
\bullet & & \ar@/_1.5pc/[ll]_{g}="1"
\ar@/^1.5pc/[ll]^{\t ( \tilde{h} h ) g }="2" \ar@{}"1";"2"|(.15){\,}="3"
\ar@{}"1";"2"|(.85){\,}="4" \ar@{=>}"3";"4"^{(\tilde{h} h, g)}  \bullet
}
\ee
and
\be
\label{eq:horizontal}
\xymatrix{
\bullet & & \ar@/_1.5pc/[ll]_{g'}="5" \ar@/^1.5pc/[ll]^{\t(h') g'}="6"
\ar@{}"5";"6"|(.15){\,}="7" \ar@{}"5";"6"|(.85){\,}="8"
\ar@{=>}"7";"8"^{(h',g')}  \bullet & & \ar@/_1.5pc/[ll]_{g}="1"
\ar@/^1.5pc/[ll]^{\t(h) g}="2" \ar@{}"1";"2"|(.15){\,}="3"
\ar@{}"1";"2"|(.85){\,}="4" \ar@{=>}"3";"4"^{(h,g)}  \bullet
}
\quad
\mapsto 
\quad
\xymatrix{
\bullet & & \ar@/_1.5pc/[ll]_{g' g}="1"
\ar@/^1.5pc/[ll]^{\t(h') g' \t(h) g}="2" \ar@{}"1";"2"|(.15){\,}="3"
\ar@{}"1";"2"|(.85){\,}="4" \ar@{=>}"3";"4"|-{(h' \a_{g'} (h) ,g' g)}
\bullet
},
\ee
respectively. When writing 2-group multiplication, we will always drop
the composition symbol $\circ,$ which is a common practice for ordinary
group multiplication. 
\end{proof}

The above proof sets up our convention for 2-group multiplication.
Equation (\ref{eq:(h,g)}) shows that what's needed to specify a
2-morphism is an element of $G,$ the source of the 2-morphism, and an
element of $H.$ Thus, if the source is already known, the element in
$H$ specifies the 2-morphism. Equation (\ref{eq:vertical}) defines
vertical composition and equation (\ref{eq:horizontal}) defines
horizontal composition. Please be aware that different authors have
different conventions (since the 2-categories $\mathrm{CrsMod}$ and
$2$-$\mathrm{Grp}$ are equivalent in many ways). 

The following is a simple but important fact (which we use in studying
gauge invariance, mainly Corollary \ref{cor:gaugeinvariantflux}). 

\begin{lem}
\label{lem:central}
Let $(H, G , \t, \a)$ be a crossed module. Then
$\ker \t := \{ h \in H \ | \ \t(h) = e \}$ is a central subgroup of $H.$
\end{lem}

\begin{proof}
Let $k \in \ker \t$ and $h \in H.$ Then 
\be
k h = k h k^{-1} k = \a_{\t(k)}(h) k = \a_{e}(h) k = h k.
\ee
\end{proof}

\begin{defn}
A \emph{\uline{Lie crossed module}} is a crossed module $(H,G,\t,\a)$ with
$G$ and $H$ Lie groups and where $\t$ and $\a$ are smooth maps, where
$\a$ being smooth technically means that the adjoint map
$G \times H \to H$ is smooth.
\end{defn}

\begin{defn}
\label{defn:Lie2groupoid}
A \emph{\uline{Lie 2-groupoid}} is a strict 2-category $\mathrm{Gr}$ whose
objects, 1-morphisms, and 2-morphisms are all smooth spaces and all
structure maps are smooth. Furthermore, all 1- and 2-morphisms are
invertible and the inversion maps are all smooth. 
\end{defn}

\begin{defn}
A \emph{\uline{Lie 2-group}} is a Lie 2-groupoid with a single object.
\end{defn}

\begin{rmk}
Lie crossed modules form the objects of a 2-category 
and Lie 2-groups form the objects of a 2-category. 
A similar proof shows that these 2-categories are also equivalent. 
\end{rmk}

\subsection{A \v Cech description of principal $\mathfrak{G}$-2-bundles}
\label{sec:2Cech}

Let $\mathcal{B}\mathfrak{G}$ be a Lie 2-group and denote the associated
crossed module by $(H,G,\t,\a).$ Principal $\mathfrak{G}$-2-bundles over
a manifold $M$ can be described in terms of 2-functors using the \v Cech
groupoid as well (this also comes from Remark II.13. of \cite{Wo}).
However, since we are dealing with 2-categories we need to slightly
modify the \v Cech groupoid of Definition \ref{defn:Cechgroupoid}. The
way we do this is just by throwing on identity 2-morphisms. In other
words, given an open cover $\{ U_{i} \}_{i \in I}$ of $M,$ a 2-morphism
from $(x,i,j)$ to $(x',i',j')$ exists only if $x' = x, i' =i,$ and
$j' = j$ and in this case there is only the identity 2-morphism.
Composition is uniquely defined by this. This defines the
\emph{\uline{\v Cech 2-groupoid}}, also written as $\mathfrak{U}.$ This is
a Lie 2-groupoid. 

\begin{defn}
2-functors between Lie 2-groupoids are \emph{\uline{smooth}} if they assign
data smoothly. Similarly, pseudonatural transformations and
modifications are smooth when the assignments defining them are smooth. 
\end{defn}

Any smooth 2-functor $\mathfrak{U} \to \mathcal{B} \mathfrak{G}$ gives
the \v Cech cocycle data of a principal $\mathfrak{G}$-2-bundle over $M$
subordinate to the cover $\{ U_{i} \}_{i \in I}.$ To see this, simply
recall what a 2-functor does (see Definition A.5. of \cite{SW3}). To
each object $(x,i)$ in $\mathfrak{U}$ it assigns the single object
$\bullet$ in $\mathcal{B} \mathfrak{G}.$ To each jump $(x,i,j),$ it
assigns an element denoted by $g_{ij} (x) \in G$ in such a way that we
get a smooth 1-cochain $g_{ij} : U_{ij} \to G$ as in Section
\ref{sec:1Cech}. However, to each triple intersection $U_{ijk},$ which
corresponds to the composition of $U_{ij}$ with $U_{jk},$ it assigns an
element $f_{ijk} (x) \in H$ in such a way that we get a smooth 2-cochain
$f_{ijk} : U_{ijk} \to H$
\be
\xy 0;/r.25pc/:
(-13,-10)*{k};
(0,12)*{j};
(13,-10)*{i};
(-10,-9)*\xycircle(15,15){.};
(10,-9)*\xycircle(15,15){.};
(0,8.32)*\xycircle(15,15){.};
(-10,-9)*{\bullet}="k";
(0,8.32)*{\bullet}="j";
(10,-9)*{\bullet}="i";
"i";"j"**\dir{-} ?(0.55)*\dir{>}+(3,0)*{ij};
"j";"k"**\dir{-} ?(0.55)*\dir{>}+(-3,2)*{jk};
"i";"k"**\dir{-} ?(0.55)*\dir{>}+(1,-3)*{ik};
\endxy
\qquad
\mapsto
\qquad  
\xy 0;/r.30pc/:
(-10,-9)*{\bullet}="k";
(0,8.32)*{\bullet}="j";
(10,-9)*{\bullet}="i";
"i";"j"**\dir{-} ?(0.55)*\dir{>}+(4,0)*{g_{ij}};
"j";"k"**\dir{-} ?(0.55)*\dir{>}+(-3,2)*{g_{jk}};
"i";"k"**\dir{-} ?(0.55)*\dir{>}+(1,-3)*{g_{ik}};
{\ar@{=>}|-(0.66){(f_{ijk}, g_{jk} g_{ij})} (0,6);(0,-8)};
\endxy
,
\ee
which says
\be
\t ( f_{ijk} ) g_{jk} g_{ij} = g_{ik}.
\ee
The 2-functor satisfies an associativity condition which is translated
into a condition on quadruple intersections giving a
``cocycle condition'' 
\be
\xy 0;/r.18pc/:
(-30,0)*+{\bullet}="k";
(-8,-30)*+{\bullet}="l";
(30,-10)*+{\bullet}="i";
(-1,40)*+{\bullet}="j";
"i";"j"**\dir{-} ?(0.55)*\dir{>}+(4,0)*{ij};
"j";"k"**\dir{-} ?(0.55)*\dir{>}+(-3,3)*{jk};
"i";"k"**\dir{.} ?(0.55)*\dir{>}+(2,-4)*{ik};
"i";"l"**\dir{-} ?(0.55)*\dir{>}+(3,-3)*{il};
"j";"l"**\dir{-} ?(0.55)*\dir{>}+(3,2)*{jl};
"k";"l"**\dir{-} ?(0.55)*\dir{>}+(-4,-1)*{kl};
\endxy
\quad 
\mapsto
\quad 
\xy 0;/r.18pc/:
(-30,0)*+{\bullet}="k";
(-8,-30)*+{\bullet}="l";
(30,-10)*+{\bullet}="i";
(-1,40)*+{\bullet}="j";
"i";"j"**\dir{-} ?(0.55)*\dir{>}+(5,0)*{g_{ij}};
"j";"k"**\dir{-} ?(0.55)*\dir{>}+(-4,3)*{g_{jk}};
"i";"k"**\dir{.} ?(0.55)*\dir{>}+(2,-4)*{g_{ik}};
"i";"l"**\dir{-} ?(0.55)*\dir{>}+(3,-3)*{g_{il}};
"j";"l"**\dir{-} ?(0.55)*\dir{>}+(4,2)*{g_{jl}};
"k";"l"**\dir{-} ?(0.55)*\dir{>}+(-5,0)*{g_{kl}};
{\ar@{:>}|-(0.65){f_{ijk} } (-1,34);(-1,-2)};
{\ar@{:>}|-(0.70){ f_{ikl} } (-25,-3);(7,-19)};
{\ar@{=>}|-(0.5){ f_{jkl} } (-25,2);(-8,4)};
{\ar@{=>}|-(0.65){ f_{ijl} } (0,34);(11,-18)};
\endxy
\ee
where $f_{ijk}$ is short for $(f_{ijk} , g_{jk} g_{ij} ),$ etc. This
condition says 
\be
\begin{matrix}
(f_{jkl}, g_{kl} g_{jk} ) ( e, g_{ij} ) \\
(f_{ijl}, g_{jl} g_{ij} )   
\end{matrix}
=
\begin{matrix} 
(e, g_{kl}) (f_{ijk}, g_{jk} g_{ij}) \\
(f_{ikl}, g_{kl} g_{ik}) 
\end{matrix}
,
\ee
which after multiplying out (using the rules of Section
\ref{sec:2-groups}) and projecting both sides to $H$ gives
\be
f_{ijl} f_{jkl} = f_{ikl}  \a_{g_{kl} } ( f_{ijk} ) .
\ee
The 2-functor also assigns 0-cochains $\psi_{i} : U_{i} \to H$ 
\be
\xy 0;/r.25pc/:
(-10,3)*{i};
(10,3)*{i};
(-10,0)*\xycircle(15,15){.};
(10,0)*\xycircle(15,15){.};
(-10,0)*{\bullet}="j";
(10,0)*{\bullet}="i";
"i";"j"**\dir{-} ?(0.55)*\dir{>}+(1,3)*{ii};
\endxy
\qquad
\mapsto
\qquad  
\xy 0;/r.25pc/:
(-15,0)*{\bullet}="j";
(15,0)*{\bullet}="i";
"i";"j"**\crv{(0,15)} ?(0.50)*\dir{>}+(0,3)*{e};
"i";"j"**\crv{(0,-15)} ?(0.50)*\dir{>}+(0,-3)*{g_{ii}};
{\ar@{=>}|-(0.5){( \psi_{i}, e )} (0,7);(0,-7)};
\endxy
,
\ee
which says
\be
\t (\psi_{i} ) = g_{ii}.
\ee
These satisfy two ``degenerate'' cocycle conditions 
on each double intersection $U_{ij}$ of $M$ for 
the two ways one edge can be collapsed on the triangle. One is 
\be
\xy 0;/r.50pc/:
(-10,-9)*{\bullet}="k";
(0,8.32)*{\bullet}="j";
(10,-9)*{\bullet}="i";
"i";"j"**\crv{(11.5,3.83)} ?(0.50)*\dir{>}+(0,2)*{e};
"i";"j"**\crv{(-2,-3.33)} ?(0.45)*\dir{>}+(1,-3)*{g_{ii}};
"j";"k"**\dir{-} ?(0.55)*\dir{>}+(-1,1)*{g_{ij}};
"i";"k"**\dir{-} ?(0.55)*\dir{>}+(1,-2)*{g_{ij}};
{\ar@{=>}|-(0.64){(f_{iij}, g_{ij} g_{ii} ) \; } (-1.5,3);(0,-8)};
{\ar@{=>}|-(0.45){( \psi_{i}, e)} (7.5,1.83);(2.5,-1.83)};
\endxy
\quad
=
\quad  
\xy 0;/r.50pc/:
(-10,-9)*{\bullet}="k";
(0,8.32)*{\bullet}="j";
(10,-9)*{\bullet}="i";
"i";"j"**\dir{-} ?(0.50)*\dir{>}+(2,0)*{e};
"j";"k"**\dir{-} ?(0.50)*\dir{>}+(-2,0)*{g_{ij}};
"i";"k"**\dir{-} ?(0.50)*\dir{>}+(0,-2)*{g_{ij}};
{\ar@{=>}|-(0.65){(e, g_{ij} )} (0,6);(0,-8)};
\endxy
,
\ee
which after multiplying out and projecting to $H$ says  
\be
f_{iij} \a_{g_{ij}} ( \psi_{i} ) =  e .
\ee
The other cocycle condition is 
\be
\xy 0;/r.50pc/:
(-10,-9)*{\bullet}="k";
(0,8.32)*{\bullet}="j";
(10,-9)*{\bullet}="i";
"j";"k"**\crv{(1,-3.33)} ?(0.55)*\dir{>}+(-1,-3)*{g_{jj}};
"j";"k"**\crv{(-12,3.33)} ?(0.55)*\dir{>}+(-1,1)*{e};
"i";"j"**\dir{-} ?(0.50)*\dir{>}+(2,0)*{g_{ij}};
"i";"k"**\dir{-} ?(0.50)*\dir{>}+(1,-2)*{g_{ij}};
{\ar@{=>}|-(0.64){ ( f_{ijj}, g_{jj} g_{ij} ) } (1,3);(0,-8)};
{\ar@{=>}|-(0.45){( \psi_{j},e)} (-7.5,1.83);(-2.5,-1.83)};
\endxy
\quad
=
\quad
\xy 0;/r.50pc/:
(-10,-9)*{\bullet}="k";
(0,8.32)*{\bullet}="j";
(10,-9)*{\bullet}="i";
"i";"j"**\dir{-} ?(0.50)*\dir{>}+(2,0)*{g_{ij}};
"j";"k"**\dir{-} ?(0.50)*\dir{>}+(-2,0)*{e};
"i";"k"**\dir{-} ?(0.50)*\dir{>}+(0,-2)*{g_{ij}};
{\ar@{=>}|-(0.65){(e, g_{ij} )} (0,6);(0,-8)};
\endxy
,
\ee
which after multiplying out and projecting to $H$ says 
\be
f_{ijj} \psi_{j} = e .
\ee

Refinements and 1-morphisms between two such 2-functors is similar to
the ordinary functor case from Section \ref{sec:1Cech} but a bit more
subtle due to modifications (which we won't discuss now anyway). Let
$\{ U_{i'} \}_{i' \in I'}$ be another cover of $M$ with associated
\v Cech 2-groupoid $\mathfrak{U}'.$ Let $P: \mathfrak{U} \to \mathcal{B}
\mathfrak{G}$ and $P' : \mathfrak{U}' \to \mathcal{B} \mathfrak{G}$ be
two smooth 2-functors. A 1-morphism from $P$ to $P'$ consists of a
common refinement $\{ V_{a} \}_{a \in A}$ of both
$\{ U_{i} \}_{i \in I}$ and $\{ U_{i'} \}_{i' \in I'}$ along with a
smooth pseudo-natural transformation 
\be
\xy 0;/r.20pc/:
(0,-15)*+{\mathfrak{U}'}="3";
(15,0)*+{\mathcal{B} \mathfrak{G}}="4";
(,15)*+{ \mathfrak{U}}="2";
(-15,0)*+{ \mathfrak{V}}="1";
{\ar^{\a} "1";"2" };
{\ar^{P} "2";"4" };
{\ar_{\a'} "1";"3"};
{\ar_{P} "3";"4"};
{\ar@{=>}^{h} (0,10);(0,-10)};
\endxy
.
\ee
By definition (see Definition A.6. in \cite{SW3}), to each object
$(x,a)$ in $\mathfrak{V}$ such a pseudo-natural transformation gives a
smooth function $h_{a} : V_{a} \to G$ as before, but also to each jump
$(x,a,b)$ in $\mathfrak{V},$ it gives another smooth function
$\e_{ab} : V_{a b } \to H$ fitting in the diagram 
\be
\xy 0;/r.15pc/:
(15,15)*+{\bullet}="1";
(-15,15)*+{\bullet}="2";
(15,-15)*+{\bullet}="3";
(-15,-15 )*+{\bullet}="4";
{\ar_{g_{a b}} "1";"2" };
{\ar_{h_{b}} "2";"4" };
{\ar^{h_{a}} "1";"3"};
{\ar^{g'_{a b}} "3";"4"};
{\ar@{=>}|-{ (\e_{ab} , h_{b} g_{ab} )} "2";"3"};
\endxy
,
\ee
which says that 
\be
\t ( \e_{ab} ) h_{b} g_{ab} = g'_{ab} h_{a}. 
\ee
The higher naturality conditions of a pseudo-natural transformation are
given as follows. In general, to every 2-morphism,  there is an
associated naturality condition, but because the 2-morphisms in
$\mathfrak{U}$ are all identities, this condition is vacuously true. To
every pair of composable 1-morphisms $(x,i,j)$ and $(x,j,k)$ we get 
\be
\xy 0;/r.50pc/:
(-12,7)*{\bullet}="k";
(0,12.32)*{\bullet}="j";
(8,1)*{\bullet}="i";
"i";"j"**\dir{-} ?(0.55)*\dir{>}+(2,1)*{g_{ab}};
"j";"k"**\dir{-} ?(0.55)*\dir{>}+(-2,1)*{g_{bc}};
"i";"k"**\dir{-} ?(0.55)*\dir{>}+(0,-1)*{g_{ac}};
{\ar@{=>}|-(0.55){ f_{abc} } (-0.5,11);(-2,5)};
(-12,-8)*{\bullet}="kb";
(0,-3.32)*{\bullet}="jb";
(8,-14)*{\bullet}="ib";
"ib";"jb"**\dir{.} ?(0.55)*\dir{>}+(1,1)*{g'_{ab}};
"jb";"kb"**\dir{.} ?(0.55)*\dir{>}+(1,2)*{g'_{bc}};
"ib";"kb"**\dir{-} ?(0.55)*\dir{>}+(-1,-2)*{g'_{ac}};
{\ar@{:>}|-(0.55){ f'_{abc} } (-0.5,-5);(-2,-10)};
"i";"ib"**\dir{-} ?(0.55)*\dir{>}+(2,0)*{h_{a}};
"j";"jb"**\dir{.} ?(0.55)*\dir{>}+(1,1.5)*{h_{b}};
"k";"kb"**\dir{-} ?(0.55)*\dir{>}+(-2,0)*{h_{c}};
{\ar@{:>}|-(0.5){ \e_{ab} } (1,9.62);(7,-11)};
{\ar@{:>}|-(0.5){ \e_{bc} } (-11.3,6);(-.68,-2.32)};
{\ar@{=>}|-(0.5){ \e_{ac} } (-11.3,6);(5.32,-12)};
\endxy
.
\ee
Commutativity of this diagram says 
\be
\begin{matrix}
( e, h_{c} ) ( f_{abc}, g_{bc} g_{ab} ) \\
(  \e_{ac},  h_{c} g_{ac})    
\end{matrix}
=
\begin{matrix} 
(\e_{bc}, h_{c} g_{bc} ) ( e, g_{ab} ) \\
(e, g'_{bc} ) ( \e_{ab} , h_{b} g_{ab} )  \\
(f'_{abc} , g'_{bc} g'_{ab} ) 
\end{matrix}
,
\ee
which after equating both projections to $H$ gives 
\be
\e_{ac} \a_{h_{c}} (f_{abc}) = f'_{abc} \a_{g'_{bc}} (\e_{ab}) \e_{bc}
\ee
for all $a,b,c \in A.$  

Finally, to every object $(x,i)$ we get on each open set $U_{i}$
\be
\xy 0;/r.40pc/:
(-10,10)*{\bullet}="j";
(10,10)*{\bullet}="i";
"i";"j"**\crv{(0,20)} ?(0.55)*\dir{>}+(1,2)*{e};
"i";"j"**\crv{(0,0)} ?(0.55)*\dir{>}+(1,-2)*{g_{ii}};
(-10,-10)*{\bullet}="jb";
(10,-10)*{\bullet}="ib";
"ib";"jb"**\crv{~*=<6pt>{.} (0,0)} ?(0.55)*\dir{>}+(1,1.5)*{e};
"ib";"jb"**\crv{(0,-20)} ?(0.55)*\dir{>}+(1,-2)*{g'_{ii}};
"i";"ib"**\dir{-} ?(0.55)*\dir{>}+(3,1)*{h_{a}};
"j";"jb"**\dir{-} ?(0.55)*\dir{>}+(-3,1)*{h_{a}};
{\ar@{=>}|-(0.5){ ( \e_{aa} , h_{a}  g_{aa} ) } "j"-(-1,3);"ib"+(-3,1)};
{\ar@{=>}|-(0.5){ (  \psi_{a}, e ) } (0,14);(0,6)};
{\ar@{:>}|-(0.5){ ( \psi'_{a}, e ) } (0,-6);(0,-14)};
\endxy
\quad ,
\ee
where the back face of the cylinder is the identity 2-morphism
$(e,h_a).$ This reads
\be
\begin{matrix}
(e, h_{a} ) ( \psi_{a}, e ) \\
( \e_{aa}, h_{a} g_{aa} ) 
\end{matrix}
\quad 
=
\quad
\begin{matrix}
(e, h_{a} )  \\
( \psi'_{a} ,  e ) ( e, h_{a} ) 
\end{matrix}
\quad ,
\ee
which after projecting to $H$ says 
\be
\e_{aa} \a_{h_{a}} ( \psi_{a} ) = \psi'_{a}. 
\ee
Therefore, a 1-morphism of such principal 2-bundles as described above
defines an equivalence of principal 2-bundles as described in \cite{Wo}.

We won't discuss 2-morphisms now because we will see that the above
construction is a special case of the concept of limits
of 2-categories in Section \ref{sec:2limit}.

\subsection{Local triviality of 2-functors}
Just as transport functors describe parallel transport along paths,
transport 2-functors describe parallel transport along paths and
surfaces. They exhibit a formulation of a generalization of bundles with
connection that describe such transport. We start by generalizing the
thin path groupoid $\mathcal{P}_{1} (X)$ to the thin path 2-groupoid
$\mathcal{P}_{2} (X).$ At this point, one should recall Definition
\ref{defn:thinloop} where bigons are introduced.

\begin{defn}
\label{defn:bigonthinhomotopy}
Let $X$ be a smooth manifold. Two bigons $\G$ and $\G'$ are said to be
\emph{\uline{thinly homotopic}} if there exists a smooth map
$A : [ 0, 1 ] \times [ 0, 1 ] \times [ 0, 1 ] \to X$ with the following
two properties. 
\begin{enumerate}[i)]
\item
First, there exists an  $\e$ with $\frac{1}{2} > \e > 0$ such that 
\be
A ( t, s, r ) = 
	\begin{cases}
	x&\mbox{ for all }(t,s,r) \in [0,\e]\times [0,1]\times [0,1]\\
	y&\mbox{ for all }(t,s,r)\in [1-\e,1]\times [0,1]\times [0,1]\\
	\g(t)&\mbox{ for all }(t,s,r)\in[0,1]\times[0,\e]\times [0,1]\\
	\g'(t)&\mbox{ for all }(t,s,r)\in[0,1]\times[1-\e,1]\times[0,1]
\\
	\G(t,s)& \mbox{ for all } (t,s,r) \in [0,1] \times [0,1] 
\times [ 0, \e ] \\
	\G'(t,s) &\mbox{ for all } (t,s,r) \in [0,1] \times [0,1]
\times [ 1 - \e, 1 ] 
	\end{cases}  
\ee
\item
Second, the rank of $A$ is strictly less than 3 for all
$(t,s,r ) \in [0,1] \times [0,1] \times [0,1]$ and is strictly less than
2 for all $(t,s,r)\in[0,1]\times([0,\e]\cup [1-\e, 1]) \times [0,1].$ 
\end{enumerate}
In this case, $A$ is said to be a \emph{\uline{thin homotopy}} from $\G$ to
$\G'.$ The set of equivalence classes of bigons under thin homotopy is
denoted by $P^2 X.$ Elements of $P^2 X$ are called
\emph{\uline{thin bigons}}. 
\end{defn}

\begin{defn}
\label{defn:path2groupoid}
Let $X$ be a smooth manifold. The \emph{\uline{thin path-2-groupoid}} is a
2-category $\mathcal{P}_{2} (X)$ defined as follows. The set of objects
and 1-morphisms of $\mathcal{P}_{2} (X)$ coincide with that of
$\mathcal{P}_{1}(X).$ The set of 2-morphisms of $\mathcal{P}_{2} (X)$ is
$P^2 X.$ Let $[\G]$ be a thin bigon. The source and targets are defined
by choosing a representative bigon $\G$ and taking the thin homotopy
equivalence classes of the paths $t \mapsto \G(t,0)$ and
$t \mapsto \G(t,1),$ respectively. For a thin path $[\g],$ the identity
at $[\g]$ is the thin homotopy class of the bigon $(t,s) \mapsto \g(t).$

The various compositions in $\mathcal{P}_{2} (X)$ are the usual ones of
composing paths and homotopies by either stacking squares vertically or
horizontally and parametrizing via double speed vertically or
horizontally, respectively. More concretely, given two vertically
composable thin bigons%
\footnote{To be absolutely clear, we write square brackets to denote the
thin homotopy equivalence classes. After this definition, we will
generally \emph{not} do this, unless otherwise specified.}
\be
\xymatrix{
y & & \ar@/_2pc/[ll]_{[\g]}="1" \ar[ll]|-{[\de]}="2"
\ar@{}"1";"2"|(.25){\,}="3" \ar@{}"1";"2"|(.85){\,}="4"
\ar@{=>}"3";"4"^{[\G]}  \ar@/^2pc/[ll]^{[\e]}="5"
\ar@{}"2";"5"|(.15){\,}="6" \ar@{}"2";"5"|(.75){\,}="7"
\ar@{=>}"6";"7"^{[\D]}   x
}
\ee
the vertical composition is given by first choosing representatives
$\de$ for the target of $\G$ and $\de'$ for the source of $\D.$ Then,
there exists a thin (rank strictly less than 2) homotopy
$\S: \de \Rightarrow \de.'$ Using this thin homotopy, the vertical
composition is the thin homotopy class associated to the bigon 
\be
\label{eq:verticalcompositionbigons}
\begin{smallmatrix}
\G \\
\overset{\circ}{\D}
\end{smallmatrix} (t,s) := 
	\begin{cases}
	\G(t,3s) & \mbox{ for } 0 \le s \le \frac{1}{3} \\
	\S (t, 3s-1) & \mbox{ for } \frac{1}{3} \le s \le \frac{2}{3} \\
	\D (t,3s-2) & \mbox{ for } \frac{2}{3} \le s \le 1
	\end{cases} 
	\ , \ t \in [0,1] .
\ee
Given two horizontally composable thin bigons
\be
\xymatrix{
z & & \ar@/_1.5pc/[ll]_{[\g']}="5" \ar@/^1.5pc/[ll]^{[\de']}="6"
\ar@{}"5";"6"|(.15){\,}="7" \ar@{}"5";"6"|(.85){\,}="8"
\ar@{=>}"7";"8"^{[\G']}  y & & \ar@/_1.5pc/[ll]_{[\g]}="1"
\ar@/^1.5pc/[ll]^{[\de]}="2" \ar@{}"1";"2"|(.15){\,}="3"
\ar@{}"1";"2"|(.85){\,}="4" \ar@{=>}"3";"4"^{[\G]}  x
}
\ee
the horizontal composition is given by the thin homotopy class
associated to 
\be
\label{eq:horizontalcompositionbigons}
(\G' \circ \G ) (t,s) := 
	\begin{cases}
	\G(2t,s) & \mbox{ for } 0 \le t \le \frac{1}{2} \\
	\G' (2t-1,s) & \mbox{ for } \frac{1}{2} \le t \le 1
	\end{cases} 
	\ , \ s \in [0,1] . 
\ee
\end{defn}

All such compositions are well-defined, smooth, associative, have left
and right units given by constant bigons for horizontal composition and
paths viewed as bigons for vertical composition respectively, and
satisfy the interchange law.
$\mathcal{P}_{2} (X)$ is a Lie 2-groupoid since thin homotopy classes of
bigons are invertible in both ways and the functions that assign every
class to its vertical and horizontal inverses are both smooth. 

\begin{rmk}
In the definition of vertical composition
(\ref{eq:verticalcompositionbigons}), we can always choose
representatives of $\G$ and $\D$ so that $\de = \de'$ and we can ignore
$\S$ for all practical purposes of this paper. Therefore, we will always
write the vertical composition as 
\be
\begin{smallmatrix}
\G \\
\overset{\circ}{\D}
\end{smallmatrix} (t,s) := 
	\begin{cases}
	\G(t,2s) & \mbox{ for } 0 \le s \le \frac{1}{2} \\
	\D (t, 2s-1) & \mbox{ for } \frac{1}{2} \le s 1 
	\end{cases} 
	\ , \ t \in [0,1] .
\ee
\end{rmk}

\begin{defn}
\label{defn:trivialization2}
Let $\Gr$ be a Lie 2-groupoid, $T$ be a 2-category, $i:\mathrm{Gr}\to T$
a 2-functor, and $M$ a smooth manifold. Fix a surjective submersion
$\pi : Y \to M.$ A \emph{\uline{$\pi$-local $i$-trivialization}} of a
2-functor $F : \mathcal{P}_{2} (M) \to T$ is a pair $(\triv, t)$ of a
strict 2-functor $\triv : \mathcal{P}_{2} (Y) \to \Gr$ and a
pseudonatural equivalence 
\be
\xy 0;/r.15pc/:
(15,-15)*+{\Gr}="3";
(-15,-15 )*+{T}="4";
(-15,15)*+{\mathcal{P}_{2} (M)}="2";
(15,15)*+{\mathcal{P}_{2} (Y)}="1";
{\ar_{\pi_{*}} "1";"2" };
{\ar_{F} "2";"4" };
{\ar^{\triv} "1";"3"};
{\ar^{i} "3";"4"};
{\ar@{=>}^{t} "2";"3"};
\endxy
\ee
meaning that there exist a weak inverse $\ov t$ along with modifications
(see Definition A.8. in \cite{SW3}) $i_{t} : \begin{smallmatrix} t \\
\overset{\circ}{\ov t} \end{smallmatrix} \Rrightarrow \id_{\pi^* F }$
and $j_{t} : \id_{\triv_{i}} \Rrightarrow \begin{smallmatrix} \ov t \\
\overset{\circ}{t} \end{smallmatrix}$ satisfying the zig-zag identities
(see Definition 7. of \cite{BL} and particularly their discussion on
string diagrams). The 2-groupoid $\Gr$ is called the
\emph{\uline{structure 2-groupoid}} for $F.$ 
\end{defn}

2-functors $F : \mathcal{P}_{2} (M) \to T$ equipped with $\pi$-local
$i$-trivializations $(\triv, t)$ form the objects, written as triples
$(F, \triv, t),$ of a 2-category denoted by
$\mathrm{Triv}^{2}_{\pi} (i).$ 

\begin{defn}
\label{defn:trivialization2mor}
A \emph{\uline{1-morphism of $\pi$-local $i$-trivializations}}
$\a : (F, \triv, t) \to (F', \triv', t')$ in
$\mathrm{Triv}^{2}_{\pi} (i)$ is a pseudo-natural transformation
$\a : F \Rightarrow F'.$ A \emph{\uline{2-morphism}} $\a \Rightarrow \a'$
is a modification. 
\end{defn}

\begin{defn}
\label{defn:descent2}
Let $\Gr$ be a Lie 2-groupoid, $T$ a 2-category, $i : \Gr \to T$ a
2-functor and $\pi : Y \to M$ a surjective submersion. A
\emph{\uline{descent object}} is a quadruple $(\triv, g , \psi,  f )$
consisting of a strict 2-functor $\triv : \mathcal{P}_{2} (Y) \to \Gr,$
a pseudonatural equivalence 
\be
\xy 0;/r.15pc/:
(17.5,-15)*+{\mathcal{P}_{2} (Y)}="3";
(-17.5,-15 )*+{T}="4";
(-17.5,15)*+{\mathcal{P}_{2} (Y)}="2";
(17.5,15)*+{\mathcal{P}_{2} (Y^{[2]})}="1";
{\ar_{\pi_{1*}} "1";"2" };
{\ar_{\triv_{i}} "2";"4" };
{\ar^{\pi_{2*}} "1";"3"};
{\ar^{\triv_{i}} "3";"4"};
{\ar@{=>}^{g} "2";"3"};
\endxy
\ee
and invertible modifications 
\be
f : \begin{smallmatrix}
\pi^{*}_{12} g \\
\overset{\circ}{\pi^{*}_{23} g}
\end{smallmatrix} \Rrightarrow  \pi^{*}_{13} g 
\ee
and
\be
\psi : \id_{\triv_{i}} \Rrightarrow \D^* g. 
\ee
These modifications must satisfy the coherence conditions which are
explicitly given in Definition 2.2.1. of \cite{SW3} (in the examples of
this current paper, the above modifications will actually be trivial and
the coherence conditions will automatically be satisfied, which is why
we leave them out). 
\end{defn}

Descent objects form the objects of a 2-category denoted by
$\mathfrak{Des}^{2}_{\pi} (i).$ Morphisms and 2-morphisms are defined as
follows. 

\begin{defn}
\label{defn:descent2mor}
A \emph{\uline{descent 1-morphism}} from $(\triv, g, \psi, f)$ to
$(\triv', g', \psi', f')$ is a pair $(h, \e)$ with $h$ a pseudo-natural
transformation $h : \triv_{i} \Rightarrow \triv'_{i}$ and $\e$ an
invertible modification 
\be
\e : 
\begin{matrix}
g \\
\overset{\circ}{\pi_{2}^{*}  h}
\end{matrix}
\Rrightarrow
\begin{matrix}
\pi_{1}^{*} h \\
\overset{\circ}{g'}
\end{matrix}
.
\ee
These must satisfy certain identities explained in Definition 2.2.2. of
\cite{SW3}. 
\end{defn}

\begin{defn}
\label{defn:descent2mor2}
Let $(h, \e)$ and $(h', \e')$ be two descent 1-morphisms from
$(\triv, g, \psi, f)$ to $(\triv', g', \psi', f').$  A
\emph{\uline{descent 2-morphism}} from $(h, \e)$ to $(h', \e')$ is a
modification $E : h \Rrightarrow h'$ satisfying a certain identity
explained in Definition 2.2.3. of \cite{SW3}. 
\end{defn}

There is a 2-functor $\mathrm{Ex}^{2}_{\pi} : \mathrm{Triv}^{2}_{\pi}(i)
\to \mathfrak{Des}^{2}_{\pi} (i)$ that extracts descent data from
trivialization data. At the level of objects, this functor is defined as
follows. Let $(F, \triv, t)$ be an object in
$\mathrm{Triv}^{2}_{\pi}(i).$ For the quadruple $(\triv, g, \psi, f),$
take $\triv$ to be exactly the same. For $g$ take the composition
$g := \begin{smallmatrix} \pi_{1}^{*} \ov t \\
\overset{\circ}{\pi_{2}^{*} t} \end{smallmatrix}$ coming from the
composition in the diagram
\be
\xy 0;/r.15pc/:
(15,0)*+{\mathcal{P}_{2} (Y)}="3";
(-15,0 )*+{\mathcal{P}_{2} (M)}="4";
(-15,30)*+{\mathcal{P}_{2} (Y)}="2";
(15,30)*+{\mathcal{P}_{2} (Y^{[2]})}="1";
(-10,25)*+{}="a";
(10,5)*+{}="b";
{\ar_{\pi_{1*}} "1";"2" };
{\ar_{\pi_{*}} "2";"4" };
{\ar^{\pi_{2*}} "1";"3"};
{\ar^{\pi_{*}} "3";"4"};
{\ar@{=}^{\id} "a";"b"};
(-40,10)*+{\Gr}="5";
(-40,-20)*+{T}="7";
(-5,-20)*+{\Gr}="6";
{\ar^{\triv} "3";"6"};
{\ar^{i} "6";"7"};
{\ar^{F} "4";"7"};
{\ar_{i} "5";"7"};
{\ar_{\triv} "2";"5"};
{\ar@{=>}^(0.4){\ov t} "5";"4"};
{\ar@{=>}^{t} "4";"6"};
\endxy
\ee
just as before 
but this time $\ov t$ is a \emph{weak} (vertical) inverse to $t.$ 
By definition $f$ should be a modification $f : \begin{smallmatrix}
\pi^{*}_{12} g \\ \overset{\circ}{\pi^{*}_{23} g} \end{smallmatrix}
\Rrightarrow  \pi^{*}_{13} g.$ Using our definition of $g,$ this
means that we can break it down as follows
\be
f : \begin{smallmatrix} \pi^{*}_{12} g \\
\overset{\circ}{\pi^{*}_{23} g} \end{smallmatrix} = 
\begin{smallmatrix} 
	\pi^{*}_{12} \left( \begin{smallmatrix} \pi_{1}^{*} \ov t \\
\overset{\circ}{\pi_{2}^{*} t} \end{smallmatrix} \right ) \\
	\overset{\circ}{\pi^{*}_{23}  \left ( \begin{smallmatrix}
\pi_{1}^{*} \ov t \\
\overset{\circ}{\pi_{2}^{*} t} \end{smallmatrix} \right )}
\end{smallmatrix}
= 
\begin{smallmatrix} 
	\left ( \begin{smallmatrix} \pi_{1}^{*} \ov t \\
\overset{\circ}{\pi_{2}^{*} t} \end{smallmatrix} \right ) \\
	\overset{\circ}{\left ( \begin{smallmatrix} \pi_{2}^{*} \ov t \\
\overset{\circ}{\pi_{3}^{*} t} \end{smallmatrix} \right )}
\end{smallmatrix}
\Rrightarrow 
\begin{smallmatrix} \pi^{*}_{1} \ov t \\
\overset{\circ}{\pi^{*}_{3} t} \end{smallmatrix} 
=
\pi^{*}_{13} g 
,
\ee
where all equalities hold by commutativity of certain diagrams and the
leftover $\Rrightarrow$ is specified by the following sequence of
modifications
\be
\xy0;/r.15pc/:
(-50,0)*+{\begin{smallmatrix} 
	\left ( \begin{smallmatrix}   \pi_{1}^{*} \ov t \\
\overset{\circ}{\pi_{2}^{*} t} \end{smallmatrix} \right ) \\
	\overset{\circ}{\left ( \begin{smallmatrix} \pi_{2}^{*} \ov t \\
\overset{\circ}{\pi_{3}^{*} t}\end{smallmatrix} \right )}
\end{smallmatrix}}="1";
(0,0)*+{ \begin{smallmatrix}  \left ( \begin{smallmatrix} 
	\pi_{1}^{*} \ov t \\
	\overset{\circ}{\left ( \begin{smallmatrix}  \pi_{2}^{*} t  \\
	\overset{\circ}{\pi_{2}^{*} \ov t} \end{smallmatrix} \right )}
\end{smallmatrix} \right ) \\
	\overset{\circ}{\pi_{3}^{*} t} \end{smallmatrix}}="2";
{\ar@3{->}^(0.45){\text{associators}} "1";"2"};
(60,0)*+{ \begin{smallmatrix}  \left ( \begin{smallmatrix} 
	\pi_{1}^{*} \ov t \\
	\overset{\circ}{\pi_{2}^{*} \id_{\pi^{*} F }} \end{smallmatrix}
\right ) \\
	\overset{\circ}{\pi_{3}^{*} t} \end{smallmatrix}}="3";
{\ar@3{->}^{ \begin{smallmatrix}  \left ( \begin{smallmatrix} 
	\id_{\pi_{1}^{*} \ov t} \\
	\overset{\circ}{\pi_{2}^{*} i_{t}} \end{smallmatrix} \right ) \\
	\overset{\circ}{\id_{\pi_{3}^{*} t}} \end{smallmatrix}}"2";"3"};
(110,0)*+{\begin{smallmatrix}  \pi^{*}_{1} \ov t \\
\overset{\circ}{\pi^{*}_{3} t} \end{smallmatrix} }="4";
{\ar@3{->}^(0.6){ \begin{smallmatrix} \pi_{1}^{*} l \\
\overset{\circ}{\id_{\pi_{3}^{*} t}}  \end{smallmatrix} } "3";"4"};
\endxy
,
\ee
where $i_{t}$ is part of the pseudo-natural equivalence from $t$ and
$\ov t,$ and $l$ is a left unifier.

Finally, by definition $\psi$ should be a modification
$\psi : \id_{\triv_{i}} \Rrightarrow \D^{*} g.$ Using our definition
of $g,$ we can break it down as follows
\be
\psi :  \id_{\triv_{i}} =  \D^{*} \pi_{1}^{*}  \id_{\triv_{i}}
\Rrightarrow \D^{*} \left ( \begin{smallmatrix} \pi_{1}^{*} \ov t \\
\overset{\circ}{\pi_{2}^{*} t} \end{smallmatrix} \right ) = \D^{*} g
\ee
and such a modification can be achieved by
\be
\xy0;/r.15pc/:
(-50,0)*+{\D^{*} \pi_{1}^{*}  \id_{\triv_{i}}}="1";
(35,0)*+{\D^{*} \left ( \begin{smallmatrix} \pi_{1}^{*} \ov t \\
\overset{\circ}{\pi_{2}^{*} t} \end{smallmatrix} \right )}="3";
(0,0)*+{\D^{*} \pi_{1}^{*} \left ( \begin{smallmatrix}  \ov t \\
\overset{\circ}{t} \end{smallmatrix} \right )}="2";
{\ar@3{->}^(0.5){\D^{*} \pi_{1}^{*} j_{t} } "1";"2"};
{\ar@{=} "2";"3"};
\endxy
,
\ee
where $j_{t}$ is the other part of the pseudo-natural equivalence from
$t$ and $\ov t.$ 
This indeed defines a descent object and that this assignment of descent
data to trivialization data extends to a 2-functor
$\mathrm{Ex}^{2}_{\pi} : \mathrm{Triv}_{\pi}^{2} (i) \to
\mathfrak{Des}^{2}_{\pi} (i)$ to include 1-morphisms and 2-morphisms
(see Lemma 2.3.1., Lemma 2.3.2., and Lemma 2.3.3. of \cite{SW3}). 

\begin{defn}
Let $(F, \triv, t)$ be a $\pi$-local $i$-trivialization of a 2-functor
$F : \mathcal{P}_{2} (M) \to T,$ i.e. an object of
$\mathrm{Triv}^{2}_{\pi} (i).$ The \emph{\uline{descent object associated
to the $\pi$-local $i$-trivialization}} is
$\mathrm{Ex}^{2}_{\pi} (F, \triv, t).$ A similar definition is made for
1- and 2-morphisms. 
\end{defn}

\subsection{Transport 2-functors}
\label{sec:trans2funct}

We now wish to discuss smoothness for descent data. However, to do this
is not so simple as it was for ordinary functors. We will have to make a
detour to describe how to think of natural transformations as functors
and modifications as natural transformations by altering the source and
target categories. For the purposes of this document, we will make
stricter assumptions than is done in \cite{SW4} that are sufficient for
our purposes and simplify several of the arguments and constructions. 

Let $\mathcal{C}$ and $\mathcal{D}$ be two strict 2-categories. Let
$\mathcal{C}_{0,1}$ denote the category whose objects and morphisms are
the objects and 1-morphisms of $\mathcal{C}$ respectively. Because
$\mathcal{C}$ is strict, this defines a category. Let $\L \mathcal{D}$
be the category whose objects are morphisms
$X_{f} \xrightarrow{f} Y_{f}$ of $\mathcal{D}.$ The set of morphisms in
$\L \mathcal{D}$ from $X_{f} \xrightarrow{f} Y_{f}$ to
$X_{g} \xrightarrow{g} Y_{g}$ are pairs of morphisms $(x : X_{f} \to
X_{g}, y : Y_{f} \to Y_{g})$ along with a 2-morphism $\varphi : g \circ
x \Rightarrow y \circ f$ as in the diagram
\be
\xy 0;/r.15pc/:
(15,-15)*+{Y_{f}}="3";
(-15,-15 )*+{Y_{g}}="4";
(-15,15)*+{X_{g}}="2";
(15,15)*+{X_{f}}="1";
{\ar_{x} "1";"2" };
{\ar_{g} "2";"4" };
{\ar^{f} "1";"3"};
{\ar^{y} "3";"4"};
{\ar@{=>}^{\varphi} "2";"3"};
\endxy
.
\ee
The composition is given by stacking  
\be
\xy 0;/r.15pc/:
(15,-15)*+{Y_{f}}="3";
(-15,-15 )*+{Y_{g}}="4";
(-15,15)*+{X_{g}}="2";
(15,15)*+{X_{f}}="1";
(-45,-15 )*+{Y_{h}}="5";
(-45,15)*+{X_{h}}="6";
{\ar_{x} "1";"2" };
{\ar|-{g} "2";"4" };
{\ar^{f} "1";"3"};
{\ar^{y} "3";"4"};
{\ar_{x'} "2";"6" };
{\ar_{h} "6";"5" };
{\ar^{y'} "4";"5"};
{\ar@{=>}^{\varphi} "2";"3"};
{\ar@{=>}^{\psi} "6";"4"};
\endxy
\quad = \quad 
\xy 0;/r.15pc/:
(15,-15)*+{Y_{f}}="3";
(-15,-15 )*+{Y_{h}}="4";
(-15,15)*+{X_{h}}="2";
(15,15)*+{X_{f}}="1";
{\ar_{x' \circ x} "1";"2" };
{\ar_{h} "2";"4" };
{\ar^{f} "1";"3"};
{\ar^{y' \circ y} "3";"4"};
{\ar@{=>}|-{\begin{smallmatrix} \psi \circ \id_{x} \\
\overset{\circ}{\id_{y'} \circ \varphi}\end{smallmatrix} } "2";"3"};
\endxy
.
\ee
One can check that under our assumptions, this forms a category. 

Notice that $\L \mathcal{D}$ has a bit more structure. It also has a
partially defined operation on objects and 1-morphisms given by
``stacking vertically.'' Suppose that $X_{f} \xrightarrow{f} Y_{f}$ and
$Y_{f} \xrightarrow{f'} Z_{f}$ are two 1-morphisms in $\mathcal{D}$ then
one can compose them and this gives a partially defined associative and
unital operation on objects of $\L \mathcal{D}.$ Similarly, given
morphisms in $\L \mathcal{D}$ which can be vertically stacked as in the
diagram 
\be
\xy 0;/r.15pc/:
(15,0)*+{Y_{f}}="3";
(-15,0 )*+{Y_{g}}="4";
(-15,30)*+{X_{g}}="2";
(15,30)*+{X_{f}}="1";
{\ar_{x} "1";"2" };
{\ar_{g} "2";"4" };
{\ar^{f} "1";"3"};
{\ar|-{y} "3";"4"};
{\ar@{=>}^{\varphi} "2";"3"};
(15,-30)*+{Z_{f}}="5";
(-15,-30 )*+{Z_{g}}="6";
{\ar_{g'} "4";"6" };
{\ar^{f'} "3";"5"};
{\ar^{z} "5";"6"};
{\ar@{=>}^{\varphi'} "4";"5"};
\endxy
\quad = \quad
\xy 0;/r.15pc/:
(15,-15)*+{Y_{f}}="3";
(-15,-15 )*+{Y_{g}}="4";
(-15,15)*+{X_{g}}="2";
(15,15)*+{X_{f}}="1";
{\ar_{x} "1";"2" };
{\ar_{g' \circ g} "2";"4" };
{\ar^{f' \circ f} "1";"3"};
{\ar^{z} "3";"4"};
{\ar@{=>}|-{\begin{smallmatrix} \id_{g'} \circ \varphi \\
\overset{\circ}{\varphi' \circ \id_{f}} \end{smallmatrix} } "2";"3"};
\endxy
.
\ee
This additional partially defined composition is written as $\otimes$ in
\cite{SW4} so we stick with this notation. 

Associated to a pseudo-natural transformation $\rho$ as in 
\be
\label{eq:curlyFbeginning}
\xy0;/r.15pc/:
(-20,0)*+{\mathcal{D}}="d";
(20,0)*+{\mathcal{C}}="c";
{\ar@/_1.5pc/"c";"d"_{F}};
{\ar@/^1.5pc/"c";"d"^{G}};
(0,10)*+{}="f";
(0,-10)*+{}="g";
{\ar@{=>}"f";"g"^{\rho}};
\endxy
\ee
is a functor
$\mathcal{F} (\rho) : \mathcal{C}_{0,1} \to \L \mathcal{D}$ defined by 
\be
\xy0;/r.15pc/:
(-30,0)*+{X}="1";
(30,10)*+{FX}="2";
(30,-10)*+{GX}="3";
{\ar^{\rho(X)} "2";"3"};
{\ar@{|->}^{\mathcal{F} (\rho) } "1"+(20,0);(10,0)};
\endxy
\ee
on objects $X$ in $\mathcal{C}_{0,1},$ i.e. objects in $\mathcal{C},$
and
\be
\xy0;/r.15pc/:
(-30,0)*+{X}="1x";
(-60,0)*+{Y}="1y";
{\ar_{f} "1x";"1y"};
(30,10)*+{FY}="2y";
(30,-10)*+{GY}="3y";
(60,10)*+{FX}="2x";
(60,-10)*+{GX}="3x";
{\ar^{\rho(X)} "2x";"3x"};
{\ar_{\rho(Y)} "2y";"3y"};
{\ar_{Ff} "2x";"2y"};
{\ar^{Gf} "3x";"3y"};
{\ar@{=>}^{\rho(f)} "2y";"3x"};
{\ar@{|->}^{\mathcal{F} (\rho) } "1x"+(20,0);(10,0)};
\endxy
\ee
on morphisms in $\mathcal{C}_{0,1},$ i.e. 1-morphisms in $\mathcal{C}.$
One can check this defines a functor. 

Associated to a modification $A$ as in 
\be
\xy0;/r.15pc/:
(-20,0)*+{\mathcal{D}}="d";
(20,0)*+{\mathcal{C}}="c";
{\ar@/_1.5pc/"c";"d"_{F}};
{\ar@/^1.5pc/"c";"d"^{G}};
(0,10)*+{}="f";
(0,-10)*+{}="g";
{\ar@{=>}@/^1pc/"f";"g"^{\rho}};
{\ar@{=>}@/_1pc/"f";"g"_{\s}};
(-5,0)*+{}="s";
(5,0)*+{}="r";
{\ar@3{->}"r";"s"_{A}};
\endxy
\ee
is a natural transformation
$\mathcal{F} (A) : \mathcal{F} (\rho) \Rightarrow \mathcal{F} (\s)$
defined by 
\be
\label{eq:curlyFend}
\xy0;/r.15pc/:
(-40,0)*+{X}="1x";
(30,10)*+{FX}="2y";
(30,-10)*+{GX}="3y";
(60,10)*+{FX}="2x";
(60,-10)*+{GX}="3x";
{\ar^{\rho(X)} "2x";"3x"};
{\ar_{\s(X)} "2y";"3y"};
{\ar_{\id_{FX}} "2x";"2y"};
{\ar^{\id_{GX}} "3x";"3y"};
{\ar@{=>}^{A(X)} "2y";"3x"};
{\ar@{|->}^{\mathcal{F} (A) } "1x"+(20,0);(0,0)};
\endxy
.
\ee
This defines a functor $\mathcal{F} : \mathrm{Hom} (F, G ) \to
\mathrm{Funct} (\mathcal{C}_{0,1}, \L \mathcal{D} ),$ where
$\mathrm{Hom} (F, G )$ is the category whose objects are pseudonatural
transformations and morphisms are modifications while $\mathrm{Funct}
(\mathcal{E}, \mathcal{E}' )$ (between two ordinary categories
$\mathcal{E}$ and $\mathcal{E}'$) is the category whose objects are
functors from $\mathcal{E}$ to $\mathcal{E}'$ and whose morphisms are
natural transformations. 

Separately, notice also that if $F : \mathcal{C} \to \mathcal{D}$ is a
2-functor then there is a functor $\L F : \L \mathcal{C} \to \L
\mathcal{D}$ defined by 
\be
\xy0;/r.15pc/:
(-30,10)*+{X_{f}}="1";
(-30,-10)*+{Y_{f}}="4";
{\ar^{f} "1";"4"};
(30,10)*+{FX_{f}}="2";
(30,-10)*+{FY_{f}}="3";
{\ar^{Ff} "2";"3"};
{\ar@{|->}^{\L F } (-10,0);(10,0)};
\endxy
\ee
on objects and 
\be
\xy0;/r.15pc/:
(-30,10)*+{X_{f}}="1x";
(-60,10)*+{X_{g}}="1y";
(-30,-10)*+{Y_{f}}="1x'";
(-60,-10)*+{Y_{g}}="1y'";
{\ar^{f} "1x";"1x'"};
{\ar_{g} "1y";"1y'"};
{\ar_{x} "1x";"1y"};
{\ar^{y} "1x'";"1y'"};
{\ar@{=>}^{\s} "1y";"1x'"};
(30,10)*+{FX_{g}}="2y";
(30,-10)*+{FY_{g}}="3y";
(60,10)*+{FX_{f}}="2x";
(60,-10)*+{FY_{f}}="3x";
{\ar^{Ff} "2x";"3x"};
{\ar_{Fg} "2y";"3y"};
{\ar_{Fx} "2x";"2y"};
{\ar^{Fy} "3x";"3y"};
{\ar@{=>}^{F \s} "2y";"3x"};
{\ar@{|->}^{\L F } (-10,0);(10,0)};
\endxy
\ee
on morphisms. 

\begin{defn}
\label{defn:smoothdescent2}
A descent object $(\triv, g, \psi, f)$ as in Definition
\ref{defn:descent2} is said to be \emph{\uline{smooth}} if 
\begin{enumerate}[i)]
\item
the 2-functor $\triv : \mathcal{P}_{2} (Y) \to \Gr$ is smooth, 
\item
the functor $\mathcal{F} (g) : \mathcal{P}_{1} (Y^{[2]}) \to \L T$ is a
transport functor with $\L \Gr$-structure, and 
\item
the natural transformations $\mathcal{F} (\psi) : \mathcal{F}
( \id_{\triv_{i}} ) \Rightarrow \D^{*} \mathcal{F} ( g )$ and
$\mathcal{F} (f) : \pi^{*}_{23} \mathcal{F} (g) \otimes \pi^{*}_{12}
\mathcal{F}(g) \Rightarrow \pi^{*}_{13} \mathcal{F}(g)$ are morphisms
between transport functors. 
\end{enumerate}
\end{defn}

Smooth descent objects form the objects of a 2-category denoted by
$\mathfrak{Des}^{2}_{\pi} (i)^{\infty}$ and form a sub-2-category of
$\mathfrak{Des}^{2}_{\pi} (i).$ Smoothness of descent 1-morphisms and
descent 2-morphisms is discussed in \cite{SW4} following Definition
3.1.2. 

\begin{defn}
\label{defn:smoothpilocalitrivialization2}
A $\pi$-local $i$-trivialization $(F, \triv, t)$ is said to be
\emph{\uline{smooth}} if the associated descent object
$\mathrm{Ex}_{\pi}^{2} (F, \triv, t)$ is smooth. The same can be said of
1-morphisms and 2-morphisms. 
\end{defn}

Smooth local trivializations, 1-morphisms, and 2-morphisms form a
sub-2-category denoted by $\mathrm{Triv}^{2}_{\pi} (i)^{\infty}$ of
$\mathrm{Triv}^{2}_{\pi} (i).$ Furthermore, $\mathrm{Ex}_{\pi}^{2}$
restricts to an equivalence of 2-categories of smooth data
(Lemma 3.2.3. of \cite{SW4}).  

After all this formalism, it should be more or less clear now what the
definition of a transport 2-functor is by just abstracting what we did
for the one-dimensional case (Definition 3.2.1. of \cite{SW4}). 

\begin{defn}
\label{defn:transport2functor}
Let $\mathrm{Gr}$ be a Lie 2-groupoid, $T$ a 2-category, $i :
\mathrm{Gr} \to T$ a 2-functor, and $M$ a smooth manifold. A
\emph{\uline{transport 2-functor on $M$ with values in a 2-category $T$ and
with $\mathrm{Gr}$-structure}} is a 2-functor $\mathrm{tra} :
\mathcal{P}_{2} (M) \to T$ such that there exists a
surjective submersion $\pi : Y \to M$ and a smooth $\pi$-local
$i$-trivialization $(\mathrm{triv}, t).$ 
\end{defn}

Transport 2-functors over $M$ with values in $T$ with
$\mathrm{Gr}$-structure form the objects of a 2-category
$\mathrm{Trans}^{2}_{\mathrm{Gr}} (M,T).$ A 1-morphism of transport
functors is a pseudo-natural transformation of 2-functors for which
there exists a common surjective submersion $\pi$ and smooth $\pi$-local
$i$-trivializations of both 2-functors so that the associated descent
1-morphism is smooth. A similar definition exists for 2-morphisms. 

As a short summary, in the past two sections we introduced three
categories for describing transport 2-functors. These were
$\mathfrak{Des}_{\pi}^{2} (i),$ $\mathrm{Triv}_{\pi}^{2}(i),$ and
$\mathrm{Trans}^{2}_{\mathrm{Gr}} (M, T).$ The category
$\mathrm{Triv}_{\pi}^{2}(i)$ was used to describe local triviality of
transport 2-functors and their morphisms in
$\mathrm{Trans}^{2}_{\mathrm{Gr}} (M, T).$ We then constructed a
2-functor $\mathrm{Ex}_{\pi}^{2} : \mathrm{Triv}_{\pi}^{2}(i) \to
\mathfrak{Des}_{\pi}^{2} (i)$ that allowed us to describe smoothness
via the subcategory $\mathfrak{Des}_{\pi}^{2} (i)^{\infty} \subset
\mathfrak{Des}_{\pi}^{2} (i)$ from which we defined
$\mathrm{Triv}_{\pi}^{2}(i)^{\infty}\subset\mathrm{Triv}_{\pi}^{2}(i).$

\subsection{The reconstruction 2-functor: from local to global}
\label{sec:reconstruction2}
The 2-functor $\mathrm{Ex}_{\pi}^{2} : \mathrm{Triv}^{2}_{\pi} (i) 
\rightarrow
\mathfrak{Des}^{2}_{\pi} (i)$ is an equivalence of 2-categories
(Proposition 4.1.1. of \cite{SW3}). To construct a (weak) inverse
$\mathrm{Rec}_{\pi}^{2} : \mathfrak{Des}^{2}_{\pi} (i) \to
\mathrm{Triv}^{2}_{\pi} (i),$ we need to enhance the \v Cech path
groupoid so that it includes more data. 

We do not require the full general definition of
$\mathcal{P}_{2}^{\pi} (M)$  in Section 3.1 of \cite{SW3} for our
purposes, but briefly the general definition is obtained by keeping the
same objects and morphisms but replacing the relations that we imposed
by 2-morphisms and setting relations on those. There are also additional
2-morphisms given by thin bigons, thin paths on intersections, and other
formal 2-morphisms such as associators, unitors, and 2-morphisms
relating the formal product to the usual composition of paths. We
therefore warn the reader that although the following definition is not
the same as that in \cite{SW3}, we use their general results and
theorems which in fact rely on their more general definition.  

\begin{defn}
\label{defn:Cechpath2groupoid}
Let $M$ be a smooth manifold and let $\pi : Y \to M$ be a surjective
submersion. The \emph{\uline{\v Cech path 2-groupoid of $M$}} is the
2-category $\mathcal{P}_{2}^{\pi} (M)$ whose set of objects and
1-morphisms are the objects and morphisms of $\mathcal{P}_{1}^{\pi}(M),$
respectively. The set of 2-morphisms are freely generated by 
\begin{enumerate}[i)]
\item
thin bigons $\G$ in $Y,$ 
\item
thin paths $\Theta : \a \to \b$ in $Y^{[2]}$ with sitting instants
thought of as 2-isomorphisms 
\be
\xy 0;/r.15pc/:
(-15,15)*+{\pi_{1} (\b)}="3";
(-15,-15 )*+{\pi_{2} (\b) }="4";
(15,-15)*+{\pi_{2} (\a) }="2";
(15,15)*+{\pi_{1} (\a) }="1";
{\ar^{\a} "1";"2" };
{\ar^{\pi_{2} (\Theta)} "2";"4" };
{\ar_{\pi_{1} (\Theta)} "1";"3"};
{\ar_{\b} "3";"4"};
{\ar@{=>}_{\Theta} "2";"3"};
\endxy
\ee
(one should think of this as weakening the first relation in Definition
\ref{defn:Cechpath1} of $\mathcal{P}^{\pi}_{1} (M)$---see Figure
\ref{fig:pathjumpcondition2} for a visualization of this),
\begin{figure}[h]
\centering
	\begin{picture}(0,0)
	\put(40,48){$\b$}
	\put(98,66){$\a$}
	\put(60,96){$\pi_{1} ( \Theta )$}
	\put(60,17){$\pi_{2} (\Theta)$}
	\put(65,40){$\xy 0;/r.25pc/: {\ar@{=>} (15,-15),(-5,15)} \endxy$}
	\end{picture}
    \includegraphics[width=0.33\textwidth]{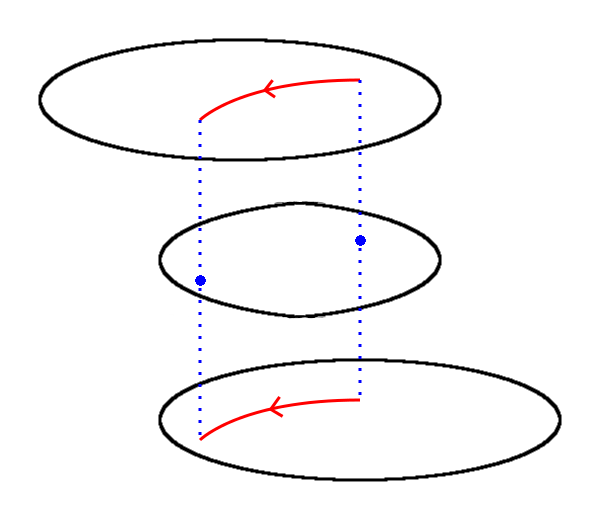}
    \caption{Thinking in terms of an open cover as a submersion,
condition ii) above says that if a thin path 
$\Theta : \a \rightarrow \b$ (with
chosen representative) is in a double intersection, there is a
relationship between going along the path first and then jumping versus
jumping first and then going along the path. The two need not be equal.}
    \label{fig:pathjumpcondition2}
\end{figure} 

\item
points $\X$ in $Y^{[3]}$ thought of as 2-isomorphisms
\be
\xy 0;/r.15pc/:
(-25,-12.5)*+{\pi_{3} (\X)}="3";
(0,12.5)*+{\pi_{2} (\X)}="2";
(25,-12.5)*+{\pi_{1} (\X)}="1";
{\ar_{\pi_{12}(\X)} "1";"2" };
{\ar_{\pi_{23}(\X)} "2";"3" };
{\ar^{\pi_{13}(\X)} "1";"3"};
{\ar@{=>}^{\X} "2";(0,-10)};
\endxy
\ee
(one should think of this as weakening the second relation in Definition
\ref{defn:Cechpath1} of $\mathcal{P}^{\pi}_{1} (M)$), 

\item
points $a$ in $Y$ thought of as 2-isomorphisms ($\id^{*}_{a}$ is the
formal identity)
\be
\xy0;/r.15pc/:
(-40,0)*+{a}="d";
(0,0)*+{a}="c";
{\ar@/_1.5pc/"c";"d"_{\id^{*}_{a}}};
{\ar@/^1.5pc/"c";"d"^{\D (a)}};
(-20,10)*+{}="f";
(-20,-10)*+{}="g";
{\ar@{=>}"f";"g"^{\D_{a}}};
\endxy
\ee
(one should think of this as weakening part of the third relation in
Definition \ref{defn:Cechpath1} of $\mathcal{P}^{\pi}_{1} (M)$), 

\item
and several other more technical generators that will not be discussed
here.  

\end{enumerate}

There are several relations imposed on the set of 1-morphisms and
2-morphisms. We will not discuss any of them, and the reader is referred
to Section 3.1 of \cite{SW3} for the details. As before, the
compositions will be written with $*$ and will be drawn vertically or
horizontally when dealing with 2-morphisms. 

\end{defn}

As before, we associate to every object $(\triv, g, \psi, f )$ in
$\mathfrak{Des}^{2}_{\pi} (i)$ a functor $R_{(\triv, g, \psi, f)} :
\mathcal{P}_{2}^{\pi} (M) \to T$ defined as follows. It sends $y \in Y$
to $\triv_{i} (y),$ thin paths $\g$ in $Y$ to $\triv_{i} (\g),$ and
jumps $\a \in Y^{[2]}$ to $g ( \a ) : \triv_{i} (  \pi_{1} (\a) ) \to
\triv_{i} ( \pi_{2} (\a) ).$ For the basic 2-morphisms, it makes the
following assignments 
\be
\label{eq:Rtrivgpsifbeginning}
\xy0;/r.15pc/:
(-40,0)*+{y}="d";
(0,0)*+{x}="c";
{\ar@/_1.5pc/"c";"d"_{\g}};
{\ar@/^1.5pc/"c";"d"^{\de}};
(-20,10)*+{}="f";
(-20,-10)*+{}="g";
{\ar@{=>}"f";"g"^{\G}};
{\ar@{|->}^{R_{(\triv, g, \psi, f)}} (10,0);(40,0)};
(60,0)*+{\triv_{i} (y)}="d2";
(100,0)*+{\triv_{i} (x)}="c2";
{\ar@/_1.5pc/"c2";"d2"_{\triv_{i} ( \g ) }};
{\ar@/^1.5pc/"c2";"d2"^{\triv_{i} (\de)}};
(80,10)*+{}="f2";
(80,-10)*+{}="g2";
{\ar@{=>}"f2";"g2"|-{\triv_{i} (\G)}};
\endxy
\ee
for thin bigons $\G : \g \Rightarrow \de$ in $Y,$ 
\be
\xy 0;/r.15pc/:
(-5,-15)*+{\pi_{1} (\b)}="3";
(-35,-15 )*+{\pi_{2} (\b) }="4";
(-35,15)*+{\pi_{2} (\a) }="2";
(-5,15)*+{\pi_{1} (\a) }="1";
{\ar_{\a} "1";"2" };
{\ar_{\pi_{2} (\Theta)} "2";"4" };
{\ar^{\pi_{1} (\Theta)} "1";"3"};
{\ar^{\b} "3";"4"};
{\ar@{=>}^{\Theta} "2";"3"};
{\ar@{|->}^{R_{(\triv, g, \psi, f)}} (15,0);(45,0)};
(110,-15)*+{\triv_{i} ( \pi_{1} (\b) )}="3m";
(65,-15 )*+{\triv_{i} (\pi_{2} (\b) ) }="4m";
(65,15)*+{\triv_{i} ( \pi_{2} (\a) ) }="2m";
(110,15)*+{\triv_{i} ( \pi_{1} (\a) ) }="1m";
{\ar_{g ( \a )} "1m";"2m" };
{\ar|-{\triv_{i} ( \pi_{2} (\Theta) ) } "2m";"4m" };
{\ar|-{\triv_{i} (\pi_{1} (\Theta) )} "1m";"3m"};
{\ar^{g ( \b )} "3m";"4m"};
{\ar@{=>}|-{g ( \Theta )} "2m";"3m"};
\endxy
\ee
for thin paths $\Theta : \a \to \b$ in $Y^{[2]},$ 
\be
\xy 0;/r.15pc/:
(-40,-12.5)*+{\pi_{3} (\X)}="3";
(-15,12.5)*+{\pi_{2} (\X)}="2";
(10,-12.5)*+{\pi_{1} (\X)}="1";
{\ar_{\pi_{12}(\X)} "1";"2" };
{\ar_{\pi_{23}(\X)} "2";"3" };
{\ar^{\pi_{13}(\X)} "1";"3"};
{\ar@{=>}^{\X} "2";(-15,-10)};
{\ar@{|->}^{R_{(\triv, g, \psi, f)}} (20,0);(50,0)};
(70,-12.5)*+{\triv_{i} ( \pi_{3} (\X) )}="3m";
(95,12.5)*+{\triv_{i} ( \pi_{2} (\X) )}="2m";
(120,-12.5)*+{\triv_{i} ( \pi_{1} (\X) )}="1m";
{\ar_{g ( \pi_{12}(\X) )} "1m";"2m" };
{\ar_{g ( \pi_{23}(\X) )} "2m";"3m" };
{\ar^{g ( \pi_{13}(\X) )} "1m";"3m"};
{\ar@{=>}|-{f ( \X)} "2m";(95,-10)};
\endxy
\ee
for points $\X$ in $Y^{[3]},$ and 
\be
\label{eq:Rtrivgpsifend}
\xy0;/r.15pc/:
(-40,0)*+{a}="d";
(0,0)*+{a}="c";
{\ar@/_1.5pc/"c";"d"_{\id^{*}_{a}}};
{\ar@/^1.5pc/"c";"d"^{\D (a)}};
(-20,10)*+{}="f";
(-20,-10)*+{}="g";
{\ar@{=>}"f";"g"^{\D_{a}}};
{\ar@{|->}^{R_{(\triv, g, \psi, f)}} (10,0);(40,0)};
(60,0)*+{\triv_{i} (a)}="d2";
(100,0)*+{\triv_{i} (a)}="c2";
{\ar@/_1.5pc/"c2";"d2"_{1_{ \triv_{i} ( a ) }} };
{\ar@/^1.5pc/"c2";"d2"^{g ( \D (a) ) }};
(80,10)*+{}="f2";
(80,-10)*+{}="g2";
{\ar@{=>}"f2";"g2"|-{\psi(a)}};
\endxy
\ee
for points $a$ in $Y.$ This defines a 2-functor
$R : \mathfrak{Des}^{2}_{\pi} (i) \to \mathrm{Funct}
(\mathcal{P}_{2}^{\pi} (M) , T )$ at the level of objects.  The rest of
this 2-functor is defined in Proposition 3.3.2. of \cite{SW3}.

There is a canonical projection functor
$p^{\pi}:\mathcal{P}^{\pi}_{2}(M) \to  \mathcal{P}_{2} (M)$ defined in
the same way as $p^{\pi} :  \mathcal{P}^{\pi}_{1} (M) \to
\mathcal{P}_{1} (M)$ on the level of objects and morphisms. On the level
of 2-morphisms, $p^{\pi}$  sends a thin bigon $\G$ in $Y$ to a the thin
bigon $\pi (\G)$ in $M.$ It sends a thin path $\Theta$ in $Y^{[2]}$ to
the identity thin bigon $\id_{\pi (\Theta) }$ (the vertical identity) in
$M$ and it sends a point $\Xi$ in $Y^{[3]}$ to to the constant thin
bigon at the point $\pi ( \Xi )$ in $M.$ Finally, it sends a point $a$
in $Y$ to the constant thin bigon at the point $\pi (a)$ in $M.$ We now
move on to defining, as before, a weak inverse
$s^{\pi} : \mathcal{P}_{2} (M) \to \mathcal{P}^{\pi}_{2} (M)$ of the
canonical projection functor. To define $s^{\pi},$ we will constantly
use  the following important fact (Lemma 3.2.2. of \cite{SW3}). 

\begin{lem}
\label{lem:unique}
Let $\g : x \to x'$ be a thin path in $M$ and let $\tilde{\g}$ and
$\tilde{\g}'$ be two lifts of $\g$ as 1-morphisms in
$\mathcal{P}_{2}^{\pi} (M)$ (the existence follows from our choices
above when we defined $s^{\pi} : \mathcal{P}_{1} (M) \to
\mathcal{P}_{1}^{\pi} (M)$). Then there exists a unique 2-isomorphism
$A : \tilde{\g} \Rightarrow \tilde{\g}'$ in $\mathcal{P}_{2}^{\pi} (M)$
such that $p^{\pi} ( A ) = \id_{\g}.$ 
\end{lem} 

We will use this to define $s^{\pi} : \mathcal{P}_{2} (M) \to
\mathcal{P}_{2}^{\pi} (M)$ on thin bigons (we have already defined
$s^{\pi}$ near (\ref{eq:pathlifting}) on objects and 1-morphisms). Let
$\G : \g \Rightarrow \de$ be any thin bigon in $M$ as in Figure
\ref{fig:bigon}. 

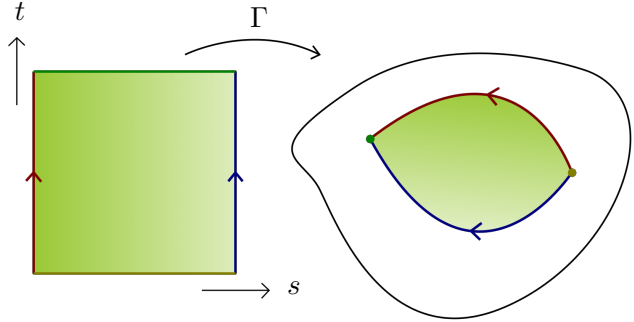
\begin{figure}[h]
\centering
\begin{picture}(0,0)
\put(112,11){$s$}
\put(9,115){$t$}
\put(98,112){$\G$}
\end{picture}
    \includegraphics[width=0.55\textwidth]{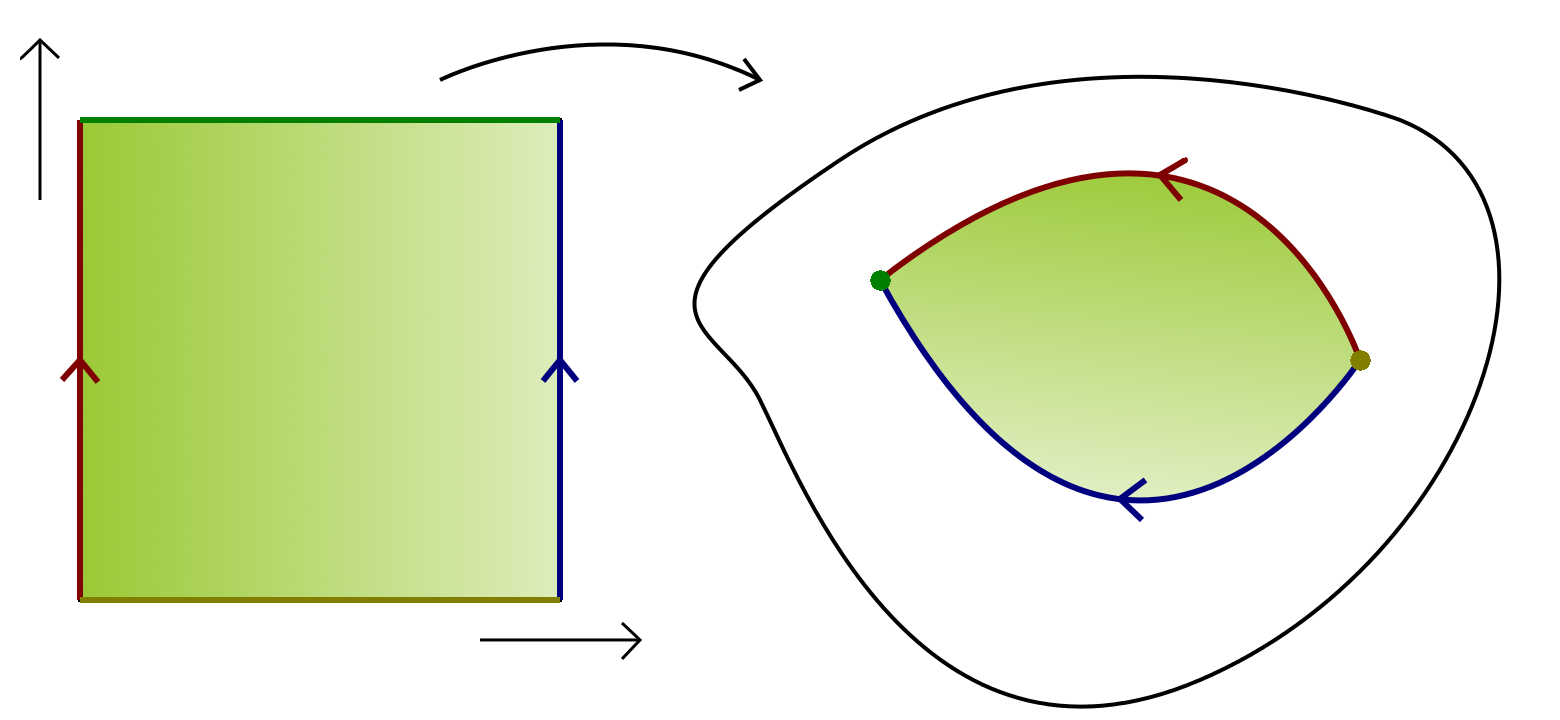}
    \vspace{-3mm}
    \caption{A representative of a thin bigon $\G$ in $M$ drawn as a map
of a square into $M.$ The $s=0$ line is drawn on top in the figure on
the right while the $s=1$ line is drawn on the bottom. The entire $t=0$
line gets mapped to the source point and the $t=1$ line gets mapped to
the target point.}
    \label{fig:bigon}
\end{figure} 

As in the case of a path, because the domain is compact, there exists a
decomposition of the bigon $\G$ (we abuse notation and write $\G$ to
mean a bigon and its thin homotopy class relying on context to
distinguish them) into smaller bigons $\{ \G_{j} \}_{j},$ as in Figure
\ref{fig:bigondecomp}, each of which fits into an open set $U_{j}.$ We
use the same notation $s_{j} : U_{j} \to Y$ as before for our local
sections. 

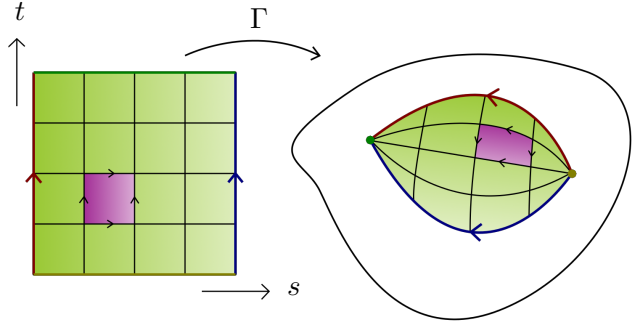
\begin{figure}[h]
\centering
\begin{picture}(0,0)
\put(112,11){$s$}
\put(9,115){$t$}
\put(98,112){$\G$}
\end{picture}
    \includegraphics[width=0.55\textwidth]{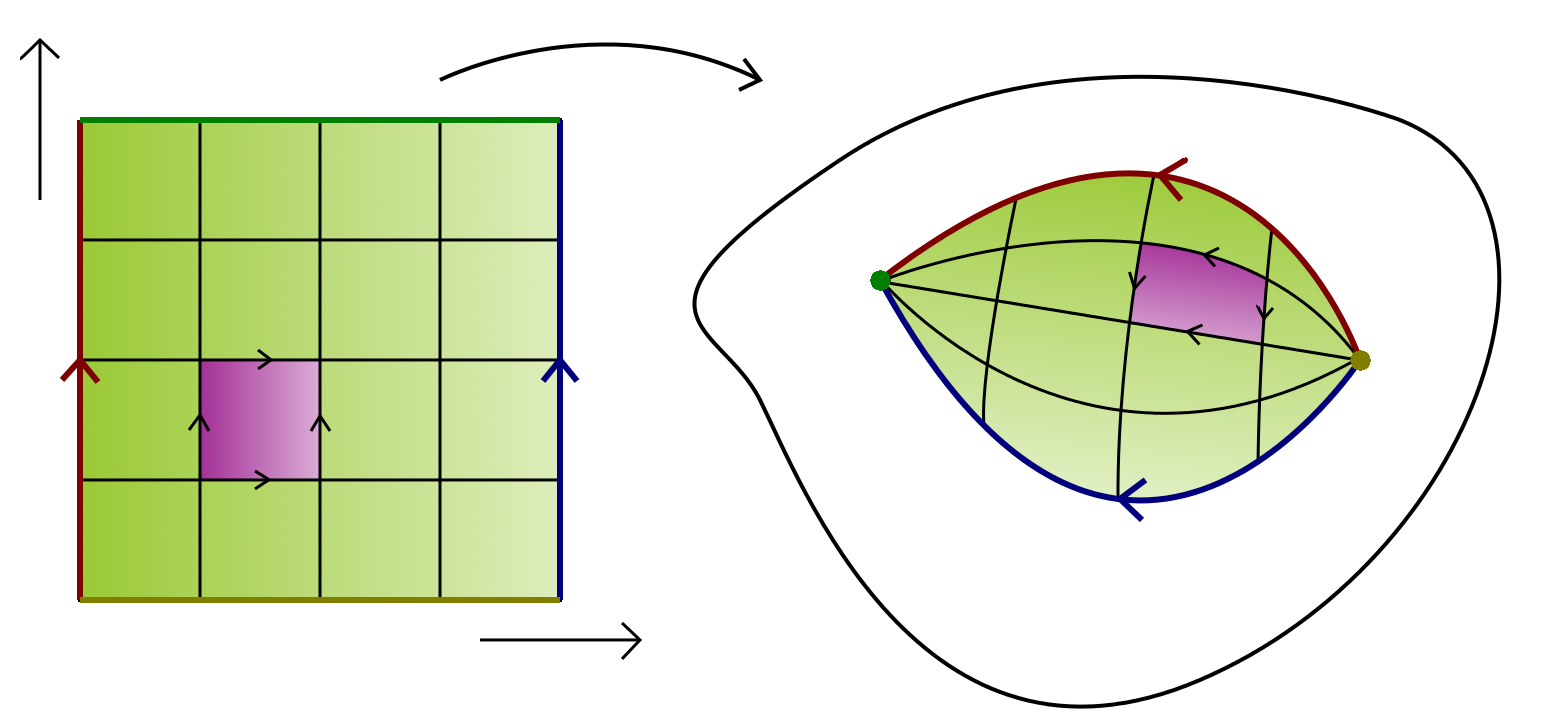}
    \vspace{-3mm}
    \caption{A decomposition of a representative of a thin bigon $\G$ in
$M$ with a single sub-bigon $\G_{j}$ highlighted. $s^{\pi} (\G)$ will be
defined as a composition of several $s^{\pi} ( \G_{j} ).$ Of course, a
general decomposition would not necessarily look like this, but such a
decomposition always exists 
by a thin homotopy so that the decomposed pieces are bigons.}
     \label{fig:bigondecomp}
\end{figure} 

Therefore, it suffices to define $s^{\pi} (\G_{j})$ for a single one of
the associated thin bigons provided that we match up all sources and
targets for the individual ones. Denote the thin bigon by 
\be
\xy0;/r.15pc/:
(-40,0)*+{x'_{j}}="d";
(0,0)*+{x_{j}}="c";
{\ar@/_1.5pc/"c";"d"_{\g_{j}}};
{\ar@/^1.5pc/"c";"d"^{\de_{j}}};
(-20,10)*+{}="f";
(-20,-10)*+{}="g";
{\ar@{=>}"f";"g"^{\G_{j}}};
\endxy
.
\ee
Then the image of this under $s^{\pi}$ is defined as the composition 
\be
\label{eq:bigonlift}
\xy0;/r.15pc/:
(-70,0)*+{s^{\pi}(x'_{j})}="left";
(30,0)*+{s^{\pi}(x_{j})}="right";
(-40,0)*+{s_{j} ( x'_{j} ) }="d";
(0,0)*+{s_{j} ( x_{j} ) }="c";
{\ar@/_1.5pc/"c";"d"_{s_{j} ( \g_{j} )}};
{\ar@/^1.5pc/"c";"d"^{s_{j} ( \de_{j} ) }};
(-20,10)*+{}="f";
(-20,-10)*+{}="g";
{\ar@{=>}"f";"g"|-{s_{j} ( \G_{j} ) }};
{\ar@/_4pc/"right";"left"_{s^{\pi}(\g_{j})}};
{\ar@/^4pc/"right";"left"^{s^{\pi}(\de_{j})}};
{\ar"right";"c"};
{\ar"left";"d"};
{\ar@{=>}(-20,25);(-20,16)};
{\ar@{=>}(-20,-16);(-20,-25)};
\endxy
.
\ee
In other words, we have lifted $\G_{j}$ using the section $s_{j} : U_{j}
\to Y,$ but to make sure that this image matches up with how $s^{\pi}$
was already defined on objects and 1-morphisms, we use the obvious jumps
and the unique 2-isomorphisms from Lemma \ref{lem:unique} to match
everything (these are the unlabeled 1-morphisms and 2-morphisms).  The
image of the entire thin bigon $\G$ is then defined by vertical and
horizontal compositions of all the $s^{\pi} ( \G_{j} )$ so that
$s^{\pi}$ respects compositions. 

The 2-functor $s^{\pi}$ is a weak inverse to $p^{\pi}$ as in the case
for the path groupoid (Proposition 3.2.1. of \cite{SW3}). However, a
weak inverse in 2-category theory in this case means that there exists
a pseudo-natural equivalence $\z : s^{\pi} \circ p^{\pi} \Rightarrow
\id_{\mathcal{P}_{2}^{\pi} (M) }$ since $p^{\pi} \circ s^{\pi} =
\id_{\mathcal{P}_{2} (M)}.$ This means there exists a weak inverse to
$\z$ which is written as $\xi : \id_{\mathcal{P}_{2}^{\pi} (M) }
\Rightarrow s^{\pi} \circ p^{\pi}.$ ``Weak'' means that
there are invertible modifications $i_{\z} : \xi \circ \z \Rrightarrow
\id_{ s^{\pi} \circ p^{\pi} }$ and $j_{\z} :
\id_{\id_{\mathcal{P}_{2}^{\pi} (M)} } \Rrightarrow \z \circ \xi$
that satisfy the zig-zag identities. The details are irrelevant for our
purposes but can be found in Section 3.2 of \cite{SW3}. An important
consequence of $s^{\pi}$ being a weak inverse to $p^{\pi}$ is the
following (general categorical) fact reproduced here  for convenience
(Corollary 3.2.5. of \cite{SW3}).

\begin{cor}
\label{cor:section}
Any two weak inverses $s^{\pi}, s'^{\pi} : \mathcal{P}_{2} (M) \to
\mathcal{P}_{2}^{\pi} (M)$ of $p^{\pi}$ are pseudo-naturally equivalent.
\end{cor}

We can define such a pseudo-natural equivalence $\h : s^{\pi}
\Rightarrow s'^{\pi}$ by the following assignment $M \ni x \mapsto$
the jump from $s^{\pi}(x)$ to $s'^{\pi}(x)$ and $P^{1}M \ni \g \mapsto$
the unique 2-isomorphism $s^{\pi}(\g) \Rightarrow s'^{\pi}(\g)$
specified by Lemma \ref{lem:unique}. We will exploit this fact when
discussing examples of higher holonomy in Section \ref{sec:examples}. 

As before, the 2-functor $s^{\pi} : \mathcal{P}_{2} (M) \to
\mathcal{P}^{\pi}_{2} (M)$ induces a 2-functor $s^{\pi *} :
\mathrm{Funct} ( \mathcal{P}^{\pi}_{2} (M), T ) 
\rightarrow \mathrm{Funct}(\mathcal{P}_{2} (M),T),$ 
the pullback along $s^{\pi}.$ Similarly,
$\mathrm{Rec}^{2}_{\pi}$ is defined as the composition in the diagram 
\be
\label{eq:rec2functor}
\xy 0;/r.15pc/:
(-25,10)*+{\mathrm{Funct} ( \mathcal{P}_{2} (M) , T ) }="2";
(25,10)*+{\mathfrak{Des}^{2}_{\pi} (i) }="1";
(0,-10)*+{\mathrm{Funct} ( \mathcal{P}^{\pi}_{2} (M) , T ) }="3";
{\ar_{\qquad \mathrm{Rec}^{2}_{\pi}} "1";"2"};
{\ar^{R} "1";"3"};
{\ar^{s^{\pi *}} "3";"2"};
\endxy
.
\ee
As before, the image of $\mathfrak{Des}^{2}_{\pi} (i)$ under
$\mathrm{Rec}_{\pi}^{2}$ lands in $\mathrm{Triv}^{2}_{\pi} (i)$
and the definition is the same as it was before, only this time $\z$ is
a pseudo-natural equivalence between 2-functors between 2-categories. 

As a short summary, in this section we introduced a weak inverse functor
$\mathrm{Rec}_{\pi}^{2} : \mathfrak{Des}_{\pi}^{2} (i) \rightarrow
\mathrm{Triv}_{\pi}^{2} (i)$ for $\mathrm{Ex}_{\pi}^{2} :
\mathrm{Triv}_{\pi}^{2} (i) \to  \mathfrak{Des}_{\pi}^{2} (i)$ by using
the 2-groupoid  $\mathcal{P}_{2}^{\pi} (M)$ associated to the surjective
submersion $\pi : Y \to M$ to lift points, thin paths, and thin bigons
in $M$ to points, thin paths and/or jumps, and thin bigons and/or jumps
in $\mathcal{P}_{2}^{\pi} (M),$ respectively.

\subsection{Differential cocycle data}
\label{sec:2cocycle}

In this section, we will give a brief review of an equivalence between
differential forms and smooth 2-functors following Section 2 of
\cite{SW2}. This will allow us to describe parallel transport locally
in terms of differential cocycle data. We will leave out several proofs
but will provide pictures that we find illustrate the necessary ideas
behind the statements. We first remind the reader of the ``Lie algebra''
of a Lie crossed module. 

Given a Lie crossed module $(H, G, \t ,\a)$ 
(recall Definition \ref{defn:crossedmodule})
there is an associated
\emph{\uline{differential Lie crossed module}}  $(\un H, \un G, \un \t ,
\un \a),$ where $\un \t : \un H \to \un G$ is the differential of
$\t : H \to G,$ $\un \a : \un G \to \mathrm{Der} (\un H)$ is the
differential of the associated action (given the same name)
$\a : G \times H \to H$ (``Der'' stands for derivations). The
differential Lie crossed module data satisfy 
\be
\un \a_{\un \t (B') } ( B ) = [ B', B ] 
\ee
and 
\be
\un \t ( \un \a_{A} (B) ) = [ A, \un \t (B) ]
\ee
for all $A \in \un G$ and $B, B' \in \un H.$ 

Note that by restricting the action $\a$ to $\{ g \} \times H$ for any
$g \in G$ and differentiating with respect to the second coordinate,
we obtain a Lie algebra homomorphism $\underline{\a_{g}} : \un H \to
\un H.$ Both $\un \a$ and $\underline{\a_{g}}$ are important for
understanding the differential cocycle data of Section
\ref{sec:2cocycle}. A more thorough review can be found in \cite{BaHu}.

\subsubsection{From 2-functors to 2-forms} 

Let $\mathcal{B}\mathfrak{G}$ be a Lie 2-group and $(H, G , \t , \a)$
its corresponding crossed module. 
Given a strict smooth 2-functor $F : \mathcal{P}_{2} (X) \to \mathcal{B}
\mathfrak{G},$ we will obtain differential forms
$A \in \W^{1} (X; \un G)$ and $B \in \W^{2} ( X; \un H ).$ These will
form the objects of a 2-category $Z^{2}_{X} (\mathfrak{G}).$ By our
previous discussion and since our 2-categories $\mathcal{P}_{2} (X)$ and
$\mathcal{B} \mathfrak{G}$ are strict and the 2-functor $F$ is strict,
the restriction of $F$ to $\mathcal{P}_{1}(X)$ is smooth. Therefore, we
obtain a differential form $A \in \W^{1} (X; \un G)$ by the results of
Section \ref{sec:diffcocycledata}. To obtain the differential form
$B \in \W^{2} ( X; \un H )$ we will ``differentiate'' the composition 
\be
H \xleftarrow{p_{H}} H \rtimes G \xleftarrow{F_{2}} P^{2} X, 
\ee
where $p_{H}$ is the projection onto the $H$ factor and $F_{2}$ is $F$
restricted to 2-morphisms. 

Infinitesimally, a bigon is determined by a point and the two tangent
vectors that begin to span it. Therefore, let $x \in X$ and $v_{1},
v_{2} \in T_{x} X$ and let $\G : \R^{2} \to X$ be a smooth map such that
\be
\label{eq:xv1v2}
\G ( (0,0) ) = x, \quad \frac{\p }{\p s} \bigg |_{s=0} \G (s,t=0)
= v_{1}, \quad\&\quad \frac{\p }{\p t} \bigg |_{t=0} \G (s=0,t) = v_{2}.
\ee
Let $\S_{\R} : \R^{2} \to P^2 \R^2$ be the (smooth) map that sends
$(s,t)$ to the thin homotopy class of the bigon in Figure
\ref{fig:sigmar}. This is unambiguously defined after modding out by
thin homotopy because a thin bigon in $\R^2$ is determined by its
source and target thin paths in $\R^2.$ 

\begin{figure}[h]
\centering
\begin{picture}(0,0)
\put(105,90){$(s,t)$}
\put(253,93){$(s,t)$}
\put(110,70){$\xymatrix{\ar@/^0.5pc/@{|->}[rr]^{\S_{\R}} & & }$}
\end{picture}
    \includegraphics[width=0.60\textwidth]{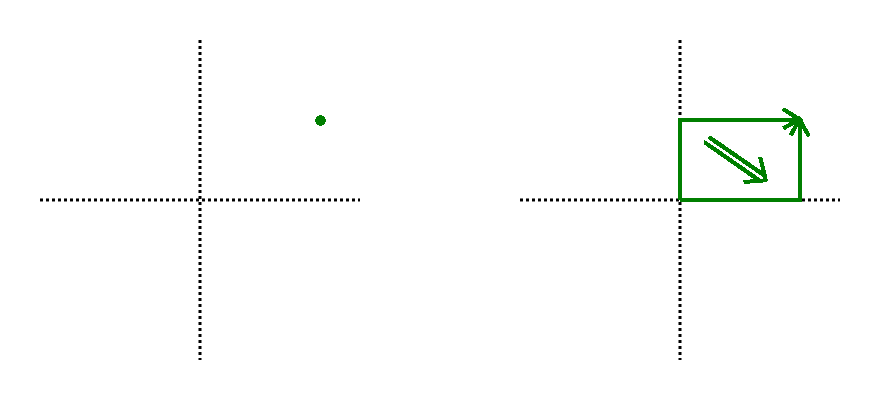}
    \vspace{-5mm}
    \caption{A point $(s,t)$ in $\R^{2}$ gets mapped to the bigon in
$\R^{2}$ shown on the right under the map $\S_{\R}.$}
    \label{fig:sigmar}
\end{figure}
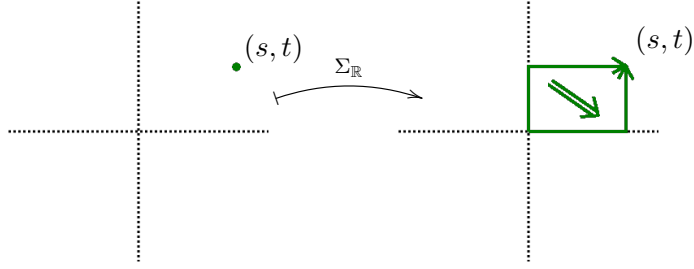 

Then we use this to define a smooth map $F_{\G}$ by the composition of
smooth maps 
\be
H \xleftarrow{p_{H}} H \rtimes G \xleftarrow{F_{2}} P^2 X
\xleftarrow{\G_{*}} P^2 \R^2 \xleftarrow{\S_{\R}} \R^{2}.
\ee
This gives an element of the Lie algebra $\un H$ by taking derivatives 
\be
B_{x} (v_{1}, v_{2} ) := - \frac{\p^2 F_{\G}}{ \p s \p t }
\bigg |_{(0,0)} \in \un H.
\ee
Furthermore, this element is independent of the choice of $\G$ provided
that equation (\ref{eq:xv1v2}) still holds. In fact, we get 
a smooth differential form $B \in \W^{2} ( X; \un H).$ 

Now let $\G : \g \Rightarrow \de$ be a thin bigon between two thin
paths. The \emph{source-target matching condition}, which says
$\t ( p_{H} ( F ( \G ) ) ) F ( \g ) = F ( \de),$ implies 
\be
\label{eq:B}
d A + \frac{1}{2} [ A, A ] = \un \t ( B ).
\ee
All of these claims are proved in Section 2.2.1 of \cite{SW2}.

\subsubsection{From 2-forms to 2-functors} 

Starting with a $\un G$-valued 1-form $A \in \W^{1} (X; \un G)$ on $X$
and a $\un H$-valued 2-form $B \in  \W^{2} (X; \un H)$ on $X$  we want
to define a smooth functor $\mathcal{P}_{2} (X) \to \mathcal{B}
\mathfrak{G}.$ From Section \ref{sec:1formstofunctors}, we have already
defined the functor at the level of objects and thin paths. What remains
is to define $F_{2} : P^{2} X \to  H \rtimes G.$  To do this, we will
define a function $k_{A,B} : B X \to H$ on bigons in $X$ (we do not mod
out by thin homotopy). Given a bigon $\S : [ 0 ,1 ] \times [0,1] \to X,$
we can pull back the 1-form $A$ and the 2-form $B$ to $[0,1] \times
[0,1],$ obtaining  $\S^{*} (A) \in \W^{1} ([0,1] \times [0,1] ; \un G )$
and $\S^{*} (B) \in \W^{2} ([0,1] \times [0,1] ; \un H).$ 

To define $k_{A,B},$ we first introduce an $\un H$-valued 1-form
$\mathcal{A}_{\S} \in \W^{1} ( [ 0 , 1 ] ; \un H)$ obtained by
integrating over one of the directions for the bigon. It is defined by
\be
\label{eq:Asigma}
( \mathcal{A}_{\S} )_{s} \left ( \frac{d}{ds} \right )
:= -\int_{0}^{1} dt \; \underline{\a_{F_{1} (\S_{*} \g_{s,t} )^{-1} } }
\left ( ( \S^{*} B )_{(s,t)} \left ( \frac{\p}{\p s} , \frac{\p }{\p t }
\right ) \right ) , 
\ee
where $\g_{s,t}$ is defined to be the straight vertical path from
$(s,0)$ to $(s,t)$ in $[0,1] \times [0,1]$ as in Figure
\ref{fig:gammast}. Note that in expression (\ref{eq:Asigma}), it is
assumed that $\S_{*} \g_{s,t}$ refers to the thin homotopy class of the
path (otherwise, applying the function $F_{1}$ would not make sense).
Therefore, the parametrization of $\g_{s,t}$ is irrelevant. 

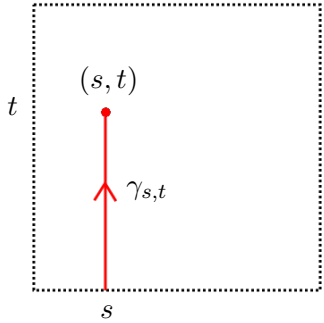
\begin{figure}[h]
\centering
  \begin{picture}(0,0)
\put(34,90){$(s,t)$}
\put(42,3){$s$}
\put(7,80){$t$}
\put(52,50){$\g_{s,t}$}
\end{picture}
    \includegraphics[width=0.3\textwidth]{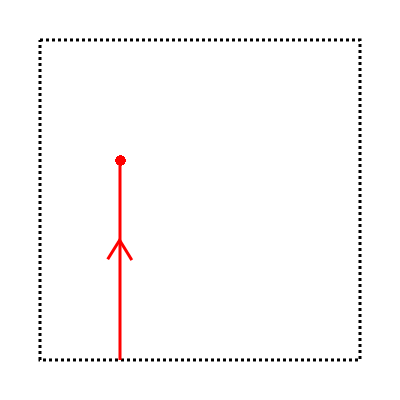}
    \vspace{-3mm}
    \caption{The path $\g_{s,t}$ is the straight vertical path from the
point $(s,0)$ to $(s,t)$ in $[0,1] \times [0,1].$}
    \label{fig:gammast}
\end{figure} 

Besides the path-ordered integral expression from the term
$F_{1} ( \S_{*} ( \g_{s,t} ) ),$ the expression for
$\mathcal{A}_{\S}$ is an ordinary integral. Also note that
$\mathcal{A}_{\S}$ depends on $\S.$ In particular, it is \emph{not}
invariant under thin homotopy. 

\begin{rmk}
Incidentally, although Schreiber and Waldorf in \cite{SW2} made their
own arguments for how to obtain such a formula for $\mathcal{A}_{\S},$
this formula appears in a special case as early as 1977 in the work of
Goddard, Nuyts, and Olive on magnetic monopoles \cite{GNO} on the
right-hand side of equation (2.9) and it may have been known earlier
\cite{Ch}. The special case \cite{GNO} considered is the case of the
crossed module $(G, G,\id, \a)$ with $\a$ being the ordinary
conjugation action.
\end{rmk}

Finally, to every bigon $\S : \g \Rightarrow \de,$ we define 
\be
\label{eq:surfaceintegral}
k_{A,B} (\S) := \a_{F_{1} ( \g ) } \left ( \mathcal{P} \exp
\left \{ - \int_{0}^{1} \mathcal{A}_{\S} \right \} \right ) .
\ee
 
In Figure \ref{fig:doublepathorderedintegral},
this integral is schematically drawn as a power series of graphs
with marked points and paths analogous to Figure
\ref{fig:pathorderedintegral}. Each of the paths drawn has a
path-ordered integral expression attached to it, and therefore each
expression has an additional power series of the form we
discussed for the ordinary path-ordered integral. 

\begin{figure}[h]
\centering
    \includegraphics[width=1.00\textwidth]{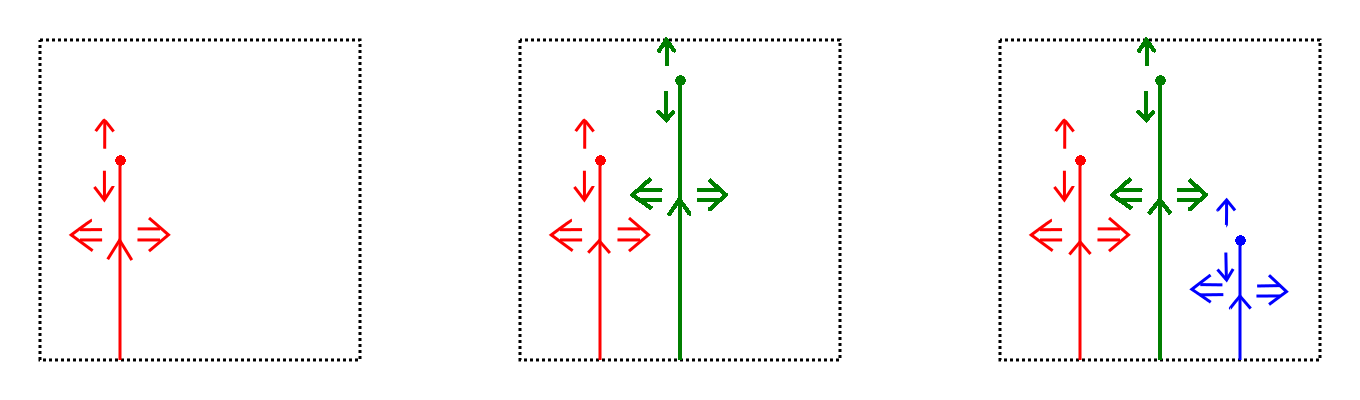}
    \vspace{-5mm}
    \caption{The path-ordered integral $\mathcal{P} \exp
\left \{-\int_{0}^{1} \mathcal{A}_{\S} \right \}$ is depicted
schematically as an infinite sum of terms expressed by placing $B$ at
the endpoints of the paths, along which we've computed parallel
transport using $A$ making sure to keep the later $s$-valued terms on
the right. The picture is to be interpreted similarly to the
one-dimensional case once we've integrated along the $t$ direction
(vertical) to obtain $\mathcal{A}_{\S}.$}
    \label{fig:doublepathorderedintegral}
\end{figure} 

\begin{defn}
The group element $k_{A, B} ( \S )$ is called the
\emph{\uline{surface transport}} associated to the bigon $\S$ and the
differential forms $A$ and $B.$
\end{defn}

$k_{A,B}$ only depends on the thin homotopy class of $\S$ and therefore
factors through a smooth map $F_{2} : P^{2} X \to H$ on thin homotopy
classes of paths. This map together with $F_{1}$ define a strict smooth
2-functor $F : \mathcal{P}_{2} (X) \to \mathcal{B} \mathfrak{G}$
(Proposition 2.17. of \cite{SW2}). 

\begin{rmk}
Historically, understanding the appropriate generalization of the
path-ordered integral to surfaces was a difficult task. It was not
obvious which formulas were correct or even what the criteria for
correctness should be. The language of functors allows one to make this
precise. The criteria for correctness is that surface holonomy should
be expressed in terms of a transport 2-functor. \emph{Any} formula that
satisfies these functorial properties, has the local constraint given
by equation (\ref{eq:B}), and changes appropriately under gauge
transformations (which we have so far only discussed globally but will
discuss differentially soon), can be rightfully called surface
transport. The specific formula in equation (\ref{eq:surfaceintegral})
is only one such formula that works. However, there could be many other,
potentially simpler formulas, that also describe 2-holonomy. In Section
\ref{sec:path-curvature} for instance, we prove that for certain
structure 2-groups, the formula  (\ref{eq:surfaceintegral}) agrees with
one that is easily computable in terms of homotopy classes of paths. 
\end{rmk}

\subsubsection{Local differential cocycles for transport 2-functors} 

By similar considerations to the previous sections, we can differentiate
transport functors and use their properties to obtain relations among
all the differential data. All the information in this section is
discussed in more detail in \cite{SW4}. In particular, the functions,
differential forms, and their relations are all derived. We merely
reproduce the results here for use in later calculations. 

\begin{defn}
\label{defn:Z2XG}
Let $Z^{2}_{X} (\mathfrak{G})^{\infty}$ be the category defined as
follows. An object of  $Z^{2}_{X} (\mathfrak{G})^{\infty}$ is a pair
$(A,B)$ of a 1-form $A \in \W^{1} (X ; \un G)$ and a 2-form
$B \in \W^{2} (X ; \un H)$ satisfying
\be
\un \t (B) = dA + \frac{1}{2} [A,A]. 
\ee 
A 1-morphism from $(A,B)$ to $(A',B')$ is a pair $(h, \varphi)$ of a
smooth map $h :  X \to G$ and a 1-form $\varphi \in \W^{1} ( X; \un H )$
satisfying 
\be
A' + \un \t ( \varphi ) = \mathrm{Ad}_{h} (A) - h^{*} \ov \theta
\ee
and
\be
B' + \un \a_{A'} ( \varphi ) + d \varphi
+ \frac{1}{2} [ \varphi, \varphi ] = \underline{\a_{g}} (B). 
\ee
The composition is defined by 
\be
(A'',B'') \xleftarrow{(h',\varphi')} (A',B') \xleftarrow{(h,\varphi)}
(A,B) := ( A'' ,B'') \xleftarrow{(h' h,\underline{\a_{h'}} (\varphi)
+ \varphi' ) } (A,B). 
\ee
A 2-morphism from $(h, \varphi)$ to $(h', \varphi'),$ which are both
1-morphisms from $(A,B)$ to $(A',B'),$ is a smooth map $f : X \to H$
satisfying 
\be
h' = \t ( f ) h
\ee
and
\be
\varphi' + ( \underline{ R_{f}^{-1} \circ \a_{f} } ) ( A')
= \mathrm{Ad}_{f} (\varphi) - f^{*} \ov \theta. 
\ee
The vertical composition is defined by 
\be
\xymatrix{
(A',B') & & \ar@/_2pc/[ll]_{(h, \varphi)}="1"
\ar[ll]|-{(h', \varphi')}="2" \ar@{}"1";"2"|(.25){\,}="3"
\ar@{}"1";"2"|(.85){\,}="4" \ar@{=>}"3";"4"^{f}
\ar@/^2pc/[ll]^{(h'', \varphi'')}="5" \ar@{}"2";"5"|(.15){\,}="6"
\ar@{}"2";"5"|(.75){\,}="7" \ar@{=>}"6";"7"^{f'}   (A,B)
}
\quad := \quad 
\xymatrix{
(A',B') & & \ar@/_1.5pc/[ll]_{(h, \varphi)}="1"
\ar@/^1.5pc/[ll]^{(h'', \varphi'') }="2" \ar@{}"1";"2"|(.15){\,}="3"
\ar@{}"1";"2"|(.85){\,}="4" \ar@{=>}"3";"4"^{f'f}  (A,B)
}
.
\ee
The horizontal composition is defined by 
\be
\xymatrix{
(A'',B'') &  \ar@/_1.5pc/[l]_{(h_1, \varphi_1)}="5"
\ar@/^1.5pc/[l]^{(h'_1, \varphi'_1)}="6" \ar@{}"5";"6"|(.15){\,}="7"
\ar@{}"5";"6"|(.85){\,}="8" \ar@{=>}"7";"8"^{f_1}  (A',B') &
\ar@/_1.5pc/[l]_{(h_2, \varphi_2)}="1"
\ar@/^1.5pc/[l]^{(h'_2, \varphi'_2)}="2" \ar@{}"1";"2"|(.15){\,}="3"
\ar@{}"1";"2"|(.85){\,}="4" \ar@{=>}"3";"4"^{f_2}  (A,B)
}
\
:=
\
\xymatrix{
(A'',B'') & & \ar@/_1.5pc/[ll]_{(h_1 h_2,\underline{\a_{h_1}}(\varphi_2)
+ \varphi_1)}="1" \ar@/^1.5pc/[ll]^{(h'_1 h'_2,\underline{\a_{h'_1}}
(\varphi'_2) + \varphi'_1)}="2" \ar@{}"1";"2"|(.15){\,}="3"
\ar@{}"1";"2"|(.85){\,}="4" \ar@{=>}"3";"4"|-{f_1 \a_{h_1} (f_2) } (A,B)
}
.
\ee
\end{defn}

As in Section \ref{sec:diffcocycles}, these arguments define 2-functors 
\be
\label{eq:integratedifferentiate2}
\xymatrix{ 
Z_{X}^{2} (\mathfrak{G})^{\infty} \ar@<+1ex>[r]^(0.4){\mathcal{P}_{X}} &
\ar@<+1ex>[l]^(0.6){\mathcal{D}_{X}} \mathrm{Funct}^{\infty}
(X, \mathcal{B} \mathfrak{G} )
}
,
\ee
which turn out to be strict inverses of each other (Theorem 2.21 of
\cite{SW2}). 

As before, this was for globally defined differential data corresponding
to globally trivial transport 2-functors. Transport 2-functors on $M$
are not necessarily of this type, but they are locally trivializable
via some surjective submersion $\pi : Y \to M$ and a $\pi$-local
$i$-trivialization. By similar arguments to the discussion in Section
\ref{sec:diffcocycles}, we are led to the following, rather long and
complicated, definition. 

\begin{defn}
\label{defn:diffcocyclespi2}
Let $\pi : Y \to M$ be a surjective submersion. Define the 2-category
$Z_{\pi}^{2} (\mathfrak{G})^{\infty}$ of \emph{\uline{differential cocycles
subordinate to $\pi$}} as follows. An object of $Z_{\pi}^{2}
(\mathfrak{G})^{\infty}$ is a tuple $( (A, B), (g, \varphi), \psi, f),$
where $(A,B)$ is an object in $Z_{Y}^{2} (G),$ $(g,\varphi)$ is a
1-morphism from $\pi_{1}^{*} (A,B)$ to $\pi_{2}^{*} (A,B)$ in
$Z_{Y^{[2]}}^{2} (\mathfrak{G}),$ $\psi$ is a 2-morphism from
$\id_{(A,B)}$ to $\D^{*} (g, \varphi)$ in $Z_{Y}^{2} (\mathfrak{G}),$
and $f$ is a 2-morphism from $\pi^{*}_{23} (g, \varphi) \circ
\pi_{12}^{*} (g, \varphi)$ to $\pi_{13}^{*} (g, \varphi).$ A 1-morphism
from $( (A,B), (g,\varphi ),\psi,f)$ to
$( (A',B'),(g',\varphi' ),\psi',f' )$ is tuple $( ( h, \phi ), \e ),$
where $(h, \phi)$ is a 1-morphism from $(A,B)$ to $(A',B')$ in
$Z_{Y}^{2} (\mathfrak{G} )$ and $\e$ is a 2-morphism from
$\pi_{2}^{*} ( h, \phi ) \circ (g, \varphi)$ to $(g', \varphi') \circ
\pi_{1}^{*} ( h, \phi )$ in $Z_{Y^{[2]}}^{2} (\mathfrak{G} ).$ A
2-morphism from $( (h, \phi ), \e)$ to $( (h', \phi' ), \e')$ is a
2-morphism $E$ from $(h, \phi)$ to $(h', \phi')$ in
$Z_{Y}^{2} (\mathfrak{G}).$ 
\end{defn}

The above generalizations produce functors 
\be
\label{eq:localintegratedifferentiate2}
\xymatrix{ 
Z_{\pi}^{2} (\mathfrak{G})^{\infty}
\ar@<+1ex>[r]^(0.65){\mathcal{P}_{\pi}} &
\ar@<+1ex>[l]^(0.35){\mathcal{D}_{\pi}} 
} \mathfrak{Des}_{\pi}^{2} (i)^{\infty}
\ee
exhibiting an equivalence of 2-categories whenever
$i : \mathcal{B} \mathfrak{G} \to T$ is an equivalence.

\subsection{Direct limits}
\label{sec:2limit}

In this section, we get rid of the dependence on the surjective
submersion in the categories introduced in the prequel. Several of our
2-categories depended on the choice of a surjective submersion. These
2-categories were $ \mathrm{Triv}_{\pi}^{2} (i)^{\infty},
\mathfrak{Des}^{2}_{\pi} (i)^{\infty},$ and
$Z^{2}_{\pi} (\mathfrak{G})^{\infty}.$ Changing the surjective
submersion gives a collection of 2-categories dependent on this
surjective submersion. One can take a limit over the collection of
surjective submersions in this case. This will be a slight
generalization of what was done in Section \ref{sec:1limits}. However,
there are subtle issues in terms of defining the many compositions. 

The general construction is done as follows. Let $S_{\pi}$ be a family
of 2-categories parametrized by surjective submersions $\pi : Y \to M$
and let $F (\z) : S_{\pi} \to S_{\pi \circ \z}$ be a family of
2-functors for every refinement $\z : Y' \to Y$ of $\pi$ satisfying the
condition that for any iterated refinement $\z ' : Y'' \to Y'$ and
$\z : Y' \to Y$ of $\pi : Y \to M$ then $F ( \z' \circ \z ) = F (\z')
\circ F (\z).$ In this case, an object of $S_{M} := \varinjlim_{\pi}
S_{\pi}$ is given by a pair $(\pi, X)$ of a surjective submersion
$\pi : Y \to M$ and an object $X$ of $S_{\pi}.$  A 1-morphism from
$(\pi_{1} , X_{1} )$ to $(\pi_{2} , X_{2})$ consists of a common
refinement 
\be
\xy 0;/r.15pc/:
(0,20)*+{Z}="Z";
(-20,0)*+{Y_{1}}="Y1";
(20,0)*+{Y_{2}}="Y2";
(0,-20)*+{M}="M";
{\ar^{\z} "Z";"M"};
{\ar_{y_{1}} "Z";"Y1"};
{\ar^{y_{2}} "Z";"Y2"};
{\ar_{\pi_{1}} "Y1";"M"};
{\ar^{\pi_{2}} "Y2";"M"};
\endxy
\ee
together with a 1-morphism $f : (F (y_{1}) ) (X_{1}) \to ( F(y_{2}) )
(X_{2})$ in $S_{\z}.$ It is written as a pair $(\z, f).$ The composition
\be
(\pi_{3}, X_{3} ) \xleftarrow{(\z_{23}, g)} (\pi_{2}, X_{2} )
\xleftarrow{(\z_{12}, f)} (\pi_{1}, X_{1} ) 
\ee
consists of the pullback refinement 
\be
\xy 0;/r.15pc/:
(0,30)*+{Z_{13}}="Z13";
(-20,10)*+{Z_{12}}="Z12";
(20,10)*+{Z_{23}}="Z23";
(-30,-10)*+{Y_{1}}="Y1";
(0,-10)*+{Y_{2}}="Y2";
(30,-10)*+{Y_{3}}="Y3";
(0,-30)*+{M}="M";
{\ar_{\z_{12}} "Z12";"M"};
{\ar^{\z_{23}} "Z23";"M"};
{\ar "Z12";"Y1"};
{\ar "Z12";"Y2"};
{\ar "Z23";"Y2"};
{\ar "Z23";"Y3"};
{\ar_{\pi_{1}} "Y1";"M"};
{\ar|-{\pi_{2}} "Y2";"M"};
{\ar^{\pi_{3}} "Y3";"M"};
{\ar_{\mathrm{pr}_{12}} "Z13";"Z12"};
{\ar^{\mathrm{pr}_{23}} "Z13";"Z23"};
\endxy
\ee
along with the composition $( F(\mathrm{pr}_{23}) ) (g) \circ
( F (\mathrm{pr}_{12}) ) (f).$ A 2-morphism from $(\z, f)$ to
$(\z', f')$ consists of an equivalence class of pairs $(\w, \a),$ where
$\w$ is a common refinement of $\z$ and $\z'$ as in the following
diagram 
\be
\xy 0;/r.15pc/:
(-15,15)*+{Z}="Z";
(15,15)*+{Z'}="Z'";
(0,35)*+{W}="W";
(-30,-10)*+{Y_{1}}="Y1";
(30,-10)*+{Y_{2}}="Y2";
(0,-30)*+{M}="M";
{\ar_(0.55){\z} "Z";"M"};
{\ar^(0.55){\z'} "Z'";"M"};
{\ar_{y_{1}} "Z";"Y1"};
{\ar|-(0.45){y_{2}} "Z";"Y2"};
{\ar|-(0.45){y'_{1}} "Z'";"Y1"};
{\ar^{y'_{2}} "Z'";"Y2"};
{\ar_{z} "W";"Z"};
{\ar^{z'} "W";"Z'"};
{\ar_{\pi_{1}} "Y1";"M"};
{\ar^{\pi_{2}} "Y2";"M"};
{\ar|-(0.6){\w} "W";"M"};
\endxy
\ee
and $\a$ is a 2-morphism $\a : F(z) (f) \Rightarrow F(z') (f').$ Two
such pairs $(\w_{1}, \a_{1})$ and $(\w_{2}, \a_{2})$ are equivalent if
they agree on the pullback. 

After getting rid of the specific choices of the surjective submersions,
we can take the limits of all the categories we have introduced. We make
the following notation, slightly differing from that of \cite{SW4}:
\begin{align}
\mathrm{Triv}^{2}_{M} (i)^{\infty} &:= \varinjlim_{\pi}
\mathrm{Triv}_{\pi}^{2} (i)^{\infty} \label{eq:triv2M}  \\
\mathfrak{Des}^{2}_{M} (i)^{\infty} &:= \varinjlim_{\pi}
\mathfrak{Des}^{2}_{\pi} (i)^{\infty} \label{eq:desc2M}\\
Z^{2} ( M; \mathfrak{G} )^{\infty} &:= \varinjlim_{\pi} Z^{2}_{\pi}
(\mathfrak{G} )^{\infty} \label{eq:Z2MG}. 
\end{align}
Then from our previous discussions, we collect the functors we have
introduced relating all these categories to
$\mathrm{Trans}^{2}_{\mathcal{B} \mathfrak{G}} (M, T)$ after taking such
limits over surjective submersions:
\be
\label{eq:Z2Des2Triv2Trans2}
\xymatrix{
Z^{2} ( M; \mathfrak{G} )^{\infty}   \ar@<+1ex>^{\mathcal{P}}[r] &
\ar@<+1ex>^{\mathcal{D}}[l]  \mathfrak{Des}^{2}_{M} (i)^{\infty}
\ar@<+1ex>^{\mathrm{Rec}^{2}}[r]   & \ar@<+1ex>^{\mathrm{Ex}^{2}}[l]
\mathrm{Triv}^{2}_{M} (i)^{\infty}  \ar@<+1ex>^(0.45){v}[r] &
\ar@<+1ex>^(0.55){c}[l] \mathrm{Trans}^{2}_{\mathcal{B} \mathfrak{G}}
(M, T)
}
\ee
Under the conditions that $i : \mathcal{B} \mathfrak{G} \to T$ is an
equivalence of categories, all of the above 2-functors are equivalence
pairs. Without the smoothness assumptions, a simpler version of some of
these equivalences is proven in Proposition 4.2.1. and Theorem 4.2.2. of
\cite{SW3} while the equivalences in (\ref{eq:Z2Des2Triv2Trans2}) are
proven in Theorem 3.2.2., Lemma 3.2.3., and Lemma 3.2.4. of \cite{SW4}.
Completely analogous versions of comments regarding the assumptions on
$i$ made before (\ref{eq:allequivalences1}) apply here as well. 

\subsection{Surface transport, 2-holonomy, and gauge invariance}
\label{sec:2-holonomy}

In Section \ref{sec:1-holonomy}, we described a procedure that began
with a transport functor and produced a group-valued parallel transport
operator for thin loops with markings. We discovered that the value of
holonomy changed by conjugation depending on the markings for the loops,
the choice of a local trivialization procedure, and by using an
isomorphic transport functor. In this section, we will analyze holonomy
along surfaces in an analogous manner. The main difference is that
bigons have source and target paths so that a closed surface has a
marking of one lower dimension, and is therefore not in general just a
point as it was for loops. For the examples we give later in this paper,
we specialize to spheres with a point marking. Such a surface is
depicted as a bigon from the constant loop at a point $x$ to itself (see
Figure \ref{fig:sphere} below and \cite{HS}). 
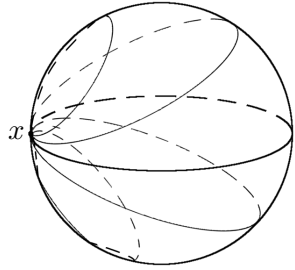
\begin{figure}[h]
\centering
  \begin{picture}(0,0)
\put(2,54){$x$}
\end{picture}
    \includegraphics[width=0.25\textwidth]{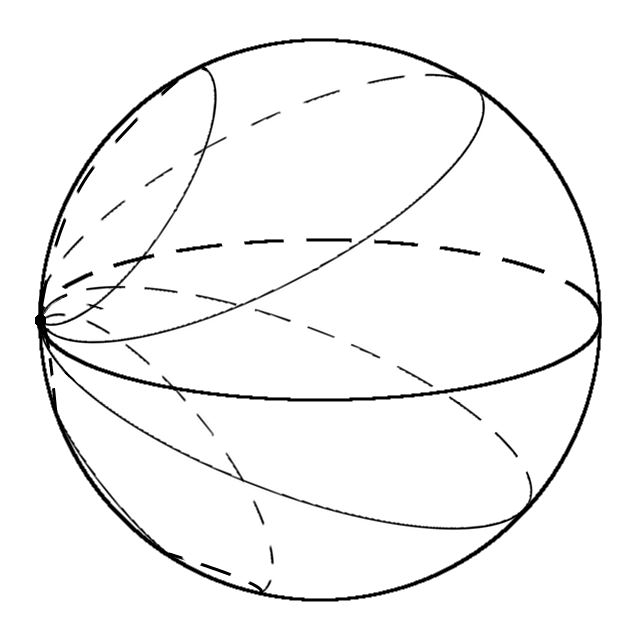}
    \vspace{-3mm}
    \caption{A based sphere viewed as a bigon $\S : \id_{x} \Rightarrow
\id_{x}.$}
    \label{fig:sphere}
\end{figure} 
However, a sphere can be more generally described as a bigon from a loop
to itself, so we analyze parallel transport for such bigons to cover
these extra cases. This analysis is completely independent of what types
of Lie 2-groups $\mathcal{B}\mathfrak{G}$ we use. For simplicity, we
assume that $i : \mathcal{B} \mathfrak{G} \to T$ is
a full and faithful 2-functor. This will differ from the presentation in
Section 5 of \cite{SW4}, where surface holonomy was defined using the
reduced 2-group. We will \emph{not} be making this restriction. 

\begin{defn}
\label{defn:groupvalued2holonomyfunctor}
A \emph{\uline{2-group-valued transport extraction}} is a composition of
functors (starting at the left and moving clockwise)
\be
\label{eq:groupvalued2holonomyfunctor}
\xy 0;/r.15pc/:
(-30,0)*+{\mathrm{Trans}_{\mathcal{B} \mathfrak{G}}^{2} (M, T)}="1";
(0,15)*+{\mathrm{Triv}^{2} (i)^{\infty}}="2";
(30,0)*+{\mathfrak{Des}^{2} (i)^{\infty}}="3";
(0,-15)*+{\mathrm{Triv}^{2} (i)^{\infty}}="4";
{\ar@/^1pc/^(0.4){c} "1";"2"};
{\ar@/^1pc/^(0.6){\mathrm{Ex}^{2}} "2";"3"};
{\ar@/^1pc/^(0.35){\mathrm{Rec}^{2}} "3";"4"};
{\ar@/^1pc/^(0.6){v} "4";"1"};
\endxy
.
\ee
\end{defn}

We write the composition (\ref{eq:groupvalued2holonomyfunctor}) as
$\scripty{t} \; .$ By the reconstruction procedure of  Section
\ref{sec:reconstruction2}, $\scripty{t} \;$ assigns $G$-valued elements
to thin paths for every transport functor $F$ as well as $H$-valued
elements to thin bigons (more on this below). Technically, thin bigons
will be assigned elements in $H \rtimes G$ but as is discussed in
Section \ref{sec:2-groups}, particularly after the proof of Theorem
\ref{thm:CrsMod2group}, such elements are completely determined by their
source, an element of $G,$ and their projection in $H.$ $\scripty{t} \;$
will also assign $G$-valued and $H$-valued gauge transformations for
every 1-morphism $\eta : F \to F'$ of transport functors. In addition,
$\scripty{t} \;$ will assign $H$-valued 2-gauge transformations for
every 2-morphism  $A : \eta \Rightarrow \eta'.$ A pseudo-natural
equivalence $\scripty{r} : \id \Rightarrow \scripty{t} \; $ describes
how to relate the transport functor to the locally trivialized one.
Although modifications of pseudo-natural transformations are allowed,
we will not analyze them in this paper. Such modifications are to be
interpreted as relating the two different ways of choosing the
pseudo-natural transformations that relate the transport functor to the
locally trivialized one. 

Just as before, we briefly review what the composition of 2-functors
defining $\scripty{t} \;$ are. For a transport 2-functor $F,$ we choose
a local trivialization $c(F) = ( \pi, F, \triv, t ).$ Then we extract
the local descent object $\mathrm{Ex}^{2} ( \pi, F , \triv, t ) = ( \pi,
\triv, g , \psi, f ).$ Then, we reconstruct a transport 2-functor
$\mathrm{Rec}^{2} ( \pi, \triv, g, \psi, f)$ and then forget the
trivialization data keeping just the 2-functor $v ( \mathrm{Rec}^{2}
(\pi, \triv, g, \psi, f ) ).$ The resulting transport 2-functor, written
as $\scripty{t}_{F},$ is defined by (see Section
\ref{sec:reconstruction2})
\be
\begin{split}
\mathcal{P}_{2}(M) &\xrightarrow{\scripty{t}_{F}} T \\
M \ni x &\mapsto i(\bullet) =: \triv_{i} ( s^{\pi} ( x) )  ,  \\
P^{1} M \ni \g &\mapsto R_{\mathrm{Ex}^{2} ( c(F) )} ( s^{\pi} (\g) ) ,
\text{ and }\\
P^{2} M \ni \S &\mapsto R_{\mathrm{Ex}^{2} ( c(F) )} ( s^{\pi} (\S) ) .
\end{split}
\ee
Points in $M$ get sent to $i(\bullet)$ by construction. Because $i$ is
full and faithful, the 1-morphisms $R_{\mathrm{Ex}^{2} ( c (F) )}
(s^{\pi} (\g) ) : i(\bullet) \to i(\bullet)$ determine unique elements
of $G.$ Similarly, the 2-morphisms $R_{\mathrm{Ex}^{2} ( c (F) )}
(s^{\pi} (\S) )$ determine unique elements in $H.$ 

The interested reader can explicitly define the compositor and the
unitor for the 2-functor $\scripty{t}_{F}.$ We won't need the precise
details for our analysis when studying surface holonomy. All we need to
know is that the 2-functors defining $\scripty{t} \; \;$ are (weakly)
invertible. 

We'd like to restrict surface holonomy to thin homotopy classes of
marked spheres for the purpose of this paper (in general, one would like
to restrict to the more general space of thin homotopy classes of marked
closed surfaces) and eventually thin free spheres. First we make a
definition of the thin marked sphere space, which should be thought of
as analogous to the thin marked loop space. 

\begin{defn}
\label{defn:thinmarkedsphere}
The \emph{\uline{marked sphere space of $M$}} is the set
\be
\label{eq:markedspherespace}
\mathfrak{S} M := \{ \S \in BM \ | \ s(\S) = t(\S)
\text{ and } s ( s ( \S ) ) = t ( t ( \S ) ) \}
\ee
equipped with the subspace smooth structure. Elements of
$\mathfrak{S} M$ are called \emph{\uline{marked spheres}}. Similarly, the
\emph{\uline{thin marked sphere space of $M$}} is the smooth space
\be
\label{eq:thinmarkedsphere}
\mathfrak{S}^{2} M = \{ \S \in P^{2} M  \ | \ s ( \S ) = t ( \S )
\text{ and } s ( s ( \S ) ) = t ( t ( \S ) ) \} . 
\ee
Elements of $\mathfrak{S}^{2} M $ are called
\emph{\uline{thin marked spheres}}.  
\end{defn}

\begin{rmk}
Note that elements of $\mathfrak{S}^{2} M$ need not look like embedded
spheres in $M.$ Indeed, they might look like pinched croissants as
Figure \ref{fig:croissant} indicates (or worse). This won't matter in
any of our calculations or proofs.
\end{rmk}

\begin{figure}[h]
\vspace{-5mm}
\centering
    \includegraphics[width=0.25\textwidth]{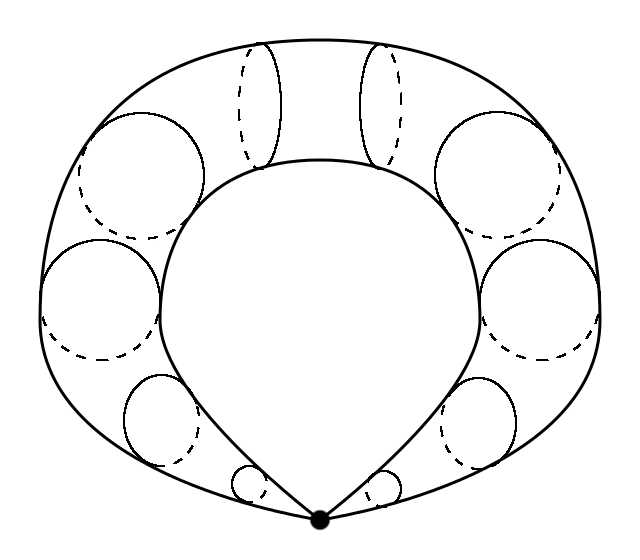}
    \vspace{-5mm}
    \caption{A pinched croissant is an example of a thin marked sphere.}
    \label{fig:croissant}
\end{figure} 

\begin{defn}
\label{defn:t2holonomy}
The \emph{\uline{$\scripty{t} \;$-2-holonomy of $F$}}, written as
$\mathrm{hol}_{\scripty{t}}^{F},$ is defined as the projection to $H$ of
the restriction of parallel transport of a transport 2-functor $F$ to
the thin marked sphere space of $M$:
\be
\mathrm{hol}_{\scripty{t}}^{F} := p_{H} \circ \scripty{t}_{F}
\Big |_{\mathfrak{S}^{2} M} : \mathfrak{S}^{2} M \to H. 
\ee
\end{defn}

\begin{rmk}
Note that $\mathrm{hol}_{\scripty{t}}^{F}$ is the same notation used for
thin loops with values in $G.$ This should cause no confusion because
thin loops are always written using lower case Greek letters such as
$\g, \de,$ etc. while thin spheres are written using upper case Greek
letters such as $\S, \G,$ etc. 
\end{rmk}

We now pose three questions analogous to those for 1-holonomy. 

\begin{enumerate}[i)]
\item
How does $\mathrm{hol}_{\scripty{t}}^{F}$ depend on the choice of a thin
marked sphere? Namely, suppose that two thin marked spheres $\S$ and
$\S',$ with possibly different markings, are thinly homotopic
\emph{without preserving the marking} (see Definition
\ref{defn:thinsphere}). Then, how is $\mathrm{hol}_{\scripty{t}}^{F}
(\S)$ related to $\mathrm{hol}_{\scripty{t}}^{F} (\S')$? 

\item
How does $\mathrm{hol}_{\scripty{t}}^{F}$ depend on $F$? Namely, suppose
that $\eta : F \to F'$ is a morphism of transport functors. How is
$\mathrm{hol}_{\scripty{t}}^{F}$ related to
$\mathrm{hol}_{\scripty{t}}^{F'}$ in terms of $\eta$?
\item
How does $\mathrm{hol}_{\scripty{t}}^{F}$ depend on $\scripty{t} \; ,$
the choice of trivialization? Namely, suppose that
$\scripty{t} \; \; {}' \;$ is another  trivialization. Then how is
$\mathrm{hol}_{\scripty{t}}^{F}$ related to
$\mathrm{hol}_{\scripty{t} \; \; {}' }^{F}$? 
\end{enumerate}

Due to the fact that we are restricting ourselves to marked spheres
instead of arbitrary surfaces, the answer will be closely related to the
1-holonomy case and will be given by a generalized version of
conjugation. As before, we need to define what we mean by thin free
sphere space and then we'll proceed to answer the above questions. 
Denote the smooth space of spheres in $M$ by 
$SM = \{ \S : S^2 \to M \ | \ \S \text{ is smooth} \}.$

\begin{defn}
\label{defn:thinsphere}
Two smooth spheres $\S$ and $\S'$  in $M$ are
\emph{\uline{thinly homotopic}} if there exists a smooth map
$h: S^2 \times [0,1] \to M$ such that 
\begin{enumerate}[i)]
\item
there exists an $\e > 0$ with $h ( t , s ) = \S (t)$ for $s \le \e$ and
$h (t, s ) = \S' (t)$ for $s \ge \e$ and for all $t \in S^2$ and  
\item
the smooth map $h$ has rank $\le 2.$ 
\end{enumerate}
The space of equivalences classes is denoted by $S^{2} M$ and is called
the \emph{\uline{thin free sphere space of $M$}}. Elements of $S^{2} M$ are
called \emph{\uline{thin spheres}}. 
\end{defn}

\begin{defn}
Define a function $f : \mathfrak{S} M \to SM$ by sending a marked sphere
$\S : [0,1] \times [0,1] \to M$ to the associated smooth map
$f(\S) : S^2 \to M$ obtained from identifying the top and bottom of the
second interval and then pinching the two ends (see Figure
\ref{fig:bigontosphere}).
$f$ is called the \emph{\uline{forgetful map}}. 
\end{defn}

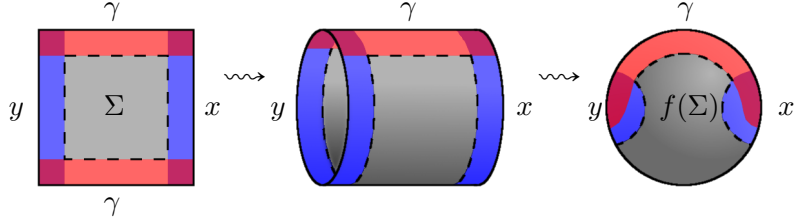
\begin{figure}[h!]
\vspace{-2mm}
\centering
	\begin{picture}(0,0)
	\put(2,36){$y$}
	\put(76,36){$x$}
	\put(38,74){$\g$}
	\put(38,0){$\g$}
	\put(38,36){$\S$}
	\put(83,48){$\rightsquigarrow$}
	\put(101,36){$y$}
	\put(194,36){$x$}
	\put(150,74){$\g$}
	\put(202,48){$\rightsquigarrow$}
	\put(221,36){$y$}
	\put(293,36){$x$}
	\put(255,74){$\g$}
	\put(248,36){$f(\S)$}
	\end{picture}
    \includegraphics[width=0.65\textwidth]{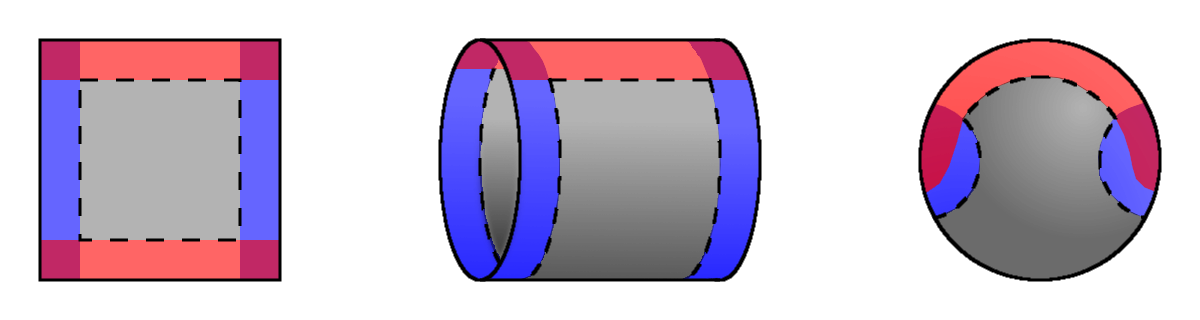}
    \caption{The definition of $f : \mathfrak{S} M \rightarrow SM.$ This
definition makes sense even when $y \ne x.$ $y=x$ is a special case.}
    \label{fig:bigontosphere}
\end{figure} 

\begin{lem}
There exists a unique map $f^{2} : \mathfrak{S}^{2}M \to S^2 M$ such
that the diagram 
\be
\xy 0;/r.25pc/:
(-10,7.5)*+{\mathfrak{S}M}="1";
(10,7.5)*+{\mathfrak{S}^{2} M}="2";
(-10,-7.5)*+{SM}="3";
(10,-7.5)*+{S^2 M}="4";
{\ar"1";"2"};
{\ar"1";"3"_{f}};
{\ar"3";"4"};
{\ar"2";"4"^{f^{2}}};
\endxy
\ee
commutes (the horizontal arrows are the projections onto thin homotopy
classes). 
\end{lem}

\begin{proof}
The proof is analogous to the case of loops. One chooses a
representative, applies $f,$ and then projects. The map is well-defined
by the thin homotopy equivalence relation on $S^2 M.$ 
\end{proof}

Note that there is also a function $\mathrm{ev}_{1} : \mathfrak{S}^2 M
\to \mathfrak{L}^1 M$ given by evaluating a thin marked sphere at its
source/target. This function forgets the sphere and remembers only the
source thin marked loop. 

\begin{defn}
\label{defn:spheremarking}
A \emph{\uline{marking of thin spheres}} is a section $\mathfrak{m} :
S^2 M \to \mathfrak{S}^{2} M $ of $f^2 : \mathfrak{S}^{2} M \to S^2 M.$
\end{defn}

\begin{lem}
A marking of thin spheres exists. 
\end{lem}

\begin{proof}
Let $[ \S ] \in \mathfrak{S}^2 M$ be a thin sphere and choose
representative $\S : S^2 \to M$ in $SM.$ Pick a point $\bullet$ on the
equator viewed as a loop $\ell : \bullet \to \bullet.$ The image of
$\ell$ under $\S$ defines a loop, $\g : x \to x,$ where
$x := \S(\bullet).$ There exists a thin homotopy
$h : S^2 \times [0,1] \to M$ from $\S$ to a smooth map
$\S_{\ell} : S^2 \to M$ such that the family of loops in Figure
\ref{fig:lamespheremarking} on the domain of $\S_{\ell}$ define a marked
sphere $\tilde{\S} : \g \Rightarrow \g.$ Projecting to thin marked
spheres defines $\mathfrak{m}([\S]).$ To see that this is well-defined,
let $\S' \in SM$ be another representative. Then there exists a thin
unmarked homotopy $\tilde{h} : \S' \Rightarrow \S.$ Composing this with
the thin homotopy $h$ gives $h \circ \tilde{h} : \S' \Rightarrow
\S_{\ell}.$ By the thin homotopy equivalence relation on $S^2 M,$ this
defines a section of $f^2.$ 
\end{proof}

\begin{figure}[h!]
\vspace{-5mm}
\centering
	\begin{picture}(0,0)
	\put(0,53){$\ell$}
	\end{picture}
    \includegraphics[width=0.25\textwidth]{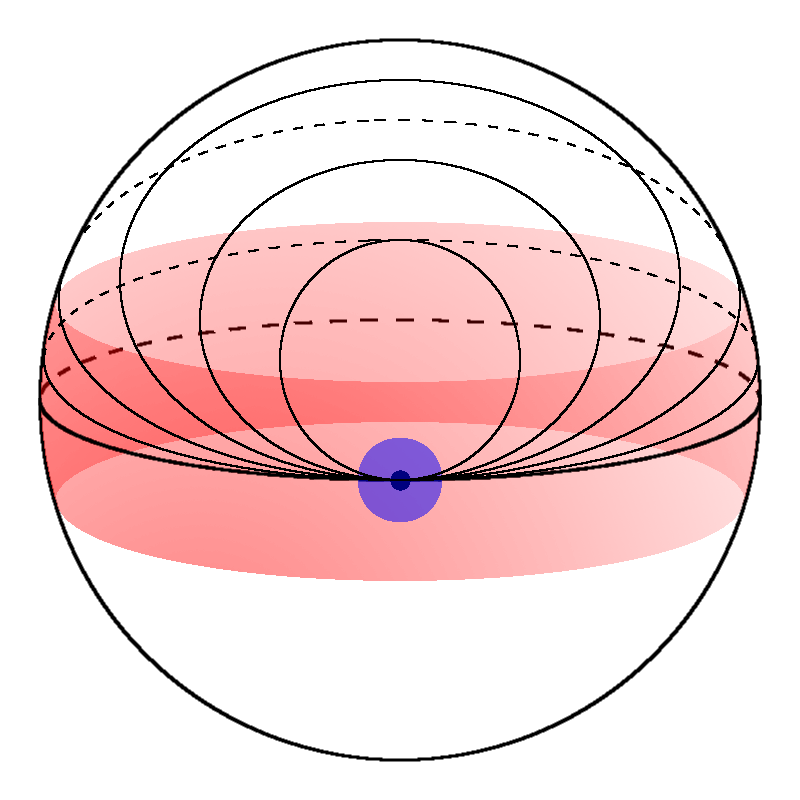}
    \vspace{-3mm}
    \caption{By a thin homotopy, the region around the equator is made
to sit at the loop $\ell$ around the equator so that the nearby loops
drawn in the shaded region agree with $\ell.$ The family of all these
loops define a marking.}
    \label{fig:lamespheremarking}
\end{figure}
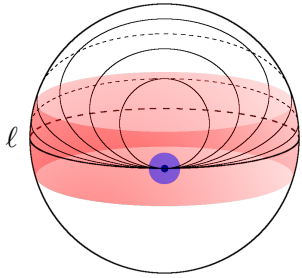 

We now proceed to answering the above questions in order. 

\begin{enumerate}[i)]
\item
Let $\mathfrak{m}, \mathfrak{m}' : S^2 M \to \mathfrak{S}^2 M$ be two
markings for thin spheres in $M.$ Let $[\S] \in S^2 M$ be a thin sphere
and let $\S : \g \Rightarrow \g$ with $\g : x \to x$ be a representative
of $\mathfrak{m}([\S])$ and $\S' : \g' \Rightarrow \g'$ with
$\g' : x' \to x'$ be a representative of $\mathfrak{m}'([\S]).$ Note
that these representatives need not have associated marked loops that
lie on some common image. Figure \ref{fig:twospheres} depicts such a
possible situation. 

\begin{figure}[h!]
\vspace{-5mm}
\centering
	\begin{picture}(0,0)
	\put(24,85){$x$}
	\put(242,77){$x'$}
	\put(238,100){$\g'$}
	\put(3,85){$\g$}
	\put(40,85){$\S$}
	\put(205,85){$\S'$}
	\end{picture}
    \includegraphics[width=0.53\textwidth]{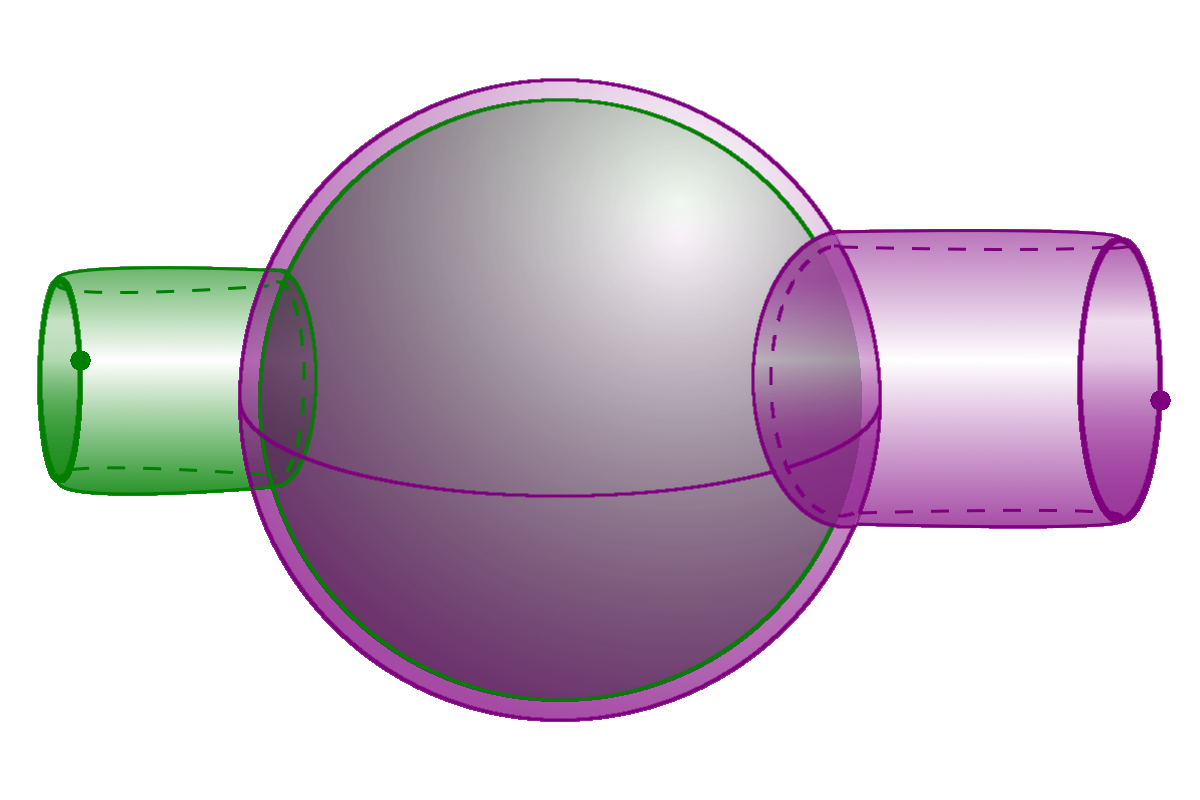}
    \vspace{-5mm}
    \caption{Two different representatives $\S$ (the `inner' sphere in
green extending left) and $\S'$ (the `outer' sphere in purple extending
right) of two markings of a thin sphere are shown. The extensions do
\emph{not} enclose any volume so that both spheres are thinly homotopic.
Their respective sources are $\g : x \rightarrow x$ 
and $\g' : x' \rightarrow x',$
neither of which lie on the other's image. Compare this to Figure
\ref{fig:basedspheres} where the two marked loops \emph{do} lie on a
common sphere.}
    \label{fig:twospheres}
\end{figure}
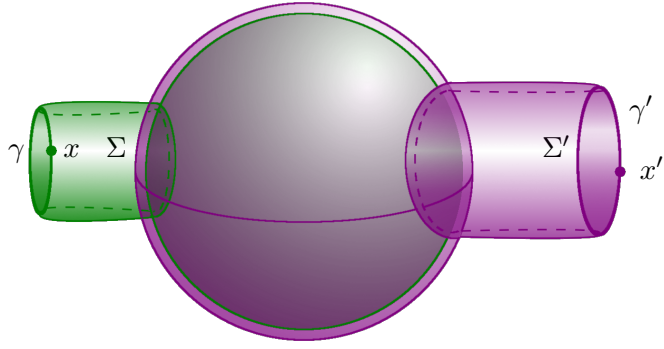 

As in the case of loops, we can use thin homotopy to draw both marked
loops on the same sphere (a more precise statement will be given
shortly). First notice that there is a thin homotopy $h : S^{2} \times
[0,1] \to M$ with $h ( \ \cdot \ , s ) = \S$  for $s \le \e$ and
$h ( \ \cdot \ , s ) = \S'$ for $s \ge 1 - \e$ for some $\e > 0.$ Such a
homotopy allows us to choose a sphere $\tilde{\S} \in SM,$ a path
$\g_{x' x} : x \to x',$ and three \emph{bigons}
$\S_{\g x} : \id_{x} \Rightarrow \g,$ $\S_{x' \g'} : \g' \Rightarrow
\id_{x'},$ and $\Delta : \g_{x' x} \circ \g \circ \overline{\g_{x' x}}$
with the following properties. First $\tilde{\S}$ can be expressed as
either of the compositions
\be
f \left (
\begin{matrix}
\S_{\g' x'} \\
\overset{\circ}{\id_{\g_{x' x}} \circ \S_{\g x } \circ
\id_{\overline{\g_{x' x}} }} \\
\overset{\circ}{\Delta}
\end{matrix}
\right )
\qquad 
\text{ or } 
\qquad 
f \left (
\begin{matrix}
\id_{\overline{\g_{x'x}}} \circ \Delta \circ \id_{\g_{x'x}} \\
\overset{\circ}{\id_{\overline{\g_{x'x}}} \circ \S_{\g' x' } \circ
\id_{\g_{x'x}}} \\
\overset{\circ}{\S_{\g x}}
\end{matrix}
\right )
\ee 
(in either order vertically). Second, the composition of \emph{bigons} 
\be
\begin{matrix}
\id_{\overline{\g_{x'x}}} \circ \Delta \circ \id_{\g_{x'x}} \\
\overset{\circ}{\id_{\overline{\g_{x'x}}} \circ \S_{\g' x' } \circ
\id_{\g_{x'x}}} \\
\overset{\circ}{\S_{\g x}}
\end{matrix}
\ee
is thinly homotopic to $\S$ preserving the marked loop $\g : x \to x.$
Third, the composition of \emph{bigons}  
\be
\begin{matrix}
\S_{\g' x'} \\
\overset{\circ}{\id_{\g_{x' x}} \circ \S_{\g x } \circ
\id_{\overline{\g_{x' x}} }} \\
\overset{\circ}{\Delta}
\end{matrix}
\ee 
is thinly homotopic to $\S'$ preserving the marked loop
$\g' : x' \to x'.$  This is depicted in Figures
\ref{fig:thinhomotopysphere} and \ref{fig:capDcup}. 

\begin{figure}[h]
\centering
	\begin{picture}(0,0)
	\put(97,105){$x$}
	\put(110,25){$x'$}
	\put(115,115){$\g$}
	\put(85,117){$\S$}
	\put(180,115){$\S'$}
	\put(172,40){$\g'$}
	\put(150,55){$\tilde{\S}$}
	\put(110,80){$\g_{x' x}$}
	\end{picture}
    \includegraphics[width=0.40\textwidth]{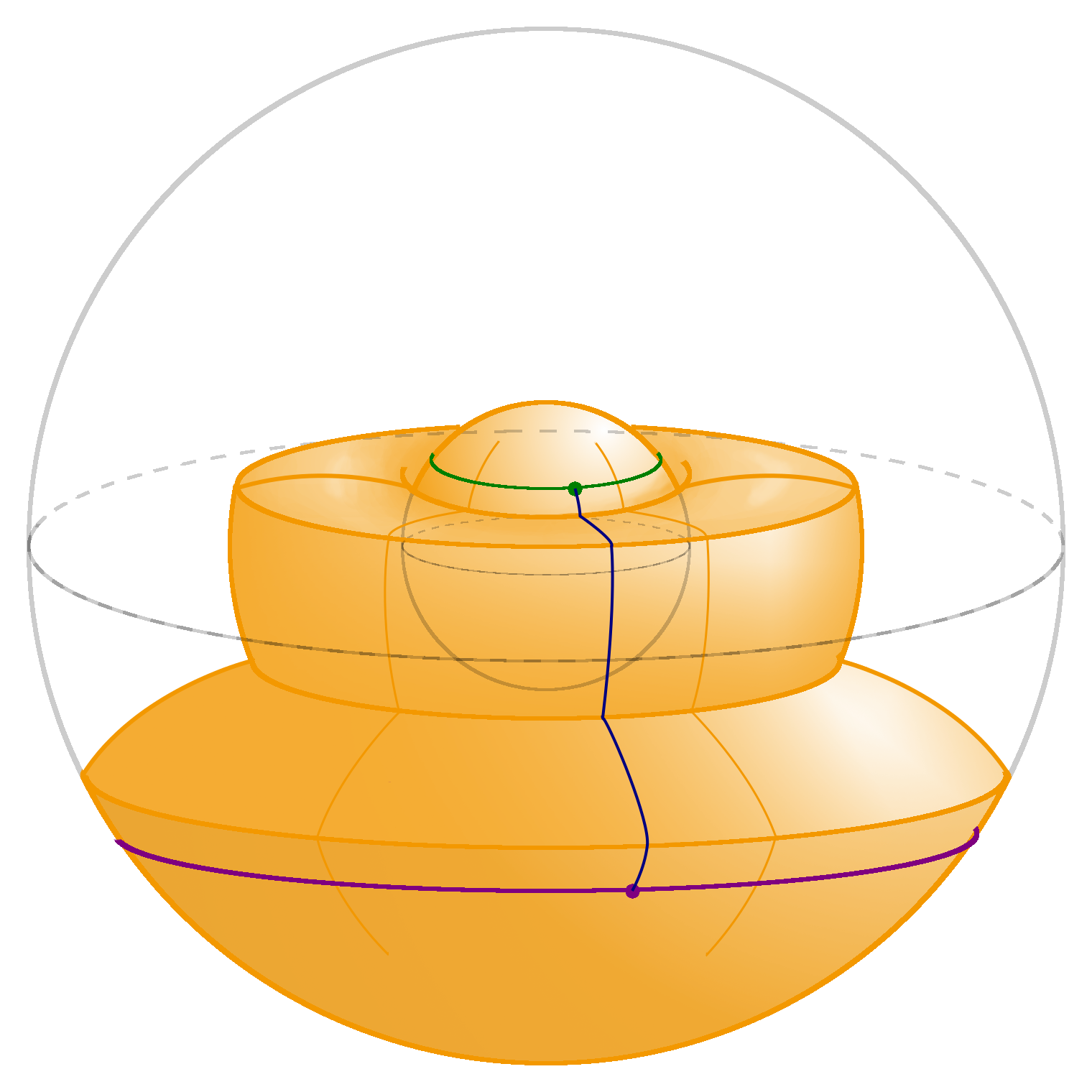}
    \vspace{-3mm}
    \caption{The domain of the homotopy 
    $h : S^{2} \times [0,1] \rightarrow M$
is drawn as a solid ball with a smaller solid ball removed from the
center. It depicts $\S$ as the inner sphere
and $\S'$ as the outer sphere. The marked loop 
$\g : x \rightarrow x$
of $\S$ is drawn on the northern hemisphere while the marked loop
$\g' : x' \rightarrow x'$ of $\S'$ is 
drawn on the southern hemisphere (by a
thin homotopy, one can always position the marked loops in this way).
The homotopy $h$ allows us to choose a sphere $\tilde{\S},$ drawn
somewhat in the middle (in orange), that contains both based loops $\g$
and $\g'$ and is thinly homotopic to both $\S$ and $\S'.$ As a result,
there exists a path $\g_{x' x} : x \rightarrow x'$ on $\tilde{\S}.$ We continue
this analysis in Figure \ref{fig:capDcup}.}
    \label{fig:thinhomotopysphere}
\end{figure}
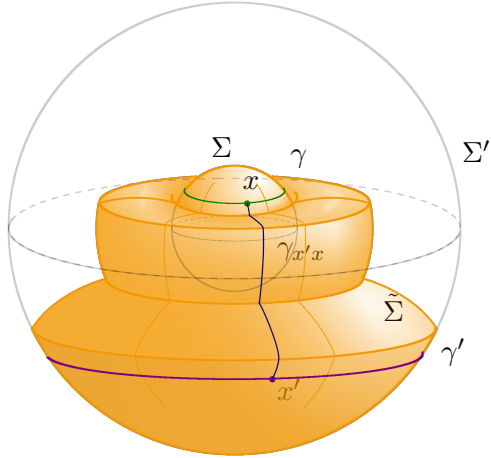 

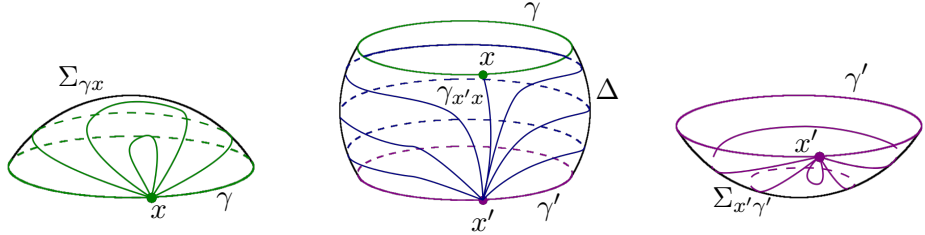
\begin{figure}[h]
	\centering
	\begin{picture}(0,0)
	\put(90,5){$\g$}
	\put(65,0){$x$}
	\put(30,50){$\S_{\g x}$}
	\end{picture}
	\begin{subfigure}[b]{0.24\textwidth}
		\includegraphics[width=\textwidth]{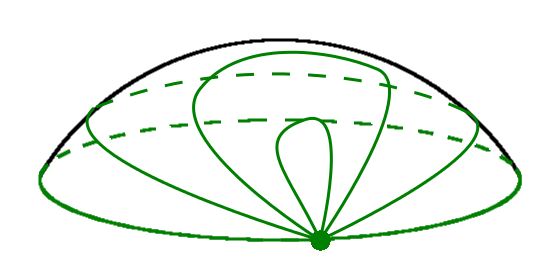}
	\end{subfigure}
	\quad
	\begin{picture}(0,0)
	\put(80,77){$\g$}
	\put(63,58){$x$}
	\put(85,1){$\g'$}
	\put(60,-5){$x'$}
	\put(47,45){$\g_{x' x}$}
	\put(107,45){$\D$}
	\end{picture}
	\begin{subfigure}[b]{0.24\textwidth}
		\includegraphics[width=\textwidth]{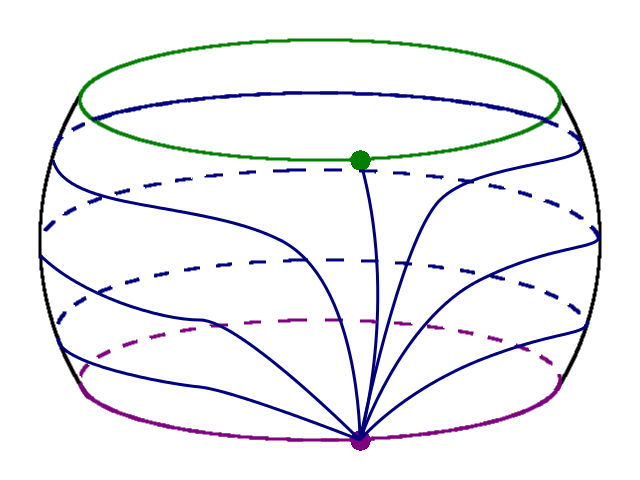}
	\end{subfigure}
	\quad
	\begin{picture}(0,0)
	\put(75,50){$\g'$}
	\put(56,25){$x'$}
	\put(25,4){$\S_{x' \g'}$}
	\end{picture}
	\begin{subfigure}[b]{0.24\textwidth}
		\includegraphics[width=\textwidth]{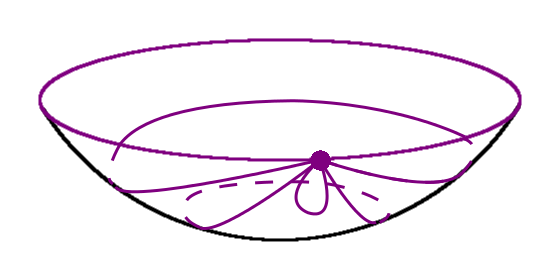}
	\end{subfigure}
    \caption{From the sphere $\tilde{\S}$ in Figure
\ref{fig:thinhomotopysphere}, the top cap defines a bigon
$\S_{\g x} : \id_{x} \Rightarrow \g,$ drawn on the left in this figure.
The path $\g_{x' x} : x \rightarrow x'$ in Figure  \ref{fig:thinhomotopysphere}
defines a bigon $\D : \g_{x' x} \circ \g \circ \overline{ \g_{x' x} }
\Rightarrow \g'$ drawn in the middle of this figure. The bottom cap
defines a bigon $\S_{x' \g'} : \g' \Rightarrow \id_{x'}$ drawn on the
right.}
    \label{fig:capDcup}
\end{figure}     

These last two equations let us write the bigon $\S$ in terms of $\S'$
and vice versa. In fact, we have
\be
\S ' = 
\begin{matrix}
\ov \Delta \\
\overset{\circ}{\id_{ \g_{x' x} } \circ \S 
\circ \id_{\overline{\g_{x'x}}}} \\
\overset{\circ}{\Delta}
\end{matrix}
\ee
up to thin homotopy preserving the marked loop $\g' : x' \to x'.$ There
is also a similar expression for $\S$ preserving the marked loop
$\g: x \to x.$  

The above argument says that given two marked spheres, with possibly
different markings, that are thinly homotopic \emph{without} preserving
the markings, one can always choose a representative of such a thin
sphere in $M$ with two marked loops so that the associated two
\emph{marked} spheres (coming from starting at either marking) are
thinly homotopic to the original two with a thin homotopy that preserves
the marking. More precisely, we proved the following. 

\begin{lem}
\label{lem:spheremarking}
Let $\mathfrak{m}, \mathfrak{m}' : S^{2} M \to \mathfrak{S}^2 M$ be two
markings. Let $[ \S ] \in S^{2} M$ be a thin sphere in $M$ and write
$[\g] : x \to x$ for $\mathrm{ev}_{1} ( \mathfrak{m} ([\S]))$ and
$[\g']:x'\to x'$ for $\mathrm{ev}_{1} ( \mathfrak{m}' ([\S])).$
Then, there exists representatives $\g$ and $\g'$ of $[\g]$ and $[\g'],$
respectively,
a path $\g_{x' x} : x \to x'$ with sitting instants and three bigons
$\S_{\g x} : \id_{x} \Rightarrow \g,$ $\S_{x' \g'} : \g' \Rightarrow
\id_{x'},$ and $\Delta : \g_{x'x} \circ \g \circ \overline{\g_{x'x}}
\Rightarrow \g',$ such that the following three properties hold (see
Figure \ref{fig:basedspheres}). 
\begin{enumerate}[i)]
\item
The vertical composition of  $\S_{\g' x'},$ $\id_{\g_{x' x}} \circ
\S_{\g x } \circ \id_{\overline{\g_{x' x}} },$ and $\Delta$ in the order
given (or a cyclic permutation of this order) and forgetting the marking
is a representative of $[\S].$  
\item
$\left ( \begin{smallmatrix}
\id_{\overline{\g_{x'x}}} \circ \Delta \circ \id_{\g_{x'x}} \\
\overset{\circ}{\id_{\overline{\g_{x'x}}} \circ \S_{\g' x' }
\circ \id_{\g_{x'x}}} \\
\overset{\circ}{\S_{\g x}}
\end{smallmatrix} \right )$ is a representative of $\mathfrak{m}([\S])$
as a bigon.  
\item
$\left ( \begin{smallmatrix}
\S_{\g' x'} \\
\overset{\circ}{\id_{\g_{x' x}}\circ\S_{\g x }\circ\id_{\overline{\g_{x'x}}}}\\
\overset{\circ}{\Delta}
\end{smallmatrix} \right )$ is a representative of $\mathfrak{m}'([\S])$
as a bigon.  
\end{enumerate}
\end{lem}

\begin{figure}[h]
	\centering
	\begin{picture}(0,0)
	\put(95,85){$\g$}
	\put(76,80){$x$}
	\put(30,54){$\S$}
	\end{picture}
	\begin{subfigure}[b]{0.24\textwidth}
		\includegraphics[width=\textwidth]{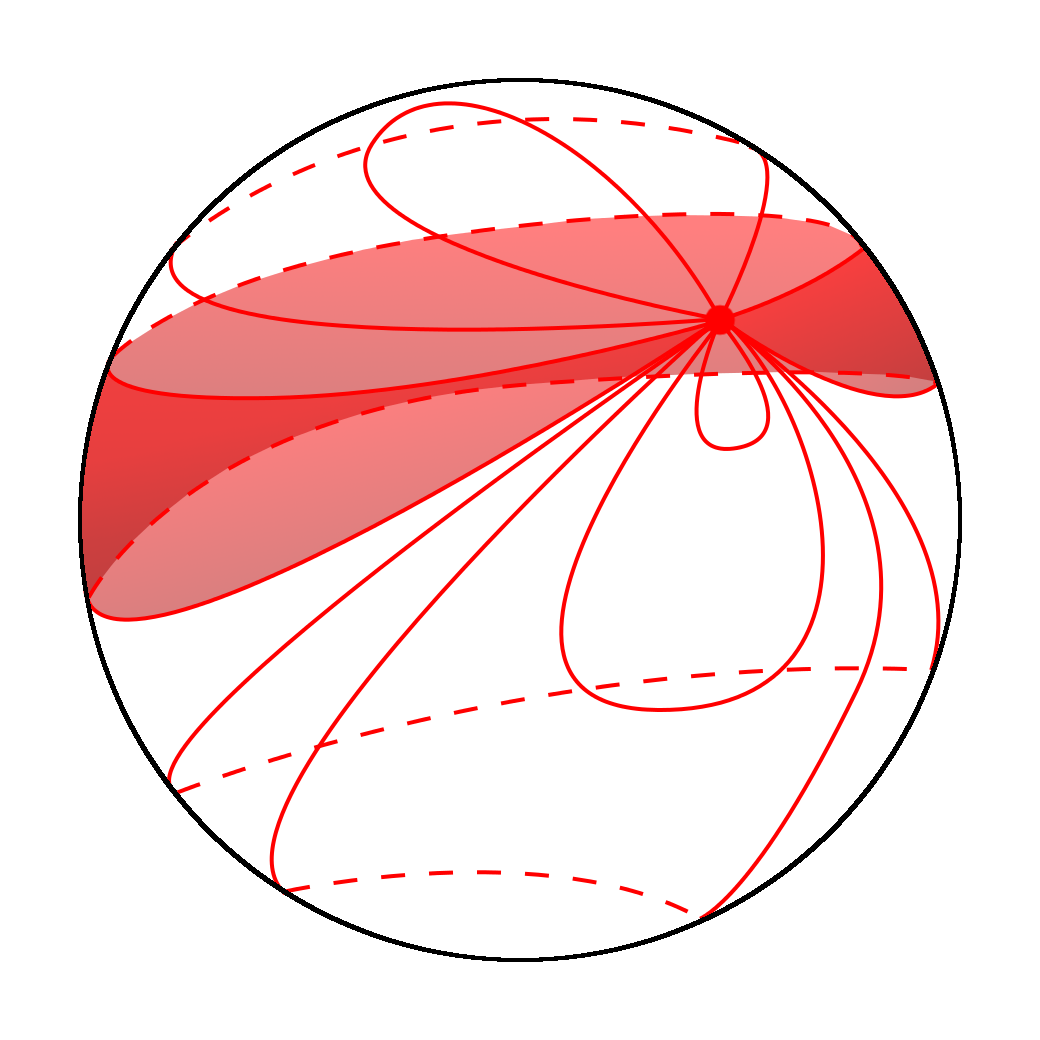}
	\end{subfigure}
	\quad
	\begin{picture}(0,0)
	\put(95,85){$\g$}
	\put(77,67){$x$}
	\put(75,5){$\g'$}
	\put(35,40){$x'$}
	\put(60,53){$\g_{x'x}$}
	\end{picture}
	\begin{subfigure}[b]{0.24\textwidth}
		\includegraphics[width=\textwidth]{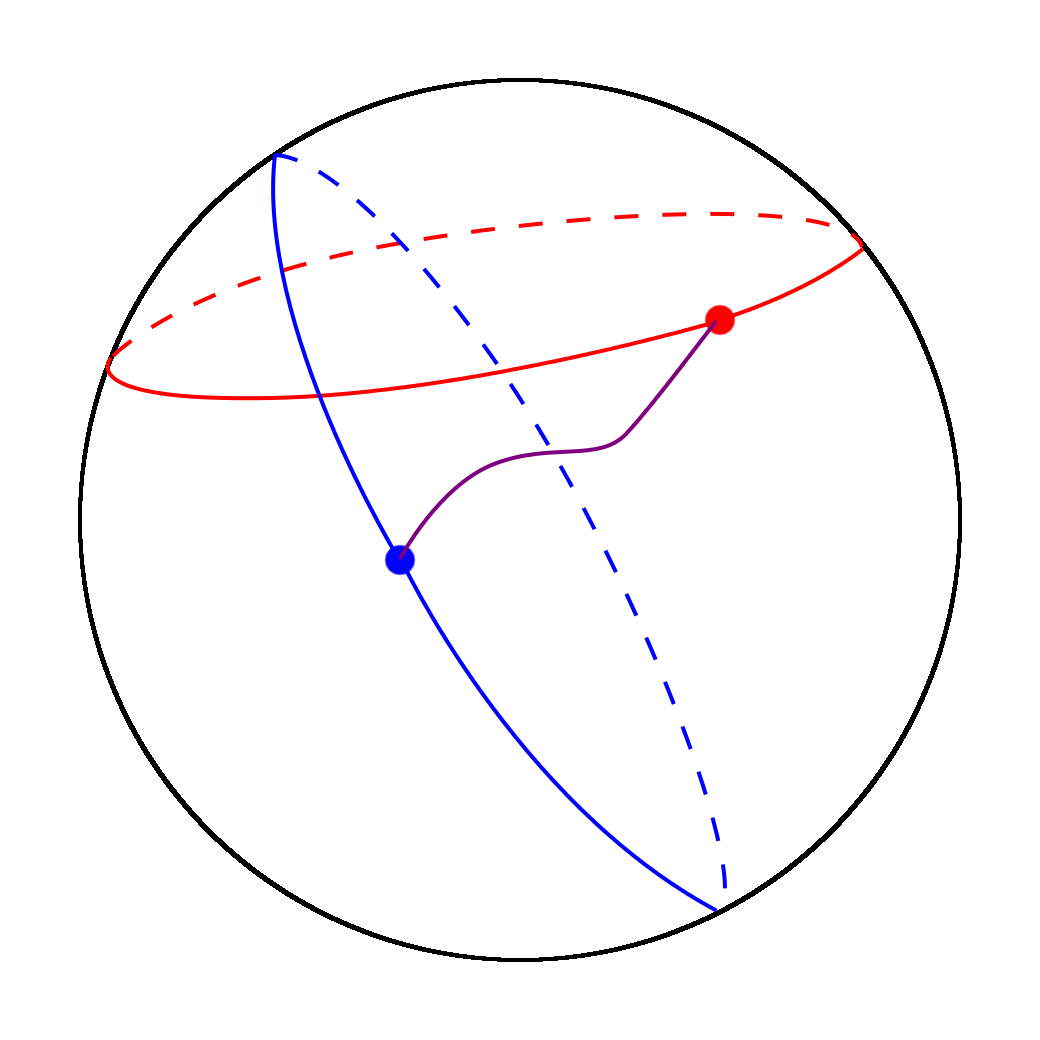}
	\end{subfigure}
	\quad
	\begin{picture}(0,0)
	\put(95,85){$\g$}
	\put(75,5){$\g'$}
	\put(35,40){$x'$}
	\put(105,45){$\D$}
	\end{picture}
	\begin{subfigure}[b]{0.24\textwidth}
		\includegraphics[width=\textwidth]{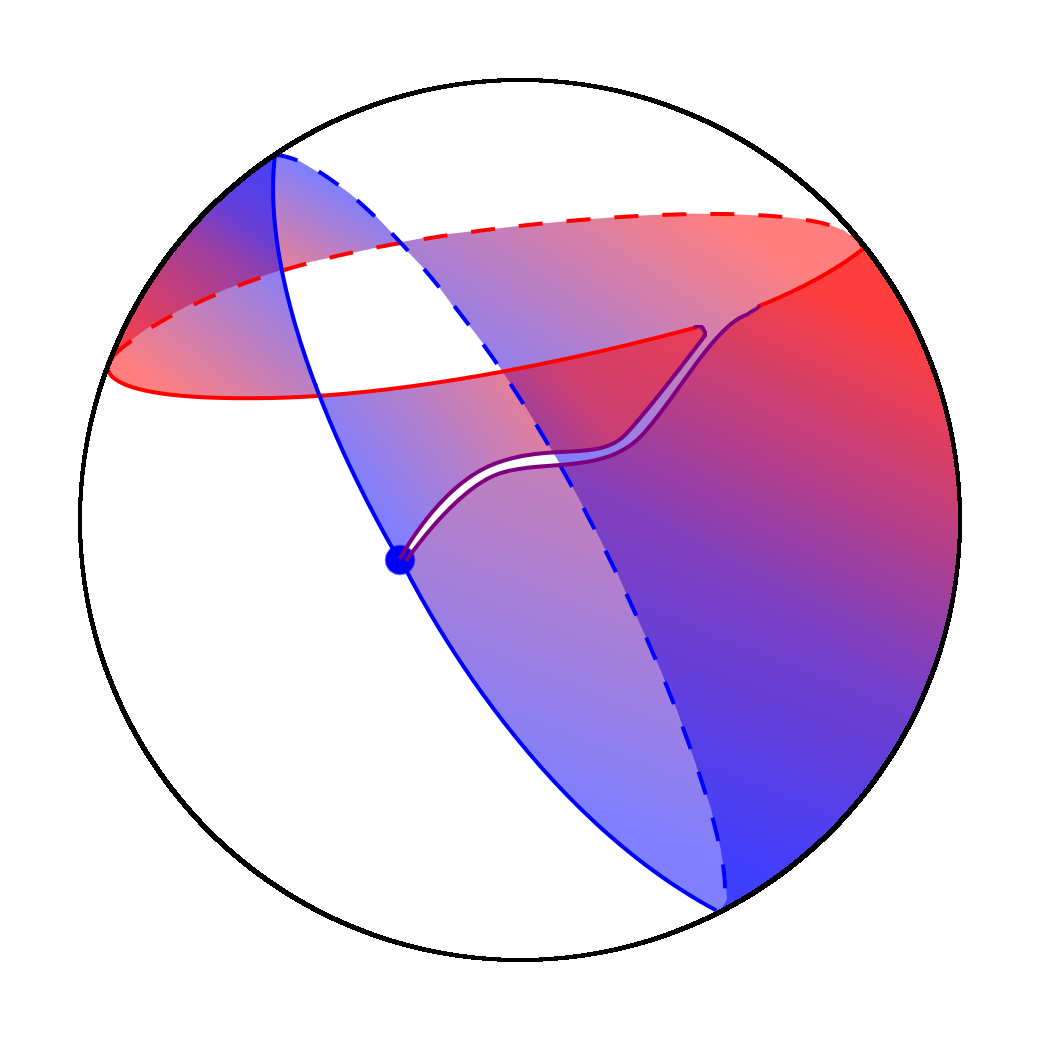}
	\end{subfigure}
	\vspace{-3mm}
	\caption{For every thin sphere and two markings, there exists a
representative with a decomposition as in Lemma \ref{lem:spheremarking}.
On the left is a bigon $\S : \g \Rightarrow \g$ with $\g : x \rightarrow x.$ The
shaded region depicts the surface swept out between $s=0$ and some small
$s.$ In the middle is another bigon $\S' : \g' \Rightarrow \g'$ with
$\g' : x' \rightarrow x'$ and a path $\g_{x' x} : x \rightarrow x'$ with sitting
instants. On the right is a bigon $\D : \g_{x'x} \circ \g \circ
\overline{\g_{x'x}}  \Rightarrow \g'$ relating the two marked loops as
in (\ref{eq:basedspheres}). 
	 }
	\label{fig:basedspheres}
\end{figure}

Therefore, without loss of generality, we can choose a \emph{single}
representative $\tilde{\S}$ of a thin free sphere $[\S]$ with a
decomposition as in the Lemma. We use $\S$ to denote the bigon in ii) of
Lemma \ref{lem:spheremarking} and $\S'$ to denote the bigon in iii). The
two are related by  
\be
\xy0;/r.15pc/:
(-40,0)*+{x'}="d";
(0,0)*+{x'}="c";
{\ar@/_1.5pc/"c";"d"_{\g' }};
{\ar@/^1.5pc/"c";"d"^{\g' }};
(-20,10)*+{}="f";
(-20,-10)*+{}="g";
{\ar@{=>}"f";"g"^{ \S' }};
\endxy
\quad 
=
\quad
\xy0;/r.15pc/:
(-70,0)*+{x'}="left";
(30,0)*+{x'}="right";
(-40,0)*+{x}="d";
(0,0)*+{x}="c";
{\ar@/_1.5pc/"c";"d"_{\g }};
{\ar@/^1.5pc/"c";"d"^{\g }};
(-20,10)*+{}="f";
(-20,-10)*+{}="g";
{\ar@{=>}"f";"g"^{ \S }};
{\ar@/_4pc/"right";"left"_{ \g' }};
{\ar@/^4pc/"right";"left"^{ \g' }};
{\ar"right";"c"_{ \overline{\g_{x'x}} }};
{\ar"d";"left"_{\g_{x'x}}};
{\ar@{=>}(-20,25);(-20,16)^{\ov \D}};
{\ar@{=>}(-20,-16);(-20,-25)^{\D}};
\endxy
\ee
i.e.
\be
\label{eq:basedspheres}
\S_{y} = \begin{matrix} \ov \D \\ \overset{\circ}{\id_{\g_{x'x}} \circ
\S \circ \id_{\overline{\g_{x'x}}}} \\ \overset{\circ}{\D}
\end{matrix} . 
\ee
By functoriality of the transport 2-functor $\scripty{t}_{F},$ we have 
\be
\begin{split}
\mathrm{hol}_{\scripty{t}}^{F} (\S') &= p_{H} \left(\scripty{t}_{F}
(\S' ) \right) \\
	&= p_{H} \left (  \begin{smallmatrix} 
		\scripty{t}_{F} ( \ov \D ) \\
		C^{-1} \\
		\id_{\scripty{t}_{F} ( \g_{yx} ) } \scripty{t}_{F}
( \S_{x} ) \id_{\scripty{t}_{F} ( \overline{ \g_{yx} } ) } \\
		C \\
		\scripty{t}_{F} (  \D )
		\end{smallmatrix}
		\right ) 
\end{split}
,
\ee
where $C : \scripty{t}_{F} ( \g_{x'x} ) \scripty{t}_{F} ( \g_{x} )
\scripty{t}_{F} ( \overline{ \g_{x'x} } ) \Rightarrow \scripty{t}_{F}
( \g_{x'x} \circ \g_{x} \circ \overline{ \g_{x'x} } )$ is  a combination
of  compositors and associators. Writing out this composition in the
2-group $\mathcal{B} \mathfrak{G}$   gives
\be
\begin{matrix}
\left( ( p_{H} ( \scripty{t}_{F} (  \D ) ) )^{-1} , \scripty{t}_{F}
( \g' )\right ) \\
\left( p_{H} ( C )^{-1} , \scripty{t}_{F} ( \g_{x'x} \circ \g \circ
\overline{ \g_{x'x} } )  \right) \\
\left( e , \scripty{t}_{F} ( \g_{x'x} ) \right) \left(
\mathrm{hol}_{\scripty{t}}^{F} ( \S ) , \scripty{t}_{F} (\g) )
(e,\scripty{t}_{F} ( \overline{ \g_{x'x} } ) \right) \\
\left( p_{H} ( C ) , \scripty{t}_{F} ( \g_{x'x} ) \scripty{t}_{F} (\g)
\scripty{t}_{F} ( \overline{ \g_{x'x} } ) \right) \\
\left( p_{H} ( \scripty{t}_{F} ( \D ) ) , \scripty{t}_{F} ( \g_{x'x}
\circ \g \circ \overline{ \g_{x'x} } ) \right )
\end{matrix}
\quad.
\ee
Multiplying these results out using the rules of 2-group multiplication
(see equations (\ref{eq:vertical}) and (\ref{eq:horizontal})) and taking
the $H$ component gives
\be
\begin{split}
\mathrm{hol}_{\scripty{t}}^{F} (\S') &= 
p_{H} ( \scripty{t}_{F} ( \D ) ) p_{H} ( C ) \a_{ \scripty{t}_{F}
(\g_{x'x}) } \left ( \mathrm{hol}_{\scripty{t}}^{F} ( \S ) \right )
p_{H} (C)^{-1} \left ( p_{H} \left ( \scripty{t}_{F} ( \D ) \right )
\right )^{-1} \\
&= \a_{ \t (  p_{H} \left ( \scripty{t}_{F} ( \D ) \right ) p_{H} (C) )
\scripty{t}_{F} ( \g_{x'x} ) } \left (  \mathrm{hol}_{\scripty{t}}^{F}
( \S ) \right) .
\end{split}
\ee
This result says that the 2-holonomy changes by $\a$-conjugation under a
change of marking for a thin sphere. 

\item
Now suppose that $\h : F \to F'$ is a 1-morphism of transport
2-functors. Then, for every thin path $\g : x \to y$  we have a
2-isomorphism  (remember that $\scripty{t}_{F'}(x) = i(\bullet)$ and
$\scripty{t}_{F}(x) = i(\bullet)$ for all $x \in M$)
\be
\xy 0;/r.15pc/:
(20,-15)*+{i(\bullet)}="3";
(-20,-15 )*+{i(\bullet)}="4";
(-20,15)*+{i(\bullet)}="2";
(20,15)*+{i(\bullet)}="1";
{\ar_{\scripty{t}_{\eta} (x) } "1";"2" };
{\ar_{\scripty{t}_{F'} (\g) } "2";"4" };
{\ar^{\scripty{t}_{F} (\g) } "1";"3"};
{\ar^{\scripty{t}_{\eta} (y)} "3";"4"};
{\ar@{=>}|-{\scripty{t}_{\eta} (\g)} "3";"2"};
\endxy
\ee
satisfying the condition that for any thin bigon $\S : \g \Rightarrow
\de,$ with $\de : x \to y$ another path, the diagram 
\be
\xy 0;/r.25pc/:
(20,-15)*+{i(\bullet)}="3";
(-30,-15 )*+{i(\bullet)}="4";
(-30,15)*+{i(\bullet)}="2";
(20,15)*+{i(\bullet)}="1";
{\ar_{\scripty{t}_{\eta} (x) } "1";"2" };
{\ar@/^2pc/^{\scripty{t}_{F'} (\g) } "2";"4" };
{\ar@/_2pc/_{\scripty{t}_{F'} (\de) } "2";"4" };
{\ar@/^2pc/^{\scripty{t}_{F} (\g) } "1";"3"};
{\ar@{-->}@/_2pc/_{\scripty{t}_{F} (\de) } "1";"3"};
{\ar^{\scripty{t}_{\eta} (y)} "3";"4"};
{\ar@{=>}|-{\scripty{t}_{\eta} (\g)} "3";"2"};
{\ar@{==>}_{\scripty{t}_{F} (\S)} (25,0);(15,0)};
{\ar@{=>}_{\scripty{t}_{F'} (\S)} (-25,0);(-35,0)};
\endxy
\ee
commutes. In this diagram, the $\scripty{t}_{\eta} (\de)$ in the back is
not shown. This diagram commuting means that  
\be
\begin{matrix}
\scripty{t}_{\eta} (\g) \\
\overset{\circ}{\scripty{t}_{F'} (\S ) \id_{\scripty{t}_{\eta}(x)}}
\end{matrix} 
\quad = \quad
\begin{matrix}
\id_{\scripty{t}_{\eta} (y) } \scripty{t}_{F} (\S) \\
\overset{\circ}{\scripty{t}_{\eta} (\de)}
\end{matrix}
\ee
and writing this out using group elements gives 
\be
\qquad
\begin{matrix} 
\left( p_{H} ( \scripty{t}_{\eta} (\g) ) , \scripty{t}_{\eta} (y)
\scripty{t}_{F} (\g) \right) \\
( p_{H} ( \scripty{t}_{F'} (\S) ),\scripty{t}_{F'} (\g) )
(e,\scripty{t}_{\eta} (x)   ) 
\end{matrix} 
\quad = \quad 
\begin{matrix}
(e,\scripty{t}_{\eta} (y) ) ( p_{H} (\scripty{t}_{F} (\S) ),
\scripty{t}_{F} (\g) ) \\
( p_{H} ( \scripty{t}_{\eta} (\de) ), \scripty{t}_{\eta} (y)
\scripty{t}_{F} (\de) )   
\end{matrix}
, \;
\ee
which after evaluating both sides and projecting to $H$  yields
\be
p_{H} (  \scripty{t}_{F'} (\S) )  p_{H} ( \scripty{t}_{\eta} (\g) )
= p_{H} ( \scripty{t}_{\eta} (\de) ) \a_{  \scripty{t}_{\eta} (y) }
\left ( p_{H} ( \scripty{t}_{F} (\S) ) \right ) .
\ee
Solving for $p_{H} (  \scripty{t}_{F'} (\S) )$ gives
\be
p_{H} (\scripty{t}_{F'} (\S) ) = p_{H} ( \scripty{t}_{\eta} (\de) )
\a_{\scripty{t}_{\eta} (y) } \left ( p_{H} (\scripty{t}_{F} (\S) )
\right ) p_{H} ( \scripty{t}_{\eta} (\g)  )^{-1}. 
\ee
Now, after specializing to the case where the source and targets of
$\S$ are all the same, i.e. $y = x$ and $\de = \g,$ so that we are
comparing this transport along thin marked spheres, this reduces to 
\be
\begin{split}
\mathrm{hol}_{\scripty{t}}^{F'} (\S) &= p_{H} (\scripty{t}_{\eta} (\g) )
\a_{\scripty{t}_{\eta} (x) } \left ( \mathrm{hol}_{\scripty{t}}^{F} (\S)
\right ) p_{H} ( \scripty{t}_{\eta} (\g)  )^{-1} \\
	&= \a_{ \t ( p_{H} ( \scripty{t}_{\eta} (\g) ) )
\scripty{t}_{\eta} (x)  } \left ( \mathrm{hol}_{\scripty{t}}^{F} (\S)
\right ) .
\end{split}
\ee
This says that $\mathrm{hol}^{F}_{\scripty{t}}$ when restricted to thin
marked spheres changes under $\a$-conjugation when the functor $F$ is
changed to a gauge equivalent one $F'.$ 

\item
Suppose that another 2-group transport extraction procedure
$\scripty{t} \; \; '$ was chosen.  Any two such procedures are
pseudo-naturally equivalent, i.e. if $\scripty{t} \; \; '$ was another
such choice, then there exists a weakly invertible pseudonatural
transformation $\scripty{s} : \scripty{t} \; \; ' \Rightarrow
\scripty{t} \; .$ This follows from the fact that each 2-functor in the
composition of 2-functors that define $\scripty{t} \; \;$ is an
equivalence of 2-categories and weak inverses are unique up to
pseudo-natural equivalences. Therefore, for every transport 2-functor
$F$ we have a 1-morphism of transport functors $\scripty{s}_{F} :
\scripty{t}_{F} {} '  \to \scripty{t}_{F}.$ Of course, we also have a
map assigning to every 1-morphism of transport functors $\eta : F \to
F'$ a 2-morphism $\scripty{s}_{\eta} : \scripty{s}_{F} \Rightarrow
\scripty{s}_{F'}$ satisfying naturality, but we will not need this fact
for the following observation because we are dealing with strict Lie
2-groups. The 1-morphism of transport functors $\scripty{s}_{F}$ assigns
to every point $x \in M$ a morphism $\scripty{s}_{F} (x) :
\scripty{t}_{F} {} ' (x) \to \scripty{t}_{F}(x)$ and to every path
$\g : x \to y$ a 2-isomorphism 
\be
\xy 0;/r.15pc/:
(20,-15)*+{i(\bullet)}="3";
(-20,-15 )*+{i(\bullet)}="4";
(-20,15)*+{i(\bullet)}="2";
(20,15)*+{i(\bullet)}="1";
{\ar_{\scripty{s}_{F} (x) } "1";"2" };
{\ar_{\scripty{t}_{F} (\g) } "2";"4" };
{\ar^{\scripty{t}_{F} {} ' (\g) } "1";"3"};
{\ar^{\scripty{s}_{F} (y)} "3";"4"};
{\ar@{=>}|-{\scripty{s}_{F} (\g)} "3";"2"};
\endxy
\ee
satisfying the condition that for a thin bigon $\S : \g \to \de$ between
two thin paths $\g, \de : x \to y$ the diagram 
\be
\xy 0;/r.25pc/:
(20,-15)*+{i(\bullet)}="3";
(-30,-15 )*+{i(\bullet)}="4";
(-30,15)*+{i(\bullet)}="2";
(20,15)*+{i(\bullet)}="1";
{\ar_{\scripty{s}_{F} (x) } "1";"2" };
{\ar@/^2pc/^{\scripty{t}_{F} (\g) } "2";"4" };
{\ar@/_2pc/_{\scripty{t}_{F} (\de) } "2";"4" };
{\ar@/^2pc/^{\scripty{t}_{F} {}' (\g) } "1";"3"};
{\ar@{-->}@/_2pc/_{\scripty{t}_{F} {}' (\de) } "1";"3"};
{\ar^{\scripty{s}_{F} (y)} "3";"4"};
{\ar@{=>}|-{\scripty{s}_{F} (\g)} "3";"2"};
{\ar@{==>}_{\scripty{t}_{F} {}' (\S)} (25,0);(15,0)};
{\ar@{=>}_{\scripty{t}_{F} (\S)} (-25,0);(-35,0)};
\endxy
\ee
commutes. This result is very similar to the previous one and is given
by
\be
\mathrm{hol}_{\scripty{t}  \; {}' }^{F}(\S)
= \a_{ \t ( p_{H} ( \scripty{s}_{F} (\g) ) ) \scripty{s}_{F} (x) }
\left ( \mathrm{hol}_{\scripty{t}}^{F} (\S)  \right ) , 
\ee
which is again just $\a$-conjugation. 
\end{enumerate}

In conclusion, when restricted to a sphere, 2-holonomy changes under
$\a$-conjugation in each of the three situations described above. This
should therefore also be called gauge covariance as in the case for
loops. This motivates the following definition. 

\begin{defn}
\label{defn:alphaconj}
Let $(H, G, \t, \a)$ be a crossed module. The \emph{\uline{$\a$-conjugacy
classes in $H$}}, denoted by $H / \a,$ is defined to be the quotient of
$H$ under the equivalence relation 
\be
h \sim h' \iff \text{ there exists a $g \in G$ such that } h
= \a_{g} (h'). 
\ee
Denote the quotient map by $q : H \to H / \a.$ 
\end{defn}

As before, we have a similar theorem for gauge-invariance of 2-holonomy.

\begin{thm}
\label{thm:invariancesurfaceholonomy}
Let $M$ be a smooth manifold, $\mathcal{B}\mathfrak{G}$ a Lie 2-group,
$T$ a 2-category, and suppose that $i : \mathcal{B} \mathfrak{G} \to T$
is a full and faithful 2-functor. Let $F$ be a transport 2-functor and
$\scripty{t} \;$ a 2-group-valued transport extraction. Let $S^{2} M,
\mathfrak{S}^2 M, \mathfrak{m}, \mathrm{hol}^{F}_{\scripty{t}}$ and $q$
be defined as above. Then the composition 
\be
H / \a \xleftarrow{q} H \xleftarrow{\mathrm{hol}^{F}_{\scripty{t}}}
\mathfrak{S}^2 M \xleftarrow{\mathfrak{m}} S^{2} M
\ee
is 
\begin{enumerate}[i)]
\item
independent of $\mathfrak{m},$ 
\item
independent of the equivalence class of $F,$
\item
and independent of the equivalence class of $\scripty{t} \; .$ 
\end{enumerate}
\end{thm}

This theorem lets us make the following definition.  

\begin{defn}
\label{defn:gaugeinv2hol}
Let $[F]$ be an equivalence class of transport 2-functors.
The \emph{\uline{gauge invariant 2-holonomy}} of $[F]$ is defined to be the
smooth map in the previous theorem, namely 
\be
\mathrm{hol}^{[F]} 
:= q \circ \mathrm{hol}^{F}_{\scripty{t}} \circ \mathfrak{m} : S^{2} M
\to H / \a
\ee
where $F$ is a representative of $[F],$ $\scripty{t} \;$ is a
group-valued transport extraction, and $\mathfrak{m} : S^{2} M \to
\mathfrak{S}^2 M$ is a marking for thin spheres in $M.$ Let $\S \in
S^{2} M.$ If $\mathrm{hol}^{[F]} (\S)$ is such that
$q^{-1} ( \mathrm{hol}^{[F]} (\S))$ is a single element, we will say
that $\mathrm{hol}^{[F]} (\S)$ is \emph{\uline{gauge invariant}} and
abusively write $\mathrm{hol}^{[F]} (\S)$ instead of
$q^{-1} (\mathrm{hol}^{[F]} (\S)).$
\end{defn}

\begin{rmk}
A result analogous to Theorem \ref{thm:invariancesurfaceholonomy} was
obtained in the context of a cubical category approach to 2-bundles in
\cite{MP}. 
\end{rmk}

We now compare this result to that in \cite{SW4}, where
the \emph{reduced group} associated to a 2-group was introduced in
order to obtain a well-defined 2-holonomy independent of the marking as
well as the representative of the transport functor used. 

\begin{defn}
\label{defn:red2group}
Let $\mathcal{B}\mathfrak{G}$ be a 2-group with associated crossed
module $(H, G, \t, \a).$ The \emph{\uline{reduced group of
$\mathcal{B}\mathfrak{G}$}} is $\mathfrak{G}_{\mathrm{red}}
:= H / [G, H],$ where $[G, H] = \< h^{-1} \a_{g} (h) \ |  \ g \in G, h
\in H \>,$ i.e. the subgroup of $H$ generated by elements of the form
$h^{-1} \a_{g} (h).$ 
\end{defn}

The analogue of the reduced 2-group in the case of ordinary holonomy for
principal $G$ bundles with connection is $G / [G, G],$ the
abelianization of $G.$ Recall, $[G, G] = \< g g' g^{-1} g'^{-1} \ | \ g,
g' \in G \>$ is a normal subgroup, called the
\emph{commutator subgroup}, of $G$ so the quotient is an abelian group,
in fact in a universal sense. 

\begin{lem}
\label{lem:conjugacyclasses}
Let $G$ be a group, $[G, G]$ its commutator subgroup, and
$G/ \mathrm{Inn}(G)$ conjugacy classes in $G.$ The map
$G / \mathrm{Inn}(G) \to G/ [G,G]$ given by taking a conjugacy class
$[g],$ choosing a representative, and projecting to the quotient
$G / [G, G],$ is 
\begin{enumerate}[i)]
\item
well-defined,  
\item
surjective, 
\item
and need not be injective in general.
\end{enumerate} 
\end{lem}

\begin{proof}

\; $ $ 
{}

\begin{enumerate}[i)]

\item
The map $G / \mathrm{Inn}(G) \to G/ [G,G]$ is well-defined because if
$g'$ was another representative of $[g],$ then there would be a
$\tilde{g} \in G$ such that $\tilde{g} g \tilde{g}^{-1} = g',$ and under
the quotient map, the difference between $g$ and $g'$ is $g' g^{-1} =
\tilde{g} g \tilde{g}^{-1} g^{-1}  \in [ G, G ].$

\item
Since $G \to G/ [G,G]$ is surjective, and the map $G / \mathrm{Inn}(G)
\to G/ [G,G]$ defined by choosing a representative is well-defined, the
map $G / \mathrm{Inn}(G) \to G/ [G,G]$ is surjective.

\item
To see why the map $G / \mathrm{Inn} (G) \to G / [G, G]$ is, in general,
not injective, consider the following example \cite{DF}. Let $S_{n}$ be
the symmetric group on $n$ letters, i.e. it is the permutation group of
$n$ elements. Let $A_{n}$ be the alternating group on $n$ letters. This
group is defined as the kernel of the homomorphism $S_{n} \to \{ - 1, 1
\}$ given by taking the sign of the permutation. It turns out this
kernel is also the commutator subgroup of $S_{n}.$ Furthermore, its
index is $[ S_{n} / [ S_{n}, S_{n} ] ] = [ S_{n} / A_{n} ] \equiv
[ S_{n} : A_{n} ] = 2.$ On the other hand, let's compute the conjugacy
classes of $S_{n}$ for some small $n.$ The simplest case actually
suffices, although we'll quote some results for higher $n$ to indicate
that the difference between conjugacy classes and abelianization gets
bigger. For $n=3,$ the set of conjugacy classes in $S_{3}$ is given by
the following elements. The identity element, written as $( \ \ ),$ is in
its own class. The elements $(1,2), (1,3),$ and $(2,3)$ are in their own
class. Finally, the elements $(1,2,3)$ and $(1,3,2)$ are in their own
class. Therefore, the set of conjugacy classes for $S_{3}$ is given by a
3-element set whereas the abelianization is a 2-element group. For
$S_{4},$ the set of conjugacy classes is a set of 5 elements. For
$S_{5},$ the set of conjugacy classes is a set of 7 elements. The
abelianization, however, is always of order 2. 
\end{enumerate}
\end{proof}

Therefore, conjugacy classes contain at least as much information about
ordinary holonomy as do elements of the abelianization, and they are
exactly the elements needed to define holonomy in a gauge invariant way
due to Theorem \ref{thm:gaugeinvarianthol}. 

In a similar way, the reduced group $\mathfrak{G}_{\mathrm{red}}$ of a
2-group $\mathcal{B}\mathfrak{G}$  is analogous to the abelianization
and does not contain the full information of 2-holonomy in general. One
needs an analogue of conjugacy classes for 2-holonomy. The candidate,
for spheres at least, is $\a$-conjugacy classes, $H / \a.$ In fact, we
have a similar fact concerning $\a$-conjugacy classes and the reduced
group. 

\begin{lem}
\label{lem:alphaconjtoreducedgroup}
Let $(H, G , \t , \a)$ be a crossed module, $\mathcal{B}\mathfrak{G}$
the associated 2-group, $\mathfrak{G}_{\mathrm{red}} := H / [ G, H ]$
the reduced group of $\mathcal{B}\mathfrak{G},$ and $H/ \a$ the
$\a$-conjugacy classes in $H$. The map $H / \a \to
\mathfrak{G}_{\mathrm{red}}$ given by taking a conjugacy class $[h],$
choosing a representative, and projecting to the quotient $H / [G, H],$
is 
\begin{enumerate}[i)]
\item
well-defined,  
\item
surjective, 
\item
and need not be injective in general.
\end{enumerate} 
\end{lem}

\begin{proof}

\; $ $ 
{}

\begin{enumerate}[i)]

\item
Let $h$ and $h'$ be two representatives. Then there exists a $g \in G$
such that $\a_{g} (h) = h'$ and so the difference between $h$ and $h'$
is $h^{-1} h' = h^{-1} \a_{g} (h) \in [ G, H ] .$

\item
Since $H \to H/ [G,H]$ is surjective, and the map $H / \a \to
\mathfrak{G}_{\mathrm{red}}$ defined by choosing a representative is
well-defined, the map $H / \a \to \mathfrak{G}_{\mathrm{red}}$ is
surjective.

\item
To see why the map $H / \a \to \mathfrak{G}_{\mathrm{red}}$ is, in
general, not injective, consider the special case where $H = G,$
$\t = \id,$ and $\a$ is the ordinary conjugation. Then this case
reduces to the previous case of Lemma \ref{lem:conjugacyclasses}. 

Although the previous example suffices to show why $\a$-conjugacy
classes $H/ \a$ contain more information than the reduced group in
general, holonomy along spheres takes values in $\ker \t \le H$ by the
source-target matching condition. Therefore, it is also important to
find an example of a crossed module $(H,G,\t,\a)$ such that
$\ker \t = H$ and the map $H/ \a \to \mathfrak{G}_{\mathrm{red}}$ is not
injective.

Take $H := \Z_{p},$ the (additive) cyclic group of order $p,$ where
$p \ge 3$ is prime. Set $G := \mathrm{Aut}(\Z_{p}),$ the automorphism
group of $\Z_{p}.$ Let $\t$ be the trivial map and $\a := \id$ be
the identity map. $(\Z_{p}, \mathrm{Aut}(\Z_{p}), \t,\a)$
defines a crossed module. 

Every element of $\mathrm{Aut}(\Z_{p})$ is of the form $\s_{k}$ with
$k \in \{ 1,2, \dots, p-1\}$ and is determined by where it sends the
generator: ${\s_{k} (1 \mod p)} := {k \mod p}.$ For this proof, denote
the $\a$-conjugacy class of an element $m \in \Z_{p}$ by $[m].$ For all
$k,$ ${\s_{k} (0 \mod p)} = {0 \mod p}$ so that $0 \mod p$ is fixed
under the $\a$ action. However, since $\s_{k}(1) = k \mod p,$ the set of
$\a$-conjugacy classes of $(\Z_{p}, \mathrm{Aut}(\Z_{p}), \t, \a)$ is 
$\Z_{p} / \a = \{ [0], [1] \},$
which is just a 2-element set. However, the reduced group is trivial. To
see this, consider generators of $[ \mathrm{Aut}(\Z_{p}), \Z_{p}],$
which are of the form ${(\s_{k} (m) - m) \mod p}$ with $k \in \{ 1,2,
\dots, p-1\}$ and $m \in \{ 0, 1,2, \dots, p-1\}.$ Set $m=1$ and $k=2.$
Then ${(\s_{k} (m) - m) \mod p} = {1 \mod p}.$ Therefore, the generator
of $\Z_{p}$ is in the subgroup $[\mathrm{Aut}(\Z_{p}), \Z_{p}]$ which
means $[\mathrm{Aut}(\Z_{p}), \Z_{p}] = \Z_{p}.$ Thus 
$\Z_{p} / [\mathrm{Aut}(\Z_{p}), \Z_{p}] = \Z_{p}/\Z_{p} \cong \{e\}.$

\end{enumerate}
\end{proof}

In this case, one can make sense of gauge-invariant quantities coming
from 2-holonomy without passing to the reduced group as is done in
\cite{SW4}. 
In the case of the examples considered in Section \ref{sec:examples},
one even gets a fixed point under the $\a$ action, in which case one
does not need to pass to the $\a$-conjugacy classes. 

\begin{defn}
\label{defn:Inv}
Let $(H, G,\t, \a)$ be a crossed module. Denote the fixed points of $H$
under the $\a$ action by 
\be
\mathrm{Inv}(\a) := \{ h \in H \ | \ \a_{g} (h) = h \text{ for all } g
\in G \}. 
\ee
\end{defn}

\begin{lem}
\label{lem:fixedpoints}
In the notation of Definition \ref{defn:Inv}, $\mathrm{Inv}(\a)$ is a
central subgroup of $H.$ 
\end{lem}

\begin{proof}
Let $h, h' \in \mathrm{Inv}(\a).$ Then 
\be
\a_{g} (h h') = \a_{g} ( h ) \a_{g} (h')  = h h'
\ee
for all $g \in G.$ Thus, $\mathrm{Inv}(\a)$ is closed. $\a_{g} (e) = e$
for all $g \in G$ says  $e \in \mathrm{Inv}(\a).$ Let $h \in
\mathrm{Inv}(\a),$ then $\a_{g} (h^{-1} ) = (\a_{g} (h) )^{-1} = h^{-1}$
showing that $h^{-1} \in \mathrm{Inv}(\a).$ Finally, $\mathrm{Inv}(\a)$
is central because
\be
h k h^{-1} = \a_{\t ( h) } (k) = k 
\ee
for all $h \in H$ and $k \in \mathrm{Inv}(\a).$ 
\end{proof}

This will have physical relevance when discussing monopoles, which, as
we will show, take values in $\mathrm{Inv}(\a).$

\section{The path-curvature 2-functor associated to a transport functor}
\label{sec:path-curvature}

In this section, given a principal $G$-bundle with connection and a
choice of a subgroup of $\pi_{1} (G),$ we construct a principal 2-bundle
with connection whose structure 2-group is a covering 2-group obtained
from $G$ and the subgroup of $\pi_{1}(G).$ This assignment is
functorial. We describe it on all levels introduced in the review,
namely as a globally defined transport functor, in terms of descent
data, and via differential cocycle data. These constructions respect all
of the functors relating these different levels.

\subsection{The path-curvature 2-functor}

The transport 2-functor defined later in this section is motivated by
the study of magnetic monopoles in gauge theories as described in
\cite{HS}. Some of the earlier accounts of similar descriptions can be
found in the work of Wu and Yang in \cite{WY} under the name
`total circuit'  and also in the work of Goddard, Nuyts, and Olive in
\cite{GNO}. Of course, several others worked on understanding the
``topological quantum number'' due to a magnetic charge in terms of just
the magnetic charge alone, but the three references mentioned are the
ones that have influenced us. We argue in Section \ref{sec:examples}
that in the case where the base space is a 3-manifold, this transport
2-functor has 2-holonomy along a sphere which is given by the magnetic
flux through that sphere. Therefore, we give a mathematically rigorous
description of non-abelian flux for magnetic monopoles in a non-abelian
gauge theory. A more detailed description of the physics will be given
in that section, but first we explain the mathematical structure. 

The starting data consist of (i) a principal $G$-bundle, where $G$ is a
connected Lie group, with connection over a smooth manifold $M,$ and
(ii) a subgroup $N$ of $\pi_{1} (G).$ By the main theorem of \cite{SW1},
the first part of the data corresponds to a transport functor
$\tra : \mathcal{P}_{1} (M) \to G\text{-}\mathrm{Tor}$ with
$\mathcal{B} G$ structure. From this data, we will construct a transport
2-functor which we call the \emph{path-curvature 2-functor.} We will
discuss two interesting cases for the choice of $N$ although other
choices are important for applications in physics so we keep this
generality for future applications. When $N = \pi_{1} (G),$ the
path-curvature 2-functor coincides with the curvature 2-functor of
Schreiber and Waldorf \cite{SW4}. The  choice  $N = \{ 1 \},$ the
trivial group, will be more appropriate in the context of gauge theory
and computing invariants. This is the case we focus on for all our
computations in Section \ref{sec:examples}.

To set up this example, we introduce the following Lie 2-group
associated to any connected Lie group $G.$ Let $\tilde{G}$ be the
universal over of $G$ (we will describe what happens for arbitrary
covers later) and denote the covering map by $\t : \tilde{G} \to G.$ An
explicit construction of $\tilde{G}$ in terms of homotopy classes of
paths will be useful for our purposes 
\be
\label{eq:universalcoverG}
\tilde{G} :=
 \{h : [0,1]\to G \ | \ h(0)=e \text{ and $h$ is continuous} \}/_{\sim}
\ee
where $h \sim h'$ if $h(1) = h' (1)$ and there exists a homotopy
$h \Rightarrow h'$ relative the endpoints.
$\tilde{G}$ naturally acquires a topology as the quotient space of a
subspace of paths. 
Denote the equivalence class representing a path with square brackets as
in $[ h ]$ or $[ t \mapsto h(t) ],$ where it is understood that $t$
takes values in $[0,1].$ The multiplication in $\tilde{G}$ is defined by
choosing representatives and multiplying them pointwise (later we will
show that this multiplication can be described in another way that is
sometimes more convenient for our examples). Let $\a : G \to
\mathrm{Aut} (\tilde{G})$ be the conjugation map $\a_{g} ( [h] )
:= [ g h g^{-1} ],$ meaning 
\be
\a_{g} ([h]) := [ t \mapsto g h(t) g^{-1} ] , 
\ee 
where the concatenation means multiplication in $G.$ 
Define $\t : \tilde{G} \to G$ to be evaluation at the endpoint, 
\be
\t ( [h] ) := h(1).
\ee

\begin{prop}
$(\tilde{G}, G, \t, \a)$ defined in the previous paragraph is a Lie
crossed module. 
\end{prop}

\begin{proof}
It is useful to recall the definition of a crossed module (Definition
\ref{defn:crossedmodule}) at this point. Since the equivalence relation
involves homotopy relative endpoints, $\t$ is well-defined. $\a$ is
well-defined because $h \sim h'$ implies $ghg^{-1} \sim g h' g^{-1}.$ 
The topological space $\tilde{G}$ has a unique smooth structure making
the map $\t$ a homomorphism and a smooth covering map, i.e. a smooth
surjective submersion with the property that for every $g \in G,$ there
exists an open neighborhood $U$ containing $g$ such that each component
of $\t^{-1}(U)$ maps to $U$ diffeomorphically. This follows from some
basic differential topology (see for example Theorem 2.13 of \cite{Le}).
Conjugation in $G$ is a smooth map, and because $\a$ is well-defined,
$\a$ is therefore smooth. The only things left to check are the crossed
module identities. First, let $[h],[h'] \in \tilde{G}$ and let $h$ and
$h'$ be representatives of $[h]$ and $[h']$ respectively. Then the map 
\be
[0,1]\times[0,1] \ni (s,t)
\mapsto h\Big((1-s)+st\Big)h'(t) h\Big((1-s)+st\Big)^{-1}
\ee
is a homotopy (relative endpoints) from the path
$t \mapsto h(1)h'(t)h(1)^{-1}$ (when $s=0$) to the path
$t \mapsto h(t)h'(t)h(t)^{-1}$ (when $s=1$). Therefore, 
\be
\begin{split}
\a_{\t([h])}([h'])&=[t\mapsto h(1) h'(t) h(1)^{-1} ] \\
&=[t\mapsto h(t)h'(t)h(t)^{-1}] \\
&= [h][h'][h^{-1}], 
\end{split}
\ee
which is the first identity (\ref{eq:Pfeiffer}). For the second
identity, let $g \in G$ and $[h]\in \tilde{G}$ with a representative
$h.$ Then 
\be
\t (\a_{g}([h])) = \t [ t \mapsto g h(t) g^{-1} ] = g h(1) g^{-1}
= g \t([h]) g^{-1}, 
\ee
which proves the other identity (\ref{eq:crossedmodule2}). 
\end{proof}

\begin{defn}
The Lie crossed module $( \tilde{G} , G, \t, \a)$ defined above is
called the \emph{\uline{universal cover crossed module}} associated to a
Lie group $G.$ The associated Lie 2-group, denoted by
$\mathcal{G}_{\{1\}},$ is called the \emph{\uline{universal cover 2-group}}
associated to the Lie group $G.$ 
\end{defn}

In fact, the only way to give a smooth covering map a Lie crossed module
structure is the way we have done so above. This follows from the
following Lemma. 

\begin{lem}
\label{lem:surjectivecrossedmodule}
Let $(H, G, \t, \a)$ be a crossed module (not necessarily Lie) with
$\t:H \to G$ a surjective homomorphism. Then $\a$ \emph{is} conjugation
in $H$ by a choice of lift, namely
\be
\a_{g} (h') = h h' h^{-1} , \quad \text{ for all } g \in G, h' \in H
\ee
for some $h$ with $\t(h) = g.$ 
\end{lem}

\begin{proof}
First we prove that conjugating by a lift is well-defined. Let
$\tilde{h} \in H$ be another lift with $\t ( \tilde{h} ) = g.$ Then 
\be
\begin{split}
h h' h^{-1} \left( \tilde{h} h' \tilde{h}^{-1} \right)^{-1}
&= h h' h^{-1}  \tilde{h} h'^{-1} \tilde{h}^{-1} \\
	&= \a_{ \t ( h ) } (h' ) \a_{ \t ( \tilde{h} ) } ( h'^{-1} )
\quad \mbox{ by (\ref{eq:Pfeiffer})} \\
	&= \a_{g } (h') \a_{g } (h'^{-1}) \\
	&= \a_{g} ( h' h'^{-1} ) \\
	&= \a_{g} (e) \\
	&= e
\end{split}
\ee
since $\a_{g} : H \to H$ is a homomorphism. The claim that
$\a_{g} (h') = h h' h^{-1}$ for a choice of lift $h$ of $g$ then follows
from the identity (\ref{eq:Pfeiffer}) since $\a_{g} (h')
= \a_{\t(h)}(h') = h h' h^{-1}$ for some $h$ because $\t$ is surjective
and a lift always exists. 
\end{proof}

\begin{lem}
\label{lem:surjectiveLiecrossedmodule}
Let $(H , G, \t, \a)$ be a Lie crossed module with $\t:H \to G$
a smooth covering map. Then $\a$ \emph{is}
conjugation in $H$ by a choice of lift, namely
\be
\a_{g} (h') = h h' h^{-1} , \quad \text{ for all } g \in G, h' \in H
\ee
for some $h$ with $\t(h) = g.$ 
\end{lem}

\begin{proof}
The claim holds even if $\t$ is just surjective.
The proof follows from Lemma
\ref{lem:surjectivecrossedmodule} viewing $H$ and $G$ as groups
(ignoring smooth structure) and using the identity
$\a_{g}(h')=\a_{\t(h)}(h')$ for some lift $h$ of $g.$ 
\end{proof}

Given any subgroup $N \le \pi_{1} (G),$ we can construct another Lie
2-group in a similar way but by using a different equivalence relation.
Define 
\be
\label{eq:NcoverG}
\tilde{G}_{N} :=
\{ h : [0,1] \to G \ | \ h(0) = e \text{ and $h$ is continuous} \}
/_{\sim_{N}} , 
\ee
where $h \sim_{N} h'$ if $h(1) = h' (1)$ and
$\left [ \begin{smallmatrix} h \\
\overset{\circ}{\overline{h'}} \end{smallmatrix} \right ] \in N,$ 
where $\overline{h'}$ denotes the reverse path 
and we use a vertical representation for the concatenation of paths in
this context
\be
\begin{smallmatrix} h \\ \overset{\circ}{\overline{h'}}
\end{smallmatrix}  (t) := 
\begin{cases}
h ( 2t) &\mbox{ for }  0 \le t \le \frac{1}{2} \\
h' (2-2t) & \mbox{ for } \frac{1}{2} \le t \le 1 
\end{cases}
.
\ee

\begin{defn}
An equivalence class of paths under the $\sim_{N}$ equivalence relation
in equation (\ref{eq:NcoverG}) will be denoted by $[ h ]_{N}$ or
$[ t \mapsto h(t) ]_{N}$ and will be called an \emph{\uline{$N$-class}}. 
\end{defn}

\begin{prop}
\label{prop:Ncovercrossedmodule}
Let $G$ be a connected Lie group, $N \le \pi_{1} (G)$ a subgroup, and
$\tilde{G}_{N}$ as in (\ref{eq:NcoverG}). Then for $[h]_{N} \in
\tilde{G}_{N},$ the function $\t : \tilde{G}_{N} \to G$ given by 
\be
\t \left( [h]_{N} \right ) := h(1), 
\ee
with $h$ a choice of a representative of $[h]_{N},$ is a well-defined
homomorphism. Furthermore, $\tilde{G}_{N}$ has a unique smooth structure
so that $\t$ is a smooth covering map. Finally,
$(\tilde{G}_{N}, G, \t,\a)$ with $\a : G \to
\mathrm{Aut}(\tilde{G}_{N})$ defined by 
\be
\a_{g} \left( [h]_{N} \right) := [ t \mapsto g h(t) g^{-1} ]_{N}
\ee
is a Lie crossed module. 
\end{prop}

\begin{proof}
$\t$ is well-defined by definition of the equivalence relation
$\sim_{N}.$ $\t$ is a homomorphism because $\t ( [h]_{N} [h']_{N})
= h(1) h'(1) = \t ( [h]_{N} ) \t ( [h']_{N} ).$ $\tilde{G}_{N}$ has a
natural topology coming from the quotient space of a subspace of paths
in $G.$ Because $\pi_{1} (G)$ is abelian, the conjugacy class of $N$ is
$N$ itself. Therefore, by a standard theorem of constructing covering
spaces (see for instance Chapter 3 of \cite{Ma})
$\t : \tilde{G}_{N} \to G$ is a covering map. By another standard result
in differential topology (see Proposition 2.12 of \cite{Le}), there is a
unique smooth structure on $\tilde{G}_{N}$ making $\t$ a smooth covering
map.
By construction, $\tilde{G}_{N}$ has a continuous multiplication making
it a topological group. The only things left to prove is that the
multiplication and inversion maps in $\tilde{G}_{N}$ are smooth. This
can be done locally using the smoothness of multiplication and inversion
in $G$ and the fact that $\t$ is a local diffeomorphism.
Therefore, $\tilde{G}_{N}$ is a Lie group. Since $\t$ is smooth, $\t$ is
a Lie group homomorphism. $\a_{g}$ is a well-defined group homomorphism
for all $g \in G$ because it can be described as conjugation. It is
smooth because for any $g \in G,$ there exists an open neighborhood $U$
around $g,$ a diffeomorphism $\varphi : U \to V,$ with $V$ a component
of $\t^{-1}(U),$ so that $U \ni g' \mapsto \a_{g'}$ coincides with
conjugation by $\varphi(g')$ by the proof of Lemma
\ref{lem:surjectiveLiecrossedmodule}. Since conjugation is smooth for
any Lie group, $\a$ is smooth. Therefore, $(\tilde{G}_{N}, G, \t,\a)$
is a Lie crossed module. 
\end{proof}

Note: We use the same notation $\t$ and $\a$ for the maps instead of
$\t_{N}$ and $\a_{N}$ since we typically fix $N$ in any given context. 

\begin{defn}
\label{defn:covering2groups}
Let $G$ be a Lie group and $N$ a subgroup of $\pi_{1} (G).$ Then
$(\tilde{G}_{N}, G, \t, \a)$ as described in Proposition
\ref{prop:Ncovercrossedmodule} is called the \emph{\uline{$N$-cover crossed
module}} of $G.$ Its associated 2-group is called the
\emph{\uline{$N$-covering 2-group}} and is denoted by
$\mathcal{B}\mathcal{G}_{N}.$ We sometimes abusively say
\emph{\uline{covering crossed module}} or \emph{\uline{covering 2-group}}
without referring to $N$ explicitly. 
\end{defn} 

Let $N \le \pi_{1} (G)$ be a subgroup of the fundamental group of a Lie
group $G.$ We will now construct a 2-category
$\widehat{G\text{-}\mathrm{Tor}_{N}}$ whose underlying 1-category
$\big ( \widehat{G\text{-}\mathrm{Tor}_{N}} \big )_{0,1}$ (see the
beginning of Section \ref{sec:trans2funct}) is $G\text{-}\mathrm{Tor}.$
Although the category $G\text{-}\mathrm{Tor}$ is not a Lie groupoid,
notice that the set of morphisms between any two $G$-torsors is
isomorphic to $G$ and therefore has a unique smooth structure.
Furthermore, the composition is a smooth map and is modeled by the
group multiplication map $G \times G \to G.$ By this we mean that by
choosing basepoints $a,b,$ and $c$ in $G$-Torsors $A, B,$ and $C$
respectively, the composition 
\be
G\text{-}\mathrm{Tor} (B,C) \times G\text{-}\mathrm{Tor} (A,B) \to
G\text{-}\mathrm{Tor} (A,C)
\ee
agrees with the multiplication $G \times G \to G$ under the isomorphisms
specified by the choice of basepoints. Therefore, the composition is
smooth. 
Thus, $G\text{-}\mathrm{Tor}$ is enriched in smooth manifolds. Using
this fact, we can extend $G\text{-}\mathrm{Tor}$ to an interesting
2-category $\widehat{ G\text{-}\mathrm{Tor}_{N} }$ in a non-trivial way.

Let $A$ and $B$ be two $G$-torsors and let $\varphi, \psi : A \to B$ be
two morphisms of $G$-torsors.
We define the set of 2-morphisms from $\varphi$ to $\psi$, drawn as
\be
\xymatrix{
B & & \ar@/_1.5pc/[ll]_{\varphi}="1" \ar@/^1.5pc/[ll]^{\psi}="2"
\ar@{}"1";"2"|(.2){\,}="3" \ar@{}"1";"2"|(.8){\,}="4"
\ar@{=>}"3";"4"^{}  A
}
,
\ee
to be the set of $N$-classes of paths from $\varphi$ to $\psi$ in
$G$-$\mathrm{Tor} (A,B).$ This means the following. 

\begin{defn}
Let $N \le \pi_{1}(G)$ be a subgroup. Two paths $\S : \varphi \to \psi$ 
and $\S' : \varphi \to \psi$ in $G$-$\mathrm{Tor} (A,B),$ drawn as 
\be
\xy 0;/r.15pc/:
(-20,0)*+{B}="B";
(20,0)*+{A}="A";
(-3,9)*+{}="f1";
(-3,-8)*+{}="y1";
(3,9)*+{}="f2";
(3,-8)*+{}="y2";
{\ar@/_1.5pc/_{\varphi} "A";"B"};
{\ar@/^1.5pc/^{\psi} "A";"B"};
{\ar@{=>}@/_0.5pc/ "f1" ; "y1"_{\S}};
{\ar@{=>}@/^0.5pc/ "f2" ; "y2"^{\S'}};
\endxy
,
\ee
are said to be \emph{\uline{$N$-equivalent}} if under the diffeomorphism
defined by 
\be
\label{eq:GtoTorABiso}
\begin{split}
G\text{-}\mathrm{Tor} (A,B) &\to G \\
\varphi & \mapsto e , 
\end{split}
\ee
the homotopy class of the loop $\begin{smallmatrix}  \S \\
\overset{\circ}{\overline{\S'}}\end{smallmatrix} : \varphi \to \varphi,$
which gets sent to an element of $\pi_{1}(G)$ under this diffeomorphism,
is an element of $N.$ The class associated to $\S$ is called an
\emph{\uline{$N$-class of paths}} and is denoted by $[\S]_{N}.$ 
\end{defn}

The choice of diffeomorphism (\ref{eq:GtoTorABiso}) where
$\varphi \mapsto e$ is merely for convenience. In particular, the
element  $\left[\begin{smallmatrix} \S \\
\overset{\circ}{\overline{\S'}} \end{smallmatrix} : \varphi \to
\varphi\right]$ is independent of this diffeomorphism. To see this, if
any other diffeomorphism was chosen, say sending some other morphism
$\varphi ' : A \to B$ to $e \in G,$ then there exists a unique $g \in G$
so that $\varphi \cdot g = \varphi '$ so that $\varphi \mapsto g^{-1}.$
In this case, one gets a loop based at $g^{-1}.$ To get one at $e,$ we
merely multiply by $g$ to obtain a loop based at $e \in G.$ This loop is
exactly the same as $\begin{smallmatrix} \S \\
\overset{\circ}{\overline{\S'}} \end{smallmatrix}$ under the
diffeomorphism defined by $\varphi \mapsto e.$ Therefore, the homotopy
class is independent of the diffeomorphism chosen. 

Vertical composition is defined on representatives as concatenation of
paths.
Horizontal composition can be defined using the $G \times G \to G$
multiplication. More explicitly, for two composable 2-morphisms as in 
\be
\xymatrix{
C & & \ar@/_1.5pc/[ll]_{\varphi'}="5" \ar@/^1.5pc/[ll]^{\psi'}="6"
\ar@{}"5";"6"|(.2){\,}="7" \ar@{}"5";"6"|(.8){\,}="8"
\ar@{=>}"7";"8"^{[\S']_N} B & & \ar@/_1.5pc/[ll]_{\varphi}="1"
\ar@/^1.5pc/[ll]^{\psi}="2" \ar@{}"1";"2"|(.2){\,}="3"
\ar@{}"1";"2"|(.8){\,}="4" \ar@{=>}"3";"4"^{[\S]_N}  A
}
,
\ee
choose representatives of such paths so that $\S : [0,1] \to
G\text{-}\mathrm{Tor} (A,B)$ and 
$\S' : [0,1] \rightarrow G\text{-}\mathrm{Tor}(B,C)$ 
with $\S ( 0 ) = \varphi, \S (1) = \psi, \S' (0) = \varphi',$ and
$\S' (1) = \psi'.$ Define the horizontal composition to be the $N$-class
of the path $\S' \circ \S$ defined by 
\be
s \mapsto (\S' \circ \S) (s) := \S' (s) \circ\S (s) \quad \text{ for } s
\in [ 0,1], 
\ee
where the composition on the right-hand-side is the usual composition of
morphisms in $G\text{-}\mathrm{Tor}.$ We check that horizontal
composition is well-defined. Suppose that $\S \sim_{N} \W$ and $\S'
\sim_{N} \W'.$ We must show that $\S' \circ \S \sim_{N} \W' \circ \W,$
i.e. 
\be
\left [ \begin{matrix}  \S' \circ \S \\
\overset{\circ}{\overline{\W' \circ \W}} \end{matrix} \right ] \in N
\ee
but a representative of this is given by 
\be
\begin{split}
\begin{matrix} \S' \circ \S \\
\overset{\circ}{\overline{\W' \circ \W}} \end{matrix} \;  (s) &= 
\begin{cases}
\S' (2s) \circ \S (2s) & \mbox{ for } 0 \le s \le \frac{1}{2} \\
\W' (2-2s) \circ \W (2-2s) & \mbox{ for } \frac{1}{2}  \le s \le 1
\end{cases}
\\
&= 
\begin{pmatrix} \S' \\ \overset{\circ}{\overline{\W'}} \end{pmatrix}
(s)\circ \begin{pmatrix} \S \\
\overset{\circ}{\overline{\W}} \end{pmatrix} (s) 
\end{split}
\ee
which gives two elements of $N$ (after taking the homotopy class) and
since $N$ is a subgroup the result is also an element of $N.$ A similar
argument is used to show that the interchange law holds. Therefore,
$\widehat{G\text{-}\mathrm{Tor}_{N}}$  defines a strict 2-category. We
summarize this as a definition. 

\begin{defn}
\label{defn:NGTor}
Let $G$ be a Lie group and $N \le \pi_{1} (G)$ a subgroup of the
fundamental group. The 2-category $\widehat{G\text{-}\mathrm{Tor}_{N}}$
has objects and 1-morphisms that of $G\text{-}\mathrm{Tor}.$ The
composition of 1-morphisms is the same as that in
$G\text{-}\mathrm{Tor}.$ The set of 2-morphisms between $G$-torsor
morphisms $\varphi$ and $\psi$ in $G\text{-}\mathrm{Tor} (A,B)$ are
$N$-classes of paths from $\varphi$ to $\psi.$ The vertical composition
of 2-morphisms is concatenation of representative paths. The horizontal
composition of 2-morphisms is the pointwise composition of $G$-torsor
morphisms after choosing representatives. 
\end{defn}

\begin{rmk}
When $N= \pi_{1} (G)$ the 2-categories
$\widehat{G\text{-}\mathrm{Tor}_{N}}$ and
$\widehat{G\text{-}\mathrm{Tor}}$ of \cite{SW4} are equivalent because
there is a unique $\pi_{1}(G)$-class of paths between any two morphisms
of $G$-torsors (since every loop is $\pi_{1} (G)$-equivalent to every
other loop). 
\end{rmk}

We will now start describing the path-curvature 2-functor, the structure
2-groupoid, and prove that it is indeed a transport 2-functor in the
sense of Definition \ref{defn:transport2functor}.

\begin{lem}
\label{lem:pathcurv}
Let $\tra \in \mathrm{Trans}_{\mathcal{B} G}^{1}
(M, G\text{-}\mathrm{Tor})$ be a transport functor and let
$N \le \pi_{1} (G)$ be a subgroup. Let $K_{N} (\tra) : \mathcal{P}_{2}
(M) \to \widehat{G\text{-}\mathrm{Tor}_{N}}$ be the following
assignment. At the level of objects and 1-morphisms $K_{N} (\tra)$
agrees with $\tra : \mathcal{P}_{1} (M) \to G\text{-}\mathrm{Tor}.$ For
every thin bigon $\G : \g \Rightarrow \de$ in $\mathcal{P}_{2} (M),$
choose a representative bigon, also denoted by $\G$, and let
\be
\label{eq:KNtrahomotopy}
K_{N} (\tra) ( \G )
:= \left [ s \mapsto \tra (\G ( \ \cdot \ , s) ) \right ]_{N}, 
\ee
i.e. the $N$-class of the path from $\tra (\g)$ to $\tra (\de)$ going
along $\tra ( \G ( \ \cdot \ , s) )$ as a function of $s \in [0,1].$
The notation means that $\G ( \ \cdot \ , s)$ is a thin path with
respect to the first coordinate for each fixed $s,$ and is depicted as a
one-parameter family of $G$-torsor morphisms  
\be
\xymatrix{
B & & & \ar@/_1.5pc/[lll]_{\varphi}="1" \ar@/^1.5pc/[lll]^{\psi}="2"
\ar@/_1pc/[lll] \ar@/_0.5pc/[lll]|-{\tra ( \G(\ \cdot \ , s) )} \ar[lll]
\ar@/^0.5pc/[lll] \ar@/^1pc/[lll]  A
}
.
\ee
This assignment is well-defined, i.e. the function
$s \mapsto \tra (\G ( \ \cdot \ , s) )$ defines a continuous path and
$K_{N} (\tra) ( \G )$ is independent of the choice of representative
bigon.
\end{lem}

$K_{N}$ is called the \uline{path-curvature 2-functor} 
associated to $\tra$ and $N \le \pi_{1}(G).$ 

\begin{proof}
The assignment in (\ref{eq:KNtrahomotopy}) is well-defined since
ordinary homotopy is a special case of thin homotopy. More explicitly
first notice that for a given bigon $\G : \g \Rightarrow \de$ the
function $s \mapsto \tra (\G (\ \cdot \ , s ) )$ is smooth because
$\tra$ is a transport functor (this follows for instance from Theorem
3.12 of \cite{SW1} and the fact that $\GTor(\tra(x),\tra(y))$ is
diffeomorphic to $G$). 
Now, suppose that $\G'$ is another representative bigon for the thin
bigon $\G.$ Then there exists a thin homotopy
$H : [0,1]\times[0,1]\times[0,1] \to M$ with $H(t,s,0) = \G(t,s)$ and
$H(t,s,1) = \G'(t,s).$ Thus $(s,r) \mapsto \tra(H(\ \cdot \ , s, r ) )$
is a smooth homotopy from $s \mapsto \tra (\G (\ \cdot \ , s ) )$ to
$s \mapsto \tra (\G' (\ \cdot \ , s ) ),$ which in particular is a
homotopy. Thus $K_{N} (\tra) ( \G )$ is well-defined. 
Similar arguments show that vertical and horizontal compositions are
respected under this assignment. Therefore, $K_{N} (\tra)$ defines a
strict 2-functor. 
\end{proof}

We construct a 2-functor $i_{N} : \mathcal{B} \mathcal{G}_{N} \to
\widehat{G\text{-}\mathrm{Tor}_{N}}$ as follows. By definition, a
2-morphism in  $\mathcal{B} \mathcal{G}_{N}$ is of the form 
\be
\xymatrix{
\bullet & & \ar@/_1.5pc/[ll]_{g}="1" \ar@/^1.5pc/[ll]^{h(1) g}="2"
\ar@{}"1";"2"|(.2){\,}="3" \ar@{}"1";"2"|(.8){\,}="4"
\ar@{=>}"3";"4"|-{([h]_{N},g)}  \bullet
}
\ee
where $[h]_{N}$ is viewed as an $N$-class of a path $h$ in $G$ starting
at the identity $e$ in $G$ and ending at a point written as
$\t ([h]_{N}) \equiv h(1).$ The image of this under $i_{N}$ is defined
to be
\be
\label{eq:iNimage}
\xymatrix{
G & & & \ar@/_2pc/[lll]_{L_{g}}="1" \ar@/^2pc/[lll]^{L_{ h(1) g}}="2"
\ar@{}"1";"2"|(.2){\,}="3" \ar@{}"1";"2"|(.8){\,}="4"
\ar@{=>}"3";"4"|-{[ s \mapsto L_{h(s) g} ]_{N} }  G
}
,
\ee
where $s \mapsto L_{h(s) g}$ is the path in ${G}\text{-}\mathrm{Tor}
(G,G) \cong G$ corresponding to the path $s \mapsto h (s) g$ in $G$
under this isomorphism.  At this point it is not clear why the vertical
composition is respected under $i_{N}.$

\begin{lem}
\label{lem:coveringmultiplication}
Let $(H, G, \t, \a)$ be a covering crossed module with elements of $H$
thought of as certain equivalence classes of paths in $G$ starting at
the identity $e \in G.$ Let $h$ and $h'$ be two representatives of
elements $[h], [h'] \in H.$ Denote the targets of $h$ and $h'$ by $g$
and $g',$ respectively. Then
\be
[h'] [h] = \left [ \begin{smallmatrix} h \\ \overset{\circ}{h' g}
\end{smallmatrix} \right ] , 
\ee
where $(h' g) (t) := h'(t) g$ for all $t \in [0,1],$ and the vertical
composition is the composition of paths starting with the one on top. 
\end{lem}

\begin{proof}
A homotopy is given by 
\be
(t,s) \mapsto 
\begin{cases}
h' ( st ) h \big ( (2-s) t \big )
& \mbox{ for } 0 \le t \le \frac{1}{2} \\
h' \big ( (2-s) t - 1 + s  \big ) h ( st + 1 - s )
& \mbox{ for } \frac{1}{2} \le t \le 1
\end{cases} 
\ee
with $s=0$ projecting to $\left [ \begin{smallmatrix} h \\
\overset{\circ}{h'  g}  \end{smallmatrix} \right ]$ and $s = 1$
projecting to $[h' h] = [h'] [h].$ 
\end{proof}

We now come to one of our main theorems. 

\begin{thm}
\label{thm:pathcurvature2functor}
The path-curvature 2-functor $K_{N} (\tra)$ defined above is a transport
2-functor with $\mathcal{B} \mathcal{G}_{N}$-structure. 
\end{thm}

\begin{proof}
To prove this, we must provide a $\pi_{N}$-local $i_{N}$-trivialization
of $K_{N} (\tra )$ and show that the associated descent object is
smooth. This will be done in several steps, outlined as follows. 
\begin{enumerate}[i)]
\item
Define $\triv_{N} : \mathcal{P}_{2} (Y) \to \mathcal{B} \mathcal{G}_{N}$
and show it is a smooth strict 2-functor. 
\item
Define a natural equivalence $t_{N} : \pi^{*} K_{N}(\tra) \Rightarrow
i_{N} \circ \triv_{N}.$ 
\item
Explicitly construct the associated descent object $(\triv_{N}, g_{N},
\psi_{N}, f_{N}).$ 
\item
Show that the descent object is smooth. 
\end{enumerate}

\begin{enumerate}[i)]
\item 
To start, $\tra : \mathcal{P}_{1} (M) \to G\text{-}\mathrm{Tor}$ is
assumed to be a transport functor, so there exists a $\pi$-local
$i$-trivialization $(\triv : \mathcal{P}_{1} (Y) \to \mathcal{B} G, t :
\pi_{1}^{*} \triv_{i} \Rightarrow \pi_{2}^{*} \triv_{i} ),$ where $\pi :
Y \to M$ is a surjective submersion, and whose associated descent object
$\mathrm{Ex}^{1}_{\pi} ( \tra, \triv, t )$ is smooth. We first define
$\pi_{N} : Y \to M$ to be $\pi$. Then we define $\triv_{N} :
\mathcal{P}_{2} (Y) \to \mathcal{B} \mathcal{G}_{N}$ by making it agree
with $\triv$ on the 1-category $\mathcal{P}_{1} (Y)$ inside
$\mathcal{P}_{2} (Y).$ For a thin bigon $\G : \g \Rightarrow \de$ in $Y$
we define
\be
\label{eq:triv}
\triv_{N} ( \G ) :=
\Big ( \left [ s \mapsto \triv ( \G ( \ \cdot \ , s ) )
\triv ( \g )^{-1} \right ]_{N} , \triv ( \g ) \Big ) \in \tilde{G}_{N}
\rtimes G
\ee
Note that $\left [ s \mapsto \triv ( \G ( \ \cdot \ , s ) ) \triv
( \g )^{-1} \right ]_{N}$ makes sense as an element of $\tilde{G}_{N}$
because $\tilde{G}_{N}$ is precisely defined to be the set of
$N$-classes of paths in $G$ starting at the identity of $G.$ This is
well-defined because thin homotopy factors through ordinary homotopy
(see Proof of Lemma \ref{lem:pathcurv}). 

We first prove that $\triv_{N}$ as defined is a strict 2-functor. It is
already a strict 2-functor at the level of objects and 1-morphisms. We
first check that vertical composition of bigons goes to vertical
composition of bigons. Consider two bigons $\G : \g \Rightarrow \de$
and%
\footnote{Technically, $\D : \de' \Rightarrow \e$ and there is a thin
homotopy $\S : \de \Rightarrow \de'$ but this means $\triv(\de) = \triv
(\de')$ so the above statement still holds.
}
 $\D : \de \Rightarrow \e.$ Their respected images under the assignment
above gives 
\be
\xymatrix{
\triv ( y ) & & \ar@/_2pc/[ll]_{\triv ( \g )}="1"
\ar[ll]|-{\triv ( \de )}="2" \ar@{}"1";"2"|(.25){\,}="3"
\ar@{}"1";"2"|(.85){\,}="4" \ar@{=>}"3";"4"|-{ \triv_{N} ( \G ) }
\ar@/^2pc/[ll]^{\triv ( \e )}="5" \ar@{}"2";"5"|(.15){\,}="6"
\ar@{}"2";"5"|(.75){\,}="7"
\ar@{=>}"6";"7"|-{\triv_{N} ( \D )} \triv ( x ) 
}
\quad = \quad 
\xymatrix{
\triv ( y ) & & \ar@/_1.5pc/[ll]_{\triv ( \g )}="1"
\ar@/^1.5pc/[ll]^{\triv ( \e )}="2" \ar@{}"1";"2"|(.15){\,}="3"
\ar@{}"1";"2"|(.85){\,}="4" \ar@{=>}"3";"4"|-{\begin{smallmatrix}
\triv_{N} ( \G )\\  \triv_{N} ( \D )
\end{smallmatrix} } \triv ( x )
}
,
\ee
which, after taking the $\tilde{G}_{N}$ component, gives 
\be
\left [ s \mapsto \triv ( \D ( \ \cdot \ , s ) ) \triv ( \de )^{-1}
\triv ( \G ( \ \cdot \ , s ) ) \triv ( \g )^{-1} \right ]_{N} 
\ee
while first composing in $\mathcal{P}_{2} (Y)$ and then applying
$\triv_{N}$ gives%
\footnote{Again, this is technically not correct. One has to use a thin
homotopy $\S : \de \Rightarrow \de'$ but the reader can check that the
proof follows through with a slightly modified homotopy.}
\be
\triv_{N} \left ( \begin{matrix} \G \\ \overset{\circ}{\D} \end{matrix}
\right ) = \left (  \left [ s \mapsto \triv \left ( \begin{matrix} \G \\
\overset{\circ}{\D} \end{matrix} ( \ \cdot \ ,s) \right )
\triv (\g)^{-1} \right ]_{N} , \triv (\g) \right ) .
\ee
A homotopy between these two representatives is 
given by $H (s,r) := $
\be
{\footnotesize
\begin{cases}
\triv ( \G ( \ \cdot \ , (r+1)s ) ) \triv (\g)^{-1}
& \mbox{ for } 0 \le s \le \frac{r}{2} \\
\triv ( \D ( \ \cdot \ , (r+1)s -r ) )
\triv (\de)^{-1} \triv ( \G ( \ \cdot \ , (r+1)s ) ) \triv (\g)^{-1}
& \mbox{ for } \frac{r}{2} \le s \le 1- \frac{r}{2} \\
\triv ( \D ( \ \cdot \ , (r+1)s -r ) )  \triv (\g)^{-1}
& \mbox{ for } 1- \frac{r}{2} \le s \le 1 
\end{cases} 
}
\ee
which indeed satisfies 
\be
H ( s, 0) = \triv ( \D ( \ \cdot \ , s ) ) \triv (\de)^{-1}
\triv ( \G ( \ \cdot \ , s ) ) \triv (\g)^{-1} 
\ee
and 
\be
H (s, 1) = 
\begin{cases}
\triv ( \G ( \ \cdot \ , 2s ) ) \triv (\g)^{-1}
& \mbox{ for } 0 \le s \le \frac{1}{2}  \\
\triv ( \D ( \ \cdot \ , 2s -1 ) )  \triv (\g)^{-1}
& \mbox{ for } 1- \frac{1}{2} \le s \le 1 
\end{cases} 
.
\ee
This proves more than what we needed since all we had to show was that
the two elements are in the same $N$-class. Showing that the two
representatives are homotopic is stronger and implies they're in the
same $N$-class. 

Now consider the horizontal composition of $\G : \g \Rightarrow \de$ and
$\Pi : \a \Rightarrow \b$ written as $\Pi \circ \G : \a \circ \g
\Rightarrow \b \circ \de.$ First composing the thin bigons and then
applying the map $\triv_{N}$ gives 
\be
\triv_{N} ( \Pi \circ \G ) = \left ( \left [ s \mapsto \triv \Big ( (\Pi
\circ \G) ( \ \cdot \ , s ) \Big ) \triv ( \a \circ \g ) ^{-1}
\right ]_{N} , \triv(\a \circ \g) \right ) 
\ee
while first applying the map $\triv$ to each thin bigon and then
multiplying in $\mathcal{B} \mathcal{G}_{N}$ gives 
\be
\begin{split}
&p_{\tilde{G}_{N}} \left ( \triv_{N} ( \Pi ) \triv_{N} ( \G ) \right )\\
&= \left [ s \mapsto \triv ( \Pi ( \ \cdot \ , s ) ) \triv ( \a )^{-1}
\triv ( \a ) \triv ( \G ( \ \cdot \ , s ) ) \triv ( \g )^{-1}
\triv ( \a )^{-1} \right ]_{N} \\
&= \left [ s \mapsto \triv ( \Pi ( \ \cdot \ , s ) ) \triv ( \G ( \
\cdot \ , s ) ) \triv ( \a \circ \g ) ^{-1}  \right ]_{N} \\
&= \left [ s \mapsto \triv \Big ( (\Pi \circ \G) ( \ \cdot \ , s )
\Big ) \triv ( \a \circ \g ) ^{-1} \right ]_{N}
\end{split}
\ee
because for every fixed $s,$ parallel transport of paths is a
homomorphism.
Therefore, $\triv_{N}$ defines a strict 2-functor. 

We now show that $\triv_{N}$ is a smooth 2-functor. We already know
$\triv_{N}$ is smooth at the level of objects and 1-morphisms. We must
therefore show $\triv_{N} : P^2 Y \to \tilde{G}_{N} \rtimes G$ is
smooth. At this point, the reader should review Appendix
\ref{sec:smooth} because we will recall several facts in the proof of
this claim. By Definition \ref{defn:smoothmaps}, $\triv_{N}$ is smooth
if and only if for every plot $\varphi : U \to P^{2} Y,$ the composition
$\triv_{N} \circ \varphi : U \to \tilde{G}_{N} \rtimes G$ is a plot. By
Example \ref{ex:manifold}, $\triv_{N} \circ \varphi$ is a plot if and
only if it is smooth. By Example \ref{ex:product}, $\triv_{N} \circ
\varphi$ is smooth if and only if both projections $p_{G} \circ\triv_{N}
\circ \varphi$ and $p_{\tilde{G}_{N}} \circ \triv_{N} \circ \varphi$ are
smooth. Since we already showed that $p_{G} \circ \triv_{N} \circ
\varphi = \triv \circ s \circ \varphi$ is smooth (here $s$ is the source
of a thin bigon), it remains to show that $p_{\tilde{G}_{N}} \circ
\triv_{N} \circ \varphi$ is smooth. 

For convenience for this proof, set $f:= p_{\tilde{G}_{N}} \circ
\triv_{N}.$ By definition, $f \circ \varphi$ is given by 
\be
U \ni u \mapsto \left [ s \mapsto \triv \big ( \varphi (u) ( \ \cdot \ ,
s ) \big ) \triv \big ( \varphi (u) ( \ \cdot \ , 0)  \big )^{-1}
\right ]_{N},
\ee 
where we've chosen a representative bigon $\varphi (u) : [0,1] \times
[0,1] \to Y,$ fixed $s$ to get a thin path $\varphi (u) ( \ \cdot \ ,
s ),$ and then applied $\triv$ (unfortunately, there is a lot of abuse
of notation to avoid an overabundance of brackets and symbols). The
problem with this is that although we know we can always choose bigons
$\varphi(u),$ these choices need not form a \emph{smooth} family of
bigons in an open neighborhood of $u \in U.$ Therefore, proving
smoothness this way will not work. 

Instead, we use the smooth structures we've defined to \emph{construct}
such a smooth family of bigons. $P^{2} Y$ is the quotient of $BY,$
bigons in $Y,$ under thin homotopy and its smooth structure was defined
as such. Therefore, by Example \ref{ex:quotient}, $\varphi : U \to
P^{2} Y$ is a plot if and only if there exists an open cover
$\{ U_{j} \}_{j \in J}$ of $U$ and plots $\{ \varphi_{j} : U_{j} \to
BY \}_{j \in J}$ such that
\be
\label{eq:liftingthinbigons}
\xy0;/r.15pc/:
(-15,15)*+{BY}="X";
(-15,-15)*+{P^{2} Y}="Y";
(15,15)*+{U_{j}}="j";
(15,-15)*+{U}="U";
{\ar@{^{(}->} "j";"U"};
{\ar_(0.50){\varphi_{j}}"j";"X"};
{\ar^(0.45){\varphi}"U";"Y"};
{\ar"X";"Y"_{q}};
\endxy
\ee
commutes for all $j \in J.$ For the purposes of this proof, $q$ is the
quotient map. 

Now, $BY$ itself is a subspace of the space of smooth squares
$Y^{[0,1]^{2}}$ in $Y.$ Denote the inclusion of $BY$ into
$Y^{[0,1]^{2}}$ by $k.$ By Example \ref{ex:subspace}, $\varphi_{j} :
U_{j} \to BY$ is a plot if and only if $k \circ \varphi_{j} : U_{j} \to
Y^{[0,1]^{2}}$ are plots. By Example \ref{ex:maps}, $k \circ \varphi_{j}
: U_{j} \to Y^{[0,1]^{2}}$ is a plot if and only if the associated
function $\widetilde{k\circ\varphi_{j}} : U_{j} \times [0,1]^{2} \to Y$
defined by $\widetilde{k\circ\varphi_{j}} (u,t,s) :=
\Big(k\big(\varphi_{j}(u)\big)\Big)(t,s)$ is smooth. This gives us our
first desired fact: the plot $\varphi : U \to P^2 Y$ gives a smooth
family of bigons $\varphi_{j} : U_{j} \to BY$ such that $q\circ
\varphi_{j} = \varphi |_{U_{j}}.$ Furthermore, since
$\widetilde{k\circ\varphi_{j}}$ is a smooth map of finite-dimensional
manifolds, it is continuous and therefore the smooth family of bigons
is also continuous. 

By using another adjunction, the smooth map
$\widetilde{k\circ\varphi_{j}}$ can be turned into a plot
$\widehat{k\circ\varphi_{j}} : U_{j} \times [0,1] \to Y^{[0,1]}$ that
factors through paths with sitting instants and is defined by
$\widehat{k\circ\varphi_{j}} (u,s) \big (t\big) :=
\Big(k\big(\varphi_{j}(u)\big)\Big)(t,s).$ Using this, we get a smooth
map $U_{j} \times [0,1] \to G$ given by 
\be
(u,s) \mapsto \triv \left ( \widehat{k\circ\varphi_{j}} (u,s) \right )
\triv \left ( \widehat{k\circ\varphi_{j}} (u,0) \right )^{-1} 
\ee
because $\triv$ is smooth on thin paths (we've taken the thin homotopy
classes of the paths $\widehat{k\circ\varphi_{j}} (u,s)$ and
$\widehat{k\circ\varphi_{j}} (u,0)$ in the arguments of $\triv$). For
each fixed $u \in U_{j} \subset U,$ this gives a path in $G$ starting at
$e$ and the $N$-class of this path coincides with $f (\varphi(u))$ by
commutativity of the diagram in (\ref{eq:liftingthinbigons}). By
continuity (which we proved in the previous paragraph), for each $u$
there exists a (sufficiently small) contractible open neighborhood $V$
of $u$ with $u \in V \subset U_{j}$ together with a (sufficiently small)
contractible open neighborhood $W$ of $f (\varphi(u))$ in
$\tilde{G}_{N}$ such that $f \big(\varphi (V)\big) \subset W$ and $W$
maps diffeomorphically to $\t(W) \subset G$ under the smooth covering
map $\t.$ But we just showed that the projection $\t \circ f \circ
\varphi |_{V} : V \to G$ is smooth and since all neighborhoods are small
and contractible, a lift is uniquely specified, is smooth, and agrees
with $f \circ \varphi |_{V}.$ Therefore, $f \circ \varphi$ is smooth at
the point $u \in U.$ By applying this argument to all plots at all
points, this proves that $f = p_{\tilde{G}_{N}} \circ \triv_{N} : P^2 Y
\to \tilde{G}_{N}$ is smooth. 
 
\item
Our goal now is to define a natural equivalence $t_{N} : \pi^{*}
K_{N}(\tra) \Rightarrow i_{N} \circ \triv_{N}.$ Note that since $\tra$
is a transport functor, we have a natural isomorphism $t : \pi^{*} \tra
\Rightarrow i \circ \triv.$ Therefore, on points $y \in Y,$ i.e. objects
of $\mathcal{P}_{2} (Y),$ define $t_{N} (y) := t (y).$ For $\g \in P^{1}
Y,$ since $t$ was a natural transformation for ordinary functors, the
required diagram already commutes so we set $t_{N} (\g) := \id.$ 

\item 
Because of our definition of $\triv_{N}$ and $t$ and since our target
category is a strict 2-category, the associated descent data will not be
too different from the ordinary transport functor case. Namely, the
modifications $\psi_{N}$ and $f_{N}$ are both trivial, i.e. they are the
identity 2-morphisms on objects.  $g_{N}$ is also completely specified
by $g$ since $t_{N}$ is specified by $t.$ 

\item
As mentioned above, $\triv_{N}$ is smooth. What's left to show is that
$\mathcal{F} (g_{N}) : \mathcal{P}_{1} ( Y^{[2]} ) \rightarrow \L
\widehat{G\text{-}\mathrm{Tor}_{N}}$ is a transport functor with
$\L \mathcal{B} \mathcal{G}_{N}$-structure. First let's explicitly
describe $\L \mathcal{B} \mathcal{G}_{N}$ and $\L
\widehat{G\text{-}\mathrm{Tor}_{N}}.$ The objects of $\L \mathcal{B}
\mathcal{G}_{N}$ are 1-morphisms of $\mathcal{B} \mathcal{G}_{N}$ which
are precisely elements of $G.$ A morphism from $g_{1}$ to $g_{2}$ in $\L
\mathcal{B} \mathcal{G}_{N}$ is a pair of elements $g_{3}$ and $g_{4}$
of $G$ along with an element $h \in H$ fitting into the diagram 
\be
\xy 0;/r.15pc/:
(15,-15)*+{\bullet}="3";
(-15,-15 )*+{\bullet}="4";
(-15,15)*+{\bullet}="2";
(15,15)*+{\bullet}="1";
{\ar_{g_{3}} "1";"2" };
{\ar_{g_{2}} "2";"4" };
{\ar^{g_{1}} "1";"3"};
{\ar^{g_{4}} "3";"4"};
{\ar@{=>}|-{(h, g_{2} g_{3})} "2";"3"};
\endxy
.
\ee
Similarly an object of $\L \widehat{G\text{-}\mathrm{Tor}_{N}}$ is a
pair of objects $P$ and $P'$ in $\widehat{G\text{-}\mathrm{Tor}_{N}}$
and a $G$-equivariant map $P \xrightarrow{f} P'.$ A morphism from
$P \xrightarrow{f} Q$ to $P' \xrightarrow{g} Q'$ in
$\L \widehat{G\text{-}\mathrm{Tor}_{N}}$ is a pair of $G$-equivariant
maps $p : P \to P'$ and $q : Q \to Q'$ along with an $N$-class of a path
$\a : g \circ p \Rightarrow q \circ f$ as in the diagram 
\be
\xy 0;/r.15pc/:
(15,-15)*+{Q}="3";
(-15,-15 )*+{Q'}="4";
(-15,15)*+{P'}="2";
(15,15)*+{P}="1";
{\ar_{p} "1";"2" };
{\ar_{g} "2";"4" };
{\ar^{f} "1";"3"};
{\ar^{q} "3";"4"};
{\ar@{=>}^{\a} "2";"3"};
\endxy
.
\ee
By applying the general definition of $\mathcal{F} (g_{N}),$ we have
\be
\xy0;/r.15pc/:
(-30,0)*+{Y^{[2]} \ni y}="1";
(40,10)*+{i ( \triv ( \pi_{1} (y)))=G}="2";
(40,-10)*+{i ( \triv ( \pi_{2} (y)))=G}="3";
{\ar^{L_{g(y)}} "2";"3"};
{\ar@{|->}^{\mathcal{F} (g_{N}) } "1"+(20,0);(10,0)};
\endxy
\ee
and
\be
P^{1} Y^{[2]} \ni 
\xy0;/r.15pc/:
(-60,0)*+{\Big (y'}="1y";
(-30,0)*+{y\Big )}="1x";
{\ar_{\g} "1x";"1y"};
(30,10)*+{G}="2y";
(35,5)*+{}="2ya";
(30,-10)*+{G}="3y";
(70,10)*+{G}="2x";
(70,-10)*+{G}="3x";
(65,-5)*+{}="3xb";
{\ar^{L_{g(y)}} "2x";"3x"};
{\ar_{L_{g(y')}} "2y";"3y"};
{\ar_{i ( \triv ( \pi_{1} ( \g) ) )} "2x";"2y"};
{\ar^{i ( \triv ( \pi_{2} ( \g) ) )} "3x";"3y"};
{\ar@{=>}^{\id} "2ya";"3xb"};
{\ar@{|->}^{\mathcal{F} (g_{N}) } "1x"+(15,0);(5,0)};
\endxy
.
\ee
Now, since $g$ is part of the smooth descent object for the functor
$\tra,$ there exists a smooth natural isomorphism $\tilde{g} :
\pi_{1}^{*} \triv \Rightarrow \pi_{2}^{*} \triv$ such that $g = \id_{i}
\circ \tilde{g}.$ Using this fact, one can define $\widetilde{g_{N}} :
\pi_{1}^{*} \triv_{N} \Rightarrow \pi_{2}^{*} \triv_{N}$ in an analogous
way to how $g_{N}$ was defined from $g$ but this time using $\tilde{g}.$
Furthermore, $\mathcal{F}(g_{N})$ factors through $\L ( i_{N} )$ via
$\mathcal{F}(g_{N}) = \L ( i_{N} ) \circ \mathcal{F}(\widetilde{g_{N}})$
since $g = \id_{i} \circ \tilde{g}.$ 

Therefore, this defines a global trivialization with the identity
surjective submersion $\id : Y^{[2]} \to Y^{[2]}$ with the
trivialization functor being $\mathcal{F}( \widetilde{g_{N}} ) :
\mathcal{P}_{1} ( Y^{[2]} ) \to \L \mathcal{B} \mathcal{G}_{N}.$ This
functor is smooth since $\tilde{g}$ is smooth. Furthermore, the descent
object associated to this transport functor is trivial because of the
global trivialization. Thus $\mathcal{F}( g_{N} )$ defines a transport
functor. 
\end{enumerate}

Thus $K_{N} (\tra)$ defines a transport 2-functor with $\mathcal{B}
\mathcal{G}_{N}$ structure. 
\end{proof}

\begin{defn}
\label{defn:pathcurvature2functor}
Let $\tra : \mathcal{P}_{1} (M) \to G\text{-}\mathrm{Tor}$ be a
transport functor over $M$ with $\mathcal{B} G$ structure and values in
$G\text{-}\mathrm{Tor}$ and let $N \le \pi_{1} (G)$ be a subgroup. Then
the transport 2-functor $K_{N} (\tra) : \mathcal{P}_{2} (M) \to
\widehat{G\text{-}\mathrm{Tor}_{N}}$ defined by 
\be
\xymatrix{ y & & \ar@/_1pc/[ll]_{\g}="1"  \ar@/^1pc/[ll]^{\de}="2"
\ar@{}"1";"2"|(.20){\,}="3" \ar@{}"1";"2"|(.80){\,}="4"
\ar@{=>}"3";"4"^{ \G }  x } 
\mapsto 
\xymatrix{ \tra (y) & & & \ar@/_1.5pc/[lll]_{\tra(\g)}="1"
\ar@/^1.5pc/[lll]^{\tra(\de)}="2"   \ar@{}"1";"2"|(.20){\,}="3"
\ar@{}"1";"2"|(.80){\,}="4"
\ar@{=>}"3";"4"|-{ [s\mapsto\tra (\G ( \ \cdot \ ,s) ]_{N} } \tra(x)}
\ee
is called the \emph{\uline{path-curvature transport 2-functor}} associated
to $\tra$ and $N.$ 
\end{defn}

More can be said, although we will not prove the details since the
proofs are simple. The above construction is functorial. Namely, for
any morphism of parallel transport functors $h : \tra \Rightarrow \tra'$
with $\mathcal{B} G$-structure with values in $G$-$\mathrm{Tor},$ there
is a corresponding 1-morphism of parallel transport 2-functors $h_{N} :
K_{N} ( \tra ) \Rightarrow K_{N} (\tra' )$ with $\mathcal{B}
\mathcal{G}_{N}$-structure with values in
$\widehat{G\text{-}\mathrm{Tor}_{N}}.$ By viewing
$\mathrm{Trans}^{1}_{\mathcal{B}G} (M, G\text{-}\mathrm{Tor})$ as a
2-category whose 2-morphisms are all identities, this defines a
2-functor 
\be
K_{N} : \mathrm{Trans}^{1}_{\mathcal{B}G} (M, G\text{-}\mathrm{Tor}) \to
\mathrm{Trans}^{2}_{\mathcal{B}\mathcal{G}_{N}}
(M, \widehat{G\text{-}\mathrm{Tor}_{N}}). 
\ee
In fact, in the above proof, in steps i) and ii), we have also outlined
the definition of a 2-functor (see equation (\ref{eq:triv}) and
surrounding text)
\be
K_{N}^{\mathrm{Triv}} : \mathrm{Triv}^{1}_{\pi} ( i )^{\infty}
\to \mathrm{Triv}^{2}_{\pi} ( i_{N} )^{\infty} 
\ee
given by the assignment 
\be
(\tra, \triv, t ) \mapsto (K_{N} (\tra), \triv_{N} , t_{N} := t )
\ee
on objects (see Definitions \ref{defn:trivialization1} and
\ref{defn:trivialization2}) and
\be
\a \mapsto \a_{N} := \a
\ee 
on morphisms (see Definitions \ref{defn:trivialization1mor} and
\ref{defn:trivialization2mor}). In these two assignments, we are
viewing a natural transformation as a pseudonatural transformation by
assigning the identity 2-morphism to every 1-morphism. 

In steps iii) and iv) we have also outlined the definition of a
2-functor 
\be
K_{N}^{\mathfrak{Des}} : \mathfrak{Des}^{1}_{\pi} ( i )^{\infty}
\to \mathfrak{Des}^{2}_{\pi} ( i_{N} )^{\infty} 
\ee
given by the assignment 
\be
(\triv, g ) \mapsto (\triv_{N} , g_{N} := g, \psi_{N} := 1, f_{N} := 1 )
\ee
on objects (see Definitions \ref{defn:descent1} and \ref{defn:descent2})
and 
\be
h \mapsto (h_{N} := h, \e_{N} := 1)
\ee 
on morphisms (see Definitions \ref{defn:descent1mor} and
\ref{defn:descent2mor}). 

By definition, both squares in the diagram
\be
\label{eq:commutingpathcurvature1}
\xy 0;/r.15pc/:
(-80,-15)*+{\mathfrak{Des}^{2} (i_{N} )^{\infty}}="6";
(-80,15)*+{\mathfrak{Des}^{1} (i )^{\infty}}="5";
(30,-15)*+{\mathrm{Trans}^{2}_{\mathcal{B}\mathcal{G}_{N}}
(M, \widehat{G\text{-}\mathrm{Tor}_{N}})}="3";
(-30,-15 )*+{ \mathrm{Triv}^{2} ( i_{N} )^{\infty} }="4";
(-30,15)*+{\mathrm{Triv}^{1} ( i )^{\infty}}="2";
(30,15)*+{\mathrm{Trans}^{1}_{\mathcal{B}G}
(M, G\text{-}\mathrm{Tor}) }="1";
{\ar_(0.6){c} "1";"2" };
{\ar^{K_{N}} "1";"3"};
{\ar|-{K_{N}^{\mathrm{Triv}}} "2";"4"};
{\ar_(0.6){c} "3";"4" };
{\ar_{\mathrm{Ex}^{1}} "2";"5" };
{\ar_{K_{N}^{\mathfrak{Des}}} "5";"6"};
{\ar_{\mathrm{Ex}^{2}} "4";"6" };
\endxy
\ee
commute (on the nose). 

The path-curvature 2-functor associated to a transport functor is
\emph{flat}. To explain this, we first define a modified version of the
thin path 2-groupoid. 

\begin{defn}
\label{defn:fundamental2groupoid}
Let $X$ be a smooth manifold. If one drops condition ii) from Definition
\ref{defn:path2groupoid}, then one obtains a 2-groupoid $\Pi_{2} (X)$
that has points of $X$ as objects, thin paths for 1-morphisms, and
(ordinary) homotopy classes of bigons for 2-morphisms. 
\end{defn}

\cite{SW4} call this 2-groupoid the \emph{\uline{fundamental 2-groupoid}}.
Although we prefer to use that terminology for the usual fundamental
2-groupoid (whose 1-morphisms are \emph{also} ordinary homotopy classes
of paths), we use this terminology for the purposes of this paper to
avoid confusion. 

\begin{defn}
A transport 2-functor $F : \mathcal{P}_{2} (M) \to T$ with
$\mathrm{Gr}$-structure is said to be \emph{\uline{flat}} if it factors
through the fundamental 2-groupoid $\Pi_{2} (M).$  
\end{defn}

The curvature 2-functor $K ( \tra ) \equiv K_{\pi_{1}(G)} ( \tra)$
introduced in \cite{SW4} is completely determined on bigons by the
boundary of the bigon. It is therefore obviously flat, but it is even
more restrictive than just that. Not only does it not depend on the
homotopy class of the bigon, it doesn't depend on the bigon at all. On
the other hand, the path-curvature 2-functor $K_{N} ( \tra )$ introduced
here depends on the homotopy class of the bigon.

\begin{cor}
\label{cor:pcflat}
The path-curvature 2-functor $K_{N} (\tra)$ is flat. 
\end{cor}

\begin{proof}
Let $\G$ and $\G'$ be two bigons that are smoothly homotopic (as opposed
to just thinly homotopic).
Let $H : [0,1]^{3} \to Y$ be a smooth homotopy from $\G$ to $\G'$ so
that $H(t,s,0) = \G(t,s)$ and $H(t,s,1) = \G'(t,s).$ By compactness of
$[0,1]^{3},$ one can choose $H$ so that it has sitting instants around
its boundary so that $\tra( H(\ \cdot \ , s, r) )$ is well-defined for
each $s, r \in [0,1].$ Then 
\be
\begin{split}
[0,1]\times[0,1] &\to G \\
(s,r) &\mapsto \tra ( H ( \ \cdot \ , s, r ) )
\end{split}
\ee
is a smooth homotopy from the path $s\mapsto\tra (\G( \ \cdot \ ,s) )$
to the path $s \mapsto \tra ( \G'( \ \cdot \ , s ) ).$ Therefore, since
$N$-classes of paths is a quotient of the universal cover $\tilde{G},$
the $N$-classes of these paths agree.
\end{proof}

\subsection{A description in terms of differential form data}
\label{sec:general}

In this section, we prove several important and useful facts. The first
theorem says that \emph{locally} transport functors whose structure
2-group is a covering 2-group can be described in terms of the
path-curvature 2-functor. The second part of this section contains a
discussion about the relationship between the path-curvature 2-functor
specifically and its differential cocycle data. As before, let
$\pi : Y \to M$ denote a surjective submersion, $G$ a connected Lie
group, $N \le \pi_{1} (G)$ a subgroup, and $\t : \tilde{G}_{N} \to G$
the cover of $G$ defined by $N.$ It is important to note that
$\un \t : \underline{\tilde{G}_{N}} \to \un G,$ the induced map of Lie
algebras, is an isomorphism of Lie algebras because $\t$ is a local
diffeomorphism. Denote the 2-group associated to the Lie crossed module
$(\tilde{G}_{N}, G, \t, \a)$ by $\mathcal{B} \mathcal{G}_{N}.$ 

First, we define a 2-functor $K_{N}^{Z} : Z^{1}_{\pi} (G) \to
Z^{2}_{\pi} ( \mathcal{G}_{N} )$ by
\be
\begin{split}
(A,g) &\mapsto \left( \left( A, B := \un \t^{-1} \left(dA
+ \frac{1}{2}[A,A]\right) \right), (g,\varphi := 1), (\psi := 1, f := 1)
\right) \\
h &\mapsto (h, \varphi := 0)
\end{split}
\ee
on objects and morphisms, respectively.

Second, notice that specifically for the path-curvature 2-functor
$K_{N} (\tra),$ and particularly its associated descent object
$K_{N}^{\mathfrak{Des}} ( \tra ),$ the analysis in Section
\ref{sec:2cocycle} gives the following differential cocycle data
associated to $K_{N} (\tra).$ The assignment on thin paths induces a
1-form $A$ with values in $\un G$ 
since the functor $K_{N} (\tra)$ agrees precisely with $\tra$ on thin
paths. On thin bigons, the assignment induces a 2-form $B$ with values
in 
$\underline{\tilde{G}_{N}}$ 
satisfying $dA + \frac{1}{2} [ A, A ] = \un \t (B).$ Since $\un \t$ is
an isomorphism, $B$ is determined by this condition and is given by
$B = \un \t^{-1} \left (  dA + \frac{1}{2} [ A, A ]  \right ).$
Therefore, the associated differential cocycle data to the
path-curvature 2-functor $K_{N} (\tra)$ is 
\be
{\small
\begin{split}
\mathcal{D} ( K_{N}^{\mathfrak{Des}} ( \triv, g ) )
= \left ( A, B:= \un \t^{-1} \left ( d A + \frac{1}{2} [A,A] \right ),g,
\varphi = 0, f = 1, \psi = 1 \right ).
\end{split}
}
\ee
Therefore, the two descriptions agree showing that the diagram
\be
\label{eq:thediagram}
\xy 0;/r.15pc/:
(-120,-15)*+{Z^{2}_{\pi} (\mathcal{G}_{N} )}="8";
(-120,15)*+{Z^{1}_{\pi} (G)}="7";
(-80,-15)*+{\mathfrak{Des}^{2}_{\pi} (i_{N} )^{\infty}}="6";
(-80,15)*+{\mathfrak{Des}^{1}_{\pi} (i )^{\infty}}="5";
{\ar_{\mathcal{D}} "5";"7" };
{\ar "7";"8"_{K_{N}^{Z}}};
{\ar "5";"6"^{K_{N}^{\mathfrak{Des}}}};
{\ar^{\mathcal{D}} "6";"8" };
\endxy
\ee
commutes. This analysis is actually a bit more general as the following
theorem shows. 

\begin{thm}
\label{thm:main}
Let $X$ be a smooth manifold and  $F_{N} : \mathcal{P}_{2} (X) \to
\mathcal{B} \mathcal{G}_{N}$ be \emph{any} smooth strict 2-functor. Then
there exists a unique smooth functor $F : \mathcal{P}_{1} (X) \to
\mathcal{B} G$ such that $F_{N} = K_{N} ( F).$ 
\end{thm}

\begin{proof}
The functor $\mathcal{D}_{X} : \mathrm{Funct}^{\infty}
(\mathcal{P}_{2} (X), \mathcal{B} \mathcal{G}_{N} ) \to Z^{2}_{X}
(\mathcal{G}_{N})^{\infty}$ (defined around
(\ref{eq:integratedifferentiate2}) in Section \ref{sec:2cocycle})
produces $(A \in \W^{1} (X; \un G), B \in \W^{2} (X;
\underline{\tilde{G}_{N}}) )$ satisfying $dA + \frac{1}{2} [ A, A ]
= \un \t (B).$ Since $\un \t : \underline{\tilde{G}_{N}} \to \un G$ is
an isomorphism, $B = \un \t^{-1} \left (  dA + \frac{1}{2} [ A, A ]
\right ).$ Restricting $F_{N}$ to $\mathcal{P}_{1} (X)$ gives a unique
$F : \mathcal{P}_{1} (X) \to \mathcal{B} G$ that satisfies
$\mathcal{D}_{X} ( F ) = A.$ By the the same token, we have
$\mathcal{D}_{X} ( K_{N} (  F ) ) = \left ( A, \un \t^{-1}
\left ( d A + \frac{1}{2} [ A , A ] \right ) \right )$ which coincides
with $\mathcal{D}_{X} (F_{N} ).$ Since  $\mathcal{P}_{X} : Z^{2}_{X}
(\mathcal{G}_{N})^{\infty} \rightarrow
 \mathrm{Funct}^{\infty} (\mathcal{P}_{2}
(X), \mathcal{B} \mathcal{G}_{N} )$ is a strict inverse to
$\mathcal{D}_{X}$ by Theorem 2.21 of \cite{SW2}, we conclude that
$F_{N} = K_{N} (F).$ 
\end{proof}

This theorem implies the following interesting and simple explicit
formula for local 2-holonomy for transport 2-functors with covering
2-groups as their structure 2-groupoids. This is another one of our
main results.

\begin{cor}
\label{cor:main}
The formula for local parallel transport for any bigon under \emph{any}
smooth 2-functor $F_{N} : \mathcal{P}_{2} (X) \to \mathcal{B}
\mathcal{G}_{N}$ is given by the formula 
\be
F_{N} \left ( \xymatrix{ y & & \ar@/_1pc/[ll]_{\g}="1"
\ar@/^1pc/[ll]^{\de}="2"   \ar@{}"1";"2"|(.20){\,}="3"
\ar@{}"1";"2"|(.80){\,}="4" \ar@{=>}"3";"4"^{ \G }  x } \right ) =
\xymatrix{ \bullet & & & \ar@/_2pc/[lll]_{F (\g)}="1"
\ar@/^2pc/[lll]^{F (\de)}="2"   \ar@{}"1";"2"|(.20){\,}="3"
\ar@{}"1";"2"|(.85){\,}="4"
\ar@{=>}"3";"4"|-(0.425){ [s\mapsto F (\G ( \ \cdot \ ,s)
F (\g)^{-1} ]_{N} } \bullet  }  , 
\ee
where $F$ is the 2-functor $F_{N}$ restricted to 1-morphisms. 
\end{cor}

Finally, by Corollary 4.9 of \cite{SW1}, Theorem 2.21 of \cite{SW2}, and
Proposition 4.1.3 of \cite{SW4}, the functors $\mathcal{P}$ in each row
of
\be
\label{eq:thediagramP}
\xy 0;/r.15pc/:
(-120,-15)*+{Z^{2}_{\pi} (\mathcal{G}_{N})^{\infty}}="8";
(-120,15)*+{Z^{1}_{\pi} (G)^{\infty}}="7";
(-80,-15)*+{\mathfrak{Des}^{2}_{\pi} (i_{N} )^{\infty}}="6";
(-80,15)*+{\mathfrak{Des}^{1}_{\pi} (i )^{\infty}}="5";
{\ar^(0.45){\mathcal{P}} "7";"5" };
{\ar "7";"8"_{K_{N}^{Z}}};
{\ar "5";"6"^{K_{N}^{\mathfrak{Des}}}};
{\ar_(0.45){\mathcal{P}} "8";"6" };
\endxy
\ee
are (weak) inverses to $\mathcal{D}$ so this diagram commutes weakly.

\section{Examples and magnetic monopoles}
\label{sec:examples}
As briefly mentioned above, the path-curvature transport 2-functor is
motivated by constructions in physics. In 1931, Dirac \cite{Di} studied
the charge of a magnetic monopole in $\R^3$ and found it to be quantized
and proportional to $\int_{S^2} R,$ where $S^2$ is a sphere enclosing
the magnetic monopole and $R$ is the curvature of the $U(1)$ bundle with
connection over $\R^3 \setminus \{ * \}$ where $\{ * \} \subset \R^3$ is
the location of the monopole. Of course, the language of bundles and
connections was not around at the time, but the ingredients were there.
Because $R$ is well-defined globally, the integral $\int_{S^2} R$ is
unambiguously defined. Furthermore, it is a topological invariant in the
sense that it only depends on the homotopy class of the sphere in
$\R^3 \setminus \{ *\}.$ However, for a non-abelian principal $G$-bundle
with connection, $R$ is not globally defined so it was not clear how to
define the magnetic charge. In \cite{WY}, \cite{HS}, and \cite{GNO} the
authors define the charge of a magnetic monopole in terms of a magnetic
flux through a sphere by calculating the holonomy along a family of
loops as in Figure \ref{fig:sphere}. This defines a loop at the identity
of the group. Taking the homotopy class of this loop was the definition
of the magnetic charge in the physics literature. \cite{GNO} was closer
to defining this flux as a double-path-ordered integral, but stopped
short and used other means to analyze it. 

We want to point out here that it is not obvious that the methods
described in the literature make sense. For instance, is it necessary
to begin with the constant loop? What should this loop have anything to
do with a magnetic flux, which was defined in the abelian case to be
$\int_{S^2} R.$ Is the resulting quantity gauge invariant? What does
gauge invariance even mean? And how does one know that these concepts
are even correct? 

As we show in this section, the path-curvature transport 2-functor
introduced in the previous section describes magnetic flux in terms of
surface holonomy. Furthermore, since this magnetic flux is defined using
surface holonomy, for which we have \emph{proven} gauge covariance in
Section \ref{sec:2-holonomy} (specifically Theorem
\ref{thm:invariancesurfaceholonomy}), we can meaningfully ask if the
magnetic flux is a gauge invariant quantity. This would be the case if
it is invariant under $\a$-conjugation.  We review the interesting cases
considered in the physics literature, those of $U(1)$ monopoles, $SO(3)$
monopoles, and $SU(n)/Z(n)$ for all $n.$ We also consider the cases
$U(n)$ for all $n.$ For all of these examples, we take the subgroup
$N \le \pi_{1} (G)$ to be $N = \{ 1 \},$ the trivial subgroup of
$\pi_{1}(G).$ This case is interesting in its own right as the examples
will illustrate. 

We do this in two ways. We first start with a transport functor,
described in terms of its differential cocycle data, and use the methods
of Section \ref{sec:reconstruction2} and Section \ref{sec:2-holonomy} to
reconstruct a transport functor with group-valued holonomies. We then
construct the path-curvature 2-functor and compute surface holonomy. The
other method we use, which is equivalent by  Theorem \ref{thm:main} and
Corollary \ref{cor:main}, is to use the surface-ordered integral of
equation (\ref{eq:surfaceintegral}) from \cite{SW2} and the definition
of the differential cocycle data of the path-curvature 2-functor
discussed in Section \ref{sec:general}. This is unnecessary due to
Corollary \ref{cor:main} but we do it anyway for the reader's
convenience. In the process, we must choose weak inverses $s^{\pi} :
\mathcal{P}_{2} (M) \to \mathcal{P}^{\pi}_{2} (M)$ to the projections
$p^{\pi} : \mathcal{P}^{\pi}_{2} (M) \to \mathcal{P}_{2} (M)$ associated
to some surjective submersion $\pi : Y \to M.$ We will define the
2-functor $s^{\pi}$ for the paths and bigons of interest to us (rather
than defining it for all paths and bigons) in the case of the first
example of $U(1)$ monopoles. We then use the same 2-functor $s^{\pi}$
for all other examples. 

For the following discussions, we will be using the following
conventions depicted in Figure \ref{fig:sphericalcoordinates} for
describing coordinates on the sphere. 

\begin{figure}[h]
\centering
  \begin{picture}(0,0)
  \put(-6,70){$\hat{x}$}
  \put(3,20){$\hat{y}$}
  \put(73,148){$\hat{z}$}
  \put(80,115){$\theta$}
  \put(65,57){$\phi$}
\end{picture}
    \includegraphics[width=0.32\textwidth]{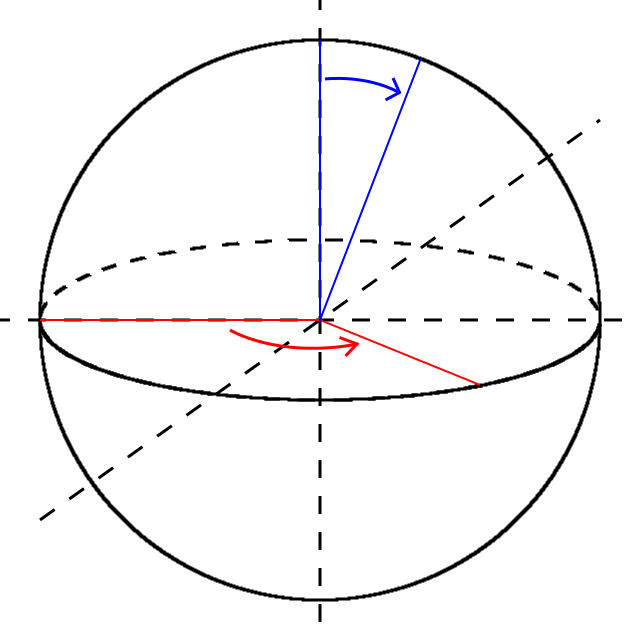}
    \vspace{-3mm}
    \caption{The azimuthal angle $\phi$ is drawn in red and extends from
the $x$ axis (pointing to the left) and goes counterclockwise in the
$xy$-plane. The zenith angle $\theta$ is drawn in blue and extends from
the $z$ axis (pointing vertically) towards the $xy$-plane.}
    \label{fig:sphericalcoordinates} 
\end{figure}
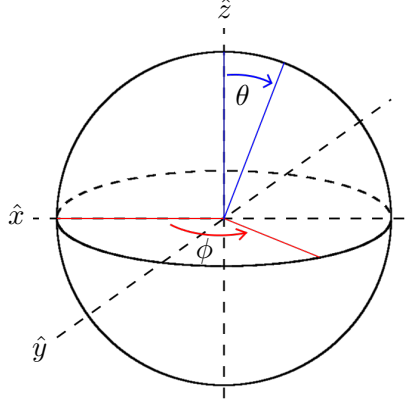

\subsection{Abelian $U(1)$ monopoles}
First, we will give an explicit example coming from abelian magnetic
monopoles. Let $P[n] \to S^2$ be the principal $U(1)$-bundle described
by the following local trivialization. Denote the northern and southern
hemispheres by $U_{N}$ and $U_{S},$ respectively. We assume that $U_{N}$
extends a little bit to the southern hemisphere so that
$U_{NS} \ne \varnothing$ (and similarly for $U_{S}$ to the northern
hemisphere). Let $Y := U_{N} \coprod U_{S}$ and $\pi : Y \to S^2$ be the
projection. Let $s_{N} : U_{N} \to Y$ and $s_{S} : U_{S} \to Y$ be the
obvious sections. Define the transition function $g_{NS} : U_{NS} \simeq
S^{1} \to U(1)$ along the equator to be 
\be
S^1 \ni \phi \mapsto g_{NS} (\phi) := e^{ i n \phi }, 
\ee
where $\phi$ is the aziumuthal angle and $n$ is an integer. Equip this
bundle with a connection $A_{N} \in \W^{1} (U_{N};\underline{U(1)})$ and
 $A_{S} \in \W^{1} (U_{S};\underline{U(1)})$ given by 
\be
A_{N} = \frac{ n }{2i} ( 1 - \cos \theta ) d \phi \aand A_{S}
= - \frac{ n }{2i} ( 1 + \cos \theta ) d \phi.
\ee
These forms satisfy the property
\be
A_{N} = g_{NS} A_{S} g_{NS}^{-1} - dg_{NS} g_{NS}^{-1} 
\ee
on $U_{NS}$ so that $g_{NS}, A_{N}, $ and $A_{S}$ are the local
differential cocycle data of a principal $U(1)$-bundle with connection.
Since $i : \mathcal{B} U(1) \to U(1)\text{-}\mathrm{Tor}$ is an
equivalence of categories, this differential cocycle data corresponds to
a global transport functor (recall (\ref{eq:allequivalences1})). 

We now consider the path-curvature 2-functor where $N = \{ 1 \} \le
\pi_{1} (S^{1}) \cong \Z$ so that the associated covering 2-group is
$(\R , U(1), \t , \a)$ with $\t : \R \to U(1)$ the  universal covering
map defined by $\phi \mapsto e^{2\pi i \phi}.$
The functor $\mathcal{P} : Z^{1}_{\pi} (G)^{\infty} \to
\mathfrak{Des}^{1}_{\pi} ( i)$ sends the differential cocycle object
$(g,A)$ to $\triv : \mathcal{P}_{1} (U_{N} \coprod U_{S} ) \to
\mathcal{B} G$ defined by the path-ordered exponential and the natural
transformation $g : \pi^{*}_{1} ( \triv_{i} ) \Rightarrow
\pi^{*}_{2} ( \triv_{i} )$ defined on components $\phi \in S^1$ by
$i ( g_{NS} ( \phi ) ).$ We partially define $s^{\pi} : \mathcal{P}_{2}
(S^2) \to \mathcal{P}^{\pi}_{2} (S^2)$ as follows. We first make the
choice  
\be
s^{\pi} ( x ) := 
\begin{cases}
s_{N}(x)  & \mbox{ if } x \in U_{N} \\
s_{S}(x) & \mbox{ if } x \in S^2 \setminus U_{N}  
\end{cases}
\ee 
for objects. We'll be a little sloppy now and define a lift of thin
paths and thin bigons on representatives of thin homotopy classes. We
only lift paths, labelled as $\g_{\theta},$ of the form depicted in
Figure  \ref{fig:hoolaloopmid}. 
\begin{figure}[h]
\centering
\begin{picture}(0,0)
\put(8,53){$\bullet$}
\put(57,20){$\g_{\q}$}
\end{picture}
    \includegraphics[width=0.25\textwidth]{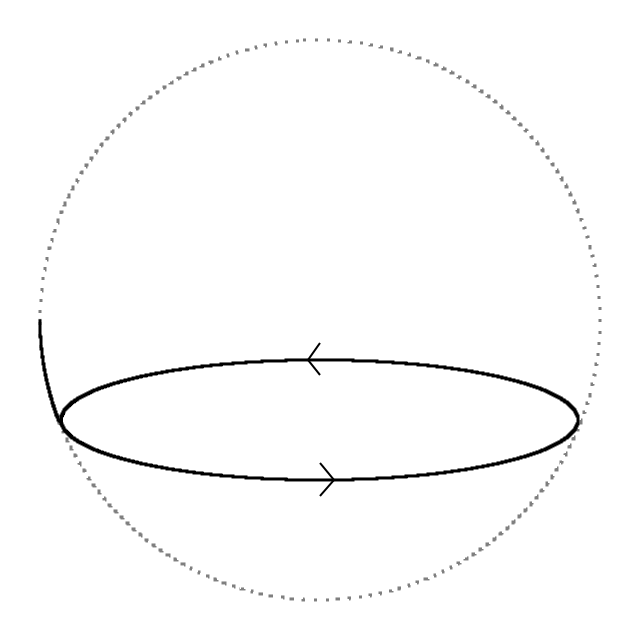}
    \vspace{-3mm}
    \caption{A loop on the sphere is made to always start at the equator
at the point $\bullet.$ In this figure, the loop is drawn for some
$\theta$ in the range $\frac{\pi}{2} < \theta < \pi.$}
    \label{fig:hoolaloopmid}
\end{figure}
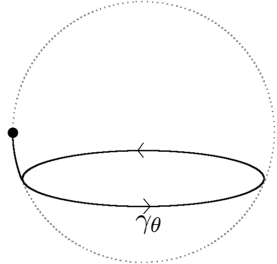 
The reason for this is because we will consider a sequence of such loops
starting at the constant loop at the point $\bullet$ on the equator (so
that $s^{\pi}(\bullet) = (\bullet, N)$) enclosing the sphere going from
$U_{N}$ to $U_{S}$ and finally ending on the constant loop at the point
$\bullet$ as depicted in Figure \ref{fig:spherelaminateflip}. 
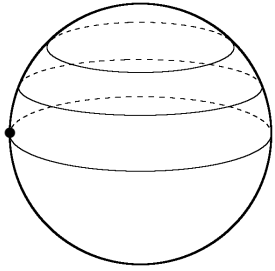
\begin{figure}[h]
\centering
\begin{picture}(0,0)
\put(8,54){$\bullet$}
\end{picture}
    \includegraphics[width=0.25\textwidth]{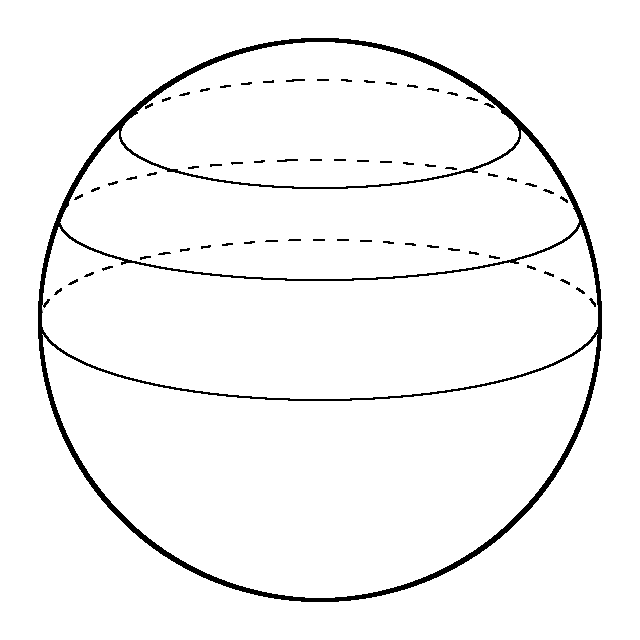}
    \vspace{-3mm}
    \caption{Loops along the $\phi$ direction on the sphere of constant
$\theta$ are drawn for $\theta = \frac{\pi}{2}$ and two intermediate
values in the range $0 < \theta < \frac{\pi}{2}.$ However, each loop is
made to start at the point $\bullet$ so that the sphere is thought of as
a bigon $S^2 : \id_{\bullet} \Rightarrow \id_{\bullet}.$}
    \label{fig:spherelaminateflip}
\end{figure} 
Therefore, we define the assignment on these loops to be 
\be
s^{\pi} ( \g_{\theta} ) 
:= 
\begin{cases}
s_{N *} (\g_{\theta}) &\mbox{ if } 0 \le \theta \le \frac{\pi}{2} \\
\a_{NS} (\bullet) * s_{S *} (\g_{\theta} ) * \a_{SN} (\bullet)
& \mbox{ if } \frac{\pi}{2} < \theta \le \pi \\
\end{cases} 
.
\ee
We now define the lift on two bigons. The first bigon $\S_{N}$ is given
by 
\be
[0,2\pi] \times [0, \pi/2] \ni (\phi, \theta)
\mapsto \S_{N} (\phi, \theta) := \g_{\theta} (\phi) 
\ee
and is a bigon $\id_{\bullet} \Rightarrow \g_{\pi/2}$ which lands in
$U_{N}$ and covers the northern hemisphere. We send this bigon to
$s^{\pi} ( \S_{N} ) := s_{N *} (\S_{N} )$ in
$\mathcal{P}_{2}^{\pi} (S^2)$ because our prescription
(\ref{eq:bigonlift}) says 
\be
s^{\pi} ( \S_{N} ) \quad := \quad 
\xy0;/r.15pc/:
(-70,0)*+{s^{\pi}(\bullet)}="left";
(30,0)*+{s^{\pi}(\bullet)}="right";
(-40,0)*+{s_{N}(\bullet) }="d";
(0,0)*+{s_{N} ( \bullet ) }="c";
{\ar@/_1.5pc/"c";"d"_{s_{N *} ( \id_{\bullet} )}};
{\ar@/^1.5pc/"c";"d"^{s_{N *} ( \g_{\pi/2} ) }};
(-20,10)*+{}="f";
(-20,-10)*+{}="g";
{\ar@{=>}"f";"g"|-{s_{N *} ( \S_{N} ) }};
{\ar@/_4pc/"right";"left"_{s^{\pi}(\id_{\bullet} )}};
{\ar@/^4pc/"right";"left"^{s^{\pi}(\g_{\pi/2})}};
{\ar"right";"c"_{\id}};
{\ar"left";"d"^{\id}};
{\ar@{=>}(-20,25);(-20,16)^{\id}};
{\ar@{=>}(-20,-17);(-20,-26)^{\id}};
\endxy
\ee
We do a similar thing for the bigon $\S_{S}$ given by 
\be
[0,2\pi] \times (\pi/2, \pi] \ni (\phi, \theta)
\mapsto \S_{S} (\phi, \theta) := \g_{\theta} (\phi) 
\ee
which is a bigon $\g_{\pi/2} \Rightarrow \id_{\bullet}$ that lands in
$U_{S}.$ This is a bigon covering the southern hemisphere. However, our
boundary data need to match up so that we'll be able to compose in
$\mathcal{P}_{2}^{\pi} (S^2).$ Again, following (\ref{eq:bigonlift})),
this is given by 
\be
s^{\pi} ( \S_{S} ) \quad := \quad 
\xy0;/r.15pc/:
(-70,0)*+{s^{\pi}(\bullet)}="left";
(30,0)*+{s^{\pi}(\bullet)}="right";
(-40,0)*+{s_{S} (\bullet) }="d";
(0,0)*+{s_{S} ( \bullet ) }="c";
{\ar@/_1.5pc/"c";"d"_{s_{S *} ( \g_{\pi/2} )}};
{\ar@/^1.5pc/"c";"d"^{s_{S *} ( \id_{\bullet} ) }};
(-20,10)*+{}="f";
(-20,-10)*+{}="g";
{\ar@{=>}"f";"g"|-{s_{S *} ( \S_{S} ) }};
{\ar@/_4pc/"right";"left"_{s^{\pi}(\g_{\pi/2} )}};
{\ar@/^4pc/"right";"left"^{s^{\pi}(\id_{\bullet})}};
{\ar"right";"c"_{\a_{SN}(\bullet)}};
{\ar"left";"d"^{\a_{SN}(\bullet)}};
{\ar@{=>}(-20,25);(-20,16)^{!}};
{\ar@{=>}(-20,-16);(-20,-25)^{!}};
\endxy
,
\ee
where the $!$ signifies the unique 2-isomorphisms from Lemma
\ref{lem:unique}. 
For the full bigon $\S : \id_{\bullet} \Rightarrow \id_{\bullet}$
depicting the full sphere as the composition $\begin{smallmatrix}
\S_{N} \\ \overset{\circ}{\S_{S}} \end{smallmatrix} : \id_{\bullet}
\Rightarrow \g_{\pi/2} \Rightarrow \id_{\bullet},$ we break it up into
the two pieces defined above and compose vertically. The result of this
is 
\be
\xy0;/r.15pc/:
(-70,0)*+{s^{\pi}(\bullet)}="left";
(70,0)*+{s^{\pi}(\bullet)}="right";
(-30,0)*+{s_{S} ( \bullet ) }="d";
(30,0)*+{s_{S} ( \bullet ) }="c";
{\ar@/_1pc/"c";"d"_{s_{S *} ( \g_{\pi/2} )}};
{\ar@/^3pc/"c";"d"^{s_{S *} ( \id_{\bullet} ) }};
(0,5)*+{}="f";
(0,-15)*+{}="g";
{\ar@{=>}"f"+(0,0.5);"g"+(0,-3)|-{s_{S *} ( \S_{S} ) }};
{\ar@/_7pc/"right";"left"_{s^{\pi}(\id_{\bullet} ) := s_{N *}
(\id_{\bullet} ) }};
{\ar@/_4pc/"right";"left"_{s^{\pi}( \g_{\pi/2} ) := s_{N *}
( \g_{\pi/2} ) }};
{\ar@/^6pc/"right";"left"^{s^{\pi}(c_{ \bullet } ) := s_{N *}
(\id_{\bullet} )}};
{\ar"right";"c"_{\a_{SN}(\bullet)}};
{\ar"left";"d"^{\a_{SN}(\bullet)}};
{\ar@{=>}(0,45);(0,33)|-{s^{\pi} ( \S_{N} ) := s_{N *} (\S_{N}) } };
{\ar@{=>}(0,-27);(0,-38.5)^{!}};
{\ar@{=>}(0,25);(0,14)^{!}};
\endxy
,
\ee
We rescale our angle $\theta$ to $s = \frac{\theta}{\pi}$ to be
consistent with our earlier notation. Going from
$Z^{2}_{\pi}(\mathcal{B}\mathcal{G}_{\{1\}})^{\infty}$ to
$\mathrm{Trans}^{2}_{\mathcal{B} \mathcal{G}_{\{1\}}}
(M,\widehat{G\text{-}\mathrm{Tor}_{ \{1\} }})$ from above to define the
global transport functor applied to the sphere, we obtain the following
diagram in $\widehat{G\text{-}\mathrm{Tor}_{ \{1\} }}$ 
\be
\label{eq:U1bigon}
\xy0;/r.15pc/:
(-70,0)*+{G}="left";
(70,0)*+{G}="right";
(-30,0)*+{G}="d";
(30,0)*+{G}="c";
{\ar@/_1pc/"c";"d"_{L_{\triv ( \g_{\pi/2} ) } }};
{\ar@/^3pc/"c";"d"^{\id_{G} }};
(0,5)*+{}="f";
(0,-15)*+{}="g";
{\ar@{=>}"f";"g"+(0,-2)|-{[ s
\mapsto L_{\triv ( \S_{S} ( \ \cdot \ , s ) ) } ] }};
{\ar@/_7pc/"right";"left"_{\id_{G}}};
{\ar@/_4pc/"right";"left"_{L_{\triv ( \g_{\pi/2} ) } }};
{\ar@/^6pc/"right";"left"^{\id_{G}}};
{\ar"right";"c"_{\id_{G}}};
{\ar"left";"d"^{\id_{G}} };
{\ar@{=>}(0,45);(0,34)|-{[ s
\mapsto L_{\triv ( \S_{N} ( \ \cdot \ , s ) ) } ] } };
{\ar@{=>}(0,-26);(0,-38)|-{\id_{\id_{G}} }};
{\ar@{=>}(0,25);(0,14)|-{\id_{L_{\triv ( \g_{\pi/2} )} }} };
\endxy
\ee
since $i_{\{1\}} \circ \triv (y) = G$ for all $y$ and so on for paths
and bigons (see the definition of $R_{(\triv, g, \psi, f)}$ in Section
\ref{sec:reconstruction2}) and $g_{NS} (\phi = 0) \equiv g_{NS}
(\bullet) = 1.$ Furthermore, $g_{NS}$ on paths is the identity since
$g_{NS}$ came from a natural transformation of ordinary functors between
ordinary categories. With these simplifications, the composition in
(\ref{eq:U1bigon}) is given by 
\be
\left [ s \mapsto 
	\begin{cases} 
	L_{\triv ( \S_{N} ( \ \cdot \ , 2s ) ) }
& \mbox{ for } 0 \le s \le \frac{1}{2} \\
	L_{\triv ( \S_{S} ( \ \cdot \ , 2s -1 ) ) }
& \mbox{ for } \frac{1}{2} \le s \le  1 
	\end{cases} 
\right ] 
\ee
which reduces to a computation on the group level. Therefore, all we
have to do is compute the homotopy class of the path 
\be
s \mapsto 
	\begin{cases} 
	\triv ( \S_{N} ( \ \cdot \ , 2s ) )  & \mbox{ for } 0 \le s \le
\frac{1}{2} \\
	\triv ( \S_{S} ( \ \cdot \ , 2s -1 ) )  & \mbox{ for }
\frac{1}{2} \le s \le  1 
	\end{cases} 
\ee
in the group $U(1)$ thanks to Lemma \ref{lem:coveringmultiplication}.
This is easily calculable
\be 
\begin{split}
\triv ( \S_{N} ( \ \cdot \ , 2s ) ) &= \triv \left ( \S_{N}
\left ( \ \cdot \ , 2 \frac{\theta}{\pi} \right )  \right ) \\
	&= e^{\frac{n}{2i}\int_{0}^{2\pi} (1-\cos\theta) d \phi } \\
	&= e^{ - i n \pi ( 1 - \cos \theta ) } 
\end{split}
\ee
since the paths going along $\theta$ do not contribute to the parallel
transport since the connection form only has a $d \phi$ contribution.
Similarly, 
\be
\triv ( \S_{S} ( \ \cdot \ , 2s -1 ) ) = \triv \left ( \S_{S}
\left ( \ \cdot \ , 2 \frac{\theta}{\pi} - 1 \right )  \right )
= e^{ i n \pi ( 1 + \cos \theta ) }  . 
\ee
As a sanity check, notice that 
\be
e^{ - i n \pi ( 1 - \cos \frac{\pi}{2} ) }  = e^{ - i n \pi}
= e^{ i n \pi } = e^{ i n \pi ( 1 + \cos \frac{\pi}{2} ) } 
\ee
showing that the matching condition (so that our path is continuous) is
satisfied. This matching condition was the one used, for instance, in
\cite{WY} (see equation (47)). 

Notice that $1 - \cos \theta$ is a monotonically increasing function of
$\theta$ for $0 \le \theta \le \frac{\pi}{2}$ starting at $0$ when
$\theta = 0$ and ending at $1$ when $\theta = \frac{\pi}{2}.$ Therefore,
$e^{ - i n \pi ( 1 - \cos \theta ) }$ winds around the circle starting
at $1$ and ending at $e^{- i n \pi } = (-1)^{n}$ winding around
monotonically $ \frac{n}{2}$ times clockwise if $n$ is positive and
counterclockwise otherwise. Now, the functon $1 + \cos \theta$ is a
monotonically decreasing function of $\theta$ for $\frac{\pi}{2} \le
\theta \le \pi$ starting at $(-1)^{n}$ when $\theta = 0$ and ending at
$1$ when $\theta = \pi.$  Therefore, $e^{ i n \pi ( 1 + \cos \theta ) }$
winds around the circle starting at $e^{i n \pi } = (-1)^{n}$ and ending
at $1$ winding around monotonically $ \frac{n}{2}$ times clockwise if
$n$ is positive and counterclockwise otherwise. In other words, the loop
goes a total of $n$ times around clockwise if $n$ is positive and $n$
times counterclockwise if $n$ is negative and the 2-holonomy along $S^2$
is given by (using the notation of Definition \ref{defn:gaugeinv2hol})
\be
\mathrm{hol}^{[n]} (S^2) = -n.
\ee
If we wanted to, we could have also computed this using differential
forms and the formula for 2-transport (\ref{eq:surfaceintegral}) of
Schreiber and Waldorf \cite{SW2} locally and pasted the group elements
together vertically as above. Of course, by the equivalence between
local smooth functors and differential forms, our formula in terms of
ordinary holonomy bypasses the rather (a-priori) complicated surface
holonomy formula (\ref{eq:surfaceintegral}) due to Corollary
\ref{cor:main}. It'll actually turn out that the surface holonomy
formula (\ref{eq:surfaceintegral}) is not so complicated in this
particular case due to our choice of bigon representing the sphere and
the differential forms representing the connection. We will subsequently
do this analysis strictly in terms of the differential forms associated
to the path-curvature 2-functor discussed in Section \ref{sec:general}. 

The curvature is given by 
\be
R_{N} = \frac{n}{2i} \sin \theta d \theta \wedge d \phi \in \W^2
(U_{N} ; \underline{U(1)})
\ee
and similarly for $R_{S} \in \W^2 (U_{S} ; \underline{U(1)}).$
Therefore, the connection 2-form is given by 
\be
B_{N} = \un \t^{-1} ( R_{N} ) = \frac{1}{2 \pi i} R_{N}
= - \frac{n}{4\pi} \sin \theta d \theta \wedge d \phi 
\ee
and similarly for $B_{S}.$ The 1-form $\mathcal{A}_{\S_{N}}$ (see
equation (\ref{eq:Asigma})) is given by 
\be
(\mathcal{A}_{\S_{N}} )_{\theta} \left ( \frac{d}{d \theta} \right )
= - \int_{0}^{2\pi} d\phi \; B_{(\theta, \phi)}
\left ( \frac{\p}{\p \theta} ,\frac{\p}{\p \phi} \right )
= \frac{n}{2} \sin \theta
\ee
and the 2-transport along $\S_{N}$ is given by 
\be
\begin{split}
k_{A,B} ( \S_{N} ) &= \mathcal{P} \exp \left \{ - \int_{0}^{\pi/2}
d\theta \; ( \mathcal{A}_{\S_{N}} )_{\theta} \left ( \frac{d}{d \theta}
\right )  \right \} \\
	&= - \int_{0}^{\pi/2} d\theta \; \frac{n}{2}  \sin \theta  
\end{split}
\ee
because the exponential map $\un \R \to \R$ is the identity. The
2-transport along $\S_{S}$ is done similarly and is given by 
\be
k_{A,B} ( \S_{S} ) = -\int_{\pi/2}^{\pi} d\theta \;
\frac{n}{2} \sin \theta .
\ee
Vertically composing these results yields 
\be
\begin{split}
k_{A,B} (\S_{S} ) + k_{A,B} (\S_{N})
&= - \int_{\pi/2}^{\pi} d\theta  \; \frac{n}{2} \sin \theta
- \int_{0}^{\pi/2} d\theta  \; \frac{n}{2} \sin \theta \\
	&= - \int_{0}^{\pi} d\theta  \; \frac{n}{2} \sin \theta   \\
	&= - n  
\end{split}
\ee
because the group operation in $\R$ is addition. Therefore, the result
obtained in terms of the path-curvature 2-functor in terms of homotopy
classes of paths in $G$ agrees with the double path-ordered exponential
formula (\ref{eq:surfaceintegral}) of Schreiber and Waldorf \cite{SW2}
from the differential cocycle data, which is what we expect due to
Corollary \ref{cor:main}.

\subsection{SO(3) monopoles}
Now we will give examples for non-abelian magnetic monopoles. The first
example will be similar to the abelian case since we will consider  the
following principal $SO(3)$ bundle over $S^2$ defined by the two open
sets $U_{N}$ and $U_{S}$ with transition function  ${g_{NS} : U_{NS}
\simeq S^1 \to SO(3)}$ to be 
\be
g_{NS} (\phi ) := e^{ - \phi J_{3} }
\ee
where 
\be
J_{1} :=  
\begin{pmatrix}
0 & 0 & 0 \\
0 & 0 & -1 \\
0 & 1 & 0  
\end{pmatrix}
, \quad 
J_{2} := 
\begin{pmatrix}
0 & 0 & 1 \\
0 & 0 & 0 \\
-1 & 0 & 0  
\end{pmatrix}
, \quad \& \quad 
J_{3} := 
\begin{pmatrix}
0 & -1 & 0 \\
1 & 0 & 0 \\
0 & 0 & 0  
\end{pmatrix}
\ee
form a set of generators for the Lie algebra $\underline{SO(3)}.$ 
One can give explicit connection forms $A_{N}$ and $A_{S}$ on $U_{N}$
and $U_{S}$ respectively as follows
\be
A_{N} :=  \frac{J_{3}}{2}  ( 1 - \cos \theta ) d \phi  \aand A_{S} := - 
\frac{J_{3}}{2} ( 1 + \cos \theta ) d \phi .
\ee
These define local curvature 2-forms $R_{N}$ and $R_{S}.$  Indeed, the
gauge transformation defined above shows that 
\be
\begin{split}
g_{NS} A_{S} g_{NS}^{-1} - dg_{NS} g_{NS}^{-1} &= A_{S} + J_{3} d\phi \\
	&= -\frac{J_{3}}{2} ( 1 + \cos \theta ) d \phi + J_{3} d \phi \\
	&=   \frac{J_{3}}{2}  ( 1 - \cos \theta ) d \phi \\
	&= A_{N}
\end{split}
\ee
because all elements commute. The curvature 2-form is given by 
\be
R_{N} = d A_{N} + \frac{1}{2} [ A_{N}, A_{N} ] = \frac{J_{3}}{2}
\sin \theta \; d\theta \wedge d \phi 
\ee
again because the elements commute. Since $R_{N} = R_{S}$ on $U_{NS},$
this defines a $\underline{SO(3)}$-valued closed 2-form on $S^2.$ 
Let $\t : SU(2) \to SO(3)$ be the double cover map so that $N = \{1 \}
\le \pi_{1} (SO (3) ) \cong \Z_{2}.$ Recall that the induced map on the
level of Lie algebras $\un \t : \underline{SU(2)} \to \underline{SO(3)}$
is an isomorphism and is given by 
\be
\un \t \left ( \frac{1}{2i} \s_{i} \right ) = J_{i},
\ee
where the $\s_{i}$ are the Pauli matrices 
\be
\s_1 =
\begin{pmatrix}
0 & 1 \\
1 & 0
\end{pmatrix}
, \qquad
\s_2 =
\begin{pmatrix}
0 & -i \\
i & 0  
\end{pmatrix}
, \aand 
\s_3 = 
\begin{pmatrix}
1 & 0 \\
0 & -1  
\end{pmatrix}.
\ee
As in the general case, define $B_{N} := \un \t^{-1} ( R_{N} ) $ and
$B_{S} := \un \t^{-1} ( R_{S} ),$ or explicitly 
\be
B = \frac{\s_{3} }{4i} \sin \theta \; d\theta \wedge d \phi 
\ee
since $B_{N} = B_{S}$ on $U_{NS}.$ By our analysis in Section
\ref{sec:general}, this defines the differential cocycle data of the
path-curvature 2-functor. We will compute the 2-holonomy in two
different ways. We will follow the same procedure as in the $U(1)$ case
and compute 2-holonomy in terms of homotopy classes of paths and then we
will use formula (\ref{eq:surfaceintegral}). 

To help us with the first task, we first recall how $SU(2)$ the way
described above in terms of the Pauli spin matrices is isomorphic to the
universal cover of $SO(3)$ described in terms of homotopy classes of
paths starting at the identity in $SO(3).$ An isomorphism
$\widetilde{SO(3)} \cong SU(2)$ from the universal cover of $SO(3)$ to
$SU(2)$ can be given by using the universal property and the fact that
$SU(2)$ is simply connected. Given any path $\g : [0,1] \to SO(3)$
starting at $\g(0) = \mathbb{I}_{3},$ the $3\times 3$ identity matrix,
there exists a unique lift $\tilde{\g} : [ 0,1] \to SU(2)$ starting at
$\tilde{\g} (0) = \mathbb{I}_{2}$ and such that the diagram 
\be
\xy0;/r.25pc/:
(10,7.5)*+{SU(2)}="SU2";
(-10,-7.5)*+{[0,1]}="01";
(10,-7.5)*+{SO(3)}="SO3";
{\ar@{-->}"01";"SU2"^{\tilde{\g}}};
{\ar"01";"SO3"_{\g}};
{\ar"SU2";"SO3"}; 
\endxy
\ee
commutes. In this way, we can define a map 
\be
\begin{split}
\widetilde{SO(3)} &\to SU(2) \\
[\g] &\mapsto \tilde{\g} (1) 
.
\end{split}
\ee
By using the universal property one more time, one can show that this
map is well-defined. Finally, it is a smooth diffeomorphism of covering
spaces. 

We can now check what the value of the path-curvature transport
2-functor is on the sphere by doing the same computations as above but
using the new $\underline{SO(3)}$-valued differential forms. The result
for the bigon describing the northern hemisphere is given by 
\be
\begin{split}
\triv ( \S_{N} ( \ \cdot \ , 2s ) ) &= \triv \left ( \S_{N}
\left ( \ \cdot \ , 2 \frac{\theta}{\pi} \right )  \right ) \\
	&= e^{\frac{J_{3}}{2} \int_{0}^{2\pi} (1-\cos\theta)d\phi } \\
	&= e^{ \pi J_{3}  ( 1 - \cos \theta ) } 
\end{split}
\ee
since the paths going along $\theta$ do not contribute to the parallel
transport since the connection form only has a $d \phi$ contribution.
The path-ordered exponential is reduced to an ordinary exponential of an
integral because only $J_{3}$ is involved and $J_{3}$ commutes with
itself. Similarly, the southern hemisphere gives
\be
\triv ( \S_{S} ( \ \cdot \ , 2s -1 ) ) = \triv \left ( \S_{S}  \left ( \
\cdot \ , 2 \frac{\theta}{\pi} - 1 \right )  \right )
= e^{ - \pi J_{3} ( 1 + \cos \theta ) } . 
\ee
Again, as a sanity check we show that the boundary values match up
between the two hemispheres along the equator: 
\be
e^{ \pi J_{3} (1-\cos\frac{\pi}{2} ) } = e^{ \pi J_{3} }
= - \mathbb{I}_{3} =  e^{ - \pi J_{3} }
= e^{ - \pi J_{3} ( 1 + \cos \frac{\pi}{2} ) } .
\ee
Now we can compute the homotopy class of the path as $\theta$ ranges
from $0$ to $\pi.$ Using similar arguments, namely that
$1 - \cos \theta$ is a monotonically increasing function of $\theta$ for
$\theta$ between $0$ and $\frac{\pi}{2},$ we see that this defines a
nontrivial loop in $SO(3)$ at the identity which agrees with our
previous calculation. Therefore, the 2-holonomy along the sphere is 
\be
\mathrm{hol} ( S^2 ) = - \mathbb{I}_{2}. 
\ee

Now we will use the differential cocycle data and integrate using
formula (\ref{eq:surfaceintegral}).  First, we compute
$\mathcal{A}_{\S_{N}}$ for the northern hemisphere bigon. Because only
$\s_{3}$ is involved in the computation, everything commutes and
conjugation is trivial. Therefore, 
\be
\left ( \mathcal{A}_{\S_{N}} \right )_{\theta} \left ( \frac{d}{d\theta}
\right ) = - \int_{0}^{2\pi} d\phi \; B_{(\theta, \phi)}
\left ( \frac{\p}{\p \theta} , \frac{\p}{\p \phi} \right )
= - \frac{ \pi \s_{3} }{ 2i }  \sin \theta 
\ee
and the 2-transport along $\S_{N}$ is given by 
\be
\begin{split}
k_{A,B} (\S_{N}) &=\mathcal{P}\exp\left\{-\int_{\theta=0}^{\theta=\pi/2}
\left ( \mathcal{A}_{\S_{N}} \right )_{\theta} \left ( \frac{d}{d\theta}
\right ) \right \} \\
	&= \exp \left \{ \int_{\theta = 0}^{\theta = \pi/2}
\frac{ \pi \s_{3} }{ 2i }  \sin \theta  \right \} .
\end{split}
\ee
The 2-transport along $\S_{S}$ is done similarly and is given by 
\be
k_{A,B} (\S_{S})=\exp\left \{\int_{\theta=\pi/2}^{\theta = \pi}
\frac{ \pi \s_{3} }{ 2i }  \sin \theta  \right \}
\ee
Vertically composing these results yields 
\be
k_{A,B} (\S_{S})  k_{A,B} (\S_{N})
= \exp \left \{ \int_{\theta = 0}^{\theta = \pi}
\frac{ \pi \s_{3} }{ 2i }  \sin \theta  \right \} = e^{ \pi i \s_{3}}
= - \mathbb{I}_{2\times 2}
\ee
because again every term commutes. We will discuss what these group
elements mean after we finish a few more examples.

\subsection{$SU(n)/Z(n)$ monopoles}

Another collection of non-abelian examples arise from the Lie group
$SU(n).$ The center of $SU(n)$ is $Z(n)$ where, in the fundamental
representation, elements in $Z(n)$ are of the form 
\be
\exp \left \{ \frac{ 2 \pi k i }{ n } \right \} \mathbb{I}_{n},
\ee
where $k \in  \{ 0, 1, \dots, n -1 \}$ and $\mathbb{I}_{n}$ is the
$n\times n$ unit matrix.
$SU(n)/Z(n)$ is a Lie group with fundamental group
$\pi_{1} ( SU(n)/Z(n) )$ isomorphic to $Z(n).$ To see this, recall that
the universal cover $\widetilde{SU(n)/Z(n)}$ constructed via paths in
$SU(n)/Z(n)$ and modding out by homotopy is naturally isomorphic to
$SU(n),$ which is simply connected, by the universal property of
universal covers. The isomorphism preserves the fibers over the identity
in $SU(n)/Z(n)$ and  restricts to the isomorphism between
$\pi_1 (SU(n)/Z(n))$ and $Z(n).$  The previous example was the special
case $n =2.$ 

The equivalence relation on $SU(n) / Z(n)$ says that two elements $A$
and $B$ of $SU(n)$ are equivalent if there exists a
$k \in \{ 0, 1, \dots, n -1 \}$ such that 
\be
A B^{-1} = \exp \left \{ \frac{2\pi k i}{n} \right \} \mathbb{I}_{n}. 
\ee
We denote the elements of equivalence classes with square brackets such
as $[A].$ 

The possible $SU(n) / Z(n)$ principal bundles over the sphere are
determined by the clutching function along the equator, which is a
homotopy class of a loop $S^{1} \to SU(n) / Z(n)$ which by the
isomorphism above is precisely an element of $Z(n).$ The quotient map is
written as $\t : SU(n) \to SU(n) / Z(n)$ and is a covering map of Lie
groups. Therefore, it defines a Lie 2-group. 

Let's first consider the case for $n=3,$ which is relevant in the theory
of quarks and gluons (see Section 1.4 of \cite{HS}). We fix
$k \in \{0,1,2 \}.$ Define $X$ to be the element in the Lie algebra of
$SU(3)$ to be 
\be
X := \frac{i}{3} \begin{pmatrix} 1 & 0 & 0 \\ 0 & 1 & 0 \\ 0 & 0 & -2
\end{pmatrix}. 
\ee
The exponential of this matrix is unitary. We define transition
functions by 
\be
g_{NS} ( \phi ) := \exp \left \{ - k \un \t ( X) \right \}
= [\exp \{ - k \phi X \} ] = \left [\begin{pmatrix}
e^{-\frac{k \phi i}{3} } & 0 & 0 \\ 0 & e^{-\frac{k \phi i}{3} } & 0 \\
0 & 0 &  e^{ \frac{2k \phi i}{3} }\end{pmatrix} \right ] . 
\ee
The element $X$ is a scalar multiple of the Gell-Mann matrix $\l_{8}.$
Note we have 
\be
g_{NS} \left ( 0 \right ) = g_{NS} \left ( 2 \pi \right )
= g_{NS} \left ( 4 \pi \right )  = [ \mathbb{I}_{3} ] \in SU(3) / Z(3). 
\ee
The transition function defines a map $\phi \mapsto g_{NS} ( \phi )$
whose homotopy class determines a principal $SU(3)/Z(3)$ bundle
characterized by the integer $k \in  \{ 0, 1, 2 \}.$ 

We define a connection on this bundle analogously to the $SO(3)$ case by
setting 
\be
A_{N} := \frac{k \un \t (X) }{2} (1-\cos \theta ) d\phi \aand
A_{S} := -  \frac{k \un \t (X) }{2} ( 1 + \cos \theta ) d \phi .
\ee
A similar computation shows that this collection of 1-forms is
consistent with the transition function. The connection 2-form is
similarly given by 
\be
B_{N} =  \frac{kX}{2} \sin \theta \; d\theta \wedge d \phi 
\ee
and likewise for $B_{S}.$ This defines an
$\underline{SU(3)/Z(3)}$-valued closed 2-form on $S^2.$ 

Again, we can do the computation for the 2-holonomy in the two ways
described earlier. The first case is done by computing the homotopy
class of the path of holonomies using the definition of the
path-curvature 2-functor of Definition \ref{defn:pathcurvature2functor}.
The second way is via the differential forms associated to the
path-curvature 2-functor described in Section \ref{sec:general} and
equation (\ref{eq:surfaceintegral}). The computation is completely
analogous to the previous two examples. 

For the first case, we have
\be 
\begin{split}
\mathrm{hol}^{k} ( S^2 ) &= \left [ \theta \mapsto 
\begin{cases} 
\left [e^{\frac{k}{2} X \int_{0}^{2\pi} (1-\cos\theta) d\phi } \right ]
 & \mbox{ if } 0 \le \theta \le \frac{\pi}{2} \\
\left [ e^{ - \frac{k}{2} X \int_{0}^{2\pi} ( 1 + \cos \theta ) d \phi }
\right ] & \mbox{ if } \frac{\pi}{2} \le \theta \le \pi
\end{cases}
\right ] \\
	&= \left [ \theta \mapsto 
\begin{cases} 
\left [ e^{ k \pi X ( 1 - \cos \theta ) } \right ]
& \mbox{ if } 0 \le \theta \le \frac{\pi}{2} \\
\left [ e^{ - k \pi  X ( 1 + \cos \theta ) } \right ]
& \mbox{ if } \frac{\pi}{2} \le \theta \le \pi
\end{cases}
\right ] \\
&= e^{ \frac{2\pi i k}{ 3} } \mathbb{I}_{3}. 
\end{split}
\ee

As for the computation in terms of differential forms, also by analogous
computations to previous cases, 
\be
\left ( \mathcal{A}_{\S_{N} } \right )_{\theta}
\left ( \frac{d }{d \theta} \right )
= - \int_{0}^{2\pi} d \phi \frac{k X}{2} \sin \theta
= - k \pi X \sin \theta
\ee
and likewise for $\left ( \mathcal{A}_{\S_{S} } \right )_{\theta}
\left ( \frac{d }{d \theta} \right ).$ Also 
\be
k_{A,B} (\S_{N} )
=\exp\left \{\int_{0}^{\pi/2} k \pi X \sin \theta \; d \theta \right \} 
\ee
and finally the 2-holonomy along the sphere is 
\be
\mathrm{hol}^{[k]} ( S^2 ) = k_{A,B} (\S_{S} ) k_{A,B} (\S_{N} )
= \exp \{ 2 \pi k X \} = e^{ \frac{ 2\pi i k }{ 3 } } \mathbb{I}_{3}. 
\ee
For the general case of $SU(n),$ by using the matrix 
\be
X := \frac{i}{n} 
\begin{pmatrix} 
1 & & & & \\
 & 1 & & & \\
 & & \ddots & & \\
 & & & 1 & \\
 & & & & 1-n
 \end{pmatrix} 
\ee
the formulas for the transition function, connection 1-forms, and
connection 2-forms are all the same with this new $X$ replacing the old
one. Completely analogous computations lead to a 2-holonomy along the
sphere given by 
\be
\mathrm{hol}^{[k]} (S^2) = e^{ \frac{ 2\pi i k }{ n } } \mathbb{I}_{n}, 
\ee
where $k \in \{ 0, 1, \dots, n-1 \}.$ The result is the magnetic charge
of a magnetic monopole computed as a non-abelian flux in $SU(n)/Z(n)$
gauge theories.

\subsection{$U(n)$ monopoles}

We now discuss yet another collection of examples generalizing the
$U(1)$ case. 
Consider the group $U(n)$ of unitary $n \times n$ matrices. The Lie
algebra, $\underline{U(n)}$ consists of Hermitian matrices. The
universal cover of $U(n)$ is $SU(n) \times \R.$ The covering map
${\t : SU(n) \times \R \to U(n)}$ is defined by $\t ( A, t )
:= A e^{2\pi i t}.$ The image of $\t$ is clearly a $U(1)$ subgroup of
$U(n).$ The fiber of this covering map is given by the kernel which is 
\be
\begin{split}
\ker \t &= \left \{ ( A, t ) \ \Big | \ 
A = e^{- 2 \pi i t} \text{ and } \det A = e^{- 2 \pi i n t} = 1 \iff t =
\frac{k}{n}, \; k \in \Z \right \}  \\
	&= \left \{ \left ( e^{ \frac{ 2 \pi i k }{n } } \mathbb{I}_{n}
, \frac{k}{n} \right ) \ \Big  | \ k \in \Z \right \} \\
	&\cong \Z . 
\end{split}
\ee
Consider the Lie algebra element along this real line 
\be
X := ( 0_{n}, 1 ), 
\ee
where $0_{n}$ is the $n\times n$ zero matrix. Then its image in
$\underline{U(n)}$ under $\un \t$ is 
\be
\un \t (X) = 2 \pi i \mathbb{I}_{n} . 
\ee
With this, for every integer $k,$  we define the transition function,
connection 1-forms, and connection 2-forms completely analogously to the
previous examples (specifically the $\R \to U(1)$ example), namely 
\be
g_{NS} (\phi ) = e^{ i k \phi } \mathbb{I}_{n} ,
\ee
\be
A_{N} = \frac{k}{2i} (1-\cos \theta) \mathbb{I}_{n} d\phi \aand
A_{S} = - \frac{k}{2i} (1 + \cos \theta ) \mathbb{I}_{n}  d \phi , 
\ee
and 
\be
B = \un \t^{-1} \left (  \frac{k}{2i} \sin \theta \; \mathbb{I}_{n} \; d
\theta \wedge d \phi \right ) = - \frac{k}{4\pi} \sin \theta \; ( 0_{n},
1 ) \; d \theta \wedge d \phi . 
\ee
In terms of the path of holonomies via the path-curvature 2-functor, the
surface holonomy is 
\be
\begin{split}
\mathrm{hol}^{[k]} ( S^2 ) &= \left [ \theta \mapsto 
\begin{cases} 
e^{ \frac{k}{2i} \mathbb{I}_{n} \int_{0}^{2\pi} (1-\cos\theta) d \phi }
& \mbox{ if } 0 \le \theta \le \frac{\pi}{2} \\
e^{ - \frac{k}{2i} \mathbb{I}_{n}  \int_{0}^{2\pi} ( 1 + \cos \theta ) d
\phi } & \mbox{ if } \frac{\pi}{2} \le \theta \le \pi
\end{cases}
\right ] \\
	&= \left [ \theta \mapsto 
\begin{cases} 
e^{ \frac{k\pi}{i} \mathbb{I}_{n}  ( 1 - \cos \theta ) }
& \mbox{ if } 0 \le \theta \le \frac{\pi}{2} \\
e^{ - \frac{k\pi}{i} \mathbb{I}_{n}  ( 1 + \cos \theta ) } & \mbox{ if }
\frac{\pi}{2} \le \theta \le \pi
\end{cases}
\right ] \\
&= - k \in \Z. 
\end{split}
\ee

If we want to compute the surface holonomy in terms of formula
(\ref{eq:surfaceintegral}), we first compute
\be
\left ( \mathcal{A}_{\S_{N}} \right )_{\theta} \left ( \frac{d}{d\theta}
\right ) = \int_{0}^{2\pi} d \phi \; \frac{k}{4\pi} \sin \theta \; (
0_{n}, 1 ) = \frac{k}{2} \sin \theta \; ( 0_{n}, 1 )
\ee
so that we get 
\be
k_{A,B} ( \S_{N} )
= \mathcal{P} \exp \left \{ - \int_{0}^{\pi/2} d \theta \;  \frac{k}{2}
\sin \theta \; ( 0_{n}, 1 ) \right \} = \left ( \mathbb{I}_{n},
- \int_{0}^{\pi/2} d \theta \;  \frac{k}{2} \sin \theta \right ) 
\ee
and the 2-holonomy along the sphere is 
\be
\begin{split}
\mathrm{hol}^{k} ( S^2 ) &= k_{A,B} ( \S_{S}  ) k_{A,B} ( \S_{N}  ) \\
	&=  \left ( \mathbb{I}_{n}, - \int_{\pi/2}^{\pi} d \theta \;
\frac{k}{2} \sin \theta \right ) \left ( \mathbb{I}_{n},
- \int_{0}^{\pi/2} d \theta \;  \frac{k}{2} \sin \theta \right ) \\
	&= \left ( \mathbb{I}_{n}, - \int_{0}^{\pi} d \theta \;
\frac{k}{2} \sin \theta \right ) \\
	&= \left ( \mathbb{I}_{n}, - k \right ). 
\end{split}
\ee

\subsection{Magnetic flux is a gauge-invariant quantity}
\label{sec:magneticflux}
In this section we state a theorem that is trivial to prove in the
formalism presented above but gives an interesting physical
interpretation. 
As mentioned earlier, the definition of the magnetic flux in the
literature \cite{HS} is given as the homotopy class of a loop of
holonomies. However, it was not known \cite{GNO} how to define it as a
surface-ordered integral except in the abelian case. The constructions
in this paper use the theory of transport 2-functors as models for
2-bundles with 2-connections to describe this loop of holonomies in
terms of a transport 2-functor. The equivalence between this description
and the definition in terms of surface holonomy is made precise. This
motivates the following definition. 

\begin{defn}
\label{defn:magflux}
Let $P \to M$ be a principal $G$-bundle with connection over $M$ and
denote the associated transport functor by $\tra.$ Let $\S : S^2 \to M$
be the map of a smooth sphere in $M.$ Let $N \le \pi_{1}(G)$ be a
subgroup, $\tilde{G}_{N} \to G$ the associated $N$-cover, $\mathcal{B}
\mathcal{G}_{N}$ the associated Lie 2-group, and $K_{N}(\tra)$ the
associated path-curvature transport 2-functor. The 2-holonomy
$\mathrm{hol}^{[K_{N}(\tra)]} ( \S )$ is the \emph{\uline{magnetic flux of
any magnetic monopole enclosed by $\S$ associated to $\tra$ and $N.$}}
\end{defn}

All the previous examples relied on choices for the open cover, paths
and bigons used to describe the sphere, and choices of lifts of paths
and bigons. It is not immediately clear that the surface holonomy
computed is independent of these choices. Theorems
\ref{thm:invariancesurfaceholonomy} and \ref{thm:main} give us two
important results, the first of which tells us the magnetic flux is
indeed independent of these choices. 

\begin{cor}
\label{cor:gaugeinvariantflux}
Under the assumptions of Definition \ref{defn:magflux}, the magnetic
flux is a gauge-invariant quantity (in terms of the notation of
Definition \ref{defn:Inv})
\be
\mathrm{hol}^{[K_{N}(\tra)]} ( \S ) \in \mathrm{Inv}(\a).
\ee
\end{cor}

\begin{proof}
Choose a marking for the thin sphere as a thin bigon $\S : \g
\Rightarrow \g$ from a thin loop to itself. Then $K_{N}(\tra) (\S) \in
\ker \t$ by the source-target matching condition (recall comment
preceding (\ref{eq:B})). By Theorem \ref{thm:invariancesurfaceholonomy},
2-holonomy along a sphere for \emph{any} gauge 2-group is well-defined
up to $\a$-conjugation. But $\a$-conjugation for
\emph{covering 2-groups} agrees with ordinary conjugation by a lift
by Lemma \ref{lem:surjectiveLiecrossedmodule}. Therefore, the
$\a$-conjugation action restricted to $G \times \ker \t$ is trivial
because $\ker \t$ is a central subgroup of $\tilde{G}_{N}$ by Lemma
\ref{lem:central}. 
\end{proof}

A corollary of this and Theorem \ref{thm:main} is the following which
relates the magnetic flux to a surface integral of the magnetic field.
This is more of a physics corollary than a math corollary. 

\begin{cor}
The magnetic flux (Definition \ref{defn:magflux}) can be computed as a
surface integral by using (\ref{eq:surfaceintegral}) locally. This
surface integral, which lands in the covering group, is the analogue of
$\int_{S^2} R$ where in electromagnetism $R$ is the electromagnetic
field strength due to the local potential $A.$ 
\end{cor}

Therefore, the surface holonomies of transport 2-functors give a
mathematically rigorous explanation for the topological quantum number
(the magnetic charge) associated to magnetic monopoles for gauge
theories with any structure/gauge group in the language of magnetic
flux. It is topological in the sense that it only depends on the
homotopy class of the sphere by Corollary \ref{cor:pcflat}. Furthermore,
it expresses this quantity as a group element in the center of the
universal cover of the gauge group. We emphasize that no Higgs field was
introduced to do these computations. This therefore gives a rigorous
mathematical result first mentioned by Goddard, Nuyts, and Olive at the
end of Section 2 of their paper \cite{GNO} by using the notion of
transport 2-functors introduced by Schreiber and Waldorf in \cite{SW4}
to describe magnetic flux generalizing the notion from the theory of
electromagnetism to non-abelian gauge theories.

\appendix 
\addcontentsline{toc}{section}{Appendices}

\section{Smooth spaces}
\label{sec:smooth}
We will briefly state important definitions and smooth structures
needed in this paper. The category of finite-dimensional manifolds is
not suitable for our purposes, nor is the category of certain
infinite-dimensional manifolds. This section reviews diffeological
spaces, which constitute one candidate for a notion of smooth spaces. 
For a review of smooth spaces that also compares several other
candidates, please refer to \cite{BH}. 

\begin{defn}
A \emph{\uline{smooth space}} is a set $X$ together with a collection of
\emph{\uline{plots}} $\{ \varphi : U \to X \},$ called its
\emph{\uline{smooth structure}}, where each $U$ is an open set in some
$\R^{n}$ ($n$ can vary) satisfying the following conditions. 
\begin{enumerate}[i)]
\item
If $\varphi : U \to X$ is a plot and $\theta : V \to U,$ where $V$ is an
open set of some $\R^{m},$ is a smooth map, then $\varphi \circ \theta :
V \to X$ is a plot.
\item
Every map $\R^{0} \to X$ is a plot.
\item
Let $\varphi : U \to X$ be a function and let $\{ U_{j} \}_{j \in I}$ be
a collection of open sets covering $U$ with $i_{j} : U_{j} \to U$
denoting the inclusion. Then if $\varphi \circ i_{j}   : U_{j} \to X$ is
a plot for all $j \in I,$ then $\varphi : U \to X$ is a plot. 
\end{enumerate}
\end{defn}

\begin{defn}
\label{defn:smoothmaps}
A function $f : X \to Y$ between two smooth spaces is \emph{\uline{smooth}}
if for every plot $\varphi : U \to X$ of $X,$ $f \circ \varphi : U \to
Y$ is a plot of $Y.$
\end{defn}

\begin{ex}
\label{ex:manifold}
Let $M$ be a smooth manifold. The \emph{\uline{manifold smooth structure}}
has as its collection of plots all infinitely differentiable functions
$\varphi : U \to M$ for various open sets $U$ in Euclidean space. $M$
with this collection of plots forms a smooth space.  With this smooth
structure, for any two manifolds $M$ and $N,$ a function $M \to N$ is
smooth if and only if it is differentiable in the usual sense. 
\end{ex}

\begin{ex}
\label{ex:subspace}
Let $A$ be a subset of a smooth space $X$ and denote the inclusion by
$i : A \hookrightarrow X.$  The \emph{\uline{subspace smooth structure}} on
$A$ has as its collection of plots all functions $\varphi : U \to A$
such that $i \circ \varphi : U \to X$ are plots of $X.$ With this smooth
structure, the inclusion $i : A \to X$ is smooth. 
\end{ex}

\begin{ex}
\label{ex:quotient}
Let $X$ be a smooth space, $\sim$ an equivalence relation on $X,$ and
$q : X \to X/_{\sim}$ the quotient map. The \emph{\uline{quotient smooth
structure}} on $X/_{\sim}$ has as its collection of plots all functions
$\varphi : U \to X/_{\sim}$ such that there exists an open cover
$\{ U_{j} \}_{j \in J} $ along with plots $\varphi_{j} : U_{j} \to X$
for $X$ such that 
\be
\xy0;/r.15pc/:
(-15,15)*+{X}="X";
(-15,-15)*+{X/_{\sim}}="Y";
(15,15)*+{U_{j}}="j";
(15,-15)*+{U}="U";
{\ar@{^{(}->} "j";"U"};
{\ar_{\varphi_{j}}"j";"X"};
{\ar^{\varphi}"U";"Y"};
{\ar"X";"Y"};
\endxy
\ee
commutes for all $j \in J.$ With this smooth structure, the quotient map
$q : X \to X/_{\sim}$ is smooth. 
\end{ex}

\begin{ex}
\label{ex:product}
Let $X$ and $Y$ be smooth spaces. The \emph{\uline{product smooth
structure}} on $X \times Y$ has as its collection of plots all functions
$\varphi : U \to X \times Y$ such that $\pi_{X} \circ\varphi : U \to X$
and $\pi_{Y} \circ \varphi : U \to Y$ are both plots of $X$ and $Y,$
respectively. Here $\pi_{X} : X \times Y \to X$ and $\pi_{Y} : X \times
Y \to Y$ are the projection maps and are smooth with respect to this
smooth structure. 
\end{ex}

\begin{ex}
\label{ex:maps}
Let $X$ and $Y$ be two smooth spaces. The \emph{\uline{mapping smooth
structure}} on the set of functions $Y^{X}$ of $X$ into $Y$ is defined
as follows. A function $\varphi : U \to Y^{X}$ is a plot if and only if
the associated function $\tilde{\varphi} : U \times X \to Y,$ defined by
$\tilde{\varphi} (u,x) := \varphi ( u ) ( x ),$ is smooth. With this
smooth structure and the smooth structure on a product, the adjunction
$Z^{X\times Y} \cong (Z^{Y})^{X}$ is an isomorphism in the category of
smooth spaces for all $X,Y,Z.$ 
\end{ex}

\section*{Index of (frequently used) notation}

\begin{longtable}{c|c|c|c}
\hline
Notation & Name/description & Location & Page \\
\hline
$G$ & a Lie group & Def \ref{defn:BG} & \pageref{defn:BG} \\
$\mathcal{B}G$ & a one-object groupoid & Def \ref{defn:BG}
& \pageref{defn:BG} \\
$\mathrm{Gr}$ & Lie groupoid/2-groupoid
& Def \ref{defn:2-space}/\ref{defn:Lie2groupoid}
& \pageref{defn:2-space}/\pageref{defn:Lie2groupoid} \\
$\GTor$ & the category of $G$-torsors & Def \ref{defn:GTor}
& \pageref{defn:GTor} \\
$PX$ & paths with sitting instants in $X$ & Def \ref{defn:pwsi}
& \pageref{defn:pwsi} \\
$BX$ & bigons in $X$ & Def \ref{defn:thinhomotopy}
& \pageref{defn:thinhomotopy} \\
$P^{1}X$ & smooth space of
thin paths in $X$ & Def \ref{defn:thinhomotopy}
&\pageref{defn:thinhomotopy} \\
$\mathcal{P}_{1}(X)$ & thin path groupoid of $X$
& Def \ref{defn:thinpathgroupoid} &\pageref{defn:thinpathgroupoid} \\
$L_{g}$ & left multiplication by $g$ & Eqn (\ref{eq:Lg})
& \pageref{eq:Lg} \\
$T$ & ``target'' category/2-category & Def \ref{defn:trivialization1}
& \pageref{defn:trivialization1} \\
$i : \mathrm{Gr} \to T$ & realization of structure
groupoid in $T$ & Def \ref{defn:trivialization1}
& \pageref{defn:trivialization1} \\
$\pi : Y \to M$ & a surjective submersion
& Def \ref{defn:trivialization1} & \pageref{defn:trivialization1} \\
$\triv$ & local trivialization functor & Def \ref{defn:trivialization1}
& \pageref{defn:trivialization1} \\
$\triv_{i}$ & $\triv_{i} := i \circ \triv$
& Def \ref{defn:trivialization1} & \pageref{defn:trivialization1} \\
$\mathrm{Triv}^{1}_{\pi} (i)$ & category of
$\pi$-local $i$-trivializations
& Def \ref{defn:trivialization1}, \ref{defn:trivialization1mor}
& \pageref{defn:trivialization1}\\
$Y^{[n]}$ & $n$-fold fiber product 
of $\pi : Y \to M$ & Eqn (\ref{eq:Yn}) & \pageref{eq:Yn} \\
$\mathfrak{Des}^{1}_{\pi}(i)$ & descent category
& Def \ref{defn:descent1}, \ref{defn:descent1mor}
& \pageref{defn:descent1} \\
$\mathrm{Ex}^{1}_{\pi}$ & extraction functor
& After Def \ref{defn:descent1mor} & \pageref{defn:descent1mor} \\
$\mathfrak{Des}^{1}_{\pi}(i)^{\infty}$ & smooth descent category
& After Def \ref{defn:smoothdescent1} & \pageref{defn:smoothdescent1} \\
$\mathrm{Triv}^{1}_{\pi} (i)^{\infty}$
& category of smooth $\pi$-local $i$-trivializations
& After Def \ref{defn:smoothtrivialization1}
& \pageref{defn:smoothtrivialization1} \\
$\mathrm{Trans}^{1}_{\mathrm{Gr}}(M,T)$
& category of transport functors
& After Def \ref{defn:transportfunctor1}
& \pageref{defn:transportfunctor1} \\
$\mathcal{P}^{\pi}_{1}(M)$ & \v Cech path groupoid of $M$
& Def \ref{defn:Cechpath1} & \pageref{defn:Cechpath1} \\
$p^{\pi}$ & canonical projection
& Eqn (\ref{eq:pathlifting})/Lem \ref{lem:unique}
& \pageref{eq:pathlifting}/\pageref{lem:unique} \\
$s^{\pi}$ & weak inverse to $p^{\pi}$
& Eqn (\ref{eq:pathlifting})/Lem \ref{lem:unique}
& \pageref{eq:pathlifting}/\pageref{lem:unique} \\
$\mathrm{Rec}^{1}_{\pi}$ & Reconstruction functor
& Eqn (\ref{eq:Rec1pi}) & \pageref{eq:Rec1pi} \\
$\un G$ & Lie algebra of $G$ & Sec \ref{sec:diffcocycledata}
& \pageref{sec:diffcocycledata} \\
$k_{A}$ & path transport & Eqn (\ref{eq:kA}) & \pageref{eq:kA} \\
$\mathcal{P} \exp$ & path-ordered exponential
& Eqn (\ref{eq:pathexp}) & \pageref{eq:pathexp} \\
$Z_{\pi}^{1} (G)^{\infty}$
& \begin{tabular}{c}category of differential \\
cocycles subordinate to $\pi$\end{tabular}
& Def \ref{defn:diffcocyclespi} & \pageref{defn:diffcocyclespi} \\
$\mathrm{Rec}^{1}$ \& $\mathrm{Ex}^{1}$
& limit of $\mathrm{Rec}^{1}_{\pi}$ \& $\mathrm{Ex}^{1}_{\pi}$ over
$\pi$ & Eqn (\ref{eq:allequivalences1})
& \pageref{eq:allequivalences1} \\
$v$ \& $c$ & forgets trivialization 
\& its weak inverse & Eqn (\ref{eq:allequivalences1})
& \pageref{eq:allequivalences1} \\
$\scripty{t}$ & group-valued
transport extraction
& Def \ref{defn:groupvaluedholonomyfunctor}/%
\ref{defn:groupvalued2holonomyfunctor}
& \pageref{defn:groupvaluedholonomyfunctor}/%
\pageref{defn:groupvalued2holonomyfunctor} \\
$\mathfrak{L}^{1} M$ & thin marked loop space of $M$
& Eqn (\ref{eq:thinmarkedloopspace})
& \pageref{eq:thinmarkedloopspace} \\
$\mathrm{hol}_{\scripty{t}}^{F}$
& $\scripty{t}\;$-holonomy of 
a transport functor $F$
& Def \ref{defn:tcsholonomyF}/\ref{defn:t2holonomy}
& \pageref{defn:tcsholonomyF}/\pageref{defn:t2holonomy} \\
$\mathfrak{m}$ & thin loop/sphere markings
& Def \ref{defn:loopmarking}/\ref{defn:spheremarking}
& \pageref{defn:loopmarking}/\pageref{defn:spheremarking} \\
$G / \mathrm{Inn}(G)$ & conjugacy classes in $G$
& Before Thm \ref{thm:gaugeinvarianthol}
& \pageref{thm:gaugeinvarianthol} \\
$\mathrm{hol}^{[F]}$ & gauge-invariant
holonomy/2-holonomy
& Def \ref{defn:gaugeinvhol}/\ref{defn:gaugeinv2hol}
& \pageref{defn:gaugeinvhol}/\pageref{defn:gaugeinv2hol} \\
$(H, G , \t, \a)$ & crossed module & Def \ref{defn:crossedmodule}
& \pageref{defn:crossedmodule} \\
$\mathcal{P}_{2} (X)$ & path 2-groupoid of $X$
& Def \ref{defn:path2groupoid} & \pageref{defn:path2groupoid} \\
$P^{2} X$
& \begin{tabular}{c}smooth space of thin \\bigons in $X$\end{tabular}
& Def \ref{defn:path2groupoid} & \pageref{defn:path2groupoid} \\
$\mathrm{Triv}^{2}_{\pi} (i)$ & 2-category of
$\pi$-local $i$-trivializations
& After Def \ref{defn:trivialization2}
& \pageref{defn:trivialization2} \\
$\mathfrak{Des}^{2}_{\pi} (i)$ & descent 2-category
& Def \ref{defn:descent2}--\ref{defn:descent2mor2}
& \pageref{defn:descent2} \\
$\mathrm{Ex}^{2}_{\pi}$ & extraction 2-functor
& After Def \ref{defn:descent2mor2} & \pageref{defn:descent2mor2} \\
$\mathfrak{Des}^{2}_{\pi} (i)^{\infty}$ & smooth descent 2-category
& Def \ref{defn:smoothdescent2} & \pageref{defn:smoothdescent2} \\
$\mathrm{Triv}^{2}_{\pi} (i)^{\infty}$ & \begin{tabular}{c}2-category of
smooth \\$\pi$-local $i$-trivializations\end{tabular}
& Def \ref{defn:smoothpilocalitrivialization2}
& \pageref{defn:smoothpilocalitrivialization2} \\
$\mathrm{Trans}^{2}_{\mathrm{Gr}} (M,T)$
& 2-category of transport 2-functors
& Def \ref{defn:transport2functor} & \pageref{defn:transport2functor} \\
$\mathcal{P}_{2}^{\pi} (M)$ & \begin{tabular}{c}\v Cech path 2-groupoid
of $M$ \\ subordinate to $\pi : Y \to M$ \end{tabular}
& Def \ref{defn:Cechpath2groupoid} & \pageref{defn:Cechpath2groupoid} \\
$\mathrm{Rec}^{2}_{\pi}$ & Reconstruction 2-functor
& Eqn (\ref{eq:rec2functor}) & \pageref{eq:rec2functor} \\
$(\un H, \un G, \un \t , \un \a)$ & differential Lie crossed module
& Sec \ref{sec:2cocycle} & \pageref{sec:2cocycle} \\
$\mathcal{A}_{\S}$ & \begin{tabular}{c}$\un H$-valued 1-form used \\
for surface transport \end{tabular} & Eqn (\ref{eq:Asigma})
& \pageref{eq:Asigma} \\
$k_{A,B}$ & surface transport & Eqn (\ref{eq:surfaceintegral})
& \pageref{eq:surfaceintegral} \\
$Z_{\pi}^{2} (\mathfrak{G})^{\infty}$
& \begin{tabular}{c}2-category of differential \\
cocycles subordinate to $\pi$\end{tabular}
& Def \ref{defn:diffcocyclespi2} & \pageref{defn:diffcocyclespi2} \\
$\mathfrak{S}^2 M$ & thin marked sphere space
& Def \ref{defn:thinmarkedsphere} & \pageref{defn:thinmarkedsphere} \\
$S^{2} M$ & thin free sphere space & Def \ref{defn:thinsphere}
& \pageref{defn:thinsphere} \\
$H / \a$ & $\a$-conjugacy classes in $H$ & Def \ref{defn:alphaconj}
& \pageref{defn:alphaconj} \\
$\mathrm{Inv}(\a)$ & $\a$-fixed points & Def \ref{defn:Inv}
& \pageref{defn:Inv} \\
$\tilde{G}$ & universal cover of $G$ & Eqn (\ref{eq:universalcoverG})
& \pageref{eq:universalcoverG} \\
$\tilde{G}_{N}$ & $N$-cover of $G$ & Eqn (\ref{eq:NcoverG})
& \pageref{eq:NcoverG} \\
$\mathcal{G}_{N}$ & $N$-cover 2-group & Def \ref{defn:covering2groups}
& \pageref{defn:covering2groups} \\
$\widehat{G\text{-}\mathrm{Tor}_{N}}$ & modified $\GTor$
& Def \ref{defn:NGTor} & \pageref{defn:NGTor} \\
$K_{N} (\tra)$ & path-curvature 2-functor
& Def \ref{defn:pathcurvature2functor}
& \pageref{defn:pathcurvature2functor} \\
$i_{N}$ & structure map for $K_{N}(\tra)$ & Eqn (\ref{eq:iNimage})
& \pageref{eq:iNimage} \\
$\triv_{N}$ & trivialization data for $K_{N}(\tra)$
& Eqn (\ref{eq:triv}) & \pageref{eq:triv} \\
$\Pi_{2} (M)$ & fundamental 2-groupoid of $M$
& Def \ref{defn:fundamental2groupoid}
& \pageref{defn:fundamental2groupoid} 
\end{longtable}


\end{document}